\documentclass[pdflatex,sn-mathphys-num]{sn-jnl}

\usepackage{array}
\usepackage{caption}
\usepackage{dcolumn}
\usepackage{graphicx}%
\usepackage{multirow}%
\usepackage{amsmath,amssymb,amsfonts}%
\usepackage{amsthm}%
\usepackage{mathrsfs}%
\usepackage[title]{appendix}%
\usepackage{textcomp}%
\usepackage{manyfoot}%
\usepackage{booktabs}%
\usepackage{algorithmicx}%
\usepackage{algpseudocode}%
\usepackage{listings}%
\newcommand{\toolshort}{JR}
\newcommand{\tool}{Java Ranger}
\usepackage[ruled,linesnumbered]{algorithm2e}
\usepackage{etoolbox}
\usepackage{textcomp}

\newcommand{\soha}[1]{%
  \iftoggle{majorcomment}{%
    \textcolor{red}{#1}%
  }{%
    \textcolor{black}{#1}%
  }%
}

\newtoggle{majorcomment}
\newtoggle{journalreport}
\newtoggle{technicalreport}
\toggletrue{technicalreport}
\togglefalse{journalreport}
\togglefalse{majorcomment}

\newcommand{\techreport}[1]{%
  \iftoggle{technicalreport}{#1}{}%
}
\newcommand{\journalreport}[1]{%
  \iftoggle{journalreport}{#1}{}%
}

\usepackage{syntax}
\usepackage{calc}
\usepackage{url}
\usepackage{color}
\usepackage{graphicx}
\usepackage{mathtools}
\usepackage{float}
\usepackage{soul}
\usepackage{multirow}
\usepackage{adjustbox}
\usepackage{listings}
\usepackage{cleveref}
\usepackage[english]{babel}
\usepackage{longtable}
\usepackage{amsmath,amssymb,amsfonts}
\usepackage{amsthm}
\newtheorem{thm}{Theorem}[section]
\newtheorem{lem}[thm]{Lemma}
\newtheorem{cor}[thm]{Corollary}

\newtheorem{prop}[thm]{Property}
\usepackage{graphicx}
\usepackage{textcomp}
\newcolumntype{d}[1]{D..{#1}}
\newcommand{\tw}[1]{\texttt{#1}}

\usepackage{amssymb,proof}
\usepackage{inconsolata}
\usepackage{tcolorbox}
\tcbuselibrary{listings, breakable}

\lstdefinestyle{mystyle}{
    language=Java,
    basicstyle=\footnotesize\ttfamily,
    numbers=left,
    numberstyle=\footnotesize\color{gray},
    stepnumber=1,
    numbersep=0pt,
    backgroundcolor=\color{white},
    showspaces=false,
    showstringspaces=false,
    showtabs=false,
    tabsize=2,
    captionpos=b,
    breaklines=false,
    breakatwhitespace=false,
    keywordstyle=\color{blue},
    commentstyle=\color{green!50!black},
    stringstyle=\color{red},
    keepspaces=true,
    escapeinside={(*@}{@*)}
}
\newcommand{\noop}[2]{noOp_{#1}(#2)}
\usepackage{xcolor}

\usepackage{cancel}
\newcommand\hcancel[2][black]{\setbox0=\hbox{$#2$}%
\rlap{\raisebox{.45\ht0}{\textcolor{#1}{\rule{\wd0}{1pt}}}}#2}

\newcommand{\rn}[1]{\scriptscriptstyle\mathtt{#1}}

\newcommand{\return}[1]{return\; #1}

\newcolumntype{C}[1]{>{\arraybackslash}p{#1}}

\newcommand{\concestep}[4]{(#1,#2)\ \leadsto_{c}{#4}\ #3}
\newcommand{\concsstep}[4]{(#1,#2)\ \Rightarrow_{c}{#4}\ #3}

\newcommand{\vdashconcestep}[3]{#1 \vdash #2 \ \leadsto_{c} #3}

\newcommand{\rangerstep}[5]{(#1,#2)\ \longrightarrow_{#5}(#3,#4) }

\newcommand{\rangersteplhsprime}[4]{(#1,#2)\ \longrightarrow_{#4}{#3}}
\newcommand{\rangerstepwrapper}[4]{#1\ \longrightarrow_{#4}(#2,#3) }
\newcommand{\rangersteprhs}[4]{#1\ \longrightarrow_{#4}(#2,#3) }
\newcommand{\rangerstepnolookup}[3]{#1\ \longrightarrow_{#3}#2 } 
\newcommand{\rangerexpstep}[5]{(#1,#2)\ \mapsto_{#5}(#3,#4) }
\newcommand{\rangerexpstepnostate}[3]{#1\ \mapsto_{#3} #2 }

\newcommand{\func}[1]{\textit{#1}}

\newcommand{\regiondef}[2]{(\overrightarrow{#1}.#2)}

\newcommand{\lookup}[2]{#1(#2)}

\newcommand{\phiexp}[3]{\phi(#1,#2,#3)}

\newcommand{\myif}[3]{if \;#1\;then\;#2\;else\;#3}
\newcommand*\unaryOp{\mathbin{\vcenter{\hbox{\rule{.9ex}{.9ex}}}}}
\newcommand{\binaryOp}{\diamond}
\newcommand{\gammaexp}[3]{\gamma(#1,#2,#3)}
\newcommand{\binaryop}[2]{#1 \binaryOp #2}
\newcommand{\unaryop}[1]{\unaryOp #1}

\newcommand{\putfield}[3]{put\_field(#1,#2,#3)}

\newcommand{\getfield}[3]{get\_field(#1,#2,#3)}

\newcommand{\invoke}[4]{invoke(#1,#2,#3,\overrightarrow{#4})}

\newcommand{\newobj}[2]{new(#1,#2)}

\newcommand{\update}[3]{#1[#2\xleftarrow[]{}#3]}

\theoremstyle{thmstyleone}%
\theoremstyle{thmstyletwo}%

\theoremstyle{thmstylethree}%

\begin{document}

\title{A Formal Semantics for Java Symbolic Evaluation using Large-Block Encoding}


\author*[1,2]{\fnm{Soha} \sur{Hussein}}\email{soha.hussein@cis.asu.edu.eg}

\author[3]{\fnm{Stephen} \sur{McCamant}}\email{mccamant@cs.umn.edu}

\author[4]{\fnm{Kelton} \sur{OBrien}}\email{kelpobri@iu.edu}

\author[3]{\fnm{Kuen-Bang} \sur{{Hou(Favonia)}}}\email{kbh@umn.edu}

\author[5]{\fnm{Michael} \sur{Whalen}}\email{mww@amazon.com}

\author[6]{\fnm{Vaibhav} \sur{Sharma}}\email{vsharma@zoox.com}

\affil*[1]{\orgdiv{Faculty of Computer and Information Sciences}, \orgname{Ain Shams University}, \orgaddress{\country{Egypt}}}

\affil*[2]{\orgdiv{
Faculty of Computing and Information Sciences}, \orgname{Egypt University of Informatics}, \orgaddress{\country{Egypt}}}

\affil[3]{\orgdiv{Faculty of Computer Science and Engineering}, \orgname{University of Minnesota}, \orgaddress{\country{USA}}}
\affil[4]{\orgdiv{Luddy School of Informatics, Computing, and Engineering}, \orgname{Indiana University}, \orgaddress{\country{USA}}}

\affil[5]{\orgname{Amazon Web Services}, \orgaddress{\country{USA}}}

\affil[6]{\orgname{Zoox, Inc.}, \orgaddress{\country{USA}}}


\abstract{

Symbolic execution plays a critical role in software reliability, as they are used 
to find bugs, generate test cases, and provide correctness guarantees, particularly for safety-critical systems. 
Yet their own 
correctness is rarely subject to formal scrutiny, as it is typically 
established empirically by evaluating tool behavior across many programs. 
This leaves open the possibility that the tools themselves introduce 
unsoundness, potentially invalidating the verification results they produce 
and undermining the very guarantees they are meant to provide. 

In this paper, we address this 
gap by providing the formal treatment of symbolic execution with 
path-merging, an optimization that improves path explosion by summarizing 
branching code regions into disjunctive constraints rather than exploring each 
path independently. Specifically, we target Java Ranger, a path-merging tool 
for Java programs that progressively transforms imperative Java code toward 
the language of formal logic through a series of code transformations. 
We 
formalize each of these transformations and prove their soundness with respect 
to a simplified version of the Java concrete semantics, establishing that Java 
Ranger's path-merging process preserves program semantics.


}




\maketitle

\section{Introduction}\label{sec1}

Symbolic execution (SE)~\cite{symbolicsurvey} is a formal testing and verification technique that ensures reliability and safety of programs, particularly for critical software where correctness guarantees are essential. Rather than executing programs on concrete inputs, SE treats inputs as symbolic values and explores program paths by collecting path conditions; i.e., logical constraints that characterize the inputs that reach each path. This generality makes SE applicable to a wide range of tasks including test input generation~\cite{dart,cute,testjr}, program and specification repair~\cite{syminfer,contractdr}, equivalence checking~\cite{ramos,adaptorsynth}, vulnerability finding~\cite{driller,angr}, invariant discovery~\cite{feedbackinvariantdiscovery}, and protocol correctness checking~\cite{transport}.


%
Despite the widespread use of symbolic execution for assuring critical software, surprisingly little attention has been paid to formally establishing the correctness of the tools themselves. This creates a peculiar situation: we rely on symbolic execution to verify programs, yet the verifiers remain unverified. In practice, confidence in these tools is built empirically by running them on many programs and observing their behavior over time. While this approach can surface errors that manifest as visible failures (crashes or obviously incorrect results), it provides no formal guarantees about the tool's correctness.

The benchmarking practices in the field reflect this empirical approach. The Software Verification Competition (SV-COMP)~\cite{svcomp2026} provides a large suite of programs with known verdicts against which tools are tested. Notably, the competition's scoring scheme explicitly anticipates unsound behavior: a tool is penalized 16 points for incorrectly reporting a safe program as unsafe (false positive), and 32 points for incorrectly reporting an unsafe program as safe (false negative)~\cite{svcomp2026rules}. The asymmetric penalties acknowledge that claiming safety incorrectly is more dangerous than spurious warnings, but the very existence of these penalties assumes that tools will produce wrong answers.
This assumption is borne out in practice. In SV-COMP 2026, among 7 participating tools for Java verification, 4 tools produced a combined total of 21 incorrect verdicts: GDart (2), JBMC (14), jLiSA (1), and MLB (4). 

It is crucial to distinguish between two fundamentally different sources of incorrect verdicts. First, bounded search (depth limits, timeouts) may prevent complete exploration. A "safe" verdict could be sound for all explored paths but incomplete for the full program. Second, tools may contain flaws in their core reasoning: bugs in modeling language semantics, program translation, or constraint solving. Here, even explored paths are analyzed incorrectly. Without formal semantics, these failure modes are indistinguishable, if at all observable. Users cannot tell whether "safe" means the bug was unreached or incorrectly analyzed. Formal semantics and soundness proofs address the second problem (the focus of our paper): ensuring that whatever a tool analyzes, it analyzes correctly.

Ironically, even establishing ground truth for verification benchmarks is challenging due to a circular dependency: benchmarks rely on tools to reveal bugs, yet we need correct verdicts to evaluate whether tools work correctly. This is reflected in practice: verification tasks are often added to SV-COMP with provisional verdicts, which are later revised when tools discover previously undetected bugs.
Even tools with strong empirical track records can harbor subtle soundness bugs that remain undetected for years. Java-Ranger won Gold Medals at SV-COMP in 2020, 2021, and 2025, yet during this work, we discovered a bug that had existed since its inception in 2020. 
\soha{The bug was a subtle soundness flaw in the Field-SSA transformation that typically remains latent under standard coding patterns, surfacing only in specific nested-aliasing configurations: Java Ranger failed to preserve sequential execution order when resolving field-manipulation statements. More precisely, during fixpoint evaluation, it could process later field statements before earlier ones that were still awaiting reference resolution, leading to out-of-order chaining of the created SSA variables and producing incorrect symbolic summaries.}
This bug was never caught by test suites or competition evaluations. It only surfaced when we attempted to prove soundness of the transformation and the proof failed. This exemplifies what formal verification of tools provides: bugs that are invisible to empirical testing become provably impossible when the tool's correctness is established formally.

In this paper, we address this gap by targeting the most semantically complex 
component of symbolic execution tools, where the risk of unsoundness is 
highest. Formally verifying a complete symbolic execution engine end-to-end 
is an ambitious undertaking that remains an open challenge, but a natural 
first step is to focus on the component whose incorrectness would have the 
most severe consequences. We therefore target Java Ranger (JR)~\cite{javaranger,svcomp2026}, a symbolic 
execution tool with path-merging for Java programs with many applications~\cite{jrfuzz,contractdr,testjr}. JR operates by identifying 
regions of code where branching is about to occur during dynamic exploration, 
and attempts to collapse each such region into a disjunctive predicate that 
collectively describes the behavior of all paths through it. This summarization 
is achieved through a sequence of systematic \textit{transformations}, each of 
which abstracts away one language feature, progressively moving the 
representation of the region from imperative Java toward the language of 
formal logic. A bug in any one of these transformations could silently corrupt 
the summary, causing JR to report verdicts that do not reflect the behavior of 
the original program. 
\soha{We therefore formalize each of JR's transformations over a simplified subset of Java features operating on an intermediate representation (IR). Assuming soundness of the underlying symbolic execution engine, we establish that the produced summaries faithfully preserve concrete program semantics, provided that the summarization fixpoint loop terminates.}
Although the formalization we provide in this paper is not machine-checked, we view this as a foundational step toward a broader formal treatment of symbolic execution tools.


The remainder of the paper is organized as follows. Section~\ref{sec:overview} presents an overview of Java Ranger and its transformations. Section~\ref{sec:related-work} presents related work and Section~\ref{sec:concrete} presents the concrete semantics of the bytecode-like language that we use for proving soundness of JR transformations. Section~\ref{sec:DSEEnv} presents the underlying assumptions and structure of traditional symbolic execution and Section~\ref{sec:JREnv} presents the initial state of Java Ranger and its environment. The subsequent sections (Section~\ref{sec:gamma}-Section~\ref{sec:constpropagation}) formalize and prove the soundness of each transformation in turn: $\gamma$-creation (Section~\ref{sec:gamma}), which identifies the region of interest for summarization; early-return elimination (Section~\ref{sec:earlyreturns} and Section~\ref{sec:finalreturn}), which normalizes control flow by removing intermediate returns; variable renaming (Section~\ref{sec:renaming}), which ensures variable names are unique across paths; variable substitution (Section~\ref{sec:substitution}), which propagates variable definitions forward through the region; method inlining (Section~\ref{sec:inlining}), which eliminates method invocations by substituting their definitions; and field SSA creation (Section~\ref{sec:field}), which converts field operations into SSA form while accounting for possible side effects of method invocations. Finally, Section~\ref{sec:fixpoint} brings the transformations together, presents the Java Ranger fixpoint algorithm, and proves the soundness of the overall path-merging process\soha{, provided that the fixpoint summarization loop terminates and assuming soundness of the underlying symbolic execution engine.}
\section{Java Ranger Overview}
\label{sec:overview}

Java Ranger is a path-merging symbolic execution for Java programs that operates by identifying regions of code where branching is about to occur during dynamic exploration, and attempts to collapse each such region into a disjunctive predicate that collectively describes the behavior of all paths through it. This summarization is achieved through a sequence of systematic \textit{transformations}, each of which abstracts away one language feature, progressively moving the representation of the region from imperative Java toward the language of formal logic.

The transformation pipeline proceeds as follows. The \textit{higher-order} transformation eliminates method invocations by inlining their definitions into the region. The \textit{field} transformation eliminates heap references by introducing Static Single Assignment (SSA) local variables that reflect the value or computation associated with each reference. Finally, \textit{single-path-cases} identifies constructs that resist summarization, such as heap allocations and exceptions, and marks them as \textit{exit points} to be explored individually. Together, these transformations produce a disjunctive constraint that Java Ranger submits to the underlying constraint solver in place of enumerating each path individually. %
At any point during the summarization, if any of the transformations fails, then JR reverts to using traditional SE exploring each path separately.

For example, consider the code in Listing~\ref{lst:example}. The code takes in an integer and it separates the number of set bits from the remaining bits. For example, for $i=20$ the output is $11000000$, assuming that $i$ is an 8-bit integer. To do that, the program runs in a loop while $i$ is not equal zero. In each iteration, it checks if the rightmost bit is one {\bf line 4}. If yes, then the program performs an unsigned right shift, sets the leftmost bit on {\bf (line 10)}, and removes the rightmost bit in $i$ {\bf (line 11)}. Otherwise, the program tries to find the number of rightmost trailing zeros, and then shifts $i$ by that number {\bf (lines 7, 9)}.  


\begin{tcolorbox}[
    listing only,
    listing options={style=mystyle},
    colback=white,
    colframe=black!30,
    top=2pt, bottom=2pt
]
\begin{lstlisting}[caption={Example: separates
the number of set bits from the remaining bits},label={lst:example}, captionpos=b]
public int separateBits(int i) {
 int j = 9;
 while (i != 0) {
   int hasZeroTrail = (i & 1);
   if (hasZeroTrail == 0) {
     int numOfZeros = Integer.numTrailZero(i);
     i = i >> numOfZeros;
   } else {
     j = j >>> 1;
     j = j * Integer.reverse(1);
     i = i >> 1;
   } }
 return j; }
\end{lstlisting}
\end{tcolorbox}

To analyze the program with symbolic execution, at each iteration of the while loop operating over a single bit, there are two possibilities that the program needs to analyze. This results in a number of paths that are exponential in the number of bits in $i$, i.e., $2^{32}$ paths for an unrestricted 32-bit $i$. Using \toolshort\; one can explore the same program by collapsing branches inside the while loop {\bf (lines 5-12)}, as well as those that exist in \textit{numTrailZero}, resulting in exploring a maximum of 16 paths for the same input.

\toolshort\ operates in two phases, a static phase, and a dynamic phase. The static phase analyzes the code to be executed and identifies potential regions of code that could be merged. The dynamic phase attempts to summarize a specific code region into a disjunctive constraint. It does so by translating the Java features through a sequence of transformations. We now talk in detail about these two phases.

\subsection{Static Phase}
 During the static phase, the region\rq s to-be-merged code is first recovered into a statement in \toolshort's intermediate representation (IR), and processed through a series of static transformations, i.e., transformations that are not dependent on the dynamic values/objects. 
 In general, Java Ranger starts by creating the Control Flow Graph (CFG) of the program\rq s bytecode. It then scans the graph to find basic blocks with conditional instructions, i.e.,  a basic block with two successors. If one is found, it translates the basic blocks between the conditional instruction and its immediate post dominator as an {\tt if-statement} in \toolshort\rq s IR. 

 \toolshort\ avoids recovering statements with loops. This is done by disregarding a subgraph whose post-dominator that is also a predecessor of a node. While \toolshort\ will disregard recovering this subgraph, it will attempt to recover other subgraphs within the identified loop. In other words, \toolshort\ does not attempt to summarize loops but can summarize statements within a loop if no inner loop exists. 

 For example, the recovered \toolshort\; IR-statement for the example in Listing~\ref{lst:example} is:

 \begin{tcolorbox}[
    listing only,
    listing options={style=mystyle},
    colback=white,
    colframe=black!30,
    top=2pt, bottom=2pt
]
\begin{lstlisting}[caption={Initial IR statement for $seperateBits$}, label={lst:recoveredseperatebits}, captionpos=b]
 if (w7 != 0) {
   w11 := w19 >>> 1; 
   w13 := invoke(w13,Integer,reverse(I)I,1)
   w14 := w11 ^ w13; 
   w15 := w18 >> 1; 
  } else {
   invoke(w9,Integer,numTrailZero(I)I, w18)
   w10 := w18 >> 1; };
  (*@$\phi$@*)(w16, w15, w10); 
  (*@$\phi$@*)(w17, w14, w19); 
\end{lstlisting}
\end{tcolorbox}

Where {\tt w}$7$ refers to $hasZeroTrail$ variable, the first definition of $i$ is {\tt w}$18$, and finally the first definition of $j$ is {\tt w}$19$.

 After the control flow graph is translated into JR's IR, JR applies three static transformations: $\phi$-expressions elimination, early-returns elimination and finally alpha renaming of variables.

\begin{enumerate}
    \item {\bf$\gamma$-Creation}:
    
    In this transformation all if-statements with $\phi$-expressions, are rewritten to assignment statements, where the expression of assignments are $\gamma$-expressions capturing the semantics of the $\phi$-statements. 

    For example, the if-statement in Listing~\ref{lst:recoveredseperatebits} at {\bf line 1} has two trailing $\phi$-statements ($\phi({\tt w}16, {\tt w}15, {\tt w}10)$ and $\phi({\tt w}17, {\tt w}14, {\tt w}19)$) capturing the changes to the $i$ variable in Listing~\ref{lst:example} in {\bf lines 7,11}, and the changes to the $j$ variable also in the same figure at {\bf 9, 10}, respectively. In the $\gamma$-creation transformation, these $\phi$-statements are replaced with assignment statements with $\gamma$-expressions. For example, {\bf lines 10,11} in Listing~\ref{lst:gammaseperatebits} shows the corresponding assignment statements created for these $\phi$-statements.
    
     \begin{tcolorbox}[
    listing only,
    listing options={style=mystyle},
    colback=white,
    colframe=black!30,
    top=2pt, bottom=2pt
]
    \begin{lstlisting}[caption={IR after $\gamma$-creation transformation}, label={lst:gammaseperatebits},captionpos=b]
 if (w7 != 0) {
   w11 := w19 >>> 1; 
   invoke(w13, Integer, reverse(I)I, 1)
   w14 := w11 ^ w13; 
   w15 := w18 >> 1; 
 } else {
   invoke(w9, Integer,numTrailZero(I)I, w18)
   w10 := w18 >> 1; }
 w16 := (*@$\gamma$@*)(w7!=0, w15, w10); 
 w17 := (*@$\gamma$@*)(w7!=0, w14, w19);
\end{lstlisting}
\end{tcolorbox}

    \item {\bf Early-Returns Elimination}:
    For code regions where multiple return statements exists, in this transformation, we construct two expressions, a \emph{conditional-return-expression}, and a \emph{return path constraint}. The conditional-return-expression describes, for each variable, the inner relative path constraint for assigning variables. On the other hand, the return path constraint describes all relative conditions that can exist within a region. the a var. A one that describes explicit return statements, all return statements are collapsed into a single assignment of a symbolic return-variable, with the appropriate return condition captured in a $\gamma$-expression. 
    
    For example, while there are no early-returns in $seperateBits$ there are early returns in the invoked method $numTrailZero(I)I$. Listing~\ref{lst:earlyreturn1} shows \toolshort\; IR-statement recovered from the definition of this method. 

Applying early-return elimination transformation results in listing~\ref{lst:earlyreturn2} where $erResult$ is a symbolic return variable. 
\begin{minipage}[t]{\textwidth}
    \begin{minipage}[t]{0.50\textwidth}
    \begin{tcolorbox}[
    listing only,
    listing options={style=mystyle},
    colback=white,
    colframe=black!30,
    top=2pt, bottom=2pt
]
      \begin{lstlisting}[caption={IR  for $numTrailZero$}, label={lst:earlyreturn1},captionpos=b]
 if (w1 == 0) {
   return w4
 } else {
   skip; 
 }
 ...
 w27 := w24 << 1; 
 w28 := w27 >>> 31; 
 w29 := w25 - w28; 
 return w29
\end{lstlisting}
\end{tcolorbox}
\end{minipage}
\hfill
\begin{minipage}[t]{0.50\textwidth}
 \begin{tcolorbox}[
    listing only,
    listing options={style=mystyle},
    colback=white,
    colframe=black!30,
    top=2pt, bottom=2pt
]
 \begin{lstlisting}[caption={\footnotesize IR after early-return}, label={lst:earlyreturn2},captionpos=b]
  if ( w1 == 0) {
    skip; 
  } else {
    skip; }
  ...
  w27 := w24 << 1; 
  w28 := w27 >>> 31; 
  w29 := w25 - w28; 
  skip; 
  erResult := (*@$\gamma$@*)((w1 == 0), w4, w29); 
\end{lstlisting}
\end{tcolorbox}
\end{minipage}
\end{minipage}

 \item{\bf$\alpha$-Renaming}:
    In this transformation all local variables of Java Ranger are appended a unique index. More precisely, every variable is now a pair of a variable and a unique index. This index is useful to identify multiple encounters of the same region of code during different parts of the execution. Here, since this the first encounter all variables with be paired with a unique index $1$. 
    In this example all local variables are appended with the number $1$, both in the IR statement as well as inside the Ranger state.
     \begin{tcolorbox}[
    listing only,
    listing options={style=mystyle},
    colback=white,
    colframe=black!30,
    top=2pt, bottom=2pt]
        \begin{lstlisting}[caption={ IR for $seperateBits$ after $\alpha$-renaming}, label={lst:alphaseperatebits},captionpos=b]
 if ((w7,1) != 0) {
   (w11,1) := (w19,1) >>> 1; 
   invoke((w13,1),Integer,reverse(I)I,1)
   (w14,1) := (w11,1) ^ (w13,1); 
   (w15,1) := (w18,1) >> 1; 
 } else {
   invoke((w9,1),Integer,numTrailZero(I)I, (w18,1))
   (w10,1) := (w18,1) >> 1; }
 (w16,1) := (*@$\gamma$@*)((w7,1) !=0, (w15,1), (w10,1)); 
 (w17,1) := (*@$\gamma$@*)((w7,1) !=0, (w14,1), (w19,1));
\end{lstlisting}
    \end{tcolorbox}
\end{enumerate}

\subsection{Dynamic Phase}

In this phase, JR uses the state of the dynamic environment in an attempt to eliminate additional Java language features, such as method invocation, field operations and exceptions. The idea here is that using the dynamic state JR can reduce the possibilities of dispatched methods, fields and specific exceptions.
This phase consists of four main transformations

\begin{enumerate}
\item {\bf Local Variable Substitution}:
    The substitution transformation substitutes values of the region's inputs that are defined in $I$. In this example, $I$ is defined to contain {\tt w}$7$, {\tt w}$18$ and {\tt w}$19$ with instantiation-time values being the symbolic expressions $i_a$~{\tt\&}~$1$, $i_a$, and $0$ respectively.
     \begin{tcolorbox}[
    listing only,
    listing options={style=mystyle},
    colback=white,
    colframe=black!30,
    top=2pt, bottom=2pt
]
    \begin{lstlisting}[caption={IR for $seperateBits$ after substitution}, label={lst:subsseperatebits},captionpos=b]
 if (i_a & 1 != 0 ) {
   (w11,1) := 0 >>> 1; 
   invoke((w13,1),Integer,reverse(I)I,1)
   (w14,1) := (w11,1) ^ (w13,1); 
   (w15,1) := i_a >> 1; 
 } else {
   invoke((w9,1),Integer,numTrailZero(I)I, i_a)
   (w10,1) := i_a >> 1;  }
 (w16,1) := (*@$\gamma$@*)((i_a & 1) !=0, (w15,1), (w10,1)); 
 (w17,1) := (*@$\gamma$@*)((i_a & 1) !=0, (w14,1), 0);
\end{lstlisting}    
\end{tcolorbox}
    \item {\bf Method Inlining}:
    To enable method inlining two main steps need to occur. The first binds the passed parameters to the invoked method arguments, the second captures the return value if one exists as an assignment in the caller. 
    This transformation also inlines recursive functions up to a preconfigured depth.
    For example listing~\ref{lst:highorder} shows inlining, where $(${\tt w}$9,1)$ holds the return expression from invoking numTrailZero.
 \begin{tcolorbox}[
    listing only,
    listing options={style=mystyle},
    colback=white,
    colframe=black!30,
    top=2pt, bottom=2pt
]
 \begin{lstlisting}[caption={IR for $numTrailZero$ after method-inlining}, label={lst:highorder},captionpos=b]
 if ((i_a & 1) != 0 ) {
   (w11,1) := 0 >>> 1 ; 
   ...
 } else {
   if (!(i_a != 0)) {
     skip; 
   } else {
     skip; }
   ...
   (w9,1) :=  (*@$\gamma$@*)(!(i_a!=0), 32, else (w29,3)); 
   (w10,1) := i_a >> 1; }
 (w16,1) := (*@$\gamma$@*)(i_a&1!=0, (w15,1), (w10,1)); 
 (w17,1) := (*@$\gamma$@*)(i_a&1!=0, (w14,1), 0); 
\end{lstlisting}
\end{tcolorbox}
    \item {\bf Field and Arrays GSA Creation}:
    The intuition behind this transformation is that, all fields updates can replaced with some fresh ssa local variable that are uniquely indexed. Likewise all fields look-ups are replaced with the corresponding ssa variable. Creating these local variables and using them sensibly to reflect field updates and lookups, is the purpose of this transformation. Therefore the output of this transformation is a bunch of local variable assignments, with no field lookup or update statements. 
    \item {\bf Single Path Cases}:
    \tool\ avoids summarizing object creation and exceptions due to some architectural limitations. However, giving up on path-merging of all code regions that are enclosing object creations or exceptions, would mean that \tool\ will give up on potentially many useful merges. Thus, we have introduced the Single Path Cases transformation.  
    In this transformation, \tool\; identifies two constraints, the first capturing the relative conditions of statements it can summarize, while the second capturing the relative conditions of statements that it cannot summarize, and where traditional symbolic execution needs to dynamically execute. 
    
    In our motivational example, \toolshort\; was able to summarize all the code region. However introducing a print statements on the else side would result in two predicates, one that describe the summary of the code region without the branch that contains the print statement, and another that is used to target the execution of SSE to it.
    \item {\bf Linearization}:
    In this transformation, {\tt if-statements} are removed, and statements on the then and the else side are composed in a sequential composition. This is possible now at this transformation because all local variable assignments are guarded, i.e., are captured by the condition in $\gamma$-statements.
     \begin{tcolorbox}[
    listing only,
    listing options={style=mystyle},
    colback=white,
    colframe=black!30,
    top=2pt, bottom=2pt
]

   \begin{lstlisting}[caption={IR's snippet for $numTrailZero$ after linearization}, label={lst:linearization},captionpos=b]
  (w15,1) := i_a >> 1 ; 
  (w7,3) := i_a << 16 ; 
  (w9,3) := (*@$\gamma$@*)((w7,3)!=0, (w7,3), i_a); 
  (w10,3) := (*@$\gamma$@*)((w7,3)!=0, 15, 31); 
  (w12,3) := (w9,3) << 8; 
  (w13,3) := (w10,3) - 8; 
   . . .  
  (w9,1) :=  (*@$\gamma$@*)( (i_a==0) , 32,  (w29,3)); 
  (w10,1) := i_a >> 1; 
  (w16,1) := (*@$\gamma$@*)(i_a&1!=0, (w15,1), (w10,1)); 
  (w17,1) := (*@$\gamma$@*)(i_a&1!=0, -2147483648, 0); 
\end{lstlisting}
\end{tcolorbox}
\end{enumerate}

After all the transformations are executed successfully, \tool\ proceeds as follows. First, \toolshort\ will create a conjunctive constraint of all the expressions that resulted from the linearization transformation. Then, it will create an execution path to explore the unsummarized behavior from the single path cases, as well as, another one to explore {\tt returns} from the early return transformation. Finally, \tool\ will set up the state of variables in each case, and will allow SE to resume execution by exploring these paths.

Next, we will show the formal semantics of JR's transformations, and the soundness proofs of the main transformations.  
The general form of the proof of the soundness of \toolshort\; is to establish correspondence between any of \toolshort\; transformations and the execution of the concrete semantics. For if \toolshort \; is sound with respect to concrete semantics, then using \toolshort\; in a symbolic execution would be sound as long as the symbolic execution is sound and correct. 

Thus in the rest of the paper is organized as follows: first we present the concrete semantics which we will use to show the soundness of JR's transformations, next, we will present the semantics of main JR's transformations along with their soundeness proofs, and finally, present the overall algorithm of JR.
\section{Related Work}
\label{sec:related-work}
Path explosion is one of the most commonly observed scalability
challenges of symbolic execution, and works prior to
Java Ranger have explored a variety of path merging approaches.
In symbolic execution tools that represent multiple complete symbolic
execution states, one natural approach is to implement a merging operation
on these state representations.
Hansen et al.~\cite{HansenSS2009} and Kuznetsov et al.~\cite{kuznetsov}
explored such state-merging techniques, and MultiSE~\cite{multise}
achieves a similar effect within a somewhat different tool architecture.

By contrast, Java Ranger takes a different approach of statically
summarizing code regions with branching control-flow, more closely
related to the Veritesting approach of Avgerinos et al.~\cite{veritesting}
in their system MergePoint for binary executables.
The static summarization approach is easier to adopt in a symbolic
execution tool that is designed to have only on execution state at once,
like Symbolic PathFinder (SPF)~\cite{spf} that Java Ranger builds upon.
Java Ranger is significantly different from MergePoint because of the
differences between binary code and Java bytecode: for instance Java
Ranger can take advantage of structured control-flow, but needs
to perform its own method inlining because of the prevalence of small
methods in Java.

Many of the invariants (such as SSA form) and transformations (such as
method inlining and constant propagation) used by Java Ranger are
similar to techniques used in compilers, and compilers are the language
transformation tools for which correctness proof techniques are most
developed.
The two main classes of verification tools for compilers are translation
validators~\cite{pnueli1998translation,necula2000translation,lopes2021alive2},
which check the correctness of a particular compilation run,
and certified compilers, such a CompCert~\cite{leroy2009formal} for C
and CakeML~\cite{kumar2014cakeml} for Standard ML, which
are proven once-and-for-all to compile any input program correctly.
%
Our proofs about Java Ranger are once-and-for-all, but in being 
transformations within an intermediate representation, they are similar
to correctness verification of individual compiler passes~\cite{zhao2013formal,lopes2015provably,vanhattum2024lightweight,brown2020towards}.


A number of previous formal descriptions of symbolic execution have
been language-agnostic or used simple languages for clarity.
Schwartz et al.~\cite{schwartz2010all} use a simple formal language to
explain taint analysis and symbolic execution, though they also mention
a number of practical implementation concerns.
%
De Boer and Bonsangue~\cite{deBoer2021symbolic} use a transition system model
to formalize symbolic execution for a sequence of languages of increasing
complexity.
%
In a recent work, Voogd et al.~\cite{voogd2025compositional} present a denotational
semantics for symbolic execution along with a machine-checked formalization.
%
The formalization of Porncharoenwase et al.~\cite{porncharoenwase2022formal} has
the closest purpose to ours in introducing an abstract form of state merging,
which generalizes path merging and can be applied to code with loops and recursion.
They described a machine-checked implementation for a core subset of Scheme and
more realistic optimized tool that is tested against it.

\section{Semantics of the Java-Simplified Concrete Execution}
\label{sec:concrete}
In this section, we define the concrete execution semantics used to establish the soundness of Java Ranger's transformations. To make this formalization tractable, we focus on a core, simplified Java language. We model this language using an intermediate representation derived from stackless Java bytecode in static single assignment (SSA) form~\cite{bilardi2003algorithms}, where each variable is assigned exactly once and over which we formally define the concrete execution rules.


\soha{\subsection{Supported Java Features and Scope of Formalization}}
\label{sec:simplifiedjava}

\soha{
While Java Ranger operates on arbitrary Java bytecode, its summarization transformation applies to a well-defined subset of language constructs. Specifically, the current implementation does not summarize loops, switch statements, exception handling, or object instantiations.}

\soha{To formalize Java Ranger's core operational semantics, we present its rules over a simplified Java language, which includes the following constructs:
\begin{itemize}
    \item Primitive Types: booleans and integers.
    \item Class Types: basic non-static class declarations (excluding interfaces, generics, and access modifiers, treating all classes as public).
    \item Variables: local variables and fields containing either primitive values or object references.
    \item Expressions: standard arithmetic and logical operations, ternary expressions ({\tt ? :}), side-effecting assignments (including compound and increment/decrement operators, e.g., {\tt +=}, {\tt ++}), and field accesses.
    \item Statements: variable assignments, conditional branches ({\tt if}), method invocations, return statements, and object creation.
\end{itemize}
}

\subsection{Concrete Grammar}
\label{sec:concretegrammar}

\begin{figure*}[t]
    \begin{align*}
        integer\quad & n \in \mathbb{Z} \quad & subscript \quad &j \in \mathbb{Z^+}\uplus \{\bot\}  \quad & identifier \quad &a \in Id  \qquad field \quad f \in F\\
class \quad &c \in C \quad & signature \quad &g \in G \quad & reference \quad &r \in R 
    \end{align*}
\begin{align*}
literal\qquad v  &::= n \mid true \mid false  \mid null 
\qquad \qquad var  \quad x,y,i,o,t ::= (a,j) \\
method \qquad m &::=\regiondef{x}{s} \quad\qquad\qquad\qquad\qquad\qquad unary Op \qquad \unaryOp ::= - \mid ^^21   \\
binary Op \qquad \binaryOp &::= + \mid - \mid * \mid \div \mid \& \mid bitwise-or \mid \oplus \mid \% \mid == \mid ^^21= \mid \leq \mid \geq \\
& \qquad \mid \&\& \mid logical\text{-}or \mid >\: \mid <\: \mid \ll \mid \gg \\
exp\qquad e  &::=  v \mid x \mid \binaryop{e_1}{e_2} \mid \unaryop{e} \mid  \gammaexp{e}{e_1}{e_2} \\
simple\text{-}ref. \qquad z &::= r \mid x\\
stmt \qquad s &::= s_1 ; s_2 \mid x := e \mid skip \mid \myif{e}{s_1}{s_2};\phi\text{-}stmt \\
 & \qquad \mid \invoke{x}{z}{g}{y} \mid \return{e}  \mid \putfield{z}{f}{e}\\
 &\qquad \mid \getfield{x}{z}{f} \mid \newobj{x}{c} \\
\phi\text{-}stmt \quad phis &::= \epsilon \mid \phi(x,e_1,e_2);\phi\text{-}stmt\\
\end{align*}

\caption{Context Free Grammar for Concrete Execution} 
\label{fig:concgrammar}
\end{figure*}

Figure~\ref{fig:concgrammar} presents the grammar of \toolshort's IR.
The metavariables used throughout the semantics range over the following domains. Integers $n \in \mathbb{Z}$ and positive subscripts $j \in \mathbb{Z}^+ \uplus \{\bot\}$ are used to index and distinguish program variables, where $\bot$ denotes the absence of a subscript. Variable identifiers $a \in Id$ name program variables, field names $f \in F$ denote object fields, class names $c \in C$ range over the set of declared classes, method signatures $g \in G$ identify methods, and references $r \in R$ denote concrete heap addresses. Note that references are always concrete and are never permitted to take symbolic values.

The expression language of the grammar includes standard literals, variables, and both unary and binary expressions. In addition, it features the $\gamma$-expression, which serves as a compact representation of the conventional conditional (if-then-else) expression. We use $x, y, i, o, t$ to range over variables. Each variable is represented as a pair consisting of a name and a unique index, i.e., $(a, j)$, where the index distinguishes different instances of the same variable across multiple executions of a code region.

Furthermore, the IR models references through simple references $z$, which may either be concrete references $r$ or reference variables $x$. Fields are denoted by $f$, classes by $c$, and method signatures by $g$.

The statement language includes assignments, sequential composition (treated as an associative operator), the {\tt skip} statement, and conditional statements of the form {\tt if}, augmented with trailing $\phi$-functions to capture the merging of variable values across branches. Method invocation is represented as $\invoke{x}{z}{g}{y}$, where $x$ receives the result of invoking method $g$ on reference $z$ with arguments $\overrightarrow{y}$. The IR also includes return statements $\return{e}$, field update statements $\putfield{z}{f}{e}$, which assign the value of expression $e$ to field $f$ of reference $z$, and field access statements $\getfield{x}{z}{f}$, which read the value of field $f$ from reference $z$ into variable $x$. Finally, object creation is represented by $\newobj{x}{c}$, where a new object of class $c$ is allocated and assigned to reference variable $x$.

Note that the full implementation of \toolshort\ also supports array load and store operations. However, since these can be modeled as specific instances of field operations, we omit their formal treatment here without loss of generality.

\subsection{Concrete Environment}
\label{sec:concreteenv}

\begin{figure*}[t]
\[
\text{vals:}\; V=\mathbb{Z} \uplus \{true, false , null \}  \uplus R  
\]

\[
\delta=(\Theta, \eta, \Delta,H\nu,mu)
\]

$$
\begin{array}{llllllllll}
\text{method-stmt:} & \Theta: C \times G  &\longrightarrow & method(\overrightarrow{var},stmt) &
\\
 \text{ref-type:} & \eta: R & \longrightarrow & C &
 \\
\text{locals:} & \Delta: (Id,\mathbb{Z^+}\uplus \{\bot\}) &\longrightarrow & V &
\\
\text{heap:} & H: R \times F &\longrightarrow & V &
\\
\text{new-fresh:} & \nu \in \mathbb{N} & & & \\
\text{return-value:} & \mu \in V \uplus \{\bot\} & & & 
 \\
\end{array}
$$
\caption{Environments and functions of Concrete Semantics}
\label{fig:environment}
\end{figure*}

Fig.~\ref{fig:environment} presents the components of the concrete semantic environment $\delta = (\Theta, \eta, \Delta, H, \nu, \mu)$. The environment captures the essential runtime state required for execution. In particular, $\Theta$ maps method statements to their corresponding signatures, while $\eta$ maps reference identifiers to their associated types. The mapping $\Delta$ represents the local variable environment, and $H$ denotes the heap, which maps objects-fields pairs ($R \times F$) to their associated values. The component $\nu$ is a counter used to generate fresh object identifiers during allocation, and $\mu$ records return values produced during execution.

Values in the semantic domain are drawn from the set $\mathbb{Z}$ of integers, the Boolean constants $\{{\tt true}, {\tt false}\}$, the distinguished value ${\tt null}$, and the set of concrete references $R$.

\subsection{Concrete Semantic Rules}
\label{sec:concreterules}
\journalreport{

\begin{figure*}[h!t]
    \footnotesize
\fbox{%
\parbox{0.98\textwidth}{%
\[
\infer[\rn{\gamma-true}]
 {\vdashconcestep{\Delta}{\gammaexp{e}{e_1}{e_2}}{v}}
{\begin{gathered}
\vdashconcestep{\Delta}{e}{true}
\quad
\vdashconcestep{\Delta}{e_1}{v}
\end{gathered}}
\qquad
\infer[\rn{\gamma-false}]
 {\vdashconcestep{\Delta}{\gammaexp{e}{e_1}{e_2}}{v}}
{\begin{gathered}
\vdashconcestep{\Delta}{e}{false}
\quad
\vdashconcestep{\Delta}{e_2}{v}
\end{gathered}}
\]
\hrule
\vspace{6pt}
\[
\infer[\rn{composition}]
 {
     \Theta\vdash \concsstep{(\eta, \Delta, H, \nu,\bot)}{s_1;s_2}{(\eta'',\Delta'', H'', \nu'',\mu'')}{}
 }
{
\begin{gathered}
    \Theta\vdash \concsstep{(\eta, \Delta, H, \nu,\bot)}{s_1}{(\eta', \Delta', H', \nu',\bot)}{}\\
    \Theta\vdash \concsstep{(\eta', \Delta', H', \nu',\bot)}{s_2}{(\eta'',\Delta'', H'', \nu'',\mu'')}{}
\end{gathered}
}
\]
\[
\infer[\rn{composition-ret}]
 {
     \Theta\vdash \concsstep{(\eta, \Delta, H, \nu,\bot)}{s_1;s_2}{(\eta',\Delta', H', \nu',\mu')}{}
 }
{
    \Theta\vdash \concsstep{(\eta, \Delta, H, \nu,\bot)}{s_1}{(\eta', \Delta', H', \nu',\mu')}{}
    \qquad\mu' \neq \bot
}
\]

\[
\infer[\rn{assignment}]
 {
    \Theta\vdash\concsstep{(\eta,\Delta, H, \nu,\bot)}{x:=e}{(\eta, \update{\Delta}{x}{v},H,\nu,\bot)}{}
}
{   
    \vdashconcestep{\Delta}{e}{v}
}
\]

\[
\infer[\rn{invoke}]
 {
    \Theta\vdash\concsstep{(\eta,\Delta, H, \nu,\bot)}{\invoke{x}{z}{g}{y}}{(\eta', \update{\Delta}{x}{\mu'},H',\nu',\bot)}{}
}
{
 \begin{gathered}
    \vdashconcestep{\Delta}{z}{r} \qquad
    r \in R
    \qquad \overrightarrow{y'}.s'= \lookup{\Theta}{\lookup{\eta}{r},g} 
    \\
    \overrightarrow{v'}=\Delta(\overrightarrow{y}) \qquad
     \Theta\vdash \concsstep{(\eta,\{(\overrightarrow{y'},\overrightarrow{v'})\},H,\nu,\bot)}{s'}{(\eta', \Delta',H',\nu',\mu')}{} \qquad \mu' \neq \bot
 \end{gathered}}
\]

$$\infer[\rn{return}]
 {
 \Theta\vdash\concsstep{(\eta, \Delta, H, \nu,\bot)}{\return{$e$}}{(\eta, \Delta,H,\nu,v)}{}
}
{
 \vdashconcestep{\Delta}{e}{v}
}
$$

$$
\infer[\shortstack{\texttt{\tiny{if-}}
\texttt{\tiny{true}}}]
 {
    \Theta\vdash\concsstep{(\eta, \Delta, H, \nu,\bot)}{\myif{e}{s_1}{s_2};\phiexp{x_1}{e_1}{e_1'}\ldots \phiexp{x_k}{e_k}{e_k'}}{(\eta',\Delta'', H', \nu',\bot)}{}
}
{
\begin{gathered}
    \vdashconcestep{\Delta}{e}{true}
    \qquad  \vdashconcestep{\Delta}{e_1}{v_1} 
    \ldots  \vdashconcestep{\Delta}{e_k}{v_k} \\
    \Theta\vdash\concsstep{(\eta,\Delta, H, \nu,\mu)}{s_1}{(\eta',\Delta', H', \nu',\bot)}{} 
    \qquad \Delta'[x_1\leftarrow v_1,\ldots, x_k\leftarrow v_k]
\end{gathered}
}
$$
$$
\infer[\rn{put-field}]
 {
     \Theta\vdash\concsstep{(\eta, \Delta, H, \nu,\bot)}{\putfield{z}{f}{e}}{(\eta, \Delta,H[(r,f) \leftarrow v],\nu,\bot)}{}
 }
{\begin{gathered}
\vdashconcestep{\Delta}{z}{r} \qquad
    \vdashconcestep{\Delta}{e}{v}
    \qquad r \in R
\end{gathered}
}
$$
\[
\infer[\rn{get-field}]
 {
     \Theta\vdash\concsstep{(\eta, \Delta, H, \nu,\bot)}{\getfield{x}{z}{f}}{(\eta, \Delta[x \leftarrow v],H, \nu,\bot)}{}{}
 }
{\begin{gathered}
\vdashconcestep{\Delta}{z}{r}
    \qquad v = \lookup{H}{r,f}
    \qquad r \in R
\end{gathered}}
\]

\[
\infer[\rn{new}]
 {
     \Theta\vdash\concsstep{(\eta, \Delta, H, \nu,\bot)}{\newobj{x}{c}}{(\eta[r \leftarrow c], \Delta[x\leftarrow r], H, r',\bot)}{}{}
 }
{\begin{gathered}
\nu'=\nu+1 \qquad r=\text{freshRef}(\nu')
\end{gathered}}
\]
}}
\caption{Selected Evaluation Rules of Concrete Semantics}
\label{fig:conc-semantics-short}
\end{figure*}
}
\techreport{\begin{figure*}[t]
    \footnotesize
\fbox{%
\parbox{0.98\textwidth}{%
\[
\infer[\rn{val}]
 {\vdashconcestep{\Delta}{v}{v}}
{}
\qquad
\infer[\rn{var}]
 {\vdashconcestep{\Delta}{x}{v}}
{v=\lookup{\Delta}{x}}
\qquad
\infer[\rn{unary-op}]
 {\vdashconcestep{\Delta}{\unaryop{e}}{v'}}
{
\vdashconcestep{\Delta}{e}{v}
\qquad
\unaryop{v}{} = v'}
\]

\[
\infer[\rn{binary-op}]
 {\vdashconcestep{\Delta}{\binaryop{e_1}{e_2}}{v_3}}
{\vdashconcestep{\Delta}{e_1}{v_1}
\qquad
\vdashconcestep{\Delta}{e_2}{v_2}
\qquad
\binaryop{v_1}{v_2}{} = v_3}
\]
\[
\infer[\rn{\gamma-true}]
 {\vdashconcestep{\Delta}{\gammaexp{e}{e_1}{e_2}}{v}}
{\begin{gathered}
\vdashconcestep{\Delta}{e}{true}
\quad
\vdashconcestep{\Delta}{e_1}{v}
\end{gathered}}
\qquad
\infer[\rn{\gamma-false}]
 {\vdashconcestep{\Delta}{\gammaexp{e}{e_1}{e_2}}{v}}
{\begin{gathered}
\vdashconcestep{\Delta}{e}{false}
\quad
\vdashconcestep{\Delta}{e_2}{v}
\end{gathered}}
\]
\hrule
\vspace{6pt}
\[
\infer[\rn{composition}]
 {
     \Theta\vdash \concsstep{(\eta, \Delta, H, \nu,\bot)}{s_1;s_2}{(\eta'',\Delta'', H'', \nu'',\mu'')}{}
 }
{
\begin{gathered}
    \Theta\vdash \concsstep{(\eta, \Delta, H, \nu,\bot)}{s_1}{(\eta', \Delta', H', \nu',\bot)}{}\\
    \Theta\vdash \concsstep{(\eta', \Delta', H', \nu',\bot)}{s_2}{(\eta'',\Delta'', H'', \nu'',\mu'')}{}
\end{gathered}
}
\]
\[
\infer[\rn{composition-ret}]
 {
     \Theta\vdash \concsstep{(\eta, \Delta, H, \nu,\bot)}{s_1;s_2}{(\eta',\Delta', H', \nu',\mu')}{}
 }
{
    \Theta\vdash \concsstep{(\eta, \Delta, H, \nu,\bot)}{s_1}{(\eta', \Delta', H', \nu',\mu')}{}
    \qquad\mu' \neq \bot
}
\]

\[
\infer[\rn{assignment}]
 {
    \Theta\vdash\concsstep{(\eta,\Delta, H, \nu,\bot)}{x:=e}{(\eta, \update{\Delta}{x}{v},H,\nu,\bot)}{}
}
{   
    \vdashconcestep{\Delta}{e}{v}
}
\]

\[
\infer[\rn{invoke}]
 {
    \Theta\vdash\concsstep{(\eta,\Delta, H, \nu,\bot)}{\invoke{x}{z}{g}{y}}{(\eta', \update{\Delta}{x}{\mu'},H',\nu',\bot)}{}
}
{
 \begin{gathered}
    \vdashconcestep{\Delta}{z}{r} \qquad
    r \in R
    \qquad \overrightarrow{y'}.s'= \lookup{\Theta}{\lookup{\eta}{r},g} 
    \\
    \overrightarrow{v'}=\Delta(\overrightarrow{y}) \qquad
     \Theta\vdash \concsstep{(\eta,\{(\overrightarrow{y'},\overrightarrow{v'})\},H,\nu,\bot)}{s'}{(\eta', \Delta',H',\nu',\mu')}{} \qquad \mu' \neq \bot
 \end{gathered}}
\]

$$\infer[\rn{return}]
 {
 \Theta\vdash\concsstep{(\eta, \Delta, H, \nu,\bot)}{\return{$e$}}{(\eta, \Delta,H,\nu,v)}{}
}
{
 \vdashconcestep{\Delta}{e}{v}
}
$$
$$
\infer[\rn{skip}]
 {
     \Theta\vdash\concsstep{(\eta, \Delta, H,\nu, \bot)}{skip}{(\eta,\Delta, H, \nu, \bot)}{}
 }
{}
$$

$$
\infer[\shortstack{\texttt{\tiny{if-}}
\texttt{\tiny{true}}}]
 {
    \Theta\vdash\concsstep{(\eta, \Delta, H, \nu,\bot)}{\myif{e}{s_1}{s_2};\phiexp{x_1}{e_1}{e_1'}\ldots \phiexp{x_k}{e_k}{e_k'}}{(\eta',\Delta'', H', \nu',\bot)}{}
}
{
\begin{gathered}
    \vdashconcestep{\Delta}{e}{true}
    \qquad  \vdashconcestep{\Delta}{e_1}{v_1} 
    \ldots  \vdashconcestep{\Delta}{e_k}{v_k} \\
    \Theta\vdash\concsstep{(\eta,\Delta, H, \nu,\mu)}{s_1}{(\eta',\Delta', H', \nu',\bot)}{} 
    \qquad \Delta'[x_1\leftarrow v_1,\ldots, x_k\leftarrow v_k]
\end{gathered}
}
$$
$$
\infer[\shortstack{\texttt{\tiny{if-}}
\texttt{\tiny{true-ret}}}]
 {
    \Theta\vdash\concsstep{(\eta, \Delta, H, \nu,\bot)}{\myif{e}{s_1}{s_2};\phiexp{x_1}{e_1}{e_1'}\ldots \phiexp{x_k}{e_k}{e_k'}}{(\eta',\Delta', H', \nu',\mu')}{}
}
{
\begin{gathered}
    \vdashconcestep{\Delta}{e}{true} 
    \qquad \mu' \neq \bot\\
    \Theta\vdash\concsstep{(\eta,\Delta, H, \nu,\mu)}{s_1}{(\eta',\Delta', H', \nu',\mu')}{}
\end{gathered}
}
$$

$$
\infer[\shortstack{\texttt{\tiny{if-}}
\texttt{\tiny{false}}}]
 {
    \Theta\vdash\concsstep{(\eta, \Delta, H, \nu,\bot)}{\myif{e}{s_1}{s_2};\phiexp{x_1}{e_1}{e_1'}\ldots \phiexp{x_k}{e_k}{e_k'}}{(\eta',\Delta'', H', \nu',\bot)}{}
}
{
\begin{gathered}
    \vdashconcestep{\Delta}{e}{false}
    \qquad  \vdashconcestep{\Delta}{e_1'}{v_1'}
    \ldots  \vdashconcestep{\Delta}{e_k'}{v_k'}\\
     \Theta\vdash\concsstep{(\eta,\Delta, H, \nu,\bot)}{s_2}{(\eta',\Delta', H', \nu',\bot)}{} \qquad
     \Delta'[x_1\leftarrow v_1',\ldots, x_k\leftarrow v_k']
\end{gathered}
}
$$
$$
\infer[\shortstack{\texttt{\tiny{if-}}
\texttt{\tiny{false-ret}}}]
 {
    \Theta\vdash\concsstep{(\eta, \Delta, H, \nu,\bot)}{\myif{e}{s_1}{s_2};\phiexp{x_1}{e_1}{e_1'}\ldots \phiexp{x_k}{e_k}{e_k'}}{(\eta',\Delta', H', \nu',\mu')}{}
}
{
\begin{gathered}
    \vdashconcestep{\Delta}{e}{false}
    \qquad \mu' \neq \bot\\
    \qquad \concsstep{(\eta,\Delta, H, \nu,\bot)}{s_2}{(\eta',\Delta', H', \nu',\mu')}{}
\end{gathered}
}
$$

$$
\infer[\rn{put-field}]
 {
     \Theta\vdash\concsstep{(\eta, \Delta, H, \nu,\bot)}{\putfield{z}{f}{e}}{(\eta, \Delta,H[(r,f) \leftarrow v],\nu,\bot)}{}
 }
{\begin{gathered}
\vdashconcestep{\Delta}{z}{r} \qquad
    \vdashconcestep{\Delta}{e}{v}
    \qquad r \in R
\end{gathered}
}
$$
\[
\infer[\rn{get-field}]
 {
     \Theta\vdash\concsstep{(\eta, \Delta, H, \nu,\bot)}{\getfield{x}{z}{f}}{(\eta, \Delta[x \leftarrow v],H, \nu,\bot)}{}{}
 }
{\begin{gathered}
\vdashconcestep{\Delta}{z}{r}
    \qquad v = \lookup{H}{r,f}
    \qquad r \in R
\end{gathered}}
\]

\[
\infer[\rn{new}]
 {
     \Theta\vdash\concsstep{(\eta, \Delta, H, \nu,\bot)}{\newobj{x}{c}}{(\eta[r \leftarrow c], \Delta[x\leftarrow r], H, r',\bot)}{}{}
 }
{\begin{gathered}
\nu'=\nu+1 \qquad r=\text{freshRef}(\nu')
\end{gathered}}
\]
}}
\caption{Evaluation Rules of Concrete Semantics}
\label{fig:conc-semantics}
\end{figure*}}

We present the evaluation rules of the concrete semantics in Appendix~\ref{fig:conc-semantics}. The rules are standard operational semantics rules capturing the execution of statements in \toolshort's IR. In particular, the {\tt composition} rule enforces sequential execution of statements $s_1$ and $s_2$. The corresponding {\tt composition-ret} rule handles the case in which $s_1$ evaluates to a return statement, in which case execution of $s_2$ is ignored.

The {\tt invoke} rule define the semantics of method invocation. In this rule, the reference expression $z$ is first evaluated to a concrete reference $r$, which is then used together with the method signature $g$ to retrieve the corresponding method body from $\Theta$. A fresh environment is constructed by binding the formal parameters of the method to the evaluated argument values ($\{(\overrightarrow{y'},\overrightarrow{v'})\}$), after which the method body is executed. Upon completion of the method body, the local environment $\Delta$ is restored, and the return value is recorded in $\mu$ and assigned to the target variable $x$ ($\Delta[x \leftarrow \mu']$).

Conditional execution is captured by the {\tt if-true}, {\tt if-false}, {\tt if-true-ret}, and {\tt if-false-ret} rules. In the non-returning cases ({\tt if-true} and {\tt if-false}), the condition expression $e$ is evaluated to determine the branch to execute, followed by execution of the selected branch and the corresponding $\phi$-statements. In the presence of a return statement within the $s_1$ branch ({\tt if-true-ret} and {\tt if-false-ret}), execution of the $\phi$-statements is skipped, and the return value is propagated via $\mu'$.

The {\tt getfield} and {\tt putfield} statements follow the standard semantics of heap access. In both cases, the reference expression $z$ is evaluated to a concrete reference $r$, after which the corresponding field $f$ is either read from or updated in the heap $H$.

Finally, object creation is handled by the {\tt new} rule, which guarantees the generation of fresh references. Each new object is assigned a unique reference $r$ produced by the $\mathsf{freshRef}$ function, which consumes the current counter value $\nu$ to ensure global freshness.

All evaluation judgments can be extended to work for vector of expressions, that is $\vdash \rangerexpstep{\omega}{\overrightarrow{e}}{\omega'}{\overrightarrow{e'}}{t}$ is defined as 
\begin{figure}[H]
    \[
    \infer[\rn{vec}]
     {
\vdash \rangerexpstep{\omega}{\overrightarrow{e}}{\omega'}{\overrightarrow{e'}}{t}
     }
    {
     \overrightarrow{e}= e_1 \ldots e_n \qquad 
     {\vdash \rangerexpstep{\omega}{e_1}{\omega'}{e_1'}{t}}\quad \cdots \quad 
      {\vdash \rangerexpstep{\omega}{e_n}{\omega'}{e_n'}{t}}
      \quad 
     \overrightarrow{e'}= e_1'\ldots e_n'
    }
\]
    \caption{Applying evaluation/rewriting over a vector of expressions }
    \label{fig:vectorreule}
\end{figure}

\subsection{Concrete Evaluation Properties}
\label{sec:concreteprop}
We assume the following properties hold over the concrete semantics
\begin{itemize}
    \item Statement must be defined in an SSA form (Property~\ref{prop:ssa}).
    \item Statement $s$ always has at most a single return within an if-statement (Property~\ref{prop:final-at-most-final-if-return}).
    \item Each class has an class object reference which can be used to make {\tt static} invocation (Property~\ref{prop:class-references}).
    \item All formal arguments of every method body in $\Theta$ are used in its statement (Property~\ref{prop:y-must-be-used-in-s}).
\end{itemize}

Next, we show the following lemma which are easy consequence of the rules of the concrete semantics and the SSA property (full details in~\cite{techreport}). 
\begin{lem}
\label{lem:inline-eta-lookup}
If for any reference $r$ and classes $c_1, c_2$, such that
    $$\Theta \vdash \concsstep{(\eta, \Delta, H, \nu,\mu)}{s}{(\eta', \Delta', H', \nu', \mu')}{} $$
    $$(r,c_1) \in \eta \wedge (r,c_2)\in\eta'$$

then 
   $c_1=c_2$.
    \techreport{
    \begin{proof}
        The only statement that changes $\eta$ is the {\tt new}-statement, so after a reference-class pair are constructed $\eta$, it cannot be updated nor removed. 
        On the other hand, adding a new pair to $\eta$ is always guaranteed to have a unique reference, since its creation is parametric by $\nu$, which gets incremented with each invocation of a {\tt new}-statement.
    \end{proof}
    }
\end{lem}

\begin{lem}
\label{lem:delta-ssa}
    $$\func{SSA}(s)$$
    $$\Theta \vdash \concsstep{(\eta, \Delta, H, \nu,\mu)}{s}{(\eta', \Delta', H', \nu', \mu')}{} $$

then
    $$\Delta \subseteq \Delta'$$
    $$\eta \subseteq \eta'$$

    \techreport{
    \begin{proof}
        First observe that because $s$ satisfies the generalized SSA property, it must be the case that variables are never remove nor updated, they can only be added. In particular for $z$, its mapping in both $\Delta$ and $\Delta'$ must be the same, more generally we must have that $\Delta \subseteq \Delta'$.

        Following a similar argument, we can conclude that references in $\eta$ never remove nor updated, they can only be added. In particular for $r$, its mapping in both $\eta$ and $\eta'$ must be the same, more generally we must have that $\eta \subseteq \eta'$
    \end{proof}
    }
\end{lem}

\begin{lem}
\label{lem:inline-eta-delta-lookup}
If for any simple reference $z$ and classes $c_1, c_2$ 
such that
    \begin{gather*}
    \func{SSA}(s)\\
    \Theta \vdash \concsstep{(\eta, \Delta, H, \nu,\mu)}{s}{(\eta', \Delta', H', \nu', \mu')}{} \\
    \vdashconcestep{\Delta}{z}{r_1}\\
    \vdashconcestep{\Delta'}{z}{r_2}\\
    (r_1,c_1) \in \eta \wedge (r_2,c_2)\in\eta' 
    \end{gather*}
    
    then
    
    $$c_1=c_2$$

    \techreport{
    \begin{proof}
        First observe that because $s$ satisfies the generalized SSA property, it must be the case that variables are never remove nor updated, they can only be added. In particular for $z$, its concrete evaluation in both $\Delta$ and $\Delta'$ must be the same, i.e., $r_1=r_2$, more generally we must have that $\Delta \subseteq \Delta'$.
        
        The only statement that changes $\eta$ is the {\tt new}-statement, so after a reference-class pair are constructed $\eta$, it cannot be updated nor removed. 
        On the other hand, adding a new pair to $\eta$ is always guaranteed to have a unique reference, since its creation is parametric by $\nu$, which gets incremented with each invocation of a {\tt new}-statement.
    \end{proof}}
\end{lem}

\section{Dynamic Symbolic Execution's Environment}
\label{sec:DSEEnv}
\begin{figure}
    \centering
\begin{tabular}{ll}
vals: & $V_s = V  \uplus exp$ \\
method-stmt: & $\Theta_s: C \times G  \longrightarrow method(\overrightarrow{var}.stmt)$ \\
ref-type: & $\eta_s: R  \longrightarrow  C$ \\
locals:& $\Delta_s: (Id,n) \longrightarrow V_s$\\
heap: & $H_s: R \times F \longrightarrow  V_s$ \\
method-return: & $\mu_s \in V_s \uplus \{\bot\}$ \\
\colorbox{lightgray}{path-conditions:} & $\Pi=(\pi,$\colorbox{lightgray}{$\pi_r,\pi_s$}$): V_s \times$ \colorbox{lightgray}{$V_s \times V_s$} $\qquad\qquad$ (PC, \colorbox{lightgray}{returnPC, exitPC})\\
\colorbox{lightgray}{method-stmt:} & \colorbox{lightgray}{$\Theta_t: C \times G  \longrightarrow method(\overrightarrow{var}.stmt)$} \\
\colorbox{lightgray}{unique-count:} & \colorbox{lightgray}{$u \in \mathbb{N}$ \quad (initialized once)} \\
\colorbox{lightgray}{region-variables:} & \colorbox{lightgray}{$\chi=(I,O,T) : 2^{Var} \times 2^{Var} \times 2^{Var} \qquad$ (Input, Outputs, Temp)}\\
\colorbox{lightgray}{unique-count:} & \colorbox{lightgray}{$u \in \mathbb{N} \qquad$ (initialized once)}\\
\end{tabular}
\caption{Java Ranger's Environment}
\label{fig:JREnv}
\end{figure}

In this section, we present the formal foundation of Dynamic Symbolic Execution (DSE) upon which Java Ranger's path-merging capabilities are built. We assume a sound symbolic execution that operates over a concrete state $(\Theta,\eta,\Delta,H,\nu,\mu)$ and a symbolic state $(\Theta,\eta_s,\Delta_s,H_s,\pi,\mu_s)$.

The grammar of DSE follows the same structure as the concrete semantics in Fig.~\ref{fig:concgrammar}, with two key differences summarized in the unshaded variables in Fig.~\ref{fig:JREnv}. First, the value domain is extended from concrete values $V$ to symbolic values $V_s$, allowing expressions to be used as values. Second, the environment is lifted accordingly. The method table $\Theta_s$ and reference typing $\eta_s$ are carried over unchanged from their concrete counterparts $\Theta$ and $\eta$ respectively; they appear with symbolic subscripts solely to distinguish the symbolic environment from the concrete one. The local state $\Delta_s$ maps uniquely indexed variables $(Id,n)$ to values in $V_s$, and the heap $H_s$ maps reference-field pairs to values in $V_s$, where references remain concrete while field values may be symbolic. Method return values $\mu_s$ range over $V_s \uplus \{\bot\}$, where $\bot$ denotes the absence of a return value. Finally, the path condition $\pi$ captures the symbolic constraint accumulated along the current execution path.

We make the following assumptions on the symbolic environment $(\eta_s,\Delta_s,H_s,\pi,\mu_s)$:
\begin{itemize}
    \item reference typing mapping is carried over unchanged from the concrete semantics, $\eta_s=\eta$ (Property~\ref{prop:eta-unchanged}).
    \item No symbolic references (Property~\ref{prop:ref-are-conc}).
    \item The Heap and local variable mapping of the concrete state is consistent with the symbolic state (Properties~\ref{prop:heap-consistent} and~\ref{prop:variable-consistent}).
    \item All variables in original statement, that is going to be rewritten, are not indexed (Property~\ref{prop:initial-variables-not-rename}).
    \item Initial local variable mapping $\Delta_s$ cannot have any use for any input variable $i \in I$ in the $Range(\Delta_s)$ (Property~\ref{prop:initial-deltas-no-I}).
\end{itemize}



\section{Java Ranger's Environment}
\label{sec:JREnv}

Java Ranger extends the environment of the DSE with the following environment variables: Path conditions are captured by the tuple $\Pi = (\pi, \pi_r, \pi_s)$, representing the current path condition, the return path condition, and a single-path condition, respectively. Region-specific variables are grouped in $\chi = (I, O, T)$, denoting input, output, and temporary variables.

Finally, the environment includes auxiliary components that are initialized once: a global counter $u \in \mathbb{N}$ used for generating unique identifiers, and a field subscript mapping $P$, which associates reference-field pairs with symbolic terms to track field accesses.

In addition to the DSE assumption, Java Ranger makes the following assumptions about the symbolic state before it starts the summarization:



\subsection{Java Ranger Initial State}

\begin{algorithm}[h!t]
\DontPrintSemicolon
\SetAlgoLined
Input: Concrete State $\delta=(\Theta,\eta,\Delta,H,\nu,\mu)$\\
Input: Symbolic State $(\eta_s,\Delta_s,H_s,\pi)$\\
Assumption Prop~\ref{prop:initial-no-return}: $\mu=\mu_s=\bot$\\
Assumption Prop~\ref{prop:eta-unchanged}:$\eta_s \subseteq \eta$\\
Assumption Prop~\ref{prop:variable-consistent}: $\Delta \models (\Delta_s, \pi)$ \\
Assumption Prop~\ref{prop:heap-consistent}: $(\Delta,H) \models_H (\Delta_s,H_s,\pi)$\\
Assumption Prop~\ref{prop:initial-deltas-no-I}: $\forall i \in \func{getInputs}(s).\forall x \in Dom(\Delta_s).\neg isUse(i,\Delta_s(x))$\\
Output: Path-Merging State $\omega$\\
Data: $\chi=(I,O,T)$\\
Data: $\Pi=(\pi,\pi_r,\pi_s)$\\
 $\Theta_t:=\{((\pi_r',\mu_s'), \overrightarrow{y}.s') \text{ s.t. } (\overrightarrow{y}.s) \in \Theta \wedge ((\pi_r',\mu_s'),s')=transform(s)\}$\\
$\pi_r:=false$\\
$\pi_s:=false$\\
$I := getInput(s)$\\
$O := getOutput(s)$\\
$T := \varnothing $\\
$u := 0$\\
{\bf return } $\omega=(\Theta_t,\eta_s,\Delta_s,H_s, \mu_s, \Pi, \chi, u)$
\caption{Algorithm for {\tt initialize-state} where $getInput(s):=\{x \mid isUse(x,s) \wedge \lnot isDef(x,s)\}$ and $ getOutput(s):=\{x \mid isDef(x,s)\}$}
\label{fig:init-algo}
\end{algorithm}

\begin{lem}
\label{lem:delta-prime-sat}
If
    $$\func{SSA}(s)$$
    $$\Theta \vdash \concsstep{(\eta, \Delta, H, \nu,\mu)}{s}{(\eta', \Delta', H', \nu', \mu')}{} $$
    $$\rangerstepnolookup{((\Theta_t,\eta_s,\Delta_s,H_s, \mu_s, (\pi,\pi_r,\pi_s), \chi, u),s)}{((\Theta_t,\eta_s,\Delta_s',H_s', \mu_s', (\pi',\pi_r',\pi_s'), \chi', u'),s')}{t}$$
    $$\Delta \models (\Delta_s, \pi) $$

then
    $$\Delta' \models (\Delta_s, \pi) $$

\end{lem}

\subsection{Transformation Common Details}
~\label{sec:gamma}
Almost each transformation, except for the remove final-returns, has two statement judgments, one that is recursive, and another is a wrapper that initiate the transformation process. Thus, in the proof we follow the same order, for statements, we first show the 
the soundness of the recursive judgment then discuss the soundness of the wrapper. 

In the remainder of this paper, we present the formal semantics and soundness proofs of main JR's transformations. More precisely, we present eight transformations:
$\gamma$ transformation $(TR_1)$,
early-returns elimination $(TR_2)$,
final-return elimination $(TR_3)$,
renaming transformation $(TR_4)$,
substitution transformation $(TR_5)$,
method inlining transformation $(TR_6)$,
field SSA transformation $(TR_7)$ and finally, 
constant reference propagation $(TR_8)$.

\section{$TR_1$: $\gamma$-Creation Transformation}
\journalreport{
\begin{figure}[H]
    \centering
\begin{align*}
stmt \qquad s &::= s_1 ; s_2 \mid x := e \mid skip \mid \myif{e}{s_1}{s_2}\hcancel[black]{;\phi\text{-}stmt} \mid
 \invoke{x}{z}{\gamma}{y} \mid \return{e}  \\ 
 &\qquad \putfield{z}{f}{e} \mid \getfield{x}{z}{f} \mid \newobj{x}{c} \\
\end{align*}
    \caption{Grammar fragment modified by the $\gamma$-creation transformation}
    \label{fig:gamma-grammar}
\end{figure}}
\techreport{
\begin{figure}[H]
    \centering
        \begin{align*}
   &integer \quad n \in \mathbb{Z} \qquad subscript \quad j \in \mathbb{Z^+}\uplus \{\bot\}  \qquad identifier \quad a,b,c \in Id  \\
&class \quad c \in C \qquad signature \quad g \in G \qquad field \quad f \in F \qquad reference \quad r \in R
    \end{align*}
\begin{align*}
literal\quad v  &::= n \mid true \mid false  \mid null  \qquad
 var \quad x,y,i,o,t ::= (a,j)  \\
method \quad m &::=\regiondef{x}{s} \qquad\qquad\qquad\qquad unary Op \qquad \unaryOp ::= - \mid ^^21   \\
binary Op \quad \binaryOp &::= + \mid - \mid * \mid \div \mid \& \mid bitwise-or \mid \oplus \mid \% \mid == \mid ^^21= \mid \leq \mid \geq \mid \&\& \mid logical\text{-}or \mid >\: \mid <\: \mid \ll \mid \gg \\
exp\quad e  &::=  v \mid x \mid \binaryop{e_1}{e_2} \mid \unaryop{e} \mid  \gammaexp{e}{e_1}{e_2} \\
simple\text{-}ref. \qquad z &::= r \mid x\\
stmt \qquad s &::= s_1 ; s_2 \mid x := e \mid skip \mid \myif{e}{s_1}{s_2}\hcancel[black]{;\phi\text{-}stmt} \mid
 \invoke{x}{z}{\gamma}{y} \mid \return{e}  \\ 
 &\qquad \putfield{z}{f}{e} \mid \getfield{x}{z}{f} \mid \newobj{x}{c} \\
\hcancel[black]{\phi\text{-}stmt \quad phis} &::= \hcancel[black]{\epsilon \mid \phi(x,e_1,e_2);\phi\text{-}stmt}\\
\end{align*}
    \caption{$\gamma$-creation to-grammar}
    \label{fig:gamma-grammar}
\end{figure}
}
The $\gamma$-creation transformation is the first transformation applied by JR, rewriting trailing $\phi$-statements at the end of if-statements as assignments over $\gamma$-expressions, thereby eliminating $\phi$-statements from the program representation (Property~\ref{prop:no-phi}). The resulting grammar is shown in Fig.~\ref{fig:gamma-grammar}.




    

\journalreport{

\begin{figure}[h!t]
    \footnotesize
\fbox{%
\parbox{\textwidth}{%

\[
    \infer[\gamma\rn{-wrap}]
{
\rangerstepnolookup{s}{s'}{{\gamma}-wrap}
}
{
\top \vdash\rangerstepnolookup{s}{s'}{\gamma}
}
\qquad
    \infer[\gamma\rn{-comp}]
     {
     \varphi \vdash\rangerstepnolookup{s_1;s_2}{s_1';s_2'}{\gamma}}
    {
     \begin{gathered}
    \varphi\vdash\rangerstepnolookup{s_1}{s_1'}{\gamma}
        \qquad      \varphi\vdash\rangerstepnolookup{s_2}{s_2'}{\gamma}
    \end{gathered}
    }
\]

\[
\infer[\gamma\rn{-if}]
 {
 \begin{gathered}
 \varphi\vdash\rangerstepnolookup{\myif{e}{s_1}{s_2};\phiexp{x_1}{e_1}{e_1'}; \ldots;\phiexp{x_k}{e_k}{e_k'}\\}
     {\myif{e}{s_1'}{s_2'};x_1:=\gammaexp{\varphi \wedge e}{e_1}{e_1'}; \ldots ; x_k:=\gammaexp{\varphi \wedge e}{e_k}{e_k'}}{\gamma}
 \end{gathered}
 }
 { 
 \begin{gathered}
  \varphi \wedge e \vdash\rangerstepnolookup{s_1}{s_1'}{\gamma}
    \qquad
 \varphi \wedge \neg e \vdash\rangerstepnolookup{s_2}{s_2'}{\gamma}
 \end{gathered}
}
\qquad
  \infer[\gamma\rn{-noOp}]
     {
     \varphi \vdash\rangerstepnolookup{s}{s}{\gamma}}
    {\noop{\gamma}{s}}
\]

}}
\caption{Selected $\gamma$-Creation Rewrite Rules}
\label{fig:gamma-rules}
\end{figure}}
\techreport{
 \begin{figure}[H]
    \footnotesize
\fbox{%
\parbox{\textwidth}{%

\[
    \infer[\gamma\rn{-wrap}]
{
\rangerstepnolookup{s}{s'}{{\gamma}-wrap}
}
{
\top \vdash\rangerstepnolookup{s}{s'}{\gamma}
}
\qquad
    \infer[\gamma\rn{-comp}]
     {
     \varphi \vdash\rangerstepnolookup{s_1;s_2}{s_1';s_2'}{\gamma}}
    {
     \begin{gathered}
    \varphi\vdash\rangerstepnolookup{s_1}{s_1'}{\gamma}
        \qquad      \varphi\vdash\rangerstepnolookup{s_2}{s_2'}{\gamma}
    \end{gathered}
    }
    \]
    \[
     \infer[\gamma\rn{-assign}]
     {
     \varphi \vdash\rangerstepnolookup{x:=e}{x:=e}{\gamma}}
    { }
    \qquad
    \infer[\gamma\rn{-noOp}]
     {
     \varphi \vdash\rangerstepnolookup{s}{s}{\gamma}}
    {\noop{\gamma}{s}}
\]

\[
\infer[\gamma\rn{-if}]
 {
 \begin{gathered}
 \varphi\vdash\rangerstepnolookup{\myif{e}{s_1}{s_2};\phiexp{x_1}{e_1}{e_1'}; \ldots;\phiexp{x_k}{e_k}{e_k'}\\}
     {\myif{e}{s_1'}{s_2'};x_1:=\gammaexp{\varphi \wedge e}{e_1}{e_1'}; \ldots ; x_k:=\gammaexp{\varphi \wedge e}{e_k}{e_k'}}{\gamma}
 \end{gathered}
 }
 { 
 \begin{gathered}
  \varphi \wedge e \vdash\rangerstepnolookup{s_1}{s_1'}{\gamma}
    \qquad
 \varphi \wedge \neg e \vdash\rangerstepnolookup{s_2}{s_2'}{\gamma}
 \end{gathered}
}
\]

}}
\caption{Gamma Creation Rules}
\label{fig:gamma-rules}
\end{figure}

}

Fig.~\ref{fig:gamma-rules} shows the rewrite rules for the $\gamma$-transformation.
The rules define two judgments. The judgment 
$\varphi \vdash \rangerstepnolookup{s}{s'}{\gamma}$ rewrites an IR statement 
$s$ into $s'$, where $\varphi$ is a boolean flag encoding the local path 
condition within the statement, used to guide the rewriting of $\phi$-expressions 
into $\gamma$-expressions. The wrapper judgment 
$\rangerstepnolookup{s}{s'}{{\gamma}\text{-}wrap}$ is the triggering judgment 
that initiates the rewriting, invoking the $\gamma$-judgment on $s$ with an 
initial true path condition.
Also, the $\noop{\gamma}{s}$ is defined as all statements except for if and sequential composition statements.

\techreport{
\noindent$\noop{\gamma}{x:=e} = true$\\
$\noop{\gamma}{\invoke{x}{z}{\gamma}{y}} = true$\\
$\noop{\gamma}{\getfield{x}{z}{f}} = true$\\
$\noop{\gamma}{\putfield{z}{f}{e}} = true$\\
$\noop{\gamma}{\return{e}} = true$\\
$\noop{\gamma}{skip} = true$\\
$\noop{\gamma}{s} = false \qquad$ otherwise \\
}
\begin{lem}
\label{lb:lm-gamma-statement}
    Let $s$ be any statement.

If
$$  
\varphi \vdash \rangerstepnolookup{s}{s'}{\gamma}$$
$$\vdashconcestep{\Delta}{\varphi}{true}$$
$$\Theta \vdash \concsstep{(\eta, \Delta, H, \nu, \mu)}{s}{(\eta', \Delta', H', \nu', \mu')}{} $$
then 
$$\Theta \vdash \concsstep{(\eta, \Delta, H, \nu, \mu)}{s'}{(\eta', \Delta', H', \nu', \mu')}{}$$
\end{lem}

\journalreport{
\begin{proof}
The proof proceeds by structural induction on $s$. Statements satisfying the $\gamma$-{\tt no-op} and $\gamma$-{\tt assign} rules are not rewritten, so the theorem holds trivially. We show the case where the if-condition evaluates to {\tt true}; the remaining cases are analogous and the full proof appears \soha{in~\cite{techreport}}.

Consider the $\gamma$-{\tt if} rule. Applying the rewriting yields:
\begin{align*}
\varphi \vdash \rangerstepnolookup{\myif{e}{s_1}{s_2};\phiexp{x_1}{e_1}{e_1'}; \ldots;\phiexp{x_k}{e_k}{e_k'}}
{\\ \myif{e}{s_1'}{s_2'};x_1:=\gammaexp{\varphi \wedge e}{e_1}{e_1'}; \ldots ; x_k:=\gammaexp{\varphi \wedge e}{e_k}{e_k'}}{\gamma}
\end{align*}
where $\varphi \wedge e \vdash\rangerstepnolookup{s_1}{s_1'}{\gamma}$ and $\varphi \wedge \neg e \vdash\rangerstepnolookup{s_2}{s_2'}{\gamma}$.

Suppose $e$ evaluates to {\tt true} (i.e., $\Theta \vdash\concestep{\Delta}{e}{\mathtt{true}}{}$) and $s_1$ contains no {\tt return} statement. Then that means, the concrete {\tt if-phi-true} rule will evaluate, yielding:
$$(\eta_{s1}',\Delta_1'', H_1', \bot)$$
where $\vdashconcestep{\Delta}{e_i}{v_i}$ for $i=1,\ldots,k$, $\Theta\vdash\concsstep{(\eta_s,\Delta, H, \nu, \mu)}{s_1}{(\eta_{s1}',\Delta_1', H_1', \nu_1', \bot)}{}$ and $\Delta_1''=\Delta_1'[x_1\leftarrow v_1,\ldots, x_k\leftarrow v_k]$. Here, we have that $\eta'=\eta_{s1}', \Delta'=\Delta_1', H'=H_1',\nu'=\nu_1',\mu'=\mu_1'$.

Now, evaluating concretely the rewritten statement which is composed of an if-statement and a series of assignment statement with $\gamma$-expressions.
In this case, observe that $e$ is unchanged from the rewriting and that it evaluates in the same environment $\Delta$, thus, it must evaluate to $true$, and therefore execution must proceeds with $s_1'$. Applying the induction hypothesis on $s_1$ gives $\Theta\vdash\concsstep{(\eta_s,\Delta, H, \nu, \mu)}{s_1'}{(\eta_{s1}',\Delta_1', H_1', \nu_1', \bot)}{}$ and that $SSA(s')$. Next, we evaluate the trailing assignment statements, which are evaluating in the same concrete environment $\Delta$, then their corresponding expression must evaluate to the same values, i.e., $\vdashconcestep{\Delta}{e_i}{v_i}$ for $i=1,\ldots,k$, and therefore we again obtain this environment after executing the assignment statements $(\eta_{s1}',\Delta_1'[x_1\leftarrow v_1,\ldots, x_k\leftarrow v_k], H_1', \nu_1', \bot)=(\eta_{s1}',\Delta_1'', H_1', \bot)$, which is what we want to prove. 

\end{proof}

\begin{cor}
\label{cor:gamma-wrap-sound}

Let $s$ be any statement.

If
$$\rangerstepnolookup{s}{s'}{\gamma-wrap}$$
$$\Theta \vdash \concsstep{(\eta, \Delta, H, \nu, \mu)}{s}{(\eta', \Delta', H', \nu', \mu')}{} $$
then 
$$\Theta \vdash \concsstep{(\eta, \Delta, H, \nu, \mu)}{s'}{(\eta', \Delta', H', \nu', \mu')}{}$$
\end{cor}

\begin{proof}
    The proof follows directly from Lemma~\ref{lb:lm-gamma-statement}.
\end{proof}
}
\techreport{
\begin{proof}
We will proof the lemma using structural induction on $s$.
First, consider the $\gamma$-assign rule and all other statements where gamma transformation does not introduce any changes to the statement $s$. In all of these cases, the theorem statement holds trivially as there is no rewriting, i.e., changes introduced on the statement.

Next, let us consider $\gamma${\tt-if} rule. Applying the rewriting we get 
 \begin{align*}  \varphi \vdash \rangerstepnolookup{\myif{e}{s_1}{s_2};\phiexp{x_1}{e_1}{e_1'}; \ldots;\phiexp{x_k}{e_k}{e_k'}}
     {\\ \myif{e}{s_1'}{s_2'};x_1:=\gammaexp{\varphi \wedge e}{e_1}{e_1'}; \ldots ; x_k:=\gammaexp{\varphi \wedge e}{e_k}{e_k'}}{\gamma}
 \end{align*}

 such that $ \varphi \wedge e \vdash\rangerstepnolookup{s_1}{s_1'}{\gamma}$, and $\varphi \wedge \neg e \vdash\rangerstepnolookup{s_2}{s_2'}{\gamma}$.
 There are four cases in which this statement can execute concretely, depending on whether the condition $e$ evaluates to {\tt true} or {\tt false}, and whether $s_1$ or $s_2$ contains a {\tt return} statement.
Let us first consider the case where $e$ evaluates to {\tt true} ($
  \Theta \vdash\concestep{\Delta}{e}{true}{}$) and where the execution of $s_1$ does not encounter a {\tt return} statement. In this case, the result of evaluating $s$ will be  
${(\eta_s',\Delta'[x_1\leftarrow v_1,\ldots, x_k\leftarrow v_k], H', \bot)}{}
$ (applying the concrete rule {\tt if-phi-true}), where $  \vdashconcestep{\Delta}{e_1}{v_1} \ldots  \vdashconcestep{\Delta}{e_k}{v_k} $, and $\Theta\vdash\concsstep{(\eta_s,\Delta, H, \nu, \mu)}{s_1}{(\eta_s',\Delta', H', \nu', \bot)}{} $.
Now, now to execute the rewritten statement, we know that the expression $e$ must evaluate to the same {\tt true} value, since it has not changed and it is executing in the same environment. This means that the execution will proceed by executing $s_1'$. But, here we can apply the induction hypothesis on $s_1$, which then implies that $\Theta\vdash\concsstep{(\eta_s,\Delta, H, \nu, \mu)}{s_1'}{(\eta_s',\Delta', H', \nu', \bot)}{}$. Similarly we know that Then evaluating the rewritten statement we get $\Theta\vdash ((\eta_s, \Delta, H,\nu, \bot), \myif{e}{s_1'}{s_2'} \Rightarrow_{c} (\eta_s',\Delta', H', \bot)$ using if-true rule. 
Next, the remaining part of the rewritten expression will be evaluated by applying both the {\tt composition} rule, and the {\tt assignment} rule. Observe here that that all expression $e$ and $e_1,..e_k$ must evaluate to the same values concretely since the grammar uses an SSA formate and that means variables within expression cannot be redefined. More precisely, if an expression has evaluated to a particular value at one place in the code, then evaluating the same expression at a later point in the code will produce the same value. But evaluating the expression earlier in the code might not be a valid thing to try.
Moreover, recall that from the hypothesis we have that $\vdashconcestep{\Delta}{\varphi}{true}$. We also know that the $\vdashconcestep{\Delta}{\varphi\wedge e}{true}$, implying that all $\gamma$ expressions will be evaluated to the {\tt true} side, i.e., by applying the $\gamma${-true} expression rule.
This implies that the sequence of assignments will have the same effect on the environment as the if-phi-true rule, i.e., 
$  \Theta\vdash ((\eta_s', \Delta'[x_1 \leftarrow v_1,\ldots, x_{k-1} \leftarrow v_{k-1}], H', \bot), x_k:=\gammaexp{e}{e_k}{e_k'} \Rightarrow_{c} (\eta_s',\Delta'[x_1\leftarrow v_1,\ldots, x_{k-1} \leftarrow v_{k-1}, x_k\leftarrow v_k], H', \bot)$.
Similarly, the proof for the case when the condition within the if-statement evaluates to  {\tt false}, and no return-statement is encountered during the execution of $s_2$, will follow the same reasoning, however, on $s_2$ instead of $s_1$. 

The other two remaining cases for the execution of the if-statement are when the condition of the if-statement evaluates to either {\tt true}, or {\tt false}, but a return-statement is encountered during the evaluation of either $s_1$ or $s_2$. These two cases are handled by the concrete rules {\tt if-phi-true-ret}, and {\tt if-phi-false-ret}. In what follows we show the proof for {\tt if-phi-true-ret}, but {\tt if-phi-false-ret} follows the same argument. 
 In the {\tt if-phi-true-ret} case, the evaluation will be different. The main difference is that, if a return occurs during the execution of the "then" side, then the ${\phi}$s and the gamma assignments will both be skipped.
 Thus, in this case, the resulting state will be 
${(\eta_s',\Delta', H',\nu', \mu')}{}$, when evaluating the if-statement concretely, where $\mu' \neq \bot$, 
$\Theta \vdash\concestep{\Delta}{e}{true}{}$. This means that $s_1'$ will be evaluated. However, we can apply induction on $s_1'$ to obtain ${(\eta_s',\Delta', H',\nu', \mu')}{}$. 
$\Theta\vdash\concsstep{(\eta_s,\Delta, H, \nu, \mu)}{s_1}{(\eta_s',\Delta', H', \nu', \mu')}{}$. But since we have assumed that $\mu' \neq \bot$, then the rewritten if-statement then will evaluate to $(\eta_s',\Delta', H', \nu', \mu')$.

Next, consider $\gamma${\tt-comp} rule. Here we want to In this case, the statement is rewritten as: $ 
\varphi \vdash \rangerstepnolookup{s_1;s_2}{s_1';s_2'}{\gamma}$, where $\varphi\vdash\rangerstepnolookup{s_1}{s_1'}{\gamma}$, and $\varphi\vdash\rangerstepnolookup{s_2}{s_2'}{\gamma}$. Evaluating the original statement we obtain $  \Theta\vdash \concsstep{(\eta, \Delta, H, \nu,\bot)}{s_1';s_2'}{(\eta_s'',\Delta'', H'', \nu'',\mu'')}{}$, where $
      \Theta\vdash \concsstep{(\eta, \Delta, H, \nu,\bot)}{s_1}{(\eta_s', \Delta', H', \nu',\bot)}{}$ and $\Theta\vdash \concsstep{(\eta', \Delta', H', \nu',\bot)}{s_2}{(\eta_s'',\Delta'', H'', \nu'',\mu'')}{}$.
Using induction on $s_1$ we obtain
  $ \Theta\vdash \concsstep{(\eta, \Delta, H, \nu, \bot)}{s_1'}{(\eta_s', \Delta', H', \nu', \bot)}{}$. Now, we can also apply induction on $s_2$ to get 
 $\Theta\vdash \concsstep{(\eta', \Delta', H', \nu', \bot)}{s_2'}{(\eta_s'',\Delta'', H'', \nu'', \mu'')}{}$.
Thus, we can apply the concrete ({\tt composition}) rule on the results of the rewriting to obtain $\Theta\vdash \concsstep{(\eta, \Delta, H, \nu, \bot)}{s_1';s_2'}{(\eta_s'',\Delta'', H'', \nu'', \mu'')}{}$
This shows that evaluating the original statement and the rewritten version both produce the same environment.

Finally, observe that Lemma~\ref{lem:delta-prime-sat} naturally hold, thus proving that $\Delta' \models (\Delta_s, \pi)$.
\end{proof}

\begin{cor}
\label{cor:gamma-wrap-sound}

Let $s$ be any statement.

If
$$     
\rangerstepnolookup{s}{s'}{\gamma-wrap}$$
$$\Theta \vdash \concsstep{(\eta, \Delta, H, \nu, \mu)}{s}{(\eta', \Delta', H', \nu', \mu')}{} $$
then 
$$\Theta \vdash \concsstep{(\eta, \Delta, H, \nu, \mu)}{s'}{(\eta', \Delta', H', \nu', \mu')}{}$$
\end{cor}

\begin{proof}
    The proof follows directly from Lemma~\ref{lb:lm-gamma-statement}.
\end{proof}

}
\section{$TR_2$: Early-Returns Elimination}
\label{sec:earlyreturns}
In this transformation, occurrences of multiple return statements are collapsed into one.
\begin{figure}[h!t]
    \footnotesize
\fbox{%
\parbox{\textwidth}{%

\[ 
    \infer[\rn{ret-wrap}]
     {
  \rangerstepwrapper{s}{\pi_r}{s';\myif{\pi_r}{\return{\mu_s}}{skip}}{ret-wrap}
     }
    {
     \begin{gathered}
 \rangersteprhs{s}{(\pi_r,\mu_s)}{s'}{ret}  \qquad
 \mu_s\neq \bot \qquad \pi_r\neq false
    \end{gathered}
    }
\]

\[
    \infer[\rn{ret-wrap-none}]
     {
  \rangerstepwrapper{s}{\pi_r}{s}{ret-wrap}
     }
    {
     \begin{gathered}
 \rangersteprhs{s}{(\pi_r,\mu_s)}{s}{ret}  \qquad
 \mu_s= \bot \qquad \pi_r = false
    \end{gathered}
    }
\]

\[
    \infer[\rn{ret-if}]
     {\rangersteprhs{\myif{e}{s_1}{s_2}}{(\pi_r^*,\mu_s^*)}{\myif{e}{s_1'}{s_2'}}{ret}}
    {
     \begin{gathered}
        \rangersteprhs{s_1}{(\pi_r,\mu_s)}{s_1'}{ret}
        \qquad
       \rangersteprhs{s_2}{(\pi_r',\mu_s')}{s_2'}{ret} \\
        \pi_r^*={\tt simp}(\gammaexp{e}{\pi_r}{\pi_r'}\qquad
        \mu_s^*={\tt simp}(\gammaexp{e}{\mu_s}{\mu_s'}) \\
    \end{gathered}
    }
\]

\[
    \infer[\rn{ret-return}]
     {
     \rangersteprhs{\return{e}}{(\pi_r, \mu_s)}{skip}{ret}}
    {
    \pi_r=true \qquad
    \mu_s=e
    }
\]

\[
    \infer[\rn{ret-comp}]
     {
  \rangersteprhs{s_1;s_2}{(\pi_r^*,\mu_s^*)}{s_1';if(!\pi_r)\, then \, s_2' \, else \, skip}{ret}
     }
    {
     \begin{gathered}
     \rangersteprhs{s_1}{(\pi_r,\mu_s)}{s_1'}{ret}
        \qquad
 \rangersteprhs{s_2}{(\pi_r',\mu_s')}{s_2'}{ret} \\
    \pi_r^*={\tt simp}(\gammaexp{\pi_r}{true}{\pi_r'}) \qquad
        \mu_s^*={\tt simp}(\gammaexp{\pi_r}{\mu_s}{{\tt simp}(\gammaexp{\pi_r'}{\mu_s'}{false}})) \\
    \end{gathered}
    }
\]
\[
    \infer[\rn{ret-no-op}]
     {\rangersteprhs{s}{(\pi_r, \mu_s)}{s}{ret}}
    { isNoOp_{ret}(s) \qquad \pi_r=false \qquad \mu_s=\bot}
\]

\hrule
\[
 {\tt simp}(\gammaexp{e}{\bot}{\bot})=\bot
\qquad
 {\tt simp}(\gammaexp{e}{false}{false})=false
 \qquad
 {\tt simp}(e)=e \quad \text { otherwise}
\]
}}
\caption{Eliminate Early-Return Statements Rules.}
\label{fig:earlyreturn}
\end{figure}
In this transformation, occurrences of multiple return statements are collapsed into one. Fig.~\ref{fig:earlyreturn} shows the early-returns rules of the early-returns transformation.

The rules define the judgment $\rangersteprhs{s}{(\pi_r, \mu_s)}{s'}{ret}$,
where $s$ is the input statement, $\pi_r$ is a path condition indicating whether
a \texttt{return} was reached, $\mu_s$ is the return value ($\bot$ if none), and
$s'$ is the rewritten statement.
The return-elimination transformation aims to eliminate early \texttt{return}
statements by collapsing all early-return statements into a single return at the end of the rewritten statement. To ensure a sound
transformation, two pieces of information must be captured: the path condition
$\pi_r$ under which a \texttt{return} was reached, and the returned expression
$\mu_s$. These are propagated by subsequent transformations to preserve the
semantics of early returns.

The transformation is triggered by the wrapper judgment
$\rangerstepwrapper{s}{\pi_r}{\cdot}{ret\text{-}wrap}$, which has two cases
depending on whether a \texttt{return} was eliminated in $s$ or not.
Rule~{\tt ret-wrap} applies when a return was found ($\mu_s \neq \bot$ and
$\pi_r \neq false$), appending a conditional
\texttt{if}~$\pi_r$~\texttt{return}~$\mu_s$ to the rewritten statement $s'$. Observe at $s'$ have no return-statements, and that the subsequent if-statement in ($s';\myif{\pi_r}{\return{\mu_s}}{skip}$) adds a single return-statement guarded by a summarized condition ($\pi_r$) and a summarized returned expression ($\mu_s$).
Rule~{\tt ret-wrap-none} applies when no return was found, leaving the statement
unchanged.

Rule~{\tt ret-return} is the base case, handling a $\return{e}$ statement
by setting the return path condition $\pi_r = true$ and the return expression
$\mu_s = e$, replacing the statement with $skip$.

For compound control flow, the transformation must collect return information
from all reachable sub-statements. Rule~{\tt ret-if} handles conditionals by
processing both branches independently and merging their return conditions and
expressions using the $\gamma$ expression. Rule~{\tt ret-comp} similarly processes
both $s_1$ and $s_2$ in a sequential composition, collecting return information
from each side.

In sequential composition, a \texttt{return} may occur in either $s_1$ or $s_2$.
The merged return condition $\pi_r^*$ is
${\tt simp}(\gammaexp{\pi_r}{true}{\pi_r'})$: if $s_1$ returned then $\pi_r^*$
is $true$, otherwise it takes the return condition $\pi_r'$ collected from
rewriting $s_2$. The merged return expression $\mu_s^*$ follows the same spirit:
it is conditioned on $\pi_r$, taking $\mu_s$ when $s_1$ returned, and otherwise
deferring to $\mu_s'$ when $s_2$ returned, defaulting to $false$ when neither
did. The residual guards $s_2'$ with $!\pi_r$ to ensure it executes only if
$s_1$ did not return.
Finally, the {\tt ret-no-op} rule covers other types of statements that the transformation does not change/rewrite.
\techreport{
The $\noop{ret}{s}$ is defined as\\

\noindent$\noop{ret}{x:=e} = true$\\
$\noop{ret}{\invoke{x}{z}{\gamma}{y}} = true$\\
$\noop{ret}{\getfield{x}{z}{f}} = true$\\
$\noop{ret}{\putfield{z}{f}{e}} = true$\\
$\noop{ret}{skip} = true$\\
$\noop{ret}{s} = false \qquad$ otherwise \\
}

The transformation also defines the ${\tt simp}$ function which simplifies $\gamma$ expressions by collapsing degenerate
cases: it returns $\bot$ when both branches are $\bot$, and $false$ when both
branches are $false$, avoiding unnecessary symbolic branching in the collected
return environment.

\subsection{Early Returns Elimination is Sound}

\begin{lem}[Early Returns Recursive Judgment is Sound]
\label{lem:inline-recursive-sound}
    If
$$\Delta \models (\Delta_s, \pi)$$
$$\Theta \vdash \concsstep{(\eta, \Delta, H, \nu, \mu)}{s}{(\eta', \Delta', H', \nu', \mu')}{}$$
$$ \rangersteprhs{s}{(\pi_r,\mu_s)}{s'}{ret} $$
$$\Theta \vdash \concsstep{(\eta, \Delta, H, \nu, \mu)}{s'}{(\eta'', \Delta'', H'', \nu'', \mu'')}{}$$
then
\begin{gather*}
\eta''=\eta',\Delta''=\Delta',H''=H',\nu''=\nu' \qquad \tag{1}\\
\mu''=\bot \qquad \tag{2}\\
\mu' \neq \bot \implies 
\vdashconcestep{\Delta''}{\mu_s}{\mu'} \qquad \tag{3}\\
\mu' \neq \bot 
\text{ iff }
\vdashconcestep{\Delta''}{\pi_r}{true} \qquad \tag{4}\\
\end{gather*}
\techreport{
\begin{proof}
    First, let us consider {\tt ret-no-op} rule. In this case, the rewriting does not change the initial statement, i.e.,  $\Theta_t,\eta_s,\Delta_s,H_s, \chi, u \vdash \rangersteprhs{s}{(\pi, \mu_s)}{s}{ret}$, where $\pi_r=false, \mu_s=\bot$. Here, we will consider the case when $s$ is the assignment statement, i.e., $s=x:=e$, but other statements where $\func{noOp_{ret}}$ predicate is valid will follow the same argument.
    In this case, the concrete semantics will evaluate as $\Theta \vdash \concsstep{(\eta, \Delta, H, \nu, \mu)}{x:=e}{(\eta, \Delta[x \leftarrow v], H, \nu, \mu)}{}$, where $\vdashconcestep{\Delta}{e}{v}$. Here, we have that $\eta'=\eta''=\eta, H'=H''=H, \nu'=\nu''=\nu, \mu'=\mu''=\mu=\bot$, and $\Delta'=\Delta[x\leftarrow v]=\Delta''$ (conclusion 1\&2).  
    
    Since we have $\mu'=\mu=\bot$, then conclusion 3 holds trivially. Finally, we know that $\pi_r=false$, which satisfies conclusion 4.

    Now consider the {\tt ret-return} rule. Here we have $  \rangersteprhs{\return{e}}{(\pi_r, \mu_s)}{skip}{ret}$, such that $ \pi_r=true$ and $\mu_s=e$.
     In this case, the concrete semantics of the original statement will evaluate as $\Theta \vdash \concsstep{(\eta, \Delta, H, \nu, \mu)}{\return{e}}{(\eta, \Delta, H, \nu, v)}{}$, where $\vdashconcestep{\Delta}{e}{v}$.
      On the other hand, evaluating concretely the rewritten statement, we obtain $\Theta \vdash \concsstep{(\eta, \Delta, H, \nu, \mu)}{skip}{(\eta, \Delta, H, \nu, \mu)}{}$.
      Here, we have that $\eta'=\eta''=\eta, H'=H''=H, \nu'=\nu''=\nu, \mu'=v, \mu''=\bot$, and $\Delta'=\Delta''=\Delta$ (conclusion 1\&2).  
       Also, we have that $\mu_s=e$. Now observe that, we have $\mu'\neq\bot$ and that $\vdashconcestep{\Delta''}{e}{v}$ , where $v=\mu'$ (conclusion 3). Finally, we have $\pi_r=true$ satisfying conclusion 4.

%

    Now consider {\tt ret-comp} rule,  $\rangersteprhs{s_1;s_2}{(\pi_r',\mu_s')}{s_1';if(!\pi_r'')\, then \, s_2' \, else \, skip}{ret}$, where $ \rangersteprhs{s_1}{(\pi_r'',\mu_s'')}{s_1'}{ret},
 \rangersteprhs{s_2}{(\pi_r''',\mu_s''')}{s_2'}{ret}$.
 
 Such that $\pi_r'={\tt simp}(\gammaexp{\pi_r''}{true}{\pi_r'''}),
        \mu_s'={\tt simp}(\gammaexp{\pi_r''}{\mu_s''}{{\tt simp}(\gammaexp{\pi_r'''}{\mu_s'''}{false}}))$, and that  $s'=s_1';if(!\pi_r'')\, then \, s_2' \, else \, skip$.
        Now, there are two cases to execute $s=s_1;s_2$ concretely depending on whether evaluation of $s_1$ results in executing a return statement or not. That is if $\Theta\vdash \concsstep{(\eta, \Delta, H, \nu,\mu)}{s_1}{(\eta_1, \Delta_1, H_1, \nu_1,\mu_1)}{}$, then it can be that $\mu_1\neq\bot$ or $\mu_1 = \bot$.
        
        If $\mu_1 = \bot$ then the only rule the could apply from the concrete semantics is the {\tt composition} rule,
        where
        $\Theta\vdash \concsstep{(\eta, \Delta, H, \nu,\mu)}{s_1;s_2}{(\eta',\Delta', H', \nu',\mu')}{}$, where $\Theta\vdash \concsstep{(\eta, \Delta, H, \nu,\mu)}{s_1}{(\eta_1, \Delta_1, H_1, \nu_1,\mu_1)}{}, \mu_1=\bot, \Theta\vdash \concsstep{(\eta_1, \Delta_1, H_1, \nu_1,\mu_1)}{s_2}{(\eta_2,\Delta_2, H_2, \nu_2,\mu_2)}{}$.
        Here we have that $\eta_2=\eta', \Delta_2=\Delta', H_2=H', \nu_2=\nu', \mu_2=\mu'$.

        We now show that executing the rewritten statement ($s_1';if(!\pi_r'')\, then \, s_2' \, else \, skip$) must satisfy the conclusions of the lemma. First, we evaluate $s_1'$, i.e.,  we have $\Theta\vdash \concsstep{(\eta, \Delta, H, \nu,\mu)}{s_1'}{(\eta_1', \Delta_1', H_1', \nu_1',\mu_1')}{}$.
        By applying induction on $s_1$, 
        we know that $\eta_1'=\eta_1, \Delta_1'=\Delta_1, H_1'=H_1, \nu_1'=\nu_1, \mu_1'=\bot$. Also, we have that 
        we have $\mu_1 \neq \bot 
        \text{ iff }
        \vdashconcestep{\Delta_1'}{\pi_r''}{true}$, thus we can conclude $\pi_r''$ cannot be evaluated to $true$, and since we assume that the rewritten statement has executed to completion, then we can conclude that $\pi_r''$ must evaluate to $false$,  i.e.,  $\vdashconcestep{\Delta_1'}{\pi_r''}{false}$.
       
%
        Now, to evaluate $if(!\pi_r'')\, then \, s_2' \, else \, skip$ we must use the concrete rule {\tt if-phi-false}, i.e, $s_2'$ must be evaluated. 
        $$\Theta\vdash \concsstep{(\eta_1', \Delta_1', H_1', \nu_1',\mu_1')}{s_2'}{(\eta_2',\Delta_2', H_2', \nu_2',\mu_2')}{}$$
        in this case we have 
        $\eta''=\eta_2', \Delta''=\Delta_2',H''=H_2',\nu''=\nu_2',\mu''=\mu_2'$. 
        Then, by induction we get
        $\eta''=\eta_2'=\eta_2=\eta', \Delta''=\Delta_2'=\Delta_2=\Delta',H''=H_2'=H_2=H',\nu''=\nu_2'=\nu_2=\nu'$ (conclusion 1), and $\mu''=\mu_2'=\bot$ (conclusion 2).  Also we have from the induction that 
        $$\mu_2 \neq \bot \implies \vdashconcestep{\Delta_2'}{\mu_s'''}{\mu_2}$$
        $$\mu_2 \neq \bot 
        \text{ iff } \vdashconcestep{\Delta_2'}{\pi_r'''}{true}$$

        Now, we want to show that this also hold
        $$\mu_2 \neq \bot \implies \vdashconcestep{\Delta_2'}{\mu_s'}{\mu_2}$$
        $$\mu_2 \neq \bot 
        \text{ iff } \vdashconcestep{\Delta_2'}{\pi_r'}{true}$$
        
        Consider the case where $\mu'=\mu_2=\bot$, then conclusion 3 holds trivially, i.e., $\mu_2 \neq \bot \implies \vdashconcestep{\Delta_2'}{\mu_s'}{\mu_2}$. 
        
        On the other hand, since $\pi_r'={\tt simp}(\gammaexp{\pi_r''}{true}{\pi_r'''})=\gammaexp{\pi_r''}{true}{\pi_r'''}$. From Lemma~\ref{lem:delta-ssa} that $\Delta_1'\subseteq\Delta_2'$, and since we have shown that $\vdashconcestep{\Delta_1'}{\pi_r''}{false}$, then we can conclude that evaluating $\gammaexp{\pi_r''}{true}{\pi_r'''}$ in $\Delta_2'$, must results in evaluating $\pi_r'''$, but from induction we just shown that $\mu_2 \neq \bot 
        \text{ iff } \vdashconcestep{\Delta_2'}{\pi_r'''}{true}$ holds, which implies that $\pi_r'''$ must evaluate in $\Delta_2'$, i.e., $\mu_2 \neq \bot 
        \text{ iff } \vdashconcestep{\Delta_2'}{\pi_r'}{true}$ (conclusion 4).
        
        On the other hand, if $\mu'=\mu_2\neq \bot$,
        %
        then, we can conclude that $\vdashconcestep{\Delta_2'}{\pi_r'''}{true}$, and $\vdashconcestep{\Delta_2'}{\mu_s'''}{\mu_2}$. This means that the expression $\mu_s'=\gammaexp{\pi_r''}{\mu_s''}{\gammaexp{\pi_r'''}{\mu_s'''}{false}}$ will reduce, using concrete rules {\tt $\gamma$-true} and {\tt $\gamma$-false}, to evaluating $\mu_s'''$ in $\Delta_2'$, i.e., $\vdashconcestep{\Delta_2'}{\mu_s'''}{\mu'}$ (conclusion 3). Similarly, we know that the expression $\pi_r'''=\gammaexp{\pi_r'}{true}{\pi_r''}$, must reduce to $true$, since by induction, it must be the case that $\pi_r''$ evaluated to $true$ in $\Delta_2'$ (conclusion 4).

        Finally, the case where we have $\mu_1 \neq \bot$, then the concrete rule {\tt composition-ret} will apply. Here we have that $\eta'=\eta_1, \Delta'=\Delta_1, H'=H_1,\nu'=\nu_1, \mu'=\mu_1$. Similarly we have $\eta''=\eta_1', \Delta''=\Delta_1', H''=H_1',\nu''=\nu_1', \mu''=\mu_1'$.
        
        By induction on $s_1$, we know that it must be that 
         $$\mu_1 \neq \bot \implies \vdashconcestep{\Delta_1'}{\mu_s''}{\mu_1}$$
        $$\mu_1 \neq \bot 
        \text{ iff } \vdashconcestep{\Delta_1'}{\pi_r''}{true}$$

        since we have that $\mu_1 \neq \bot$ then $\pi_r''$ must evaluate to $true$ in $\Delta_1'$. Now evaluating the second part of $s'$, i.e., the if-statement $if(!\pi_r'')\, then \, s_2' \, else \, skip$, must result in evaluating the $skip$ statement, yielding the same concrete state, i.e., 
        $$\Theta\vdash \concsstep{(\eta_1', \Delta_1', H_1', \nu_1',\mu_1')}{if(!\pi_r'')\, then \, s_2' \, else \, skip}{(\eta_1',\Delta_1', H_1', \nu_1',\mu_1')}{}$$

        Which from induction satisfies the conclusions of the lemma. 

        Now, consider the {\tt ret-if} rule, we have, ${\rangersteprhs{\myif{e}{s_1}{s_2}}{(\pi_r',\mu_s')}{\myif{e}{s_1'}{s_2'}}{ret}}$, where $\rangersteprhs{s_1}{(\pi_r'',\mu_s'')}{s_1'}{ret}, 
       \rangersteprhs{s_2}{(\pi_r''',\mu_s''')}{s_2'}{ret}$, $
        \pi_r'={\tt simp}(\gammaexp{e}{\pi_r''}{\pi_r'''},
        \mu_s'={\tt simp}(\gammaexp{e}{\mu_s''}{\mu_s'''})$.
        
        There are two cases for evaluating the if-statement depending on whether the expression $e$ evaluates to true or not. We will show the proof when $e$ evaluates to $true$, but the other case follows similar reasoning.
        Thus, here we have $\vdashconcestep{\Delta}{e}{true}$, the concrete semantics must evaluate $s_1$, $\Theta\vdash\concsstep{(\eta,\Delta, H, \nu,\mu)}{s_1}{(\eta_1,\Delta_1, H_1, \nu_1,\mu_1)}{}$.
        In this case, there are two concrete evaluation rules that can apply: {\tt if-phi-true} or {\tt if-phi-true-ret}, depending on whether $\mu_1=\bot$ or not. 
        Since no phi-assignment case exist in the early return grammar, then we know that $\eta=\eta',\Delta_1=\Delta', H_1=H',\nu_1=\nu',\mu_1=\mu'$

        Now, executing the re-written statement $\myif{e}{s_1'}{s_2'}$, must initially start by evaluating $e$, since $e$ has not changed by the rewritten, and is evaluated in the same environment $Delta$, then it must evaluate to the same value, $true$. Consequently, $s_1'$ is be evaluated next, that is  $\Theta\vdash\concsstep{(\eta,\Delta, H, \nu,\mu)}{s_1'}{(\eta_1',\Delta_1', H_1', \nu_1',\mu_1')}{}$. Here we have that $\eta''=\eta_1',\Delta''=\Delta_1',H''=H_1',\nu''=\nu_1',\mu''=\mu_1'$.
        
        By induction, we know that 
        $\eta''=\eta_1'=\eta_1=\eta',\Delta''=\Delta_1'=\Delta_1=\Delta',H''=H_1'=H_1==H',\nu''=\nu_1'=\nu_1=\nu'$ (conclusion 1) and $\mu''=\bot$ (conclusion 2). Also, by the induction we have 
        $$\mu_1 \neq \bot \implies \vdashconcestep{\Delta_1'}{\mu_s''}{\mu_1}$$
        $$\mu_1 \neq \bot 
        \text{ iff }
        \vdashconcestep{\Delta_1'}{\pi_r''}{true}$$.
        
        At this point we want to show that 

         $$\mu_1 \neq \bot \implies \vdashconcestep{\Delta_1'}{\mu_s'}{\mu_1}$$
        $$\mu_1 \neq \bot 
        \text{ iff }
        \vdashconcestep{\Delta_1'}{\pi_r'}{true}$$.

        In the case where $\mu_1=\bot$, $\mu_1 \neq \bot \implies \vdashconcestep{\Delta_1'}{\mu_s'}{\mu_1}$ holds trivially. Also, we know that if $\mu_1=\bot$, then $\pi_r''$ cannot evaluate to $true$, that is $\vdashconcestep{\Delta_1'}{\pi_r''}{true}$ cannot happen. From Lemma~\ref{lem:delta-ssa}, we have that $\Delta \subseteq \Delta_1'$, which implies that the expression $e$ must still evaluate to $true$ when evaluated in $\Delta_1'$.
        Now evaluating $\pi_r'={\tt simp}(\gammaexp{e}{\pi_r''}{\pi_r'''}$ in $\Delta_1'$ must reduce to evaluating since $\pi_r''$ which we have shown that it cannot evaluate to $true$ (conclusion 4).

        On the other hand consider the case where $\mu_1\neq \bot$. 
        Here, evaluating $\mu_s'$ defined as
        $\mu_s'={\tt simp}(\gammaexp{e}{\mu_s''}{\mu_s'''})$, must reduce to evaluating $\mu_s''$ but we have already shown by the induction that $\mu_1 \neq \bot \implies \vdashconcestep{\Delta_1'}{\mu_s''}{\mu'}$ (conclusion 3).
        Similarly, evaluating $\pi_r'$ in $\Delta_1'$, defined as $\pi_r'={\tt simp}(\gammaexp{e}{\pi_r''}{\pi_r'''}$, which reduces to evaluating $\pi_r''$, which we have already shown that it holds by the induction (conclusion 4).

\end{proof}
}
\journalreport{
\begin{proof}[Proof Sketch]
We show the two interesting cases: {\tt ret-comp} and {\tt ret-if}. The 
remaining cases are deferred to the technical report~\cite{techreport}.

\textit{Case} {\tt ret-comp}: The original statement $s_1;s_2$ is rewritten to 
$s_1';\ \texttt{if}(!\pi_r'')\ \texttt{then}\ s_2'\ \texttt{else}\ \texttt{skip}$. 
There are two subcases depending on whether $s_1$ executes a return. If 
$\mu_1 = \bot$, then $s_2$ executes in the original, and since $\pi_r''$ 
evaluates to $\texttt{false}$ by the induction hypothesis on $s_1$, the guard 
$!\pi_r''$ is $\texttt{true}$ and $s_2'$ executes in the rewrite. Conclusions~(1) 
and~(2) follow by induction on $s_1$ and $s_2$. For conclusions~(3) and~(4), 
since $\vdashconcestep{\Delta_1'}{\pi_r''}{\texttt{false}}$, evaluating 
$\pi_r' = \gammaexp{\pi_r''}{true}{\pi_r'''}$ reduces to $\pi_r'''$, and 
similarly $\mu_s'$ reduces to $\gammaexp{\pi_r'''}{\mu_s'''}{false}$, both of 
which hold by induction on $s_2$. If $\mu_1 \neq \bot$, then 
\texttt{composition-ret} applies, $s_2$ is skipped in the original, and since 
$\pi_r''$ evaluates to $\texttt{true}$, the guard $!\pi_r''$ is $\texttt{false}$ 
and \texttt{skip} executes in the rewrite, yielding the same state. All 
conclusions follow directly from the induction hypothesis on $s_1$.

\textit{Case} {\tt ret-if}: The conditional $\myif{e}{s_1}{s_2}$ is rewritten 
to $\myif{e}{s_1'}{s_2'}$. We show the case where $e$ evaluates to 
$\texttt{true}$; the $\texttt{false}$ case follows symmetrically. Since $e$ is 
unchanged by the rewriting, it evaluates to $\texttt{true}$ under both $\Delta$ 
and $\Delta'$ by Lemma~\ref{lem:delta-ssa}, so $s_1$ and $s_1'$ are selected 
in both executions. Conclusions~(1) and~(2) follow by induction on $s_1$. For 
conclusions~(3) and~(4), since $\vdashconcestep{\Delta_1'}{e}{\texttt{true}}$, 
evaluating $\pi_r' = \gammaexp{e}{\pi_r''}{\pi_r'''}$ and 
$\mu_s' = \gammaexp{e}{\mu_s''}{\mu_s'''}$ both reduce to their left branches 
$\pi_r''$ and $\mu_s''$ respectively, which hold by induction on $s_1$.
\end{proof}

}
\end{lem}

\begin{thm}
\label{thm:return-wrap-sound}
    Let $s$ be a statement.

If
$$\Theta \vdash ((\eta,\Delta,H,\mu),s) \Rightarrow_c (\eta',\Delta',H',\mu')$$
$$ \rangerstepwrapper{s}{\pi_r}{s'}{ret-wrap}$$
$$\Theta \vdash ((\eta,\Delta,H,\mu),s') \Rightarrow_c (\eta'',\Delta'',H'',\mu'')$$
then
$$\eta''=\eta',\Delta''=\Delta',H''=H',\nu''=\nu',\mu''=\mu'$$

\end{thm}

\subsection{Properties}
The early-return elimination transformation assumes all the guarantees by the $\gamma$-creation transformation, and it provides the following guarantees.

Given $\rangerstepwrapper{s}{\pi_r}{\hat{s}}{ret-wrap}$,
this transformation guarantees the following properties
\begin{itemize}
    \item Rewritten Statement $\hat{s}$ satisfies the SSA (property~\ref{prop:ssa}).
    \techreport{
        \begin{proof}
           This is clear since none of the rules in Fig.~\ref{fig:earlyreturn} modify assignments within generalized assignment statements, update local/heap mappings, or introduce new variables.
        \end{proof}
        }
    \item Rewritten statement always has at most a single return within an if-statement (Property~\ref{prop:early-at-most-final-if-return}).
    \item If the original statement $s$ has a return statement, then there must be a valid return-condition, and the rewritten statement must be of a specific form (Property~\ref{prop:final-s-has-a-specific-form}).
    \item Return condition ($\pi_r$) is true on all return paths (Property~\ref{prop:final-pi-r-iff-mu-s}).
    \techreport{
        \begin{proof}
        This can be directly be shown by applying the recursive theorem of the {\tt ret} judgment.
        \end{proof}
        }    
\end{itemize}
\section{$TR_3$: Remove Final Return}
\label{sec:finalreturn}
\techreport{
\begin{figure}[H]
    \centering
          \begin{align*}
   &integer \quad n \in \mathbb{Z} \qquad subscript \quad j \in \mathbb{Z^+}\uplus \{\bot\}  \qquad identifier \quad a,b,c \in Id  \\
&class \quad c \in C \qquad signature \quad g \in G \qquad field \quad f \in F \qquad reference \quad r \in R
    \end{align*}
\begin{align*}
literal\quad v  &::= n \mid true \mid false  \mid null  \qquad
 var \quad x,y,i,o,t ::= (a,j)  \\
method \quad m &::=\regiondef{x}{s} \qquad\qquad\qquad\qquad unary Op \qquad \unaryOp ::= - \mid ^^21   \\
binary Op \quad \binaryOp &::= + \mid - \mid * \mid \div \mid \& \mid bitwise-or \mid \oplus \mid \% \mid == \mid ^^21= \mid \leq \mid \geq \mid \&\& \mid logical\text{-}or \mid >\: \mid <\: \mid \ll \mid \gg \\
exp\quad e  &::=  v \mid x \mid \binaryop{e_1}{e_2} \mid \unaryop{e} \mid  \gammaexp{e}{e_1}{e_2} \\
simple\text{-}ref. \qquad z &::= r \mid x\\
stmt \qquad s &::= s_1 ; s_2 \mid x := e \mid skip \mid \myif{e}{s_1}{s_2} \mid
\invoke{x}{z}{g}{y} \mid \hcancel[black]{\return{e}}  \\ 
 &\qquad \putfield{z}{f}{e} \mid \getfield{x}{z}{f} \mid \newobj{x}{c} \\
\end{align*}
    \caption{Remove Final Return to-grammar}
    \label{fig:inline-grammar}
\end{figure}

}
\journalreport{
\begin{figure}[H]
    \centering
          \begin{align*}
stmt \qquad s &::= s_1 ; s_2 \mid x := e \mid skip \mid \myif{e}{s_1}{s_2} \mid
\invoke{x}{z}{g}{y} \mid \hcancel[black]{\return{e}}  \\ 
 &\qquad \putfield{z}{f}{e} \mid \getfield{x}{z}{f} \mid \newobj{x}{c} \\
\end{align*}
    \caption{Grammar fragment modified by the final-returns transformation}
    \label{fig:inline-grammar}
\end{figure}

}
 \begin{figure}[h!t]
 \centering
    \footnotesize
\fbox{%
\parbox{\textwidth}{%
\[
    \infer[\rn{final-noop}]
     {
     \pi_r \vdash \rangerstepnolookup{s}{(\bot,s)}{final}}
    {\begin{gathered}
    \pi_r = false 
    \end{gathered}}
    \quad
    \infer[\rn{final-ret}]
     {
     \pi_r \vdash \rangerstepnolookup{s}{(\mu_s,s')}{final}}
    {\begin{gathered}
    \pi_r \neq false \quad s=s';\myif{\pi_r}{\return{e}}{skip} \quad \mu_s=e\\
    \end{gathered}}
\]
}}
\caption{Remove Final Return Rules.}
\label{fig:finalreturn-rules}
\end{figure}
The early-return elimination transformation does three things: it captures the 
return condition $\pi_r$ and the return expression $\mu_s$ into the environment, 
and appends a single \texttt{return} statement guarded by $\pi_r$ of the form 
$\myif{\pi_r}{\return{\mu_s}}{skip}$, replacing all early returns. 
The final-return elimination transformation then removes this guarded 
\texttt{return} statement, as $\pi_r$ and $\mu_s$ in the environment already 
capture sufficient information about the return for subsequent transformations. Fig.~\ref{fig:finalreturn-rules} shows the two rules of this transformation. Rule~{\tt final-noop} handles the case where $\pi_r = false$, indicating that no 
early return was eliminated, and thus no guarded return is present; the statement 
is left unchanged and $\mu_s$ is set to $\bot$. Rule~{\tt final-ret} handles the 
case where $\pi_r \neq false$, using pattern matching to identify the guarded 
$\myif{\pi_r}{\return{e}}{skip}$ appended by the early-return elimination 
transformation, stripping it from the statement and recording $\mu_s = e$.












\subsection{Remove Final Return Transformation is Sound}

Let us define
$evalReturn(\Delta,\pi_r,\mu_s)\triangleq$
$\begin{cases}
    v & \text{if } \concestep{\Delta}{\pi_r}{true}{} \wedge \concestep{\Delta}{\mu_s}{v}{} \\
    \bot & \text{otherwise}
\end{cases}$

\begin{lem}[evalReturn is not changed by extension of the environment]
\label{lem:evalreturn-unchanged-by-extension}

    $$\func{evalReturn}(\Delta,\pi_r,\mu_s)=\func{evalReturn}(\Delta',\pi_r,\mu_s)$$
    
    where $\Delta \subseteq \Delta'$

\techreport{
    \begin{proof}
    \end{proof}
    }
\end{lem}

\begin{cor}[Remove final-return transformation is sound]
\label{cor:final-ret-sound}
    If for any statement $s$ and  a rewriting state 
    $(\Theta_t,\eta_s,\Delta_s, H_s, \mu_s, (\pi,\pi_r,\pi_s), (I,O,T), u)$,

$$\Theta \vdash \concsstep{(\eta, \Delta, H, \nu,\mu)}{s}{(\eta', \Delta', H', \nu', \mu')}{} $$
$$\rangerstepnolookup{s}{(\mu_s,s')}{remove-ret} $$
$$\Theta \vdash \concsstep{(\eta, \Delta, H, \nu,\mu)}{s'}{(\eta'', \Delta'', H'', \nu'',\mu'')}{}$$

then 

$$\eta''=\eta',\Delta''=\Delta',H''=H',\nu''=\nu',\mu''=\bot$$
$$
\mu' = evalReturn(\Delta'',\pi_r,\mu_s) 
$$
\techreport{
\begin{proof}
   Since there are only two rewrite rules that can apply, we will proceed with our proof by considering the two possible cases for $\pi_r$: when it is the formula $false$, and when it is not.
First, consider the case when $\pi_r$ is $false$. According to assumed early-return property(Property~\ref{prop:final-pi-r-iff-mu-s}), $\mu'$ must be $\bot$, which also means that $s$ contains no return statements. 
Otherwise, if a control flow existed with $\mu' \neq \bot$ but $\pi_r = false$, this would violate the early-return (Property~\ref{prop:final-pi-r-iff-mu-s}). 
Thus, $s$ cannot be of the form $s'';\myif{\pi_r}{\return{e}}{skip}$, and the {\tt final-noop} rule applies without modifying $s$. 
Hence, $s' = s$ and $\mu' = \bot$. Consequently, executing $s'$ yields the same state, i.e., $\eta''=\eta'$, $\Delta''=\Delta'$, $H''=H'$, $\nu''=\nu'$, and $\mu''=\mu=\bot$ (conclusion~2). 
Finally, in this case, the second branch of $evalReturn$ applies, giving $evalReturn(\Delta'', \pi_r, \mu_s')=\bot=\mu'$ (conclusion~3).

    The second case is when $\pi_r \neq false$. 
From the assumed property of early-returns transformation (Property~\ref{prop:final-s-has-a-specific-form}), $s$ must be of the form $$s = s'';\myif{\pi_r}{\return{e}}{skip}$$
Executing $s$ concretely first applies {\tt composition}, evaluating $s''$ as 
$\Theta \vdash \concsstep{(\eta, \Delta, H, \nu,\mu)}{s''}{(\eta''', \Delta''', H''', \nu''',\mu''')}{}$, 
where $\mu'''=\bot$, followed by evaluating the conditional using either {\tt if-phi-true} or {\tt if-phi-false}, depending on $\pi_r$. 

If $\pi_r$ evaluates to $true$, we apply {\tt if-phi-true} (since $\neg \func{hasReturn}(s'')$, {\tt if-phi-true-ret} cannot apply), yielding
$$\Theta\vdash\concsstep{(\eta''', \Delta''', H''', \nu''',\mu'''=\bot)}{\myif{\pi_r}{\return{e}}{skip}}{(\eta''',\Delta''', H''', \nu''',\mu^4)}{}$$
where $\eta'=\eta'''$, $\Delta'=\Delta'''$, $H'=H'''$, $\nu'=\nu'''$, and $\mu'=\mu^4=v$, with $\vdashconcestep{\Delta'''}{e}{v}$.

Since $\pi_r \neq false$ and $s = s'';\myif{\pi_r}{\return{e}}{skip}$, the {\tt final-ret} rule applies, giving $s'=s''$ and $\mu_s=e$. 
Executing $s'$ gives $\Theta \vdash \concsstep{(\eta, \Delta, H, \nu,\mu)}{s''}{(\eta''', \Delta''', H''', \nu''',\mu''')}{}$, 
with $\eta''=\eta'''=\eta'$, $\Delta''=\Delta'''=\Delta'$, $H''=H'''=H'$, $\nu''=\nu'''=\nu'$, and $\mu'''=\mu=\bot$ (conclusion~2). 
Because $\pi_r \neq false$, the first condition of $\func{evalReturn}$ applies. 
Since $\Delta''=\Delta'''$, we have $\vdashconcestep{\Delta'''}{e}{v_1}$ with $v_1=v=\mu_4=\mu$ (conclusion~3).

If, instead, $\pi_r$ evaluates to $false$, {\tt if-phi-false} applies, executing $skip$ and yielding
$$\Theta\vdash\concsstep{(\eta''', \Delta''', H''', \nu''',\mu'''=\bot)}{\myif{\pi_r}{\return{e}}{skip}}{(\eta''',\Delta''', H''', \nu''',\mu^4)}{}$$
such that $\eta'=\eta'''$, $\Delta'=\Delta'''$, $H'=H'''$, $\nu'=\nu'''$, and $\mu'''=\mu^4=\bot=\mu'$, with $\vdashconcestep{\Delta'''}{\pi_r}{false}$. 
Executing $s'$ again yields $\eta''=\eta'''=\eta'$, $\Delta''=\Delta'''=\Delta'$, $H''=H'''=H'$, $\nu''=\nu'''=\nu'$, and since $\neg \func{hasReturn}(s'')$, $\mu'''=\mu=\bot$ (conclusion~2). 
Finally, because $\pi_r = false$, the second condition of $\func{evalReturn}$ applies, giving $\func{evalReturn}(\Delta'',\pi_r,\mu_s)=\bot=\mu'$ (conclusion~3).

\end{proof}

}
\journalreport{
\begin{proof}
Due to the simplicity of this transformation we skip the details of its proof, interested reader is referred to our technical report~\cite{techreport} for more details.    
\end{proof}

}
\end{cor}

\subsection{Properties of the Final-Returns Transformation}
The remove final-return transformation assumes all properties guaranteed by the early-return elimination transformation, and it provide the following guarantees.
Given $\pi_r \vdash \rangerstepnolookup{s}{(\mu_s,s')}{final-wrap}$, we have that

\begin{itemize}
    \item Rewritten statement $s'$ satisfies the SSA property (Property~\ref{prop:ssa}).
    \techreport{
    More formally, $SSA(s')$ holds.
    \begin{proof}
    This is straightforward: since $s$ satisfies the SSA property and none of the rules in Fig.~\ref{fig:finalreturn-rules} introduce new variables or update existing ones, it follows that $isSSA(s')$.
    \end{proof}
    }
    \item Statement $s'$ has no return statements (Property~\ref{prop:final-ret-no-return}).
    \techreport{
    More formally, $\neg hasReturn(s')$ is valid.
    \begin{proof}
        We know that the only case where a {\tt return} statement can exist in $s$ is when it is of the from $s=\hat{s};\myif{\pi_r}{\return{e}}{skip}$, and that $\pi_r \neq false$ (Property~\ref{prop:final-s-has-a-specific-form}). This means that $\hat{s}$ cannot contain another {\tt return} statement because 
        $\func{atMostSingleIfReturn(s)}$ would be violated.
        Therefore, when  we can only apply {\tt final-ret} rule, which removes the enclosed {\tt return} statement on the then-side, and rewrites the statement $s$ to $\hat{s}$, where $\hat{s}=s'$, which we have just shown that it cannot contain a {\tt return} statement, i.e., $\neg \func{hasReturn}(s')$.
    \end{proof}
    }    
\end{itemize}

\section{$TR_4$: Renaming Transformation}
\label{sec:renaming}

This transformation assigns a unique index to every local variable in the 
statement $s$. Specifically, each variable of the form $(a, \bot)$ is renamed 
to $(a, j)$, where $j \in \mathbb{Z}^+$, ensuring that 
distinct occurrences of the same variable across multiple summarizations of $s$, i.e., within loops or recursive calls, are uniquely 
distinguished. 

Fig.~\ref{fig:renaming-rules} shows a snippet of the main renaming rules. 
The rules define two judgments. The judgment 
$u \vdash \rangerstepnolookup{s}{s'}{rename}$ rewrites an IR statement $s$ 
into $s'$, where $u$ is the global unique index used to assign fresh version 
numbers to unversioned variables. The wrapper judgment 
$I, O \vdash \rangerstep{(\pi_r, \mu_s, T, u)}{s}{(\mu_s', \pi_r', T', u')}{s'}{rename\text{-}wrap}$ 
is the entry point of the transformation, triggering the rename judgment on 
$s$. Here, $I$ and $O$ are the sets of input and output variables, $\pi_r$ 
and $\mu_s$ are the early-return path condition and symbolic return expression, 
and $T$ is the set of temporary variables. Their primed counterparts 
$\pi_r'$, $\mu_s'$, $T'$, and $u'$ reflect their updated values after the 
rewriting.

For example, rule~{\tt rename-var} renames a variable $x$ to $x'$ by applying 
$\mathtt{unique}_u(x)$, which assigns the current counter $u$ as the version 
index if $x$ has not yet been renamed, i.e., if it is of the form $(a, \bot)$. 
If $x$ already carries a non-$\bot$ index, $\mathtt{unique}_u$ leaves it 
unchanged, as captured by the side condition $j \neq \bot$ in the definition 
of $\mathtt{unique}_u$ at the top of Fig.~\ref{fig:renaming-rules}. This 
behavior is deliberate: it allows $\rangerexpstepnostate{e}{e'}{rename}$ to be 
applied over partially renamed statements, skipping variables that have already 
been versioned and acting only on those that remain unversioned, i.e., those 
carrying $\bot$.

Rule~{\tt rename-unaryOp} propagates the renaming into the operand $e$, deriving 
$u \vdash \rangerexpstepnostate{e}{e'}{rename}$ in the premise, and reconstructs 
$\unaryop{e'}$ around the renamed result. The rule does not itself rename anything 
directly but serves as a structural rule that drives the renaming relation into 
sub-expressions. Rules for other expression forms, including binary operators and 
$\gamma$-expressions, are omitted from Fig.~\ref{fig:renaming-rules} for 
brevity, but follow the same pattern of propagating 
renaming onto their constituent sub-expressions.

The renaming transformation is propagated structurally over statements. 
Rule~{\tt rename-assign} handles assignment by renaming both the right-hand side 
expression $e$ and the left-hand side variable $x$ independently, producing the 
renamed assignment $x' := e'$. Rule~{\tt rename-comp} handles sequential 
composition by propagating the renaming relation into each sub-statement 
independently under the same counter $u$, yielding $s_1'; s_2'$.

Other statement forms follow the same pattern. Rules for conditionals, 
{\tt rename-putfield}, and {\tt rename-invoke}, etc., all propagate renaming into their sub-statements and invoke renaming on their composing expressions. Notably, {\tt rename-invoke} renames the variables and expressions appearing in the call site without inlining the body of the invoked method; the inlining of method bodies is deferred to a dedicated later transformation.

Finally, the rule~{\tt rename-wrap} extends the renaming transformation to environment 
variables that live outside the statement but must be consistently renamed to 
match their new names within it ( $\pi_r$ and $\mu_s$ are renamed 
to $\pi_r'$ and $\mu_s'$). Using the same counter $u' = u+1$, it renames 
the input/output sets $I$, $O$, and the temporary set $T$ via 
$\mathtt{unique}_{u'}$, and applies $u' \vdash \rangerstepnolookup{s}{s'}{rename}$ 
to the statement, ensuring all unversioned variables across both the environment 
and the statement receive the same index. To preserve the external interface of 
$s$, the rewritten statement $s''$ wraps $s'$ with assignments that bind the 
old input names to their renamed counterparts before $s'$, and bind the renamed 
output variables back to the old output names after, ensuring $s''$ exposes the 
same input/output interface as $s$. 

\techreport{
\begin{figure}[h!t]
    \footnotesize
\fbox{%
\parbox{\textwidth}{%
\[
\begin{gathered}
    \mathtt{unique}_u((I,O,T)) = (\mathtt{unique}_u(I), \mathtt{unique}_u(O), \mathtt{unique}_u(T))\\
    \mathtt{unique}_u(X) = \{\mathtt{unique}_u(x) \mid x \in X\} \qquad \mathtt{unique}_u((a,\bot))=(a,u)
    \qquad \mathtt{unique}_u((a,j))=(a,j) \text{\; when\; } j\ne\bot\\
\end{gathered}
\]
\hrule
\hrule
\hrule
\[
  \infer[\rn{rename-wrap}]
     {
     I,O\vdash\rangerstep{(\pi_r,\mu_s,T,u)}{s}{ \mu_s',\pi_r',T'', u')}{s''}{rename-wrap}
     }
    { 
    \begin{gathered}
        u'=u+1 \qquad 
        \chi=(I,O,T) \qquad \chi'=(I',O',T') \qquad
        \chi'=\mathtt{unique}_{u'}(\chi)
        \\
         I=\{i_1..i_n\} \qquad O=\{o_1..o_m\} \qquad u' \vdash \rangerstepnolookup{s}{s'}{rename} 
        \\
        s''= (i_1,u'):=i_1;..;(i_n, u'):=i_n;s';o_1:=(o_1, u');..;o_m:=(o_m,u');\\
        u' \vdash \rangerexpstepnostate{\mu_s}{\mu_s'}{rename}
        \qquad u' \vdash \rangerexpstepnostate{\pi_r}{\pi_r'}{rename} \qquad
        T''=(I'\cup O' \cup T')
    \end{gathered}
    }
\]
\hrule
\hrule
\hrule
\[
    \infer[\rn{rename-bot}]
     {
         u \vdash \rangerexpstepnostate{\bot}{\bot}{rename}
     }
    { }
\qquad
    \infer[\rn{rename-val}]
     {
         u \vdash \rangerexpstepnostate{v}{v}{rename}
     }
    { 
    v:V_s
    }
\qquad
        \infer[\rn{rename-unaryOp}]
     {
        u \vdash \rangerexpstepnostate{\unaryop{e}}{\unaryop{e}}{rename}
     }
    { 
    \begin{gathered}
    u \vdash \rangerexpstepnostate{e}{e'}{rename}
    \end{gathered}
    }
\]
\[
\infer[\rn{rename-var}]
     {
         u \vdash \rangerexpstepnostate{x}{x'}{rename}
     }
    { 
    \begin{gathered}
     x'=\mathtt{unique}_u(x)
    \end{gathered}
    }
\qquad
    \infer[\rn{rename-binaryOp}]
     {
        u \vdash \rangerexpstepnostate{\binaryop{e_1}{e_2}}{\binaryop{e_1'}{e_2'}}{rename}
     }
    { 
    \begin{gathered}
    u \vdash\rangerexpstepnostate{e_1}{e_1'}{rename}\qquad
    u \vdash\rangerexpstepnostate{e_2}{e_2'}{rename}
    \end{gathered}
    }
\]

\[
    \infer[\rn{rename-\gamma}]
     {
        u \vdash \rangerexpstepnostate{\gammaexp{e}{e_1}{e_2}}{\gammaexp{e'}{e_1'}{e_2'}}{rename}
     }
    { 
    \begin{gathered}
    u \vdash\rangerexpstepnostate{e}{e'}{rename}\qquad
    u \vdash\rangerexpstepnostate{e_1}{e_1'}{rename}\qquad
    u \vdash\rangerexpstepnostate{e_2}{e_2'}{rename}
    \end{gathered}
    }
\]
\hrule
\hrule

\[
    \infer[\rn{rename-assign}]
     {
         u \vdash \rangerstepnolookup{x:=e}{x':=e'}{rename}
     }
    { 
    \begin{gathered}
       u \vdash \rangerexpstepnostate{e}{e'}{rename} \qquad
       u \vdash \rangerstepnolookup{x}{x'}{rename}
    \end{gathered}
    }
\]

\[
    \infer[\rn{rename-invoke}]
     {u \vdash \rangerstepnolookup{\invoke{x}{z}{g}{y}}{\invoke{z'}{x'}{g}{y'}}{rename}}
    {
     \begin{gathered}
                  u \vdash \rangerexpstepnostate{x}{x'}{rename}
                 \qquad
                  u \vdash \rangerexpstepnostate{z}{z'}{rename} \qquad
        u \vdash \rangerexpstepnostate{\overrightarrow{y}}{\overrightarrow{y'}}{rename} \\        
    \end{gathered}
    }
\]

\[
    \infer[\rn{rename-comp}]
     {
        u \vdash \rangerstepnolookup{s_1;s_2}{s_1';s_2'}{rename}
     }
    {
     \begin{gathered}
  u \vdash \rangerstepnolookup{s_1}{s_1'}{rename}
        \qquad
         u \vdash \rangerstepnolookup{s_2}{s_2'}{rename}
    \end{gathered}
    }
\]

\[
    \infer[\rn{rename-if}]
     {
     \begin{gathered}
        u \vdash \myif{e}{s_1}{s_2}  \longrightarrow_{rename}  {\myif{e'}{s_1'}{s_2'}})
     \end{gathered}
     }
    {
     \begin{gathered}
  u \vdash \rangerexpstepnostate{e}{e'}{rename}\qquad
    u \vdash \rangerstepnolookup{s_1}{s_1'}{rename}
        \qquad
         u \vdash \rangerstepnolookup{s_2}{s_2'}{rename}
    \end{gathered}
    }
\]

\[
    \infer[\rn{rename-putfield}]
     {
        u \vdash \rangerstepnolookup{\putfield{z}{f}{e}}{\putfield{z'}{f}{e'}}{rename}
     }
    {
     \begin{gathered}
    u \vdash \rangerexpstepnostate{e}{e'}{rename}
    \qquad
    u \vdash \rangerexpstepnostate{z}{z'}{rename}
    \end{gathered}
    }
\]

\[
    \infer[\rn{rename-getfield}]
     {
        u \vdash \rangerstepnolookup{\getfield{x}{z}{f}}{\getfield{x'}{z'}{f}}{rename}
     }
    {
     \begin{gathered}
    u \vdash \rangerexpstepnostate{x}{x'}{rename}
    \qquad
    u \vdash \rangerexpstepnostate{z}{z'}{rename}
    \end{gathered}
    }
\]
\[
    \infer[\rn{rename-skip}]
     {
        u \vdash \rangerstepnolookup{skip}{skip}{rename}
     }
    {  }
\qquad
    \infer[\rn{rename-new}]
     {
        u \vdash \rangerstepnolookup{\newobj{x}{c}}{\newobj{x'}{c}}{rename}
     }
    {
      u \vdash \rangerexpstepnostate{x}{x'}{rename}
    }
\]

}}
\caption{Renaming Transformation Rules}
\label{fig:renaming-rules}
\end{figure}
}
\journalreport{
\begin{figure}[h!t]
    \footnotesize
\fbox{%
\parbox{\textwidth}{%
\[
\begin{gathered}
    \mathtt{unique}_u((I,O,T)) = (\mathtt{unique}_u(I), \mathtt{unique}_u(O), \mathtt{unique}_u(T))\\
    \mathtt{unique}_u(X) = \{\mathtt{unique}_u(x) \mid x \in X\} \qquad \mathtt{unique}_u((a,\bot))=(a,u)
    \qquad \mathtt{unique}_u((a,j))=(a,j) \text{\; when\; } j\ne\bot\\
\end{gathered}
\]
\hrule
\hrule
\hrule
\[
  \infer[\rn{rename-wrap}]
     {
     I,O\vdash\rangerstep{(\pi_r,\mu_s,T,u)}{s}{ \mu_s',\pi_r',T'', u')}{s''}{rename-wrap}
     }
    { 
    \begin{gathered}
        u'=u+1 \qquad 
        \chi=(I,O,T) \qquad \chi'=(I',O',T') \qquad
        \chi'=\mathtt{unique}_{u'}(\chi) \\
        I=\{i_1..i_n\} \qquad O=\{o_1..o_m\} \qquad u' \vdash \rangerstepnolookup{s}{s'}{rename} 
        \\
        s''= (i_1,u'):=i_1;..;(i_n, u'):=i_n;s';o_1:=(o_1, u');..;o_m:=(o_m,u');\\
        u' \vdash \rangerexpstepnostate{\mu_s}{\mu_s'}{rename}
        \qquad u' \vdash \rangerexpstepnostate{\pi_r}{\pi_r'}{rename} \qquad
        T''=(I'\cup O' \cup T')
    \end{gathered}
    }
\]
\hrule
\hrule
\hrule
\[
    \infer[\rn{rename-bot}]
     {
         u \vdash \rangerexpstepnostate{\bot}{\bot}{rename}
     }
    { }
    \qquad
\infer[\rn{rename-var}]
     {
         u \vdash \rangerexpstepnostate{x}{x'}{rename}
     }
    { 
    \begin{gathered}
     x'=\mathtt{unique}_u(x)
    \end{gathered}
    }
\qquad
    \infer[\rn{rename-unaryOp}]
     {
        u \vdash \rangerexpstepnostate{\unaryop{e}}{\unaryop{e}}{rename}
     }
    { 
    \begin{gathered}
    u \vdash \rangerexpstepnostate{e}{e'}{rename}
    \end{gathered}
    }
\]
\hrule
\hrule

\[
    \infer[\rn{rename-assign}]
     {
         u \vdash \rangerstepnolookup{x:=e}{x':=e'}{rename}
     }
    { 
    \begin{gathered}
       u \vdash \rangerexpstepnostate{e}{e'}{rename} \qquad
       u \vdash \rangerstepnolookup{x}{x'}{rename}
    \end{gathered}
    }
\]

\[
    \infer[\rn{rename-comp}]
     {
        u \vdash \rangerstepnolookup{s_1;s_2}{s_1';s_2'}{rename}
     }
    {
     \begin{gathered}
  u \vdash \rangerstepnolookup{s_1}{s_1'}{rename}
        \qquad
         u \vdash \rangerstepnolookup{s_2}{s_2'}{rename}
    \end{gathered}
    }
\]
}}
\caption{Selected Renaming Transformation Rules}
\label{fig:renaming-rules}
\end{figure}
}

\subsection{Renaming Transformation is Sound}
To show the soundness of this transformation, we first show the soundness of the expression judgment, and follow that up with our usual discussion about the soundness of the statement recursive judgment and then the wrapper judgment.
\begin{lem}[Expression Renaming is Sound]
\label{lem:renaming-expr-sound}

    Let
$\mathtt{unique}_u(\Delta)= \{(\mathtt{unique}_u(x),\Delta(x)) \mid x \in Dom(\Delta)\}$, then for any concrete expression $e$  and for any integer $u$, if

$$ \vdashconcestep{\Delta}{e}{v} \quad \wedge \quad  u \vdash \rangerexpstepnostate{e}{e'}{rename}$$

then, we have
$$\vdashconcestep{\mathtt{unique}_u(\Delta)}{e'}{v}$$

\end{lem}

That is, evaluating the rewritten expression will result to the same value $v$ that is obtained when evaluating the original expression, if and only if the evaluation used the local variable mapping from $\mathtt{unique}_u(\Delta)$.

\techreport{
\begin{proof}
    We will proceed by induction on the judgment $ u \vdash \rangerexpstepnostate{e}{e'}{rename}$.

    First, observe that the lemma proves trivially on {\tt rename-val} rule, since there are no variables contained within values, and thus no rewriting is done.
    Now, consider the {\tt rename-var} rule, we have, $ u \vdash \rangerexpstepnostate{x}{x'}{rename}$, where $e'=x'=\mathtt{unique}_u(x)=(x,u)$.
    Applying the {\tt var} rule in the concrete semantics, we obtain $\vdashconcestep{\Delta}{x}{v}$, such that $\Delta(x)=v$.  
    Now evaluating the rewritten expression $(\mathtt{unique}_u(\Delta))(x,u)=v$, directly from the definition of the $\mathtt{unique}$ function, thus, we have $\vdashconcestep{\mathtt{unique}_u(\Delta)}{(x,u)}{v}$.
    
    Consider the {\tt rename-binaryOp} rule. the rewriting will produce $ u \vdash \rangerexpstepnostate{\binaryop{e_1}{e_2}}{\binaryop{e_1'}{e_2'}}{rename}$, where $
    u \vdash\rangerexpstepnostate{e_1}{e_1'}{rename}$, $
    u \vdash\rangerexpstepnostate{e_2}{e_2'}{rename}$. Evaluating concretely the binary operation expression, we get $\vdashconcestep{\Delta}{\binaryop{e_1}{e_2}}{v_3}$, where $\vdashconcestep{\Delta}{e_1}{v_1}$, $ \vdashconcestep{\Delta}{e_2}{v_2}$, $\binaryop{v_1}{v_2}{} = v_3$. We know by induction on $e_1$, and $e_2$ that  $ \vdashconcestep{\mathtt{unique}_u(\Delta_1)}{e_1'}{v_1} $,
    $ \vdashconcestep{\mathtt{unique}_u(\Delta_2)}{e_2'}{v_2}$.
 Observe that we have $\mathtt{unique}_u(\Delta_1) = \mathtt{unique}_u(\Delta_2) = \mathtt{unique}_u(\Delta) = \{(unique_{u}(x), \Delta(x)) \mid x \in Dom(\Delta) \}$. 
The reason that they are the same is that the expression judgment does not change in the state $\Delta$ when evaluating expressions, as expressions as per the definition do not have side effects on the environment. Now,  applying the {\tt binary-Op} on the reduced values, we obtain  $\vdashconcestep{\mathtt{unique}_u(\Delta)}{\binaryop{e_1'}{e_2'}}{v}$. The same argument also follows for the 
rewriting of {\tt rename-unaryOp} rule. 

Now, consider the {\tt rename-$\gamma$} rule, the rewriting will produce $ u \vdash \rangerexpstepnostate{\gammaexp{e_1}{e_2}{e_3}}{\gammaexp{e_1'}{e_2'}{e_3'}}{rename}$, such that  $
    u \vdash\rangerexpstepnostate{e_1}{e_1'}{rename}$, $
    u \vdash\rangerexpstepnostate{e_2}{e_2'}{rename}$, $
    u \vdash\rangerexpstepnostate{e_3}{e_3'}{rename}$. There are two rules to evaluate a $\gamma$-expression concretely, depending on whether its condition evaluates to {\tt true} or {\tt false}. We show the proof when the condition evaluates to {\tt true}, but the other case follows a very similar argument. Here, evaluating that $\gamma$-expression, we get $\vdashconcestep{\Delta}{\gammaexp{e_1}{e_2}{e_3}}{v}$, such that $\vdashconcestep{\Delta}{e_1}{true}$, $\vdashconcestep{\Delta}{e_2}{v}$, and where  $e=\gammaexp{e_1}{e_2}{e_3}$. 
    By using induction on $e_1$ and $e_2$, we get $ \vdashconcestep{\mathtt{unique}_u(\Delta)}{e_1'}{v_1}$,
    $\vdashconcestep{\mathtt{unique}_u(\Delta_1)}{e_2'}{v_2} $. Therefore, we can conclude that evaluating $\gammaexp{e_1'}{e_2'}{e_3'}$ must evaluate to $v_1$, that is, $\vdashconcestep{\mathtt{unique}_u(\Delta)}{\gammaexp{e_1'}{e_2'}{e_3'}}{v_1}$.

\end{proof}
}
\journalreport{
\begin{proof}[Proof Sketch]
    
For brevity, we only give the general idea of the proof for the variable case. The remaining cases are either trivial as in {\tt rename-val} or use straight forward structural induction. Interested reader is referred to our technical report~\cite{techreport} for the complete proof.
Now, consider the {\tt rename-var} rule, we have, $ u \vdash \rangerexpstepnostate{x}{x'}{rename}$, where $e'=x'=\mathtt{unique}_u(x)=(x,u)$.
    Applying the {\tt var} rule in the concrete semantics, we obtain $\vdashconcestep{\Delta}{x}{v}$, such that $\Delta(x)=v$.  
    Now evaluating the rewritten expression $(\mathtt{unique}_u(\Delta))(x,u)=v$, directly from the definition of the $\mathtt{unique}$ function, thus, we have $\vdashconcestep{\mathtt{unique}_u(\Delta)}{(x,u)}{v}$.

\end{proof}
}

Now, we define the lemma for statements, and then discuss its proof.

\begin{lem}[Recursive statement renaming is sound]
\label{lem:rename-rec-stmt-sound}
Let Let $s$ be any statement and  a rewriting state $(\Theta_t,\eta_s,\Delta_s, H_s, \mu_s, (\pi,\pi_r,\pi_s), (I,O,T), u)$, and $\mathtt{unique}_u(\Delta)$ be defined as $\mathtt{unique}_u(\Delta)=\{(\mathtt{unique}_u(x), \Delta(x)) \mid x \in Dom(\Delta) \}
$. Then, if 

$$\Theta \vdash \concsstep{(\eta, \Delta, H, \nu, \mu)}{s}{(\eta', \Delta', H', \nu', \mu')}{} $$
$$u\vdash \rangerstepnolookup{s}{s'}{rename}$$
Then, we have $$
  \Theta \vdash \concsstep{(\eta, \mathtt{unique}_u(\Delta), H, \nu, \mu)}{s'}{(\eta', \mathtt{unique}_u(\Delta'), H', \nu', \mu')}{}$$    

  \techreport{
  
\begin{proof}

    First, observe that the {\tt rename-skip} rule holds trivially as no change to the statement actually happens.
    
    Now, consider the {\tt rename-assign} rule. 
    $ u\vdash \rangerstepnolookup{x:=e_1}{x':=e_1'}{rename} $, for some $u$, where $x'=\mathtt{unique}_u(x)$ and $ u \vdash \rangerexpstepnostate{e}{e_1'}{rename}$. Evaluating, the original statement concretely we get $\Theta \vdash \concsstep{(\eta,\Delta, H, \nu, \mu)}{x:=e_1}{(\eta, \Delta', H, \nu, \mu)}{}$, where $\Delta'=\Delta[x\leftarrow v]$, $ \vdashconcestep{\Delta}{e_1}{v}$.
     Observe that using the expression lemma~\ref{lem:renaming-expr-sound}, we have $e_1'$ evaluating to the same value $v$ when evaluated on the unique local mapping, i.e., $\vdashconcestep{\mathtt{unique}_u(\Delta)}{e_1'}{v}$.
    Now, evaluating the rewritten statement on the unique local mapping, we have
    $\Theta \vdash \concsstep{(\eta, \mathtt{unique}_u(\Delta), H, \nu, \mu)}{x':=e'}{(\eta, \Delta'', H, \nu, \mu)}{}$, where $\Delta''=\mathtt{unique}_u(\Delta)[x' \leftarrow v]$, and $x'=\mathtt{unique}_u(x)=(x,u)$. Now observe that, $\Delta''=\mathtt{unique}_u(\Delta)[x' \leftarrow v]= \mathtt{unique}_u(\Delta')$. This observation follows directly from applying the $\mathtt{unique}_u$ function on $\Delta'$, and from observing that $\Delta'(x) = \mathtt{unique}_u(\Delta)(x')$. 

    Next, consider the {\tt rename-comp} rule.   
    Applying the rewriting we get $ u\vdash \rangerstepnolookup{s_1;s_2}{s_1';s_2'}{rename} $, for some $u$, where $u \vdash \rangerstepnolookup{s_1}{s_1'}{rename},
         u \vdash \rangerstepnolookup{s_2}{s_2'}{rename}$. There are two cases for executing a composition statement in the concrete semantics, either $s_1$ has a return control flow ({\tt composition-ret} rule), or not ({\tt composition} rule). But, recall by from the previous transformation (remove final-returns), it is guaranteed that  {\tt return} statements are guaranteed to not exist in $s$~\ref{prop:final-ret-no-return} that also enforced by the grammar of the renaming transformation. This assumption is safe to assume because the renaming transformation always happens after the remove final-return transformation, which ensures that all {\tt return} statements are removed from $s$, while capturing the semantics and the conditions of the returns within the rewriting environment variables $\mu_s$, and $\pi_r$. Therefore, since $s_1$ cannot contain a {\tt return} statement, we can apply the {\tt composition} rule from the concrete semantics to get 
    $ \Theta \vdash \concsstep{(\eta, \Delta, H, \nu, \mu)}{s_1}{(\eta_1, \Delta_1, H_1, \nu_1, \mu)}{}$, and $
       \Theta \vdash \concsstep{(\eta_1, \Delta_1, H_1, \nu_1, \mu)}{s_2}{(\eta_2, \Delta_2, H_2, \nu_2, \mu)}{}$. Observe here that the concrete execution of $s_2$ also does not change the return expression $\mu$, again because $s_2$ cannot have a {\tt return} statement, also by the same guaranteed Property~\ref{prop:final-ret-no-return}.
       Using induction on $s_1$ we get 
     $\Theta \vdash \concsstep{(\eta, \mathtt{unique}_u(\Delta), H, \nu, \mu)}{s_1'}{(\eta_1, \mathtt{unique}_u(\Delta_1), H_1, \nu_1, \mu)}{} $. Using induction on $s_2$ we get 
     $\Theta \vdash \concsstep{\eta_1, \mathtt{unique}_u(\Delta_1), H_1, \nu_1, \mu)}{s_2'}{(\eta_2, \mathtt{unique}_u(\Delta_2), H_2, \nu_2, \mu)}{}$. Applying the {\tt composition} rule, we get $\Theta \vdash \concsstep{\eta', \mathtt{unique}_u(\Delta), H, \nu, \mu)}{s_1';s_2'}{(\eta_2, \mathtt{unique}_u(\Delta_2), H_2, \nu_2, \mu)}{}$. Thus, we arrive to the same state reached when executing the original statement $s$.


Consider the {\tt rename-getfield} rule, $u\vdash \rangerstepnolookup{\getfield{x}{z}{f}}{\getfield{x'}{z'}{f}}{rename} $, where $u \vdash \rangerexpstepnostate{x}{x'}{rename}
    $, and $  u \vdash \rangerexpstepnostate{z}{z'}{rename}$.
 Also, we can evaluate the original statement concretely by applying concrete rule {\tt get-field} to get $\Theta\vdash\concsstep{(\eta, \Delta, H, \nu,\mu)}{\getfield{x}{z}{f}}{(\eta, \Delta',H, \nu,\mu)}{}{}$, where $\Delta'=\Delta[x \leftarrow v]$, $\vdashconcestep{\Delta}{z}{r}$, $ v = \lookup{H}{r,f}$, and $r \in R$.
 Now, we can apply the expression lemma to get $\vdashconcestep{\mathtt{unique}_u(\Delta)}{z'}{r}$. This means that the same value will be looked up from the Heap, i.e., $v=H(r,f)$. Now we can apply the {\tt get-field} rule to get $\Theta \vdash \concsstep{(\eta, \mathtt{unique}_u(\Delta), H, \mu)}{\getfield{x'}{z'}{f}}{(\eta, \Delta'', H, \mu)}{}$, where $\Delta''=\mathtt{unique}_u(\Delta) [x' \leftarrow v]$. 
 Now,  observe that $\Delta(z)=\mathtt{unique}_u(\Delta)(z')$ following from the definition of $\mathtt{unique_u}(\Delta)$. This implies that $\mathtt{unique}_u(\Delta [x \leftarrow v])= \mathtt{unique}_u(\Delta) [x' \leftarrow v]=\mathtt{unique}_u(\Delta) [\mathtt{unique_u}(x) \leftarrow v]$.
 That is, updating the mapping of the variable $x$, then renaming it using the $unique_u$ function, is the same as, renaming the variable then updating its mapping. To see how this is possible, let $(x,v')$ be the mapping of $x$ in $\Delta$. Then, $\Delta [x \leftarrow v]$, will result in updating the $(x,v')$ mapping with $(x,v)$ mapping. Now, applying $\mathtt{unique}_u$, i.e., $\mathtt{unique}_u(\Delta [x \leftarrow v])$, will replace the pervious $x$ mapping with $((x,u),v)$. On the other hand, $\mathtt{unique}_u(\Delta)$, will result in replacing $(x,v')$ definition with $((x,u),v')$ in the mapping. Then, $\mathtt{unique}_u(\Delta) [x' \leftarrow v]$, will update the $((x,u),v')$ mapping to $((x,u),v)$.
 Thus, we we have that $\Theta \vdash \concsstep{(\eta, \mathtt{unique}_u(\Delta), H, \mu)}{\getfield{x'}{z'}{f}}{(\eta, \mathtt{unique}_u(\Delta) [x' \leftarrow v], H, \mu)}{}$.

Consider the {\tt rename-putfield} rule, $ u\vdash \rangerstepnolookup{\putfield{z}{f}{e}}{\putfield{z'}{f}{e'}}{rename} $, where 
    $u \vdash \rangerexpstepnostate{e}{e'}{rename}
    $, and $
    u \vdash \rangerexpstepnostate{z}{z'}{rename}$. Also, applying {\tt put-field}, we get $\Theta \vdash \concsstep{(\eta, \Delta, H, \mu)}{\putfield{z}{f}{e}}{(\eta, \Delta, H', \mu)}{}$, where $H'=H[(r,f) \leftarrow v]$,$\vdashconcestep{\Delta}{z}{r}$, $r \in R$, and $\vdashconcestep{\Delta}{e}{v}$.
We know from the expression lemma that $\vdashconcestep{\mathtt{unique}_u(\Delta)}{e'}{v}$, and similarly that, $\vdashconcestep{\mathtt{unique}_u(\Delta)}{z'}{r}$.
Now, if we evaluate the rewritten statement we get 
$
\Theta \vdash \concsstep{(\eta, \mathtt{unique}_u(\Delta), H, \mu)}{\putfield{z'}{f}{e'}}{(\eta, \mathtt{unique}_u(\Delta), H'', \mu)}{}    
$, where $H''=H[(r,f) \leftarrow v]=H'$.

Consider the {\tt rename-invoke} rule, we get $ u \vdash \rangerstepnolookup{(\invoke{x}{z}{g}{y})}{(\invoke{z'}{x'}{g}{y'})}{rename}$, $
              u \vdash \rangerexpstepnostate{x}{x'}{rename}
                 $, $
                  u \vdash \rangerexpstepnostate{z}{z'}{rename} $, and $
        u \vdash \rangerexpstepnostate{\overrightarrow{y}}{\overrightarrow{y'}}{rename}$
        (expanding $unique_u$ notation to work for a vector of expressions). 
We also have by applying {\tt invoke} concrete rule that, 
$ \Theta\vdash\concsstep{(\eta,\Delta, H, \nu,\mu)}{\invoke{x}{z}{g}{y}}{(\eta'', \Delta',H'',\nu'',\mu')}{}$, 
where $\Delta'=\update{\Delta}{x}{\mu''}$, $\mu'=\bot$, $\Delta(z) \in R
    $, $method(\overrightarrow{y'}.s')=\lookup{\Theta}{\lookup{\eta}{\Delta(z)},g}
    $, $\overrightarrow{v'}=\Delta(\overrightarrow{y})$ (here we are overloading lookup to work over a vector of variables), and $
     \Theta\vdash \concsstep{(\eta,\{(\overrightarrow{y'},\overrightarrow{v'})\},H,\nu,\bot)}{s'}{(\eta'', \Delta'',H'',\nu'',\mu'')}{}$. In this case, we have that $\eta'=\eta'',\Delta'=\update{\Delta}{x}{\mu''},H''=H',\nu'=\nu'',\mu'=\mu=\bot$.
     Observe here that although we initially assumed $s$ to not contain a {\tt return} statement, this condition is not valid for $s'$, because $s'$ exists in $\Theta$ which by the concrete property~\ref{prop:theta-stmt-has-return} must have a {\tt return} statement. This means, we cannot use the induction hypothesis on $s'$, instead we show that by executing this statement $s'$ in either case will result in the same state.
From the expression lemma we get that $\vdashconcestep{\mathtt{unique}_u(\Delta)}{\overrightarrow{y'}}{\overrightarrow{v'}}$, and similarly, we have that  $\vdashconcestep{\mathtt{unique}_u(\Delta)}{z}{r}$. 
We also know from the definition of $\mathtt{unique}_u(\Delta)$ that if $z$ is a variable, then it must be that $\Delta(z)=\mathtt{unique}_u(\Delta)(z')$. This implies that the same method, is looked up and pulled from $\Theta$, that is $\lookup{\Theta}{\lookup{\eta}{\mathtt{unique}_u(\Delta(z))},g} = method(\overrightarrow{y'}.s')$. 
Observe that in this case, we have that $\Theta\vdash \concsstep{(\eta,\{(\overrightarrow{y'},\overrightarrow{v'})\},H,\mu)}{s'}{(\eta', \Delta',H',\mu')}{}$. Finally, observe that, $(\eta', \mathtt{unique}_u(\update{\Delta}{x}{\mu'}),H',\mu)=(\eta', \update{\mathtt{unique}_u(\Delta)}{x'}{\mu'},H',\mu)$, again using the definition $\mathtt{unique}_u(\Delta)$.

Finally, consider the {\tt rename-if} rule. Observe that here by the guarantee from the $\gamma$-transformation, we have no $\phi$-expressions can exist in $s$ (Property~\ref{prop:no-phi}). This means that if-statements that reaches the this transformation, must have an empty list of $\phi$-statements. Thus, in the coming proof we discuss the correctness of the if-statements that do not any $\phi-$expressions. 
Applying the rewriting we get 
 $
  u \vdash (\myif{e}{s_1}{s_2})  \longrightarrow_{rename}  ({\myif{e'}{s_1'}{s_2'}})
 $, where $
  u \vdash \rangerexpstepnostate{e}{e'}{rename}$, $
    u \vdash \rangerstepnolookup{s_1}{s_1'}{rename}
        $, and $
         u \vdash \rangerstepnolookup{s_2}{s_2'}{rename}$.
This statement can execute concretely in four cases: when the if-condition evaluates to either {\tt true} or {\tt false}, and each case may or may not involve encountering a {\tt return} statement in the corresponding branch, the return case is excluded by remove final-return transformation which guarantee (Property~\ref{prop:final-ret-no-return}). We will show the proof when the condition evaluates to {\tt true}, the other case is similar.

First, let us consider the case where the if-condition evaluates to {\tt true} and we know that no return statement is encountered in $s_1$. Then, we can apply {\tt if-phi-true} to get 
$\Theta\vdash\concsstep{(\eta, \Delta, H, \nu,\mu)}{\myif{e}{s_1}{s_2}}{(\eta',\Delta', H', \nu',\mu)}{}$, where $\vdashconcestep{\Delta}{e}{true}
    $, $\Theta\vdash\concsstep{(\eta,\Delta, H, \nu,\mu)}{s_1}{(\eta',\Delta', H', \nu',\mu')}{}$, again here $\mu=\mu'=\bot$ from Property~\ref{prop:final-ret-no-return}. Now, using the expression lemma we get $\Theta \vdash\concestep{\mathtt{unique}_u(\Delta)}{e'}{true}{}$. Also, using induction on $s_1$, we get $\vdash\concsstep{(\eta,\mathtt{unique}_u(\Delta), H, \nu, \mu)}{s_1'}{(\eta',\mathtt{unique}_u(\Delta'), H', \nu', \mu)}{} $.
Thus, evaluating the rewritten statement, we get 
$\Theta\vdash ((\eta, \mathtt{unique}_u(\Delta), H, \nu, \mu), {\myif{e'}{s_1'}{s_2'}} \Rightarrow_{c} (\eta',\mathtt{unique}_u(\Delta'), H', \nu', \mu')$, which is what we wanted to show.
\end{proof}

  }

  \journalreport{
  
\begin{proof}[Proof Sketch]
Again, for brevity, we only give the general idea of the proof for the assignment-statement. The remaining cases are either straightforward as in {\tt rename-getfield} or use straight forward structural induction. Interested reader is referred to our technical report~\cite{techreport} for the complete proof.

    Now, consider the {\tt rename-assign} rule. 
    $ u\vdash \rangerstepnolookup{x:=e_1}{x':=e_1'}{rename} $, for some $u$, where $x'=\mathtt{unique}_u(x)$ and $ u \vdash \rangerexpstepnostate{e}{e_1'}{rename}$. Evaluating, the original statement concretely we get $\Theta \vdash \concsstep{(\eta,\Delta, H, \nu, \mu)}{x:=e_1}{(\eta, \Delta', H, \nu, \mu)}{}$, where $\Delta'=\Delta[x\leftarrow v]$, $ \vdashconcestep{\Delta}{e_1}{v}$.
     Observe that using the expression lemma~\ref{lem:renaming-expr-sound}, we have $e_1'$ evaluating to the same value $v$ when evaluated on the unique local mapping, i.e., $\vdashconcestep{\mathtt{unique}_u(\Delta)}{e_1'}{v}$.
    Now, evaluating the rewritten statement on the unique local mapping, we have
    $\Theta \vdash \concsstep{(\eta, \mathtt{unique}_u(\Delta), H, \nu, \mu)}{x':=e'}{(\eta, \Delta'', H, \nu, \mu)}{}$, where $\Delta''=\mathtt{unique}_u(\Delta)[x' \leftarrow v]$, and $x'=\mathtt{unique}_u(x)=(x,u)$. Now observe that, $\Delta''=\mathtt{unique}_u(\Delta)[x' \leftarrow v]= \mathtt{unique}_u(\Delta')$. This observation follows directly from applying the $\mathtt{unique}_u$ function on $\Delta'$, and from observing that $\Delta'(x) = \mathtt{unique}_u(\Delta)(x')$. 

\end{proof}

  }
\end{lem}

Finally, we prove the rename theorem for the {\tt rename-wrap} rule. \\
\begin{thm}[Renaming Transformation is Sound]
\label{thm:renaming-wrap-sound}
For any statement $s$. If
$$\Theta \vdash \concsstep{(\eta, \Delta, H, \nu, \mu)}{s}{(\eta', \Delta', H', \nu', \mu')}{} $$
$$I,O\vdash\rangerstep{(\pi_r,\mu_s,T,u)}{s}{ \mu_s',\pi_r',T'', u')}{s''}{rename-wrap}$$
then 
$$\Theta \vdash \concsstep{(\eta, \Delta, H, \nu, \mu)}{s''}{(\eta', \Delta'', H', \nu', \mu')}{} $$

$$ \forall o \in O, \Delta'(o)=\Delta''(o) $$
$$evalReturn(\Delta',\pi_r,\mu_s)=evalReturn(\Delta'',\pi_r',\mu_s')$$
$$ \forall i \in I, \Delta(i)=\Delta'(i)=\Delta''(i)$$
\end{thm}
\techreport{

\begin{proof}
    
Here, the {\tt rename-wrap} rule we get $I,O\vdash\rangerstep{(\pi_r,\mu_s,T,u)}{s}{ \mu_s',\pi_r',T'', u')}{s''}{rename-wrap}$, where $u'=u+1$, $s''= (i_1,u'):=i_1;\ldots;(i_n, u'):=i_n;s';o_1:=(o_1, u');\ldots;o_m:=(o_m,u');$, $u' \vdash \rangerexpstepnostate{\mu_s}{\mu_s'}{rename}, u' \vdash \rangerexpstepnostate{\pi_r}{\pi_r'}{rename}$, $ T''=(I'\cup O' \cup T')$ and where $s'$, is obtained from $ u' \vdash \rangerstepnolookup{s}{s'}{rename} $. Also, we know that executing $s$ concretely we get, $\Theta \vdash\concsstep{(\eta,\Delta, H, \nu, \mu)}{s}{(\eta',\Delta', H', \nu', \mu')}{}$.
Then, we want to show that $\Theta \vdash\concsstep{(\eta,\Delta, H, \nu, \mu)}{s''}{(\eta',\Delta'', H', \nu', \mu')}{}$, such that $ \forall o \in O, \Delta'(o)=\Delta''(o)$.

There are three parts to prove this theorem. First, observe that using the statement lemma (Lemma~\ref{lem:rename-rec-stmt-sound}) on $s'$ we know that

$$
  \Theta \vdash \concsstep{(\eta, \mathtt{unique}_{u'}(\Delta), H, \nu, \mu)}{s'}{(\eta', \mathtt{unique}_{u'}(\Delta'), H', \nu', \mu')}{}$$

such that
$
 \mathtt{unique}_{u'}(\Delta)=\{(\mathtt{unique}_{u'}(x), \Delta(x)) \mid x \in Dom(\Delta) \}
$

This implies that for the subset of output variables $O$, we can find a matching mapping values of their unique new name, that is 
$\forall o \in O, \Delta'(o) = ({\mathtt{unique}_{u'}}(\Delta'))(o,u')$, and similarly for input variables $i\in I$, i.e., $\forall i \in I, \Delta'(i) = ({\mathtt{unique}_{u'}}(\Delta'))(i,u')$,

Second, observe that input assignments within $s''$ does not change the value of input values for $s$, i.e., there are no def-statements that are re-defining variables in $I$. Also, observe that executing the first assignment, for example, would result in $\Theta \vdash \concsstep{(\eta, \Delta, H, \nu, \mu)}{(i_1,u'):=i_1}{(\eta,\hat{\Delta}, H, \nu, \mu)}{}$, where $\hat{\Delta}=\Delta[(i_1,u') \leftarrow \Delta(i_1)]=\Delta[unique_{u'}(i_1) \leftarrow \Delta(i_1)]$. More generally, the local variable mapping we will obtain after executing all assignment statement but before executing $s'$ will be of the form $\Delta[unique_{u'}(i_1) \leftarrow \Delta(i_1)]\ldots[unique_{u'}(i_n) \leftarrow \Delta(i_n)]$. 
Since we know that $unique_{u'}(i)$ must have been encountered for any $i$, and thus is already part of $Dom(\Delta')$. Also, it is clear that $\hat{\Delta}(unique_{u'}(i))=\hat{\Delta}(i)$.
That implies that it must be that $\Theta \vdash \concsstep{(\eta, \Delta, H, \nu, \mu)}{(i_1,u'):=i_1;\ldots;(i_n, u'):=i_n;s'}{(\eta', ({\mathtt{unique}_{u'}}(\Delta')), H', \nu', \mu')}{}$. 
%
Then, evaluating the remaining parts of $s''$,  $\Theta \vdash \concsstep{(\eta', ({\mathtt{unique}_{u'}}(\Delta')), H', \nu', \mu')}{(o_1:=(o_1, u');\ldots;o_m:=(o_m,u')}{(\eta', ({\mathtt{unique}_{u'}}(\Delta')) \cup \{(o_1, v_1), \ldots, (o_m, v_m)\}, H', \nu', \mu')}{}$, where we have that 
$v_1=({\mathtt{unique}_{u'}}(\Delta'))(o_1,u'), \ldots,  v_n=({\mathtt{unique}_{u'}}(\Delta'))(o_n,u')$. 
Let $\Delta''=({\mathtt{unique}_{u'}}(\Delta')) \cup \{(o_1, v_1), \ldots, (o_m, v_m)\}$, then using the above two components we know that it must be the case that $\forall o \in O, \Delta'(o) = \Delta''(o)$.

To show the third conclusion of the theorem, we first start by showing that $\pi_r$, and $\pi_r'$ are semantically the same, and similarly for $\mu_s$, and $\mu_s'$.
From the soundness of the expression lemma~\ref{lem:renaming-expr-sound}, we know that if $\vdashconcestep{\Delta'}{\mu_s}{v} \quad \wedge \quad  u' \vdash \rangerexpstepnostate{\mu_s}{\mu_s'}{rename}$, then we have $\vdashconcestep{\mathtt{unique}_{u'}(\Delta')}{\mu_s'}{v}$. Observe here that we have just shown that $\mathtt{unique}_{u'}(\Delta') \subseteq \Delta''$, thus we can conclude that $\vdashconcestep{\Delta'}{\mu_s'}{v}$. Same argument applies for $\pi_r$ and $\pi_r'$.
Then, we have two cases $\vdashconcestep{\Delta'}{\pi_r}{true}$ or not. First, consider the case where $\pi_r$ does not evaluate to $true$. Then, $evalReturn(\Delta',\pi_r,\mu_s)=\bot$. And we have just shown that $\pi_r'$ and has a similar evaluation in $\Delta''$, thus we get $evalReturn(\Delta'',\pi_r',\mu_s')=\bot$, so they are equal.
Now, consider the second case where $\pi_r$ evaluates to true 
$\vdashconcestep{\Delta'}{\pi_r}{true}$. Let $v$ be the evaluation of $\mu_s$, $\vdashconcestep{\Delta'}{\mu_s}{v}$. Then, $evalReturn(\Delta',\pi_r,\mu_s)=v$. But again, we have shown that both $\pi_r'$ and $\mu_s'$ will evaluate to the same values in $\Delta''$, that is $\vdashconcestep{\Delta''}{\pi_r'}{true}$, and $\vdashconcestep{\Delta''}{\mu_s'}{v}$. Then, $evalReturn$ must evaluate to the same value $v$, i.e., $evalReturn(\Delta'',\pi_r',\mu_s')=v=evalReturn(\Delta',\pi_r,\mu_s)$, which is what we want to prove.

Finally, we show that $\forall i \in I.,\Delta(i)=\Delta'(i)=\Delta''(i)$. Since $i \in I$, it follows that $i \in Dom(\Delta)$, and because $s$ is in SSA form (i.e., $isSSA(s)$), no variable in $\Delta$ can be modified during the execution of $s$. Therefore, $\forall i \in I.,\Delta(i)=\Delta'(i)$.

Furthermore, by Guarantee~\ref{prop:ssa}, the rewritten statement $s''$ is also in generalized SSA form, ensuring that no variable in $Dom(\Delta)$ can be updated during the execution of $s''$. Hence, $\forall i \in I.,\Delta(i)=\Delta'(i)=\Delta''(i)$, as required.

\end{proof}

\techreport{
\begin{cor}

Let $s$ be any statement such that

If
$$\Theta \vdash \concsstep{(\eta, \Delta, H, \nu, \mu)}{s}{(\eta', \Delta', H', \nu', \mu')}{} $$
$$I,O\vdash\rangerstep{(\pi_r,\mu_s,T,u)}{s}{ \mu_s',\pi_r',T'', u')}{s''}{rename-wrap}$$
$$s''=i_1':=i_1;..;i_n':=i_n;s';o_1:=o_1';..;o_m:=o_m';$$
then 
$$\Delta''=\Delta[i_1' \leftarrow \Delta(i_1)\ldots i_n'\leftarrow \Delta(i_n)] $$
$$\Theta \vdash \concsstep{(\eta, \Delta'', H, \nu, \mu)}{s'}{(\eta', \Delta''', H', \nu', \mu')}{}$$ 
$$
\Delta'''(o_1')=\Delta'(o_1) \wedge \cdots \wedge
\Delta'''(o_m')=\Delta'(o_m)
$$
\end{cor}

\begin{proof}
    From the renaming theorem, we know that $\Theta \vdash \concsstep{(\eta, \Delta, H, \nu, \mu)}{s''}{(\eta', \Delta'', H', \nu', \mu')}{} \wedge 
 \forall o \in O, \Delta'(o)=\Delta''(o)$. Now, observe that the state in which $s'$ executes must be of the form $\Delta[i_1' \leftarrow \Delta(i_1)\ldots i_n'\leftarrow \Delta(i_n)]$. This is clear by partially evaluating $s''$, i.e., $\Theta \vdash \concsstep{(\eta, \Delta, H, \nu, \mu)}{i_1':=i_1;..;i_n':=i_n}{(\eta, \Delta[i_1' \leftarrow \Delta(i_1)\ldots i_n'\leftarrow \Delta(i_n)], H, \nu', \mu)}{}$. Thus, concretely executing $s'$ in this local mapping must result in the same state, i.e., $\Theta \vdash \concsstep{(\eta, \Delta[i_1' \leftarrow \Delta(i_1)\ldots i_n'\leftarrow \Delta(i_n)], H, \nu, \mu)}{s';o_1:=o_1';..;o_m:=o_m'}{(\eta', \Delta''', H', \nu', \mu')}{} $, and from the renaming theorem we know that out variable mappings must be the same, i.e., $\Delta'''(o_1')=\Delta'(o_1) \wedge \cdots \wedge \Delta'''(o_m')=\Delta'(o_m)$.
\end{proof}
}
}

\journalreport{

\begin{proof}[Proof Sketch]
By {\tt rename-wrap}, the rewritten statement decomposes as $s'' = 
(i_1,u'):=i_1;\ldots;(i_n,u'):=i_n;s';o_1:=(o_1,u');\ldots;o_m:=(o_m,u')$, 
where $u'=u+1$ and $s'$ is obtained by applying the renaming transformation 
to $s$ under $u'$.

By Lemma~\ref{lem:rename-rec-stmt-sound}, executing $s'$ under 
$\mathtt{unique}_{u'}(\Delta)$ produces $\mathtt{unique}_{u'}(\Delta')$, 
preserving $\eta$, $H$, $\nu$, and $\mu$. The input assignments prepended to 
$s'$ populate the renamed input variables $(i,u')$ with their original values 
from $\Delta$, yielding a state consistent with $\mathtt{unique}_{u'}(\Delta)$ 
before $s'$ executes. The output assignments appended after $s'$ copy the 
renamed output values $(o,u')$ back to the original output names $o$, producing 
$\Delta''$ such that $\forall o \in O, \Delta'(o) = \Delta''(o)$.

For the soundness of $\pi_r'$ and $\mu_s'$, by 
Lemma~\ref{lem:renaming-expr-sound}, both $\pi_r'$ and $\mu_s'$ evaluate to 
the same values as $\pi_r$ and $\mu_s$ under $\mathtt{unique}_{u'}(\Delta') 
\subseteq \Delta''$, so $\mathtt{evalReturn}(\Delta', \pi_r, \mu_s) = 
\mathtt{evalReturn}(\Delta'', \pi_r', \mu_s')$.

Finally, since $s$ and $s''$ are both in SSA form the guaranteed property~\ref{prop:ssa}, no variable in $Dom(\Delta)$ is redefined 
during execution of either, so $\forall i \in I,\ \Delta(i) = \Delta'(i) = 
\Delta''(i)$.
\end{proof}

}

\subsection{Properties of the Renaming Transformation}

The renaming transformation assumes all properties guaranteed by earlier transformations ($\gamma$ creation, early-returns elimination and remove final-returns), and it provides the following guaranteed properties.

For any statement $s$ with $I,O,T$ representing the statement input, output, and temporary variables, and $u$ representing unique-count respectively,
if $(x,n) \in (Id, n)$ is a variable within $s'$, 
such that 
$$I,O\vdash\rangerstep{(\pi_r,\mu_s,T,u)}{s}{ \mu_s',\pi_r',T'', u')}{s''}{rename-wrap}
     $$
where 
$  u'=u+1, 
        \chi=(I,O,T) , \chi'=(I',O',T') ,
        \chi'=\mathtt{unique}_{u'}(\chi) ,
        I=\{i_1..i_n\} , O=\{o_1..o_m\} \\
        , u' \vdash \rangerstepnolookup{s}{s'}{rename} 
        ,
        s''= (i_1,u'):=i_1;..;(i_n, u'):=i_n;s';o_1:=(o_1, u');..;o_m:=(o_m,u');\\
        u' \vdash \rangerexpstepnostate{\mu_s}{\mu_s'}{rename}
        , u' \vdash \rangerexpstepnostate{\pi_r}{\pi_r'}{rename} ,
        T''=(I'\cup O' \cup T')$\\

then we have
\begin{itemize}
    \item Rewritten Statement $s''$ satisfies the SSA property (Property~\ref{prop:ssa}).
        \techreport{
        More formally, $SSN(s'')$.
        \begin{proof}
        Let $s'' = (i_1,u'):=i_1;\cdots;(i_n,u'):=i_n;\;s';\;o_1:=(o_1,u');\cdots;o_m:=(o_m,u')$. To show that $SSA(s'')$ holds, note first that by the SSA assumed property from the previous transformation (guaranteed property from remove final-returns transformation, Property~\ref{prop:ssa}) we have $SSA(s)$, and since $s'$ is obtained from $s$ by a systematic renaming, $SSA(s')$ follows. Next, the prefix assignments $(i_k,u'):=i_k$ introduce only fresh variables $(i_k,u')$, and since $i_k \in I$ implies $\neg isDef(i_k,s)\wedge isUse(i_k,s)$, it follows that $\neg isDef((i_k,u'),s'') \wedge isUse((i_k,u'),s'')$, hence no clashes arise with $s'$. Finally, the suffix assignments $o_j:=(o_j,u')$ do not introduce new definitions, because all occurrences of $o_j$ have already been renamed to $(o_j,u')$ in $s'$, and thus they preserve SSA. Therefore, $SSA(s'')$ holds. For the detailed proof, please refer to our technical report~\cite{techreport}.
        \end{proof}
        }
    \item Fresh variables have $u+1$ suffices (Property~\ref{prop:fresh-vars-increments-u}).
    \techreport{
    More formally, $(x,\bot)$ is in $s$, then all its occurrences with $s'$ must have been renamed to $(x,u+1)$. 
    \begin{proof}
        This is easy to show as by JR assumptions, we know that all variables in $s$ must have a $\bot$ index (Property~\ref{prop:initial-variables-not-rename}), for example, for an identifier $x$, its corresponding variable form that appear in $s$ must be $(x,\bot)$. Next, observe that renaming of variables are done {\tt rename-var}, which renames all variables to have an index $u+1$, i.e., $(x,u+1)$ from our continued example above. Observe that happens over all expressions in $s$, following the rules in Fig.~\ref{fig:renaming-rules}. Finally, observe that all environments that are carrying variables, i.e., $\Delta_s,\pi_r,\chi$ are also renamed.
    \end{proof}
    }
    \item Old input not in $s'$ (Property~\ref{prop:renaming-no-old-input}).
    More formally, Old input does not appear, i.e., is not used in $s'$. More formally if $(x,j) \in I$, then $\neg isUse((x,j),s')$
    \techreport{
    \begin{proof}
    This follows directly by structural induction on $s$. The renaming transformation replaces every variable $x$ in $s$ with a fresh name $x'$, yielding a statement $s'$ containing only renamed variables. Consequently, none of the original variable names, including those in $I$, appear in $s'$.
    \end{proof}
    }
    \techreport{
    \item $s''$ and $s'$ are the same if the state for $\Delta$ is set properly.
    \label{prop:same-concrete-state}
    If
    \begin{align*}
    &I,O\vdash\rangerstep{(\pi_r,\mu_s,T,u)}{s}{ \mu_s',\pi_r',T'', u')}{s''}{rename-wrap}\\
    &\Theta \vdash (\eta, \Delta[(i_1,u')\leftarrow\Delta(i_1)\ldots(i_n,u')\leftarrow\Delta(i_n)], H, \mu),s';o_1:=(o_1, u');..;o_m:=(o_m,u')\\
    &\Rightarrow_c(\eta', \Delta', H', \mu')\\
    &\Theta \vdash \concsstep{(\eta, \Delta, H, \mu)}{s''}{(\eta'', \Delta'', H'', \mu'')}{}\\
    &\mu=\bot 
    \end{align*}
    where 
    \begin{align*}
     &u'=u+1, I=\{i_1..i_n\}, O=\{o_1..o_m\}, \\
     &s''=(i_1,u'):=i_1;..;(i_n, u'):=i_n;s';o_1:=(o_1, u');..;o_m:=(o_m,u')\\
     &u' \vdash \rangerexpstep{(I,O,T)}{\mu_s}{(I',O',T')}{\mu_s'}{rename}
    \end{align*}
    then
    $$\eta'=\eta'',\Delta'=\Delta'',H'=H'',\mu'=\mu''$$
    \begin{proof}
    It is clear that to concretely evaluate $s''$, one has to first evaluate the initial sequence of assignment statements, i.e., $(i_1,u'):=i_1;..;(i_n, u'):=i_n$, which will result in, since these assignment statements have no side effects on $\eta,H$ or $\mu$, then the intermediate state just before evaluating $s';o_1:=(o_1, u');..;o_m:=(o_m,u')$ would be $(\eta, \Delta[(i_1,u')\leftarrow\Delta(i_1)\ldots(i_n,u')\leftarrow\Delta(i_n)], H, \mu)$. This then implies that indeed the evaluation of the remaining statements,i.e., $s';o_1:=(o_1, u');..;o_m:=(o_m,u')$ must result in the same state as $\eta',\Delta',H',\mu'$.  
    \end{proof}   
    }
\end{itemize}

\section{$TR_5$: Input Substitution Transformation}
\label{sec:substitution}
 The substitution transformation substitutes values of the region's inputs that are defined in $I$. 
Fig.~\ref{fig:susbstitution-rules}, shows a snippet of the rules for the substitution transformation. 
The rules define two judgments. The judgment $\Delta_s, I \vdash 
\rangerstepnolookup{s}{s'}{sub}$ rewrites an IR statement $s$ into $s'$ by 
substituting variables with their corresponding values in the symbolic local 
variable mapping $\Delta_s$, where $I$ is the set of input variables whose 
occurrences drive the substitution. The wrapper judgment $\Delta_s \vdash 
\rangerstepnolookup{((\pi_r, \mu_s, I), s)}{((\pi_r', \mu_s', I'), s')}{sub\text{-}wrap}$ 
is the entry point of the transformation, triggering the sub judgment on $s$ 
and additionally applying the substitution to the early-return path condition 
$\pi_r$ and symbolic return expression $\mu_s$, producing their updated 
counterparts $\pi_r'$ and $\mu_s'$, and updating the input set from $I$ to 
$I'$.

The substitution transformation replaces variables with their corresponding 
values in $\Delta_s$. For expressions, Rule~{\tt sub-expr} directly substitutes 
each variable occurrence with its mapped value in $\Delta_s$. For statements, 
the transformation propagates substitution over all \emph{uses} of variables, 
leaving definitions unchanged. Rule~{\tt sub-wrap} extends the substitution to 
environment variables that live outside the statement, applying it to the path 
condition $\pi$, the return condition $\pi_r$, and the return expression $\mu_s$. 
Additionally, the wrapper updates the input set $I$ to reflect the free variables 
in all rewritten expressions, as the previous input variables in $I$ have been 
replaced by their mapped values in $\Delta_s$ and are no longer needed as inputs.

\techreport{
\begin{figure}[h!t]
    \footnotesize
\fbox{%
\parbox{\textwidth}{%
\[
    \infer[\rn{sub-val}]
     {
         \Delta_s,I \vdash \rangerexpstepnostate{v}{v}{sub}
     }
    { 
    }
\qquad
    \infer[\rn{sub-var}]
     {
          \Delta_s,I \vdash \rangerexpstepnostate{x}{\Delta_s(x)}{sub}
     }
    { 
        x \in I 
    }
\qquad
    \infer[\rn{sub-var-triv}]
     {
          \Delta_s,I \vdash \rangerexpstepnostate{x}{x}{sub}
     }
    { 
     x \notin I 
    }
\]
\[
 \infer[\rn{sub-unaryOp}]
     {
        \Delta_s,I \vdash \rangerexpstepnostate{\unaryop{e}{}}{\unaryop{e}{}}{sub}
     }
    { 
    \begin{gathered}
       \Delta_s,I \vdash \rangerexpstepnostate{e}{e'}{sub}
    \end{gathered}
    }
    \qquad
    \infer[\rn{sub-binaryOp}]
     {
          \Delta_s,I \vdash \rangerexpstepnostate{\binaryop{e_1}{e_2}{}}{\binaryop{e_1'}{e_2'}{}}{sub}
     }
    { 
    \begin{gathered}
         \Delta_s,I \vdash \rangerexpstepnostate{e_1}{e_1'}{sub}\qquad
          \Delta_s,I \vdash \rangerexpstepnostate{e_2}{e_2'}{sub}
    \end{gathered}
    }
\]

\[
    \infer[\rn{sub-\gamma}]
     {
       \Delta_s,I \vdash \rangerexpstepnostate{\gammaexp{e}{e_1}{e_2}}{\gammaexp{e}{e_1'}{e_2'}}{sub}
     }
    { 
    \begin{gathered}
    \Delta_s,I \vdash\rangerexpstepnostate{e}{e'}{sub}\qquad
    \Delta_s,I \vdash\rangerexpstepnostate{e_1}{e_1'}{sub}\qquad
    \Delta_s,I \vdash\rangerexpstepnostate{e_2}{e_2'}{sub}
    \end{gathered}
    }
\]
\hrule
\hrule
\[
\infer[\rn{sub-wrap}]
     {
         \Delta_s \vdash \rangerstepnolookup{((\pi_r,\mu_s,I),s)}{(( \pi_r',\mu_s',I'),s')}{sub-wrap}
     }
     {
     \begin{gathered}
       \Delta_s,I \vdash \rangerstepnolookup{s}{s'}{sub}\qquad
         \Delta_s,I \vdash \rangerexpstepnostate{\pi_r}{\pi_r'}{sub} \\
        \Delta_s,I \vdash \rangerexpstepnostate{\mu_s}{\mu_s'}{sub} \qquad
        I'=getInput(s')\cup getUse(\mu_s') \cup getUse(\pi_r')
     \end{gathered}
     }
\]
\hrule
\hrule
\[
\infer[\rn{sub-assign}]
     {
         \Delta_s,I \vdash \rangerstepnolookup{x:=e}{x:=e'}{sub}
     }
     {
          \Delta_s,I \vdash \rangerexpstepnostate{e}{e'}{sub}
     }
\]

\[
    \infer[\rn{sub-comp}]
    {
         \Delta_s,I \vdash \rangerstepnolookup{s_1;s_2}{s_1';s_2'}{sub}
    }
    {
         \Delta_s,I \vdash \rangerstepnolookup{s_1}{s_1'}{sub}
        \quad
         \Delta_s,I \vdash \rangerstepnolookup{s_2}{s_2'}{sub}
    }
\]

\[
    \infer[\rn{sub-invoke}]
     { \Delta_s,I \vdash \rangerstepnolookup{\invoke{x}{z}{\gamma}{y}}{\invoke{x}{z'}{\gamma}{y'}}{sub}}
    {
     \begin{gathered}
       \Delta_s,I \vdash \rangerexpstepnostate{\overrightarrow{y}}{\overrightarrow{y'}}{sub}
\qquad  \Delta_s,I \vdash \rangerexpstepnostate{z}{z'}{sub}
    \end{gathered}
    }
\]

\[
    \infer[\rn{sub-ret}]
     {
        \Delta_s,I  \vdash \rangerstepnolookup{\return{e}}{\return{e'}}{sub}
     }
    {
     \begin{gathered}
     \Delta_s,I  \vdash \rangerexpstepnostate{e}{e'}{sub}
    \end{gathered}
    }
\qquad
    \infer[\rn{sub-skip}]
     {
        \Delta_s,I  \vdash \rangerstepnolookup{skip}{skip}{sub}
     }
    {    }
\]

\[
    \infer[\rn{sub-if}]
     {
     \begin{gathered}
      \Delta_s,I  \vdash \myif{e}{s_1}{s_2}\longrightarrow_{sub} {\myif{\hat{e}}{\hat{s_1}}{\hat{s_2}}}
     \end{gathered}
     }
    {
     \begin{gathered}
     \Delta_s,I  \vdash \rangerexpstepnostate{e}{\hat{e}}{sub} \qquad
    \Delta_s,I  \vdash \rangerstepnolookup{s_1}{\hat{s_1}}{sub}
        \qquad
         \Delta_s,I  \vdash \rangerstepnolookup{s_2}{\hat{s_2}}{sub}
    \end{gathered}
    }
\]

\[
    \infer[\rn{sub-putfield}]
     {
        \Delta_s,I  \vdash \rangerstepnolookup{\putfield{z}{f}{e}}{\putfield{z'}{f}{e'}}{sub}
     }
    {
     \begin{gathered}
    \Delta_s,I  \vdash \rangerexpstepnostate{e}{e'}{sub}
    \qquad \Delta_s,I  \vdash \rangerexpstepnostate{z}{z'}{sub}
    \end{gathered}
    }
\]

\[
    \infer[\rn{sub-getfield}]
     {
        \Delta_s,I  \vdash \rangerstepnolookup{\getfield{x}{z}{f}}{\getfield{x}{z'}{f}}{sub}
     }
    {
     \begin{gathered}
    \Delta_s,I  \vdash \rangerexpstepnostate{z}{z'}{sub}
    \end{gathered}
    }
\]
\[
    \infer[\rn{sub-new}]
     {
        \Delta_s,I  \vdash \rangerstepnolookup{\newobj{x}{c}}{\newobj{x}{c}}{sub}
     }
    {   }
\]
}}
\caption{Input Substitution Rules}
\label{fig:susbstitution-rules}
\end{figure}

}
\journalreport{
\begin{figure}[h!t]
    \footnotesize
\fbox{%
\parbox{\textwidth}{%
\[
    \infer[\rn{sub-var}]
     {
          \Delta_s,I \vdash \rangerexpstepnostate{x}{\Delta_s(x)}{sub}
     }
    { 
        x \in I 
    }
\qquad
    \infer[\rn{sub-var-triv}]
     {
          \Delta_s,I \vdash \rangerexpstepnostate{x}{x}{sub}
     }
    { 
     x \notin I 
    }
\]
\hrule
\hrule
\[
\infer[\rn{sub-wrap}]
     {
         \Delta_s \vdash \rangerstepnolookup{((\pi_r,\mu_s,I),s)}{(( \pi_r',\mu_s',I'),s')}{sub-wrap}
     }
     {
     \begin{gathered}
       \Delta_s,I \vdash \rangerstepnolookup{s}{s'}{sub}\qquad
         \Delta_s,I \vdash \rangerexpstepnostate{\pi_r}{\pi_r'}{sub} \\
        \Delta_s,I \vdash \rangerexpstepnostate{\mu_s}{\mu_s'}{sub} \qquad
        I'=getInput(s')\cup getUse(\mu_s') \cup getUse(\pi_r')
     \end{gathered}
     }
\]
\hrule
\hrule
\[
\infer[\rn{sub-assign}]
     {
         \Delta_s,I \vdash \rangerstepnolookup{x:=e}{x:=e'}{sub}
     }
     {
          \Delta_s,I \vdash \rangerexpstepnostate{e}{e'}{sub}
     }
\]

\[
    \infer[\rn{sub-comp}]
    {
         \Delta_s,I \vdash \rangerstepnolookup{s_1;s_2}{s_1';s_2'}{sub}
    }
    {
         \Delta_s,I \vdash \rangerstepnolookup{s_1}{s_1'}{sub}
        \quad
         \Delta_s,I \vdash \rangerstepnolookup{s_2}{s_2'}{sub}
    }
\]
}}
\caption{Selected Rule for Input Substitution Transformations}
\label{fig:susbstitution-rules}
\end{figure}

}

\begin{lem}[Expression Substitution is Sound]
\label{lem:sub-expr-sound}
For any concrete expression $e$, if
    
$$ \vdashconcestep{\Delta}{e}{v} $$ 
$$ \Delta_s,I \vdash \rangerexpstepnostate{e}{e'}{sub}$$

then for any $\Delta'$ such that 

$$\Delta' \models (\Delta_s,\pi)$$
$$\Delta \subseteq \Delta'$$

we have
$$\vdashconcestep{\Delta'}{e'}{v}$$
$$\forall i \in I.\neg isUse(i,e')$$
\end{lem}

\techreport{
\begin{proof}
    We will proceed by induction on the judgment $ \Delta_s,I \vdash \rangerexpstepnostate{e}{e'}{sub}$. For this proof, let $\Delta, H$ represents the concrete variable mapping and the concrete heap, while $\Delta_s, H_s$ and $\pi$ represents the symbolic variable mapping, the symbolic heap and the path condition, respectively. Finally, $I$ represents the input variables. Also, we have that $(\Delta',H) \models (\Delta_s,H_s,\pi)$. 

First, consider {\tt sub-val}. Since the rewrite does not modify the value $v$, and $v$ always evaluates to itself, we trivially have $\vdashconcestep{\Delta'}{v}{v}$. Moreover, $v$ does not perform any lookup in $\Delta$ during evaluation, and $\forall i \in I.\neg isUse(i,v)$ holds immediately, as $v$ is a val.

    Next, consider {\tt sub-var-triv} where $e = x$ and $x \notin I$. In this case, no rewriting occurs, i.e., $e = x = e'$.
    Now, observe that since $\Delta\subseteq\Delta'$ and the fact that $(\Delta') \models (\Delta_s,\pi)$, then evaluating $x$ in $\Delta'$ must yield the same value as the value resulting from evaluating it in $\Delta$, i.e., $\vdashconcestep{\Delta}{e'}{v}$ and $\vdashconcestep{\Delta'}{e'}{v}$ (conclusion 1).
    The second conclusion holds trivially, since $x$ is not an input.

  Now, consider {\tt sub-var}, where $e = x$ and $x \in I$, yielding $\Delta_s,I \vdash \rangerexpstepnostate{x}{e'}{sub}$ with $e' = \Delta_s(x)$. The first conclusion holds trivially from $\Delta \subseteq \Delta'$ and the fact that $(\Delta') \models (\Delta_s,\pi)$. The second conclusion holds since we have substituted the value $\Delta_s(x)$ for $x$ and we have that $\forall i \in I.\forall x \in Dom(\Delta_s).\neg isUse(i,\Delta_s(x))$, then we can conclude that $e'=\Delta_s(x)$ does not use any input variable.

  Now, consider the {\tt sub-binaryOp} rule. Let $e=\binaryop{e_1}{e_2}{c}$. We have that $\Delta_s,I \vdash \rangerexpstepnostate{\binaryop{e_1}{e_2}{c}}{\binaryop{e_1'}{e_2'}{c}}{sub}$. On the other hand, concretely evaluating $e$, we get $\vdashconcestep{\Delta}{\binaryop{e_1}{e_2}{c}}{v}$ (using {\tt binaryOp}), such that $\vdashconcestep{\Delta}{e_1}{v_1}$ and $\vdashconcestep{\Delta}{e_2}{v_2}$, and thus $v=\binaryop{v_1}{v_2}{c}$. By induction on $e_1$ and $e_2$, we get $\vdashconcestep{\Delta'}{e_1'}{v_1}$ and $\vdashconcestep{\Delta'}{e_2'}{v_2}$, where $\forall i \in I.\neg isUse(i,e_1')$ and $\forall i \in I.\neg isUse(i,e_2')$. Thus, applying the {\tt binaryOp} rule over the evaluation of $e_1'$ and $e_2'$ yields $\vdashconcestep{\Delta'}{\binaryop{e_1'}{e_2'}{c}}{v}$ (first conclusion). Also, not that  $\forall i \in I.\neg isUse(i,\binaryop{e_1'}{e_2'}{c})$, this is because $isUse(i,\binaryop{e_1'}{e_2'}{c}) = isUse(i,e_1') \vee isUse(i,e_2')$, since the binary operation does not introduce new variables.

The soundness of {\tt sub-unaryOp} and {\tt sub-$\gamma$} rewriting follows the same recursive argument as {\tt sub-binaryOp}, differing only in the number of inductive steps: one for unaryOp and three for {\tt sub-$\gamma$}.
\end{proof}
}
\journalreport{
        

}

\begin{thm}[Statement Recursive Substitution is Sound]
\label{thm:subs-rec-stmt-sound}
Let $s$ be any statement.

If
$$\Theta \vdash \concsstep{(\eta, \Delta, H, \nu,\mu)}{s}{(\eta_s, \Delta', H', \nu', \mu')}{} $$
$$  \Delta_s,I  \vdash \rangerstepnolookup{s}{s'}{sub}$$
$$\forall i \in I.\forall x \in Dom(\Delta_s).\neg isUse(i,\Delta_s(x))$$

then for any $\Delta''$ such that 

$$(\Delta'') \models (\Delta_s,\pi)$$
$$\Delta\subseteq\Delta''$$
$$\Theta \vdash \concsstep{(\eta, \Delta'', H, \nu,\mu)}{s'}{(\eta''', \Delta''', H''', \nu''',\mu''')}{}$$

Then

$$\eta'''=\eta', H'''=H', \nu'''=\nu',\mu'''=\mu'$$
$$\Delta'\subseteq\Delta'''$$
$$\forall i \in I.\neg isUse(i,s')$$

\end{thm}

\techreport{
\begin{proof}
Let $\eta_s, \Delta, H, \nu, \mu, \Delta_s, H_s, \pi, I, O, \Delta, \Delta''$ and $s$ satisfy the premises of the theorem. We will prove this theorem by using structural induction.

First, observe that the {\tt sub-skip} rule does not introduce any change to the rewritten statement. Also, observe that in {\tt skip}, no state is used/accessed during its evaluation. Therefore, we will have that  $\Delta''' = \Delta''$ and $\Delta' = \Delta$, and thus, the first and second conclusion hold, trivially. The third conclusion also holds because a $skip$ by definition cannot contain an input variable.

In the {\tt sub-new} case, the statement $s$ is unchanged ($s'=s$). Thus, we have that $(\eta', \Delta', H', \nu', \mu')$ be the state from evaluating the original statement, $\Theta \vdash \concsstep{(\eta,\Delta,H,\nu,\mu)}{\newobj{x}{c}}{(\eta',\Delta',H',\nu',\mu')}{}$, such that $\nu'=\nu+1, \Delta'=\Delta[x\leftarrow r], \eta'=\eta[r \leftarrow c], r=\text{freshRef}(\nu')$. 

Now, let us assume we have $\Delta''$ such that it satisfies the symbolic state, i.e., $(\Delta'') \models (\Delta_s,\pi)$, and $\Delta\subseteq\Delta''$. Since, obtaining a new reference is deterministically defined over $\nu'$, i.e., $r=\text{freshRef}(\nu')$. thus, on evaluating the non-changing rewritten statement $s'$ in $\Delta''$, we get $\Theta \vdash \concsstep{(\eta,\Delta'',H,\nu,\mu)}{\newobj{x}{c}}{(\eta''',\Delta''',H''',\nu''',\mu''')}{}$. Here, we also have that $r=\text{freshRef}(\nu')$, thus we know that $\nu'''=\nu+1=\nu', \eta'''=\eta[r \leftarrow c]=\eta', H'''=H=H', \mu'''=\mu=\mu'$ (conclusion 1).

Also, we have that the execution of $s$ and $s'$ introduces exactly one new mapping, the binding $x$, such that $\Delta'''=\Delta''[x\leftarrow r]$, and $\Delta'=\Delta[x\leftarrow r]$, and since we have that $\Delta\subseteq\Delta''$, then by transitivity we can conclude that $\Delta'\subseteq\Delta'''$ (conclusion 2).

By the guarantee of previous transformations, we know that $SSA(s)$ most hold (Property~\ref{prop:ssa}), $z$ cannot appear in either $\Delta$ or $\Delta''$, since otherwise there would be multiple definitions for $z$, violating the generalized SSA property; moreover, $x \notin I$ because $isDef(z,\newobj{x}{c})$ holds, while $I$ only contains variables $i$ such that $isUse(i,s)\wedge \lnot isDef(i,s)$. 
Thus, we can conclude that $\forall i \in I.\ \neg isUse(i,\newobj{x}{c})$ (conclusion 3).

Consider the {\tt sub-assign} rule. From the concrete semantics, we have $\Theta \vdash \concsstep{(\eta,\Delta,H,\nu,\mu)}{x:=e}{(\eta',\Delta',H',\nu',\mu')}{}$ by {\tt assign}, where $\vdashconcestep{\Delta}{e}{v}$ and $\Delta'=\Delta[x\leftarrow v]$, while $H'=H$, $\eta'=\eta_s$, $\nu'=\nu$, and $\mu'=\mu$. Applying {\tt sub-assign}, we obtain $\Delta_s,I \vdash \rangerstepnolookup{x:=e}{x:=e_1}{sub}$ with $\Delta_s,I \vdash \rangerexpstepnostate{e}{e_1}{sub}$. Now, let us assume that we have $\Delta''$ such that it satisfies the symbolic state, i.e., $(\Delta'') \models (\Delta_s,\pi)$, and $\Delta\subseteq\Delta''$. By the lemma on expression substitution (Lemma~\ref{lem:sub-expr-sound}), $\vdashconcestep{\Delta''}{e_1}{v}$.
Now, concretely executing the rewritten statement in $\Delta''$, we get $\Theta \vdash \concsstep{(\eta,\Delta'',H,\nu,\mu)}{x:=e_1}{(\eta''',\Delta''',H''',\nu''',\mu''')}{}$, which must be of the form $\eta'''=\eta=\eta'$, $H'''=H=H', \nu'''=\nu=\nu'$, and $\mu'''=\mu=\mu'$ (conclusion 1). 
Note that, we also have that we assumed $\Delta\subseteq\Delta''$, and that the only change introduced in the mapping for $x$, i.e., $\Delta'=\Delta[x\leftarrow v]$, $\Delta'''=\Delta''[x\leftarrow v]$. By transitivity, we can conclude that $\Delta'\subseteq\Delta'''$ (conclusion 2).
Finally, since $\forall i\in I.\neg isUse(i,e_1)$ by the expression lemma and $x$ cannot be not an input, since $x \notin getInput(x:=e)$ while $I$ only contains variables $i$ such that $isUse(i,s)\wedge \lnot isDef(i,s)$, we conclude $\forall i\in I.\neg isUse(i,x:=e_1)$ (conclusion 3).

 Next, consider the {\tt sub-comp} rule.   
 There are two cases that can occur when evaluating a compositional statement in the concrete semantics: either $s_1$ has a return control flow (thus the {\tt composition-ret} rule will apply), or not (thus the {\tt composition} rule will apply). Since, the substitution transformation happens after the remove final return transformation (Property~\ref{prop:final-ret-no-return}), thus we will skip the case where a return can happen in $s_1$. Thus, to show that the lemma holds for the {\tt composition} rule, we have
  $\Theta \vdash \concsstep{(\eta, \Delta, H, \nu, \mu)}{s_1}{(\eta_1, \Delta_1, H_1, \nu_1, \mu)}{}$ and
       $\Theta \vdash \concsstep{(\eta_1, \Delta_1, H_1, \nu_1, \mu)}{s_2}{(\eta_2, \Delta_2, H_2, \nu_2, \mu)}{}$.
Here we have that $\eta'=\eta_2, \Delta'=\Delta_2, H'=H_2, \nu'=\nu_2, \mu'=\mu$.
Now, applying the rewriting, we get $\Delta_s,I \vdash \rangerstepnolookup{s_1;s_2}{s_1';s_2'}{sub}$.
Now, let us assume we have $\Delta''$ such that it satisfies the symbolic state, i.e., $(\Delta'') \models (\Delta_s,\pi)$, and where $\Delta \subseteq \Delta''$.
 Using induction hypothesis on $s_1$, we have that
     $  \Theta \vdash \concsstep{(\eta, \Delta'', H, \nu, \mu)}{s_1'}{(\eta_1, \Delta_3, H_1, \nu_1, \mu)}{}$, such that $\Delta_1\subseteq\Delta_3$. 
     From Lemma~\ref{lem:delta-ssa}, we know that $\Delta'' \subseteq\Delta_1$, since we have established that $(\Delta'') \models (\Delta_s,\pi)$, then by transitivity, it must be that $(\Delta_3) \models (\Delta_s,\pi)$, and similarly, $(\Delta'') \models (\Delta_s,\pi)$
     Now, we can use this state and apply induction on $s_2$ to get
     $\Theta \vdash \concsstep{(\eta_1, \Delta_3, H_
     1, \nu_1, \mu)}{s_2'}{(\eta_2, \Delta_4, H_2, \nu_2, \mu)}{}$, such that $\Delta_2\subseteq\Delta_4$. Here, we have that $\eta'''=\eta_2=\eta', H'''=H_2=H', \nu'''=\nu_2=\nu',\mu'''=\mu_2=\bot=\mu'$ (conclusion 1). 
     
     Also, we have that $\Delta'''=\Delta_4$, and $\Delta'=\Delta_2$, and we have shown that $\Delta_2(x)\subseteq\Delta_4(x)$ (conclusion 2). 
     Finally, to prove $\forall i \in I.\,\neg isUse(i,s_1';s_2')$, 
we apply induction on both $s_1$ and $s_2$. From the induction hypothesis, 
$\forall i \in I.\,\neg isUse(i,s_1')$ and $\forall i \in I.\,\neg isUse(i,s_2')$. 
By the definition of $isUse$, a variable is used in a composition only if it is 
used in one of the components. Hence, 
$\forall i \in I.\,\neg isUse(i,s_1';s_2')$ (conclusion 3).

   Consider the {\tt sub-getfield} rule, we have that $\Theta \vdash \concsstep{(\eta, \Delta, H, \nu, \bot)}{\getfield{x}{z}{f}}{(\eta, \Delta [x \leftarrow v], H, \nu, \bot)}{}$, where $r \in R$, $r=\Delta(z)$, and $v=H(r,f)$, by applying the {\tt get-field} rule.   Here we have that $\eta'=\eta_s, \Delta'=\Delta [x \leftarrow v],H'=H,\nu'=\nu,\mu'=\bot$. On the other hand, we can apply the {\tt sub-getfield} rule to get $\Delta_s, I  \vdash \rangerstepnolookup{\getfield{x}{z}{f}}{\getfield{x}{z'}{f}}{sub}$, for some reference variable $z$ and a field $f$. 
   Now, let us assume we have $\Delta''$ such that it satisfies the symbolic state, i.e., $(\Delta'') \models (\Delta_s,\pi)$, where $\Delta\subseteq\Delta''$.
   Using the substitution lemma for expression soundness (Lemma~\ref{lem:sub-expr-sound}), we have $\vdashconcestep{\Delta''}{z'}{r}$. 
   Thus, evaluating concretely the rewritten statement we get, $\Theta \vdash \concsstep{(\eta, \Delta'', H, \mu)}{\getfield{x}{z'}{f}}{(\eta, \Delta'' [x \leftarrow v'], H, \bot)}{}$.
   Here, we have that $\eta'''=\eta=\eta', H'''=H=H', \nu'''=\nu=\nu',\mu'''=\mu=\bot=\mu'$ (conclusion 1). 
   Next, we have that $\Delta'''=\Delta'' [x \leftarrow v']$, $v'=H(r,f)=v$. 
   From the SSA guarantees from previous transformation, we know that $SS(s)$ most hold (Property~\ref{prop:ssa}), $x$ cannot appear in $\Delta$ or $\Delta''$, nor can $x\in I$ since $isDef(x,\getfield{x}{z'}{f})$ holds while $I$ only contains variables $i$ such that $isUse(i,s)\wedge \lnot isDef(i,s)$. As $[x\leftarrow v]$ is the only update in both $\Delta'$ and $\Delta'''$, pairing that with the assumed fact that $\Delta \subseteq \Delta''$, then we can conclude that $\Delta'\subseteq\Delta'''$(conclusion 2). 
   
   Finally, to show that $\forall i\in I.\neg isUse(i,s')$, observe that we have 
   $\forall i\in I.\neg isUse(i,z')$ from the expression lemma, then it must be the case that $\forall i\in I.\neg isUse(i,\getfield{x}{z'}{f})$ (conclusion 3). 

Consider the {\tt sub-putfield} rule, 
we apply the {\tt put-field} rule $\Theta \vdash \concsstep{(\eta, \Delta, H, \nu, \bot)}{\putfield{z}{f}{e}}{(\eta_s, \Delta, H[(r,f) \leftarrow v], \nu, \bot)}{}$, such that $r=\Delta(z), r \in R$, and $\vdashconcestep{\Delta}{e}{v}$.
Here we have that $\eta'=\eta_s,\Delta'=\Delta,H'=H[(r,f) \leftarrow v],\nu'=\nu,\mu'=\bot$.
On the other hand, we can apply the {\tt sub-putfield} rule to get $\Delta_s,I  \vdash \rangerstepnolookup{\putfield{z}{f}{e}}{\putfield{z'}{f}{e'}}{sub}$ , such that $
    \Delta_s,I  \vdash \rangerexpstepnostate{e}{e'}{sub}
    $, $\Delta_s,I  \vdash \rangerexpstepnostate{z}{z'}{sub}$. 
     Now, let us assume we have $\Delta''$ such that it satisfies the symbolic state, i.e., $(\Delta'') \models (\Delta_s,\pi)$, where $\Delta\subseteq\Delta''$.
    First, observe that both $z'$ and $e'$ will evaluate to the same value as their non-rewritten counter parts, due to Lemma~\ref{lem:sub-expr-sound}, i.e, $\vdashconcestep{\Delta''}{e'}{v}$, and $\vdashconcestep{\Delta''}{z'}{r}$
Now, evaluating the rewritten statement we get, 
$
\Theta \vdash \concsstep{(\eta, \Delta'', H, \nu, \bot)}{\putfield{z'}{f}{e'}}{(\eta''',\Delta''', H'', \nu''', \mu''')}{}$. Note here we get the same heap, that is, $H''=H[(r,f) \leftarrow v]=H'$. Also, observe that no other environment variables are changed, thus we can arrive to conclusion 1. Next, observe that no change has occurred on the local variable mapping, i.e., $\Delta'''=\Delta''$, and that $\Delta'=\Delta$  pairing that with the assumed fact that $\Delta \subseteq \Delta''$, then we can conclude that $\Delta'\subseteq\Delta'''$(conclusion 2).

Finally, one of the conclusions of applying the expression lemma on both $z$, and $e$, is that that $\forall i\in I.\neg isUse(i,z')$, and $\forall i\in I.\neg isUse(i,e')$, thus we can conclude that  $\forall i\in I.\neg isUse(i,\putfield{z'}{f}{e'})$, since the {\tt putfield} construct does not introduce additional variables (other than those in $z'$ and in $e'$) to the rewritten statement (conclusion 3).

Now, consider the {\tt sub-invoke} rule, we have that $\Theta\vdash\concsstep{(\eta,\Delta, H, \nu,\mu)}{\invoke{x}{z}{g}{y}}{(\eta', \Delta',H',\nu',\mu')}{}$, such that $\mu=\bot$, $\vdashconcestep{\Delta}{z}{r}$, $r\in R$, $\lookup{\Theta}{\lookup{\eta}{r},g} = method(\overrightarrow{y_1}.s_1)$, $
     \vdashconcestep{\Delta}{\overrightarrow{y}}{\overrightarrow{v_1}}
    $, and $
     \Theta\vdash \concsstep{(\eta,\{(\overrightarrow{y_1},\overrightarrow{v_1})\},H,\nu,\bot)}{s_1}{(\eta^5, \Delta^5,H^5,\nu^5,\mu^5)}{}$. Here, we have that $\eta'=\eta^5$, $\Delta'=\update{\Delta}{x}{\mu^5}$, $H'=H^5$, $\nu'=\nu^5=$, and $\mu^5=v', v'\neq\bot$.
     Next, we can apply the {\tt sub-invoke} to get
     $\Delta_s,I \vdash \rangerstepnolookup{\invoke{x}{z}{g}{y}}{\invoke{x}{z'}{g}{y'}}{sub}$ and we have $
      \Delta_s,I \vdash \rangerexpstepnostate{\overrightarrow{y}}{\overrightarrow{y'}}{sub}$, and $\Delta_s,I \vdash \rangerexpstepnostate{z}{z'}{sub}$.

    Now, let us assume we have $\Delta''$ such that it satisfies the symbolic state, i.e., $(\Delta'') \models (\Delta_s,\pi)$, where $\Delta\subseteq\Delta''$.
  First, observe that both $z'$ and $y'$ will evaluate to the same value as their non-rewritten counter parts, due to Lemma~\ref{lem:sub-expr-sound}, i.e, $\vdashconcestep{\Delta''}{\overrightarrow{y'}}{\overrightarrow{v_1}}$, and $\vdashconcestep{\Delta''}{z'}{r}$.  This implies that the fetched method for invocation must also match the one that is fetched when evaluating the original statement, that is $\lookup{\Theta}{\lookup{\eta}{r},g} = method(\overrightarrow{y_1}.s_1)$. Now, we can evaluate $s_1$ in exactly the same way in the original statement, to get 
 $\Theta\vdash \concsstep{(\eta,\{(\overrightarrow{y_1},\overrightarrow{v_1})\},H,\nu,\bot)}{s_1}{(\eta^6, \Delta^6,H^6,\nu^6,\mu^6)}{}$. But since we are evaluating $s_1$ in exactly the same environment then it must be the case that $\eta^6=\eta^5$, $\Delta^6=\Delta^5$, $H^6=H^5$, $\nu^6=\nu^5$, $\mu^6=\mu^5$. Moreover, we can also conclude that the resulting concrete state $(\eta''', \Delta''',H''',\nu''',\mu''')$ must be of the form $\eta'''=\eta^6=\eta^5=\eta'$, $H'''=H^6=H^5=H'$, $\nu'''=\nu^6=\nu^5=\nu'$, and finally, $\mu'''=\mu^6=\mu^5=\bot$ (conclusion 1). 

 Next, observe that $\Delta'''=\Delta''[x\leftarrow\mu^5]$, and similarly we have that $\Delta'=\Delta[x\leftarrow\mu^5]$, and since we have established that $\Delta\subseteq\Delta''$, thus adding the single mapping for $x$, must still maintain the subset relationship, i.e., $\Delta'\subseteq\Delta'''$ (conclusion 2).

Finally, it is easy to see that $\forall i \in I.\neg isUse(i,\invoke{z'}{x}{g}{y'})$ by using the expression lemma on both $z$ and $\overrightarrow{y}$, we obtain that $\forall i\in I.\neg isUse(i,z')$, and $\forall i\in I.\neg isUse(i,\overrightarrow{y'})$, thus we can conclude that  $\forall i\in I.\neg isUse(i,\invoke{z'}{x}{g}{y'})$, since the {\tt invoke} construct does not introduce additional variables (other than those in $z'$ and in $\overrightarrow{y'}$) to the rewritten statement (conclusion 3).

Finally, consider the {\tt sub-if} rule. The proof splits into several cases depending on whether the if-condition evaluates to {\tt true} or {\tt false}, and whether the executed branch contains a {\tt return}. We present the proof for two representative cases: when the condition evaluates to {\tt true} and the taken branch $s_1$ has no {\tt return}, and when the condition still evaluates to {\tt true} but $s_1$ contains a {\tt return}, noting that the latter is skipped due to SSA guarantees from pervious transformations (Property~\ref{prop:ssa}). Recall that this transformation is applied after the $\gamma$-transformation, which eliminates $\phi$-statements from if-statements; therefore, we assume the sequence of $\phi$-statements is empty.

First, let us consider the first case. Here we have that $\Theta\vdash\concsstep{(\eta, \Delta, H, \nu,\bot)}{\myif{e}{s_1}{s_2}}{(\eta',\Delta', H', \nu',\mu')}{}$, where $\vdashconcestep{\Delta}{e}{true}$, $\Theta\vdash\concsstep{(\eta,\Delta, H, \nu,\mu)}{s_1}{(\eta',\Delta', H', \nu',\mu')}{}$, and $\mu'=\bot=\mu$.
     Applying the {\tt sub-if} rule we get $ \Delta_s,I  \vdash \myif{e}{s_1}{s_2}\longrightarrow_{sub} {\myif{e'}{s_1'}{s_2'}}$, such that we have $\Delta_s,I  \vdash \rangerexpstepnostate{e}{e'}{sub}$, and also $
    \Delta_s,I  \vdash \rangerstepnolookup{s_1}{s_1'}{sub}
        $, and finally, $
         \Delta_s,I  \vdash \rangerstepnolookup{s_2}{s_2'}{sub}$.       
          First, observe that and $e'$ must evaluate to the same value as its non-rewritten counter parts, due to Lemma~\ref{lem:sub-expr-sound}, i.e, $\vdashconcestep{\Delta''}{e'}{true}$. Thus, the concrete evaluation will evaluate the then-side ($s_1$). Here, we use induction over $s_1$. That is, if we have that $\Theta \vdash \concsstep{(\eta, \Delta'', H, \nu,\bot)}{s_1'}{(\eta_1, \Delta_1, H_1, \nu_1, \mu_1)}{}$, then by induction, we know that  $\eta_1=\eta', H_1=H', \nu_1=\nu', \mu_1=\bot$, and $\Delta_1=\Delta'''$ (conclusion 1). 
         Similarly by induction we have that $\Delta'\subseteq\Delta'''$ (conclusion 2).
   Finally, to show that $\forall i\in I.\neg isUse(i,s')$, observe that we have 
   $\forall i\in I.\neg isUse(i,e')$ holds by using the expression lemma, also, by using induction over $s_1$, and $s_2$, we have that $\forall i\in I.\neg isUse(i,s_1') \wedge \neg isUse(i,s_2')$.
   then it must be the case that $\forall i\in I.\neg isUse(i,\myif{e'}{s_1'}{s_2'})$ (conclusion 3).
\end{proof}
}

\begin{cor}[The substitution-wrapper is sound] 
\label{cor:subs-wrap-stmt-sound}

If
$$\Theta \vdash \concsstep{(\eta, \Delta, H, \nu,\mu)}{s}{(\eta', \Delta', H', \nu', \mu')}{} $$
$$\Delta_s \vdash \rangerstepnolookup{((\pi_r,\mu_s, I),s)}{((\pi_r',\mu_s',I'),s')}{sub-wrap}$$
$$\forall i \in I.\forall x \in Dom(\Delta_s).\neg isUse(i,\Delta_s(x))$$
$$\mu^*=\func{evalReturn}(\Delta',\pi_r,\mu_s)$$

then for any $\Delta''$ such that 

$$(\Delta'') \models (\Delta_s,\pi)$$
$$\Delta\subseteq\Delta''$$
$$\Theta \vdash \concsstep{(\eta, \Delta'', H, \nu,\mu)}{s'}{(\eta''', \Delta''', H''', \nu''',\mu''')}{}$$
we have
$$\eta'''=\eta', H'''=H', \nu'''=\nu',\mu'''=\mu'$$
$$\Delta'\subseteq\Delta'''$$
$$\forall i \in I.\neg isUse(i,s')$$
$$(\Delta''') \models (\Delta_s,\pi') $$
$$\mu^*=\func{evalReturn}(\Delta''',\pi_r',\mu_s')$$

\techreport{

\begin{proof}
  
  Applying {\tt sub-wrap rule} on a statement $s$, we get $\Delta_s \vdash \rangerstepnolookup{((\pi_r,\mu_s,I),s)}{((\pi_r',\mu_s',I'),s')}{sub-wrap}$, such that $ \Delta_s,I \vdash \rangerstepnolookup{s}{s'}{sub}, 
         \Delta_s,I \vdash \rangerexpstepnostate{\pi_r}{\pi_r'}{sub},
        \Delta_s,I \vdash \rangerexpstepnostate{\mu_s}{\mu_s'}{sub} $, and $
        I'=getInput(s')\cup getUse(\mu_s') \cup getUse(\pi_r')$.

        Given the premise of the corollary we can apply the  of Theorem~\ref{thm:subs-rec-stmt-sound} on $s'$ to show conclusion 1-4:
        $$\Theta \vdash \concsstep{(\eta, \Delta'', H, \nu,\mu)}{s'}{(\eta''', \Delta''', H''', \nu''',\mu''')}{}$$
$$\eta'''=\eta', H'''=H', \nu'''=\nu',\mu'''=\mu'$$
$$\Delta'\subseteq\Delta'''$$
$$\forall i \in I.\neg isUse(i,s')$$
        
       Next, we want to show that $(\Delta'') \models (\Delta_s,\pi)$. Since the evaluation of this relation depends on its arguments. From Lemma~\ref{lem:delta-ssa}, we know that $\Delta''\subseteq \Delta'''$, then $\pi$ will be evaluated exactly the same in $\Delta''$, which in turn guarantees that $(\Delta''') \models (\Delta_s,\pi)$ must also hold.

       Finally, from the premise of the corollary, we have $\mu^* = \func{evalReturn}(\Delta',\pi_r,\mu_s)$. 
       The evaluation of $\func{evalReturn}$ depends on the evaluation of the expressions $\pi_r$ and $\mu_s$ under $\Delta$. 
Since we have shown that $\Delta' \subseteq \Delta'''$ and $\pi_r', \mu_s'$ are the rewritten counterparts of $\pi_r, \mu_s$, respectively, Lemma~\ref{lem:sub-expr-sound} applies. 
Therefore,
\[
\func{evalReturn}(\Delta''',\pi_r',\mu_s') = \func{evalReturn}(\Delta',\pi_r,\mu_s) = \mu^*.
\]

\end{proof}
}
\journalreport{

\begin{proof}
  
  Applying {\tt sub-wrap rule} on a statement $s$, we get $\Delta_s \vdash \rangerstepnolookup{((\pi_r,\mu_s,I),s)}{((\pi_r',\mu_s',I'),s')}{sub-wrap}$, such that $ \Delta_s,I \vdash \rangerstepnolookup{s}{s'}{sub}, 
         \Delta_s,I \vdash \rangerexpstepnostate{\pi_r}{\pi_r'}{sub},
        \Delta_s,I \vdash \rangerexpstepnostate{\mu_s}{\mu_s'}{sub} $, and $
        I'=getInput(s')\cup getUse(\mu_s') \cup getUse(\pi_r')$.

        Given the premise of the corollary we can apply the  of Theorem~\ref{thm:subs-rec-stmt-sound} on $s'$ to show conclusion 1-4:
        $$\Theta \vdash \concsstep{(\eta, \Delta'', H, \nu,\mu)}{s'}{(\eta''', \Delta''', H''', \nu''',\mu''')}{}$$
$$\eta'''=\eta', H'''=H', \nu'''=\nu',\mu'''=\mu'$$
$$\Delta'\subseteq\Delta'''$$
$$\forall i \in I.\neg isUse(i,s')$$
        
       Next, we want to show that $(\Delta'') \models (\Delta_s,\pi)$. Since the evaluation of this relation depends on its arguments. From Lemma~\ref{lem:delta-ssa}, we know that $\Delta''\subseteq \Delta'''$, then $\pi$ will be evaluated exactly the same in $\Delta''$, which in turn guarantees that $(\Delta''') \models (\Delta_s,\pi)$ must also hold.

       Finally, from the premise of the corollary, we have $\mu^* = \func{evalReturn}(\Delta',\pi_r,\mu_s)$. 
       The evaluation of $\func{evalReturn}$ depends on the evaluation of the expressions $\pi_r$ and $\mu_s$ under $\Delta$. 
Since we have shown that $\Delta' \subseteq \Delta'''$ and $\pi_r', \mu_s'$ are the rewritten counterparts of $\pi_r, \mu_s$, respectively, Lemma~\ref{lem:sub-expr-sound} applies. 
Therefore,
\[
\func{evalReturn}(\Delta''',\pi_r',\mu_s') = \func{evalReturn}(\Delta',\pi_r,\mu_s) = \mu^*.
\]

\end{proof}
}
\end{cor}

Let $\Delta^*$ be the restriction of $\Delta''$ to variables not in $I$, i.e.,
\[
Dom(\Delta^*) = Dom(\Delta'') \setminus I
\]

\begin{cor}[Substitution-wrapper is sound even if we remove input variables]
\label{cor:sub-wrap-remove-I}
Assume the premises of Corollary~\ref{cor:subs-wrap-stmt-sound}.
Let $\Delta^*$ be defined as above, we observe that $\Delta^* \subset \Delta''$.
Then, we know that if 
\[
\Theta \vdash \concsstep{(\eta, \Delta^*, H, \nu,\mu)}{s'}
{(\eta^4, \Delta^{**}, H^4, \nu^4,\mu^4)}{}
\]
then the following hold:
\[
\eta^4=\eta',\quad H^4=H',\quad \nu^4=\nu',\quad \mu^4=\mu'
\]
\[
\forall x.\in Dom(\Delta_s).\concestep{\Delta^{*}}{\Delta_s(x)}{\Delta''(x)}{} 
\]
\[
\Delta^{**} \subset \Delta'''
\]
\[
\mu^*=\func{evalReturn}(\Delta^{**},\pi_r',\mu_s')
\]
$$(\Delta^{**}) \models (\Delta_s,\pi)$$

\begin{proof}

By Corollary~\ref{cor:subs-wrap-stmt-sound}, we have $\forall i \in I.\neg isUse(i,s')$.
Since $\Delta^*$ differs from $\Delta''$ only on variables in $I$, the execution of $s'$
under $\Delta^*$ yields $\eta^4=\eta'''=\eta',\quad H^4=H'''=H',\quad \nu^4=\nu'''=\nu',\quad \mu^4=\mu'''=\mu'$ (conclusion 1).

Also, observe that since we have established that $(\Delta''') \models (\Delta_s,\pi)$, and that we also know from the premise of Corollary~\ref{cor:subs-wrap-stmt-sound} that $\forall i \in I.\forall x \in Dom(\Delta_s).\neg isUse(i,\Delta_s(x))$. Then this implies that the evaluation of any variable $x$ in $Dom(\Delta_s)$ is not dependent on the input variable mapping, and since the only difference between $\Delta^*$ and $\Delta''$ is the input variable mappings, then we can conclude that $\forall x.\in Dom(\Delta_s).\concestep{\Delta^{*}}{\Delta_s(x)}{\Delta''(x)}{}$ (conclusion 2).

We want to know that the subset relation will be preserved with the execution of $s'$, and this follows because the execution of $s'$ initially starts with the assumption that both $\Delta^*$ and $\Delta''$ match up on all variable excluding those in $I$. Thus, expressions, as they are not using any variables from $I$, must evaluate to the same values in both environments. Moreover, all state updates must also make the same updates, thus preserving the subset relation between resulting states, i.e., $\Delta^{**} \subset \Delta'''$.

Also, by Lemma~\ref{lem:sub-expr-sound}, input variables do not occur in
$\pi_r'$ or $\mu_s'$, and since we have from Corollary~\ref{cor:subs-wrap-stmt-sound} that, $\mu^*=\func{evalReturn}(\Delta''',\pi_r',\mu_s')$, then we can conclude that evaluating $\pi_r'$ and $\mu_s'$ under $\Delta^{**}$, and is unchanged.
Therefore,
\[
\mu^*=\func{evalReturn}(\Delta^{**},\pi_r',\mu_s').
\]

Finally, by Lemma~\ref{lem:sub-expr-sound}, input variables do not occur in $\pi_r'$ and that it evaluates to the same value in $\Delta'''$ according to the same lemma. We also have established that $\forall x.\in Dom(\Delta_s).\concestep{\Delta^{*}}{\Delta_s(x)}{\Delta''(x)}{}$, then that means that it must evaluate to the same value in $\Delta**$. Same reasoning applies for all expressions mapped to by any variable $x$ such that $s \in Dom(\Delta_s)$, this implies that the evaluation for $(\Delta^{**}) \models (\Delta_s,\pi)$ must hold.
\end{proof}

\end{cor}

\subsection{Substitution Properties}

The substitution assumes all guaranteed properties from $\gamma$-creation, remove final-returns and renaming transformations. And it guarantees the following properties.

Given $ \Delta_s \vdash \rangerstepnolookup{((\pi_r,\mu_s,I),s)}{((\pi_r',\mu_s',I'),s')}{sub-wrap}$, such that $ \Delta_s,I \vdash \rangerstepnolookup{s}{s'}{sub}$ and $
         \Delta_s,I \vdash \rangerexpstepnostate{\pi_r}{\pi_r'}{sub}$, $\Delta_s,I \vdash \rangerexpstepnostate{\pi}{\pi'}{sub}$ where $
        \Delta_s,I \vdash \rangerexpstepnostate{\mu_s}{\mu_s'}{sub}$ and $
        I'=getInput(s')\cup getUse(\mu_s') \cup getUse(\pi_r')$, then the substitution transformation has the following guarantees. 

\begin{itemize}
    \item Rewritten Statement $s'$ satisfies the generalized SSA property. More formally, $SSA(s')$ holds.
    \techreport{
    \begin{proof}
        This follows directly, from renaming transformation guarantee (Property~\ref{prop:ssa}) that the $SSA(s)$ holds and that none of the expression rules ($\mapsto_{sub}$) alter the local variable mapping or the heap mappings. Likewise, none of the statement rules ($\rightarrow_{sub}$) modify the variables assigned within generalized assignment statements, nor do they introduce new variables. Consequently, because neither variable mappings nor assignments are changed, the generalized SSA property is preserved for $s'$, i.e., $SSA(s')$ holds.
    \end{proof}
    }
    \item No input variables $i \in I$ exist in the rewritten statement $s'$. 
    \label{gur:sub-no-inpu-in-s-prime}
    More formally, $\forall x\in I, \neg isUse(x,s')$
    \techreport{
    \begin{proof}
        This follows directly from applying the soundness theorem~\ref{cor:sub-wrap-remove-I} of the substitution transformation.
    \end{proof}
    }
\end{itemize}

\section{$TR_6$: Method Inlining Transformation}
\label{sec:inlining}
This transformation inlines the IR of called methods one level at a time. 
Inlining is performed only when the object reference on which the method is 
invoked is a concrete reference. 

Fig.~\ref{fig:inlinerules} shows a snippet of the main rules for this 
transformation. 
The rules define two judgments. The judgment $\Theta_t, \eta_s, \Delta_s 
\vdash \rangerstep{(T, u)}{s}{(T', u')}{s'}{inline}$ rewrites an IR statement 
$s$ into $s'$ by inlining method calls, where $\Theta_t$ is the mapping of 
transformed method bodies, $\eta_s$ is the reference-to-type mapping, $\Delta_s$ 
is the symbolic local variable mapping, $T$ is the set of temporary variables, 
and $u$ is the global unique index. Their primed counterparts $T'$ and $u'$ 
reflect their updated values after rewriting. The wrapper judgment 
$\Theta_t, \eta_s, \Delta_s \vdash \rangerstep{(T, u)}{s}{(T', u')}{s'}{inline\text{-}wrap}$ 
is the entry point of the transformation, triggering the inline judgment on $s$.

Rule~{\tt inline-invoke} inlines a method call by looking up 
its transformed body in $\Theta_t$, which contains the IR of methods on which 
the $\gamma$-transformation, early-return elimination, and final-return removal 
have already been applied (see Appendix~\ref{sec:statictransform}). The rule 
retrieves the transformed body, applies the renaming transformation to ensure 
uniqueness of variable names, and applies substitution to propagate the passed 
arguments into the body. The variable $x$ is then assigned the modified return expression $\mu_s^3$ of the inlined method. Rules~{\tt inline-comp} propagates the transformation into their sub-statements, threading the temporary 
variable set $T$ and counter $u$ through each; the rule for conditionals, omitted from 
Fig.~\ref{fig:inlinerules} for brevity, follows the same pattern. All remaining statement forms, 
including assignments, getfield, and putfield, etc., are handled by 
Rule~{\tt inline-no-op}, which leaves them unchanged.

\techreport{
 \begin{figure}[H]
    \footnotesize
\fbox{%
\parbox{\textwidth}{%
\[
    \infer[\rn{inline-wrap}]
     {\Theta_t,\eta_s,\Delta_s \vdash \rangerstep{(T,u)}{s}{(T',u')}{s'}{inline-wrap}}
    {
     \begin{gathered}
    \Theta_t,\eta_s,\Delta_s\vdash \rangerstep{(T,u)}{s}{(T',u')}{s'}{inline}
    \end{gathered}
    }
\]

\[
    \infer[\rn{inline-if}]
     {\Theta_t,\eta_s,\Delta_s\vdash \rangerstep{(T,u)}{\myif{e}{s_1}{s_2}}{(T'', u'')} {\myif{e}{s_1'}{s_2'}}{inline}}
    {
     \begin{gathered}
        \Theta_t,\eta_s,\Delta_s\vdash \rangerstep{(T,u)}{s_1}{(T',u')}{s_1'}{inline}
        \\
       \Theta_t,\eta_s,\Delta_s\vdash \rangerstep{(T',u')}{s_2}{(T'',u'')}{s_2'}{inline}
    \end{gathered}
    }
\]

\[
    \infer[\rn{inline-invoke}]
     {\Theta_t,\eta_s,\Delta_s\vdash \rangerstep{(T,u)}{\invoke{x}{r}{g}{y}}{(T^3,u')}{s^3; x:=\mu_s^3}{inline}}
    {
     \begin{gathered}
     r \in R \qquad
     ((\mu_s',\pi_r'), \overrightarrow{y'}.s')=\lookup{\Theta_t}{\lookup{\eta_s}{r},g}\qquad
    O'=\varnothing \qquad I'=\overrightarrow{y'}
    \qquad T'=getOutput(s')\\
      I',O'\vdash \rangerstepnolookup{((\mu_s',\pi_r',T',u),s')}{((\mu_s'',\pi_r'', T'', u'),s'')}{rename-wrap}\\
         \sigma = (I', \overrightarrow{y'}) \qquad 
          \sigma \vdash \rangerstepnolookup{((true,\pi_r'',\mu_s'',I'),s'')}{(( true,\pi_r^3,\mu_s^3,I''),s^3)}{sub-wrap}
    \qquad T^3=T''\cup T
    \end{gathered}
    }
\]


\[
 \infer[\rn{inline-comp}]
     {\Theta_t,\eta_s,\Delta_s\vdash \rangerstep{(T,u)}{s_1;s_2}{(T'',u'')}{s_1';s_2'}{inline}}
    { 
    \begin{gathered}
    \Theta_t,\eta_s,\Delta_s,\vdash \rangerstep{(T,u)}{s_1}{(T',u')}{s_1'}{inline}
    \\
    \Theta_t,\eta_s,\Delta_s\vdash \rangerstep{(T',u')}{s_2}{(T'',u'')}{s_2'}{inline}        
    \end{gathered}
    }
\]

\[
 \infer[\rn{inline-no-op}]
     {\Theta_t,\eta_s,\Delta_s \vdash \rangerstep{(T,u)}{s}{(T,u)}{s}{inline}}
    {  \func{noOp}_{inline}(s) }
\]
}}
\caption{Method-Inlining Rules}
\label{fig:inlinerules}
\end{figure}
 }

 \journalreport{
 \begin{figure}[h!t]
    \footnotesize
\fbox{%
\parbox{\textwidth}{%
\[
    \infer[\rn{inline-wrap}]
     {\Theta_t,\eta_s,\Delta_s \vdash \rangerstep{(T,u)}{s}{(T',u')}{s'}{inline-wrap}}
    {
     \begin{gathered}
    \Theta_t,\eta_s,\Delta_s\vdash \rangerstep{(T,u)}{s}{(T',u')}{s'}{inline}
    \end{gathered}
    }
\]

\[
    \infer[\rn{inline-invoke}]
     {\Theta_t,\eta_s,\Delta_s\vdash \rangerstep{(T,u)}{\invoke{x}{r}{g}{y}}{(T^3,u')}{s^3; x:=\mu_s^3}{inline}}
    {
     \begin{gathered}
     r \in R \qquad
     ((\mu_s',\pi_r'), \overrightarrow{y'}.s')=\lookup{\Theta_t}{\lookup{\eta_s}{r},g}\qquad
    O'=\varnothing \qquad I'=\overrightarrow{y'}
    \qquad T'=getOutput(s')\\
      I',O'\vdash \rangerstepnolookup{((\mu_s',\pi_r',T',u),s')}{((\mu_s'',\pi_r'', T'', u'),s'')}{rename-wrap} \qquad \sigma = (I', \overrightarrow{y'})\\
          \sigma \vdash \rangerstepnolookup{((\pi_r'',\mu_s'',I'),s'')}{((\pi_r^3,\mu_s^3,I''),s^3)}{sub-wrap}
    \qquad T^3=T''\cup T
    \end{gathered}
    }
\]
\soha{\[
    \infer[\rn{inline-if}]
     {\Theta_t,\eta_s,\Delta_s\vdash \rangerstep{(T,u)}{\myif{e}{s_1}{s_2}}{(T'', u'')} {\myif{e}{s_1'}{s_2'}}{inline}}
    {
     \begin{gathered}
        \Theta_t,\eta_s,\Delta_s\vdash \rangerstep{(T,u)}{s_1}{(T',u')}{s_1'}{inline}
        \\
       \Theta_t,\eta_s,\Delta_s\vdash \rangerstep{(T',u')}{s_2}{(T'',u'')}{s_2'}{inline}
    \end{gathered}
    }
\]}
\[
 \infer[\rn{inline-comp}]
     {\Theta_t,\eta_s,\Delta_s\vdash \rangerstep{(T,u)}{s_1;s_2}{(T'',u'')}{s_1';s_2'}{inline}}
    { 
    \begin{gathered}
    \Theta_t,\eta_s,\Delta_s,\vdash \rangerstep{(T,u)}{s_1}{(T',u')}{s_1'}{inline}
    \\
    \Theta_t,\eta_s,\Delta_s\vdash \rangerstep{(T',u')}{s_2}{(T'',u'')}{s_2'}{inline}        
    \end{gathered}
    }
\]
\[
 \infer[\rn{inline-no-op}]
     {\Theta_t,\eta_s,\Delta_s \vdash \rangerstep{(T,u)}{s}{(T,u)}{s}{inline}}
    {  \func{noOp}_{inline}(s) }
\]
}}
\caption{Snippet of the Method-Inlining Rules}
\label{fig:inlinerules}
\end{figure}
 }
\techreport{
Here the $\noop{inline}{s}$ is defined as\\

\noindent$\noop{inline}{x:=e} = true$\\
$\noop{inline}{\getfield{x}{z}{f}} = true$\\
$\noop{inline}{\putfield{z}{f}{e}} = true$\\
$\noop{inline}{skip} = true$\\
$\noop{inline}{\newobj{x}{c}} = true$\\
$\noop{inline}{s} = false \qquad$ otherwise \\
}

%
\subsection{Method-Inlining Transformation is Sound}
\begin{lem}[transform $+$ rename $+$ substitution]
\label{lem:trans-rename-subs-sound}
Given a parametric statement $\overrightarrow{y}.s$, we have 
    $$((\mu_s',\pi_r'),s')=transform(s)$$
    $$\Theta \vdash \concsstep{(\eta, \Delta, H, \nu,\mu)}{s}{(\eta', \Delta', H', \nu', \mu')}{} $$
$$O'=\varnothing, \quad T=getOutput(s'), \quad I'=\overrightarrow{y}, \quad           \Delta_s'=(\overrightarrow{y}, \overrightarrow{x})
$$
$$    I',O'\vdash \rangerstepnolookup{((\pi_r',\mu_s',T,u),s')}{((\mu_s'',\pi_r'', T', u'),s'')}{rename-wrap}$$
$$I' \cap \overrightarrow{x} = \varnothing$$

then for any $\Delta^*$ such that 

$$
(\Delta^*) \models (\Delta_s',\pi)$$
$$\Delta\subseteq\Delta^*$$

         $$ %
 \Delta_s' \vdash \rangerstepnolookup{((\pi_r'',\mu_s'',I'),s'')}{((\pi_r^3,\mu_s^3,I''),s^3)}{sub-wrap}$$
 $$\Theta \vdash \concsstep{(\eta, \Delta^*, H, \nu,\mu)}{s^3}{(\eta'', \Delta'', H'', \nu'', \mu'')}{} $$

          %
then
$$\eta'=\eta'', H'=H'', \nu'=\nu'',\mu''=\mu'$$
$$ \mu'=evalReturn(\Delta'',\pi_r^3,\mu_s^3)$$
$$(\Delta'') \models (\Delta_s,\pi)$$

 \techreport{
    
    \begin{proof}
We structure the proof in three stages. First, we characterize the intermediate state produced by applying the {\tt transform} function. Next, we analyze the state resulting from the {\tt rename-wrap} judgment. Finally, we apply the {\tt sub-wrap} judgment, yielding a state that matches the lemma’s conclusion.

    First, we apply Lemma~\ref{thm:transform-sound}, such that
    
$$\Theta \vdash \concsstep{(\eta, \Delta, H, \nu, \mu)}{s}{(\eta', \Delta', H', \nu, \mu')}{} $$
$$((\mu_s',\pi_r'),s') =transform(s)$$
then, we get
$$\Theta \vdash \concsstep{(\eta, \Delta, H, \nu, \mu)}{s'}{(\eta^4, \Delta^4, H^4, \nu^4, \mu^4)}{} $$
$$\eta'=\eta^4,\Delta'=\Delta^4,H'=H^4,\nu'=\nu^4$$
$$
\mu' = evalReturn(\Delta^4,\pi_r',\mu_s'),  \mu^4=\bot
$$

Next, since we have that 

$$\Theta \vdash \concsstep{(\eta, \Delta, H, \nu, \mu)}{s'}{(\eta^4, \Delta^4, H^4, \nu^4, \mu^4)}{} $$
$$\mu^4=\bot$$
$$I',O'\vdash\rangerstep{(\pi_r',\mu_s',T,u)}{s'}{ \mu_s'',\pi_r'',T', u')}{s''}{rename-wrap}$$
then, we can apply renaming Theorem (Theorem~\ref{thm:renaming-wrap-sound}) to get
$$\Theta \vdash \concsstep{(\eta, \Delta, H, \nu, \mu)}{s''}{(\eta^5, \Delta^5, H^5, \nu^5, \mu^5)}{} $$
$$evalReturn(\Delta^4,\pi_r',\mu_s')=evalReturn(\Delta^5,\pi_r'',\mu_s'')$$

where $\eta^5=\eta^4=\eta', H^5=H^4=H', \nu^5=\nu^4=\nu'=\bot, \mu^5=\mu^4=\bot$.

Also, observe that we have established from before that $\mu' = evalReturn(\Delta^4,\pi_r',\mu_s')$, and we also have shown that $evalReturn(\Delta^4,\pi_r',\mu_s')=evalReturn(\Delta^5,\pi_r'',\mu_s'')$, thus by transitivity, we have $\mu'=evalReturn(\Delta^5,\pi_r'',\mu_s'')$

Next, consider the {\tt sub-wrap}. From the premise of the theorem we have that $\Delta^*$ be a concrete state such that $(\Delta^*) \models (\Delta_s',\pi)$ where $\Delta \subseteq \Delta^*$, and $\mu^5=\bot$.
Also, observe that from the statement theorem we know that $\Delta_s'$ is constructed as $\Delta_s'=(\overrightarrow{y},\overrightarrow{x})$, such that $I' \cap \overrightarrow{x}=\varnothing$,
and since $I'=\overrightarrow{y}$, thus if follows that  $\forall i \in I'.\forall x \in Dom(\Delta_s).\neg isUse(i,\Delta_s(x))$.

With all the above observations, we can now use the substitution (Corollary~\ref{cor:subs-wrap-stmt-sound}), such that we have 

$$ \Delta_s' \vdash \rangerstepnolookup{((\pi_r'',\mu_s'',I'),s'')}{((\pi_r^3,\mu_s^3,I''),s^3)}{sub-wrap}$$

we get

$$\Theta \vdash \concsstep{(\eta, \Delta^*, H, \nu,\mu)}{s^3}{(\eta^6, \Delta^6, H^6, \nu^6,\mu^6)}{}$$
$$\eta^6=\eta''=\eta^5, H^6=H''=H^5, \nu^6=\nu''=\nu^5,\mu^6=\mu''=\mu^5$$
$$\Delta^5\subseteq\Delta^6$$
$$\forall i \in I'.\neg isUse(i,s^3)$$
$$(\Delta^6) \models (\Delta_s,\pi)$$
$$\func{evalReturn}(\Delta^5,\pi_r'',\mu_s'')=\func{evalReturn}(\Delta^6,\pi_r^3,\mu_s^3)$$

Observe here that by transitivity we get, $\eta''=\eta^6=\eta^5=\eta^4=\eta',H''=H^6=H^5=H^4=H',\nu''=\nu^6=\nu^5=\nu^4=\nu',\mu''=\mu^6=\mu^5=\mu^4=\bot$ (conclusion 1). And we already established 
$(\Delta^6) \models (\Delta_s,\pi)$ (conclusion 2)

Finally, it follows directly from applying Corollary~\ref{cor:subs-wrap-stmt-sound}, i.e., $\func{evalReturn}(\Delta^5,\pi_r'',\mu_s'')=\func{evalReturn}(\Delta^6,\pi_r^3,\mu_s^3)$, since we have that $\Delta''=\Delta^6$, then by transitivity we have $\func{evalReturn}(\Delta^6,\pi_r^3,\mu_s^3)=\func{evalReturn}(\Delta^5,\pi_r'',\mu_s'')=evalReturn(\Delta^4,\pi_r',\mu_s')=\mu'$ (conclusion 3).

    \end{proof}
 }
\end{lem}

\begin{lem}[transform $+$ rename $+$ substitution - no input]
\label{lem:trans-rename-subs-sound-no-input}
Assume the premises of Corollary~\ref{lem:trans-rename-subs-sound},
where an additional local variable mapping $\Delta^{**}$ is defined as

\[
Dom(\Delta^{**}) = Dom(\Delta^*) \setminus I
\quad\text{and}\quad
\Delta^{**} \subset \Delta^*
\]
$$\Theta \vdash \concsstep{(\eta, \Delta^{**}, H, \nu,\mu)}{s^3}{(\eta''', \Delta''', H''', \nu''', \mu''')}{} $$

then
$$\eta'''=\eta''=
\eta', H'''=H''=H', \nu'''=\nu''=\nu',\mu'''=\mu''=\mu'$$
$$\Delta'''\subset \Delta''$$
$$ \mu'=evalReturn(\Delta''',\pi_r^3,\mu_s^3)$$
$$(\Delta''') \models (\Delta_s,\pi)$$
    \techreport{
    \begin{proof}
    The proof directly follow from applying Corollary~\ref{cor:sub-wrap-remove-I}.
    \end{proof}
    }
\end{lem}

%

\begin{thm}[Inlining recursive judgment is sound]
\label{thm:inline-recursive-sound}
    If
$$\Theta \vdash \concsstep{(\eta, \Delta, H, \nu,\mu)}{s}{(\eta', \Delta', H', \nu', \mu')}{} $$
$$\Theta_t,\eta_s,\Delta_s\vdash \rangerstep{(T,u)}{s}{(T',u')}{s'}{inline}$$
$$\eta_s\subseteq \eta$$
$$(\Delta) \models (\Delta_s,\pi)$$
$$\Theta \vdash \concsstep{(\eta, \Delta, H, \nu,\mu)}{s'}{(\eta'', \Delta'', H'', \nu'', \mu'')}{} $$

then

\begin{gather*}
\Delta'\subseteq \Delta'' \tag{1}\\
\eta'=\eta'',H''=H',\nu''=\nu',\mu''=\mu' \tag{2}\\
\eta_s \subseteq \eta'' \tag{3}\\
(\Delta'') \models (\Delta_s,\pi) \tag{4}
\end{gather*}

\techreport{
\begin{proof}
    We proceed by induction on the inlining rules.

    First, consider the case {\tt inline-no-op}, and {\tt inline-invoke-stop} rules. Since inlining introduces no changes to the statements or environments in either case, then, we have
\[
\Theta_t,\eta_s,\Delta_s \vdash \rangerstep{(T,u)}{s}{(T,u)}{s}{inline}
\quad \text{where } \func{noOp}_{inline}(s) \text{ holds,}
\]
and therefore $s' = s$. Executing $s$ gives the concrete step
\[
\Theta \vdash \concsstep{(\eta,\Delta,H,\nu,\mu)}{s}{(\eta',\Delta',H',\nu',\mu')}{},
\]
and since $s = s'$, we obtain $\eta' = \eta''$, $H' = H''$, and $\nu' = \nu''$. Furthermore, because $\func{noOp}_{inline}(s)$ holds, $s$ contains no {\tt return} statement (plus the guarantee from the remove final-returns (Property~\ref{prop:final-ret-no-return}), and thus $\mu' = \bot = \mu$ (conclusion 2). Since we have that $\Delta'' = \Delta'$, it follows that $\Delta' \subseteq \Delta''$ (conclusion 1). 

Also, we have established that $\eta\subseteq\eta''$, and since $\eta_s\subseteq \eta$ (from the hypothesis), then by transitivity, we conclude that the $\eta_s \subseteq \eta''$ (conclusion 3).

Finally, we have $(\Delta) \models (\Delta_s,\pi)$ from the premise, and we know from Lemma~\ref{lem:delta-prime-sat} that $\Delta \subseteq \Delta'$. Also, we have just shown that $\Delta' \subseteq \Delta''$, then that means that evaluating $\pi$ and all expressions mapped to within $\Delta_s$ results in the same values, thus we can conclude that $(\Delta'') \models (\Delta_s, \pi)$.

Now, consider the {\tt inline-if} rule. Here we have that 

$$\Theta_t,\eta_s,\Delta_s\vdash \rangerstep{(T,u)}{\myif{e}{s_1}{s_2}}{(T_2, u_2)} {\myif{e}{s_1'}{s_2'}}{inline}$$

such that $\Theta_t,\eta_s,\Delta_s\vdash \rangerstep{(T,u)}{s_1}{(T_1,u_1)}{s_1'}{inline}$, and $
       \Theta_t,\eta_s,\Delta_s\vdash \rangerstep{(T_1,u_1)}{s_2}{(T_2,u_2)}{s_2'}{inline}$.

On the other hand, using the concrete semantics either the {\tt if-phi-true} or {\tt if-phi-false} can only apply depending on the evaluation of $e$. We will show the proof for the {\tt if-phi-true} case, but the proof for the {\tt if-phi-false} case works similarly. Note that, no return statement can exist in $s$  due to the guarantee from the remove final-returns transformation (Property~\ref{prop:final-ret-no-return}) thus other rules for concretely executing an if-statement can be ignored. 

First, consider the \tw{if-phi-true} case. Then we have that 
$\vdashconcestep{\Delta}{e}{true}$ and  $\concsstep{(\eta,\Delta, H, \nu,\mu)}{s_1}{(\eta_1,\Delta_1, H_1, \nu_1,\mu_1)}{}$. Observe here that, $\eta_1=\eta', \Delta_1=\eta', H_1=H', \nu_1=\nu',\mu_1=\mu'$.
Applying the induction hypothesis to $s_1$, we have that 
$$\concsstep{(\eta,\Delta, H, \nu,\mu)}{s_1'}{(\eta_1',\Delta_1', H_1', \nu_1',\mu_1')}{}$$ where $\eta_1'=\eta_1,H_1'=H_1, \nu_1'=\nu_1, \mu_1'=\bot$ (conclusion 2), $\Delta_1\subseteq \Delta_1'$ (conclusion 1) and $\eta_s\subseteq \eta_1'$ (conclusion 3). Also, observe here that $\eta_1'=\eta'', \Delta_1'=\Delta'', H_1'=H'', \nu_1'=\nu'',\mu_1'=\mu''$. Thus, by transitivity we have our conclusions satisfied.
Finally, we have $(\Delta) \models (\Delta_s,\pi)$ from the premise, and we know from Lemma~\ref{lem:delta-prime-sat} that $\Delta \subseteq \Delta_1$. Also, we have just shown that $\Delta_1 \subseteq \Delta_1'$, then that means that evaluating $\pi$ and all expressions mapped to within $\Delta_s$ results in the same values using $\Delta_1'$, thus we can conclude that $(\Delta_1') \models (\Delta_s, \pi)$ (conclusion 4).

Next, consider the {\tt inline-comp} rule. In this case we have 
$$\Theta_t,\eta_s,\Delta_s\vdash \rangerstep{(T,u)}{s_1;s_2}{(T'',u'')}{s_1';s_2'}{inline}$$
    
 where $\Theta_t,\eta_s,\Delta_s,\vdash \rangerstep{(T,u)}{s_1}{(T',u')}{s_1'}{inline}$, and $\Theta_t,\eta_s,\Delta_s\vdash \rangerstep{(T',u')}{s_2}{(T'',u'')}{s_2'}{inline}$. Again, due to the guarantee from the remove final-returns transformation (Property~\ref{prop:final-ret-no-return}), we know that $s_1;s_2$ cannot contain return-statements, and none of the inlining rules generate a return-statement, so a return-statement cannot be in $s_1'$, nor $s_2'$. Then, then the only concrete rule that applies is the {\tt composition} rule. Thus, we have that $ \Theta\vdash \concsstep{(\eta, \Delta, H, \nu,\mu)}{s_1;s_2}{(\eta',\Delta', H', \nu',\mu')}{}$ such that $ \Theta\vdash \concsstep{(\eta, \Delta, H, \nu,\mu)}{s_1}{(\eta_1, \Delta_1, H_1, \nu_1,\mu)}{}$ and $\Theta\vdash \concsstep{(\eta_1, \Delta_1, H_1, \nu_1,\mu)}{s_2}{(\eta',\Delta', H', \nu',\mu')}{}$. By applying the induction hypothesis on $s_1$ we get $ \Theta\vdash \concsstep{(\eta, \Delta, H, \nu,\mu)}{s_1'}{(\eta_1', \Delta_1', H_1', \nu_1',\mu_1')}{}$, such that $\eta_1'=\eta_1,H_1'=H_1,\nu_1'=\nu_1,\mu_1'=\bot$, $\Delta_1 \subseteq \Delta_1'$ and $\eta_s\subseteq\eta_1'$. Now, executing $s_2'$ in the resulting state we get $ \Theta\vdash \concsstep{(\eta_1', \Delta_1', H_1', \nu_1',\mu_1')}{s_2'}{(\eta_2, \Delta_2, H_2, \nu_2,\mu_2)}{}$. By applying the induction hypothesis on $s_2$, we get $\eta_2=\eta',H_2=H',\nu_2=\nu',\mu_2=\bot$, $\Delta' \subseteq \Delta_2$ and $\eta_s\subseteq\eta_2$. Here we have that $\Delta_2=\Delta'',H_2=H'',\nu_2=\nu'',\mu_2=\mu''$, we get by transitivity $\eta''=\eta_2=\eta',H''=H_2=H',\nu''=\nu_2=\nu',\mu''=\mu_2=\bot$ (conclusion 3), $\Delta' \subseteq \Delta''$ (conclusion 1) and $\eta_s\subseteq \eta_2=\eta''$ (conclusion 3).
Finally, we have $(\Delta) \models (\Delta_s,\pi)$ from the premise, and we know from Lemma~\ref{lem:delta-prime-sat} that $\Delta \subseteq \Delta_1' \subseteq \Delta'$. Also, we have just shown that $\Delta' \subseteq \Delta''$, then that means that evaluating $\pi$ and all expressions mapped to within $\Delta_s$ using $\Delta'$, must reduce to the same values, thus we can conclude that $(\Delta'') \models (\Delta_s, \pi)$ (conclusion 4).

Finally, consider the {\tt inline-invoke} rule. In this case we have that 
$$\Theta_t,\eta_s,\Delta_s\vdash \rangerstep{(T,u)}{\invoke{x_1}{r}{g}{x_2}}{(T_3,u_1)}{s^3; x_1:=\mu_s^3}{inline}$$

such that 

$$r \in R$$
     $$((\mu_s',\pi_r'), \overrightarrow{y}.s_1)=\lookup{\Theta_t}{\lookup{\eta_s}{r},g}$$
    $$O'=\varnothing \qquad I'=\overrightarrow{y} \qquad T_1=getOutput(s_1)$$
      $$I',O'\vdash \rangerstepnolookup{((\mu_s',\pi_r',T_1,u),s_1)}{((\mu_s'',\pi_r'', T_1', u'),s_2)}{rename-wrap}$$
         $$\Delta_s'=(\overrightarrow{y}, \overrightarrow{x_2})$$ 
          $$\Delta_s' \vdash \rangerstepnolookup{((\mu_s'', \pi_r'',I'),s_2)}{(( \mu_s^3,\pi_r^3,I''),s_3)}{sub-wrap}$$
    $$ T_1''=T_1' \cup T$$

In this case, we have that $s'=s_3; x_1:=\mu_s^3, T'=T_1''$.
Executing the original statement concretely, we get

$$\Theta\vdash\concsstep{(\eta,\Delta, H, \nu,\mu)}{\invoke{x_1}{r}{g}{x_2}}{(\eta', \Delta',H',\nu',\mu')}{}$$

such that 
$$ r \in R$$
    $$\lookup{\Theta}{\lookup{\eta}{r},g} = \overrightarrow{y'}.s_1'$$
    $$\vdashconcestep{\Delta}{\overrightarrow{e}}{\overrightarrow{v'}}$$
    $$\Delta_t=\{(\overrightarrow{y'},\overrightarrow{v'})\}$$
    $$\Theta\vdash \concsstep{(\eta,\Delta_t,H,\nu,\mu)}{s_1'}{(\eta', \Delta_t',H',\nu',\mu')}{}$$
    $$\Delta'=\update{\Delta}{x_1}{\mu'}$$
    $$\mu' \neq \bot$$

    Since we have that $\eta_s\subseteq\eta$, then we know that $\eta(r)=\eta_s(r)=c$, for a class $c$. 
    
    Now, recall that $\lookup{\Theta}{c,g} = \overrightarrow{y'}.s_1'$, while the $\Theta_t$, according to its definition, is the transformed version of each class-method pair, if 
    $$((\mu_s',\pi_r'),s_1)=transform(s_1')$$
    then $((\mu_s',\pi_r'),\overrightarrow{y}.s_1) \in \Theta_t$
    according to these definitions, we shown that  $\overrightarrow{y'}=\overrightarrow{y}$
    
    Next, observe that the premise of the {\tt inline-invoke} rule is already matching up with the premise of Lemma~\ref{lem:trans-rename-subs-sound}, more importantly that we have, $O'=\varnothing, I'=\overrightarrow{y}, T_1=getOutput(s_1)$ and that $\Delta_s'=(\overrightarrow{y}, \overrightarrow{x_2})$.

    It remains to show that $I' \cap \overrightarrow{x_2} = \varnothing$. Once again, referring to the remove final-returns transformation guaranteed property (Property~\ref{prop:final-ret-no-return}), we know that all variables in $\overrightarrow{x_2}$ must have been renamed as part of renaming of the statement containing $s$. On the other hand, we know that $\overrightarrow{y}$, must not have been renamed before. More precisely, for all variables $y$ in $\overrightarrow{y}$, $y$ must be of the form $(a,\bot)$, whereas all variables $x$ in $\overrightarrow{x_2}$, must be of the form $(a,n)$, where $n$ is a constant int. Because of mismatched suffixes within the variable names, we can conclude that $I' \cap \overrightarrow{x_2} = \varnothing$.

    Also, observe that $(\Delta) \models (\Delta_s,\pi)$ by the premise of the Theorem. 
    Now let us choose $\Delta^{*}=\Delta\cup\Delta_t$. Here, we can also see that $(\Delta^*) \models (\Delta_s',\pi)$, this is because, we know that $\Delta_s'$ is constructed as $\Delta_s'=(\overrightarrow{y}, \overrightarrow{x_2})$, and since we have that $(\Delta) \models (\Delta_s,\pi)$, which means that evaluating expressions in $\overrightarrow{x_2}$ must evaluate in $\Delta$ must yield $\concestep{\Delta}{\overrightarrow{x_2}}{\overrightarrow{v}}{}$, such that $\Delta^{*}(\overrightarrow{y})=\overrightarrow{v}$, thus we can conclude that $(\Delta^{*}) \models (\Delta_s',\pi)$. And obviously we have $\Delta\subset \Delta^{*}$ by construction. 
    Now, observe that $\Delta$ satisfies the condition of $\Delta^{**}$ in Lemma~\ref{lem:trans-rename-subs-sound-no-input}. Thus, we can apply Lemma~\ref{lem:trans-rename-subs-sound-no-input} to conclude that 

    $$\Theta \vdash \concsstep{(\eta, \Delta, H, \nu,\mu)}{s_3}{(\eta''', \Delta''', H''', \nu''', \mu''')}{} $$
    $$\eta'''=\eta', H'''=H', \nu'''=\nu',\mu'''=\bot$$
    $$ \mu'=evalReturn(\Delta''',\pi_r^3,\mu_s^3)$$

    Recall that the rewritten statement was of the form $s'=s^3;x_1:=\mu_s^3$. And so far we have evaluated the first part of $s'$, i.e., the concrete execution of $s^3$. Now, we apply the concrete composition rule to evaluate the remaining part of $s'$; i.e., the assignment statement ($x:=\mu_s^3$).

    $$\Theta \vdash \concsstep{(\eta''', \Delta''', H''', \nu''', \mu''')}{x:=\mu_s^3}{(\eta''', \Delta^4, H''', \nu''', \mu''')}{} $$
    
    observe that we have $\eta''==\eta'''=\eta'$, $H''=H'''=H'$, $\nu''=\nu'''=\nu'$, $\mu''=\bot$ (conclusion 2).

    Now observe the following
    
    \begin{itemize}
        \item $\pi_r'$, must be the value $true$, by the construction of the the $\Theta_t$ (see Fig.~\ref{fig:transform}).
        \item $\pi_r'',\pi_r^3=\pi_r'=true$, this is because both the rename judgment for expression not the substitution judgment for expression change concrete values (see rule {\tt rename-val}, and {\tt sub-val}).
        \item since $\pi_r^3=true$, then according to the definition of $\func{evalReturn}$, we have that 
        $\func{evalReturn}(\Delta''',\pi_r^3,\mu_s^3)=\mu'$, this then implies that $\mu_s^3$ must evaluate to $\mu'$, i.e., $\concestep{\Delta'''}{\mu_s^3}{\mu'}{}
        $
        \item finally, observe that $\Delta^4=\Delta'''[x_1 \leftarrow \mu']$, while $\Delta'=\Delta[x_1\leftarrow \mu']$, and since we have that $\Delta\subseteq\Delta''' $ (using Lemma~\ref{lem:delta-ssa}), and that $\Delta'''\subseteq \Delta^4$, then updating $\Delta$ and $\Delta^4$ in the same way would maintain the same subset relation, i.e., $\Delta'\subset\Delta^4=\Delta''$ (conclusion 1)
    \end{itemize}
    
Now, we want to show that $\eta_s \subseteq \eta'''$. We have established from before that $\eta'''=\eta'$, and from the statement of the lemma we have $\eta_s \subseteq \eta$, and from Lemma~\ref{lem:delta-ssa} $\eta \subseteq \eta'''$, then by transitivity we can conclude that $\eta_s \subseteq \eta'''$ (conclusion 3).

Finally, we have $(\Delta) \models (\Delta_s,\pi)$ from the premise, and we know from Lemma~\ref{lem:delta-prime-sat} that $\Delta \subseteq \Delta'$. Also, we have just shown that $\Delta' \subseteq \Delta''=\Delta^4$, then that means that evaluating $\pi$ and all expressions mapped to within $\Delta_s$ using $\Delta''$, must reduce to the same values, thus we can conclude that $(\Delta'') \models (\Delta_s, \pi)$ (conclusion 4).
\end{proof}
}
\end{thm}

\begin{lem}[Inlining Transformation is Sound]
\label{lem:inline-wrapper-sound}
If
$$\Theta \vdash \concsstep{(\eta, \Delta, H, \nu,\mu)}{s}{(\eta', \Delta', H', \nu', \mu')}{} $$
$${\Theta_t,\eta_s,\Delta_s \vdash \rangerstep{(T,u)}{s}{(T',u')}{s'}{inline-wrap}}$$
$$\eta_s\subseteq \eta$$
$$\mu^*=\func{evalReturn}(\Delta',\pi_r,\mu_s)$$
$$(\Delta) \models (\Delta_s,\pi)$$
$$\Theta \vdash \concsstep{(\eta, \Delta, H, \nu,\mu)}{s'}{(\eta'', \Delta'', H'', \nu'', \mu'')}{} $$

then
$$\Delta'\subseteq \Delta'' $$
$$\eta'=\eta'',H''=H',\nu''=\nu',\mu''=\mu'$$
$$\eta_s \subseteq \eta''$$
$$\mu^*=\func{evalReturn}(\Delta'',\pi_r,\mu_s)$$
$$(\Delta'') \models (\Delta_s,\pi)$$

\techreport{
\begin{proof}
The proof for the first four conclusions directly fall from applying Theorem~\ref{thm:inline-recursive-sound}.

To show that $\mu^*=\func{evalReturn}(\Delta',\pi_r,\mu_s)$, recall that $\mu_s$, $\pi_r$ has not been changed by the inlining transformation, moreover, we have proved that $\Delta'\subseteq\Delta''$, thus evaluating $\mu_s$ and $\pi_r$ under $\Delta''$ must reduce to the same value $\mu^*$, same reasoning applies for showing that $(\Delta'') \models(\Delta_s,\pi)$ holds.
\end{proof}
    
}
\end{lem}

\subsection{Method Inlining Properties}

This transformation assumes all the guaranteed properties from pervious transformations, $\gamma$-creation, remove final-returns, renaming and substitution transformations. 
In addition, it also assumes this property which is enforced by ensuring that the method transformation must have proceeded the method-inlining transformation (Property~\ref{prop:inline-theta-is-transformed}). The method-transform is a static transformation that is done on all methods to populate $\Theta_t$ with the corresponding method statements, after applying $\gamma$-creation, early-returns elimination and remove final-returns transformations. Refer to section~\ref{sec:statictransform} for more detail about this static transformation.


Now assuming we have that 
$$\Theta_t,\eta_s,\Delta_s \vdash \rangerstep{(T,u)}{s}{(T'',u')}{s'}{inline-wrap}$$

such that 
$$\Delta_s'=\varnothing \qquad
    \Theta_t,\eta_s,\Delta_s'\vdash \rangerstep{(T,u)}{s}{(T',u')}{s'}{inline}
    \qquad T''=T\cup T'$$

Then we have the following guarantees

\begin{itemize}
    \item Rewritten Statement $s'$ satisfies the SSA property (Property~\ref{prop:ssa}).
    More formally, $SSA(s')$ holds.
        \techreport{
        \begin{proof}
            We proceed by induction over the method-inlining transformation rules.
        Consider the rule \tw{inline-if}.
        Then, from the induction hypothesis, 
        we get that $\Gamma \vdash (\mu_s,s_1) \to_{inline} (\mu_s',s_1')$ implies that the SSA property holds for $s_1'$, likewise, we get that
        $\Gamma \vdash (\mu_s,s_2) \to_{inline} (\mu_s',s_2')$ implies that the SSA property holds for $s_2'$.
        This immediately gives that $SSA(if \,\, e \,\, then \,\, s_1' \,\, else \,\, s_2')$ by definition of $SSA$.

        Consider the rule \tw{inline-invoke}.
        By the assumption that all invoked methods obey the SSA property, $s$ obeys the SSA property. 
        We then have the result that \emph{rename-wrap} preserves SSA, and generates variable names that are tuples of an Id and a numeral (and as such are disjoint from $x$), so $s''$ obeys the SSA property.
        Further, by definition of SSA, we have that any sub-statement of $s''$ obeys SSA, so $s'$ obeys SSA.
        By the assumption that the input substitution transformation preserves the SSA property, $s^3$ obeys the SSA property. 
        By the inductive hypothesis, $s^4$ obeys the SSA property.
    
        Consider \tw{inline-invoke-stop}.
        In this case, the statement is not transformed, so the result holds trivially.
        
        Consider \tw{inline-return}. 
        $SSA(skip)$ and $SSA(return \,e)$  hold by definition of $SSA$, so SSA is preserved.
        \end{proof}
        }
\end{itemize}

\section{$TR_7$: Field SSA Transformation}
\label{sec:field}
In this transformation, field statements ({\tt getfield}, {\tt putfield}) are eliminated and replaced by a sequence of assignment statements that satisfies the SSA property.
Fig.~\ref{fig:fieldSSA} shows the rules for the field transformation. 
The rules define two judgments. The judgment $\Delta_s, H_s \vdash 
\rangerstep{(T, P, \bot)}{s}{(T', P', p)}{s'}{field}$ rewrites an IR 
statement $s$ into $s'$, where $H_s$ is the symbolic heap, $P$ is the 
path-subscript map tracking the latest update to each reference-field pair 
along the current local path in $s$, and $p$ is the local path condition 
guarding reference-field updates. Initially, $P$ agrees with $H_s$ and $p = 
\bot$. The updated map $P'$ and flag $p'$ reflect the state of field accesses 
after rewriting. The wrapper judgment $\Delta_s, H_s \vdash 
\rangerstep{(I, O, T)}{s}{(I, O, T'')}{(\mathtt{addGets}(P);\mathit{skip};s';\mathit{skip};\mathtt{addPuts}(P'))}{field\text{-}wrap}$ 
is the entry point of the transformation, initializing $P$ from $H_s$ and 
triggering the field judgment on $s$, then wrapping the result with 
$\mathtt{addGets}(P)$ and $\mathtt{addPuts}(P')$ to load and store field 
values at the boundaries of the rewritten region.

More precisely, the rule~{\tt field-wrap} is the entry point for the field transformation. It 
initializes the field map $P$ via $\mathtt{initialP}(H_s)$, which assigns a 
fresh temporary variable $t$ to each field-reference pair $(r, f)$ in the 
domain of $H_s$, representing the latest SSA variable carrying the value of 
that field along the current path. These fresh temporaries are added to $T$ to 
produce $T'$, and the field transformation is then applied to $s$ under $T'$ 
and $P$, yielding a rewritten statement $s'$, an updated field map $P'$, and 
a path condition $p$. The rule wraps $s'$ with $\mathtt{addGets}(P)$ before 
and $\mathtt{addPuts}(P')$ after, which emit the necessary \texttt{getfield} 
and \texttt{putfield} instructions to load field values into their corresponding 
temporaries before execution and write them back after, respectively. The output 
temporary set $T''$ reflects all temporaries introduced during the transformation.

The wrapping of $s'$ with $\mathtt{addGets}(P)$ and $\mathtt{addPuts}(P')$ 
serves an important purpose beyond a single pass. Since field references are 
not always resolved to concrete references in one iteration, the transformation 
may require multiple fixpoint iterations until all reference variables are 
concretized. The wrapping captures the progress made in the current iteration: 
$\mathtt{addGets}(P)$ records the getfield eliminations already performed, and 
$\mathtt{addPuts}(P')$ records the putfield eliminations, so that if a later 
iteration resolves a previously unresolved reference to a concrete one, the 
transformation can resume from where it left off rather than restarting from 
scratch. The \texttt{skip} statements surrounding $s'$ in the residual serve 
as delimiters for pattern matching: in subsequent iterations, the field 
transformation is applied only to $s'$ and not to the entire residual statement. 
This is essential to avoid an infinite loop, as re-applying the transformation 
to $\mathtt{addGets}(P)$ or $\mathtt{addPuts}(P')$ would cause their 
instructions to be treated as new field accesses requiring further elimination.

The $\bot$ flag in the field judgment $\rangerstep{(T',P,\bot)}{s}{(T'',P',p)}{s'}{field}$ 
acts as a control flag that governs whether field elimination should continue. 
As long as the flag is $\bot$, the transformation proceeds, eliminating 
\texttt{getfield} and \texttt{putfield} statements and updating $P$ accordingly. 
However, upon encountering a method invocation, Rule~{\tt field-invoke} 
immediately flips the flag to $\top$ and halts further field elimination, 
flushing the current field map by prepending $\mathtt{addPuts}(P)$ before the 
invoke statement. This is necessary because inlining the invoked method may 
introduce new concrete references that affect the SSA assignments of subsequent 
statements, and proceeding with elimination before those references are resolved 
would violate the sequentialization of SSA. Once the flag is set to $\top$, it 
propagates through compound statements: Rule~{\tt field-if-2, field-if-3, field-if-4} and the 
corresponding sequential composition rule (\soha{{\tt field-comp}}) propagate $\top$ to 
prevent any further elimination in the remainder of the statement, ensuring 
that field elimination resumes only in the next fixpoint iteration after 
inlining has been applied.

The core of the field transformation lies in the elimination of 
\texttt{getfield} and \texttt{putfield} statements, replacing them with 
simple variable assignments using the field map $P$.

Rule~{\tt field-put} eliminates a \texttt{putfield} instruction by introducing 
a fresh temporary variable $t$ and updating the field map to $P' = P[(r,f) 
\leftarrow t]$, recording that $t$ now carries the latest value of the 
field-reference pair $(r,f)$. The \texttt{putfield} is replaced by the 
assignment $t := e$. Rule~{\tt field-get} eliminates a \texttt{getfield} 
instruction by looking up the latest temporary $t$ associated with $(r,f)$ 
in $P$ and replacing the \texttt{getfield} with the assignment $x := t$, 
leaving $P$ unchanged. Both rules require that $r$ is a concrete reference, 
i.e., $r \in R$, if references has been resolved to their concrete values, otherwise, no rewriting is performed, \soha{i.e, {\tt field-get-var} (also {\tt field-put-var} see our technical report~\cite{techreport} for complete set of rules), instead, further field statements are prevented using the $\top$ flag.}

Finally, the rule~{\tt field-if-1} handles the case where field elimination succeeds 
completely in both branches of a conditional, i.e., neither branch encounters 
a method invocation. Each branch is processed independently under the same 
initial field map $P$, yielding rewritten statements $s_1'$ and $s_2'$ and 
updated field maps $P'$ and $P''$ respectively. Since the two branches may 
assign different temporaries to the same field-reference pair $(r,f)$, the 
maps $P'$ and $P''$ must be reconciled after the conditional. This is handled 
by the $\mathtt{merge}$ function, which compares $P'$ and $P''$ and, for each 
field-reference pair $(r,f)$ where the two branches disagree, introduces a 
fresh temporary $t$ assigned to $\gammaexp{e}{t_1}{t_2}$, where $t_1 = P'(r,f)$ 
and $t_2 = P''(r,f)$ are the temporaries from the left and right branches 
respectively. The merge proceeds recursively until all disagreements are 
resolved, producing a unified field map $\hat{P}$ and a sequence of $\gamma$-assignments 
$\hat{s}$ appended after the conditional. If both branches produce identical 
field maps, $\mathtt{merge}$ returns $\mathtt{skip}$ and leaves $P$ unchanged.


\techreport{
\begin{figure}[h!t]
    \footnotesize
\fbox{%
\parbox{\textwidth}{%
$$
\begin{gathered}
merge(e,P,P) = (skip, P) \\
merge(e,P_1,P_2) \text{ when } P_1 \neq P_2 \\
\text{Let } (r,f) \in Dom(P_1), \text{ and } (r,f) \in Dom(P_2) \text{ s.t. } P_1(r,f) \neq P_2(r,f), \\
\indent \text{ Let } t_1=P_1(r,f)\qquad t_2=P_2(r,f)\\
P_1'=P_1[(r,f) \leftarrow t] 
\qquad P_2'=P_2[(r,f) \leftarrow t]
\\(s_r', P_r')=merge(e, P_1', P_2'), \text{fresh } t.\\
(t:=\gammaexp{e}{t_1}{t_2};s_r', P_r')
\end{gathered}
\vspace{5pt}
$$
\hrule
\hrule
\vspace{5pt}
\[
\begin{gathered}
    initialP(H_s)= \{((r,f), t) \mid (r,f) \in Dom(H_s),  \text{ fresh } t\}\\
\end{gathered}
\]
\[
\begin{gathered}
    addGets(P)=\getfield{t_1}{r_1}{f_1};\ldots;\getfield{t_n}{r_n}{f_n}\quad  \forall ((r_i,f_i),t_i) \in P, 1 \le i \le n,  n=|P| \neq 0\\
    addGets(\varnothing)=skip 
\end{gathered}
\]
\[
\begin{gathered}
    addPuts(P)=\putfield{r_1}{f_1}{t_1};\ldots;\putfield{r_n}{f_n}{t_n} \quad 
    \forall ((r_i,f_i),t_i) \in P, 1 \le i \le n,  n=|P| \neq 0\\
    addPuts(\varnothing)=skip 
\end{gathered}
\]
\vspace{5pt}
\hrule
\hrule
\vspace{10pt}
\[
\infer[\rn{field-wrap}]
 {\Delta_s, H_s \vdash \rangerstep{(I,O,T)}{s}{(I,O,T'')}{(addGets(P); skip; s'; skip; addPuts(P')}{field-wrap}}
{ 
\begin{gathered}
P=initialP(H_s) \qquad
T'=T\cup Range(P) \qquad
\Delta_s, H_s\vdash\rangerstep{(T',P,\bot)}{s}{(T'',P',p)}{s'}{field}\\
\end{gathered}
}
\]
\hrule
\hrule
\vspace{8pt}
$$
\infer[\rn{field-no-op}]
 {\Delta_s, H_s \vdash \rangerstep{(T,\varnothing,\top)}{s}{(T,\varnothing,\top)}{s}{field}}
{}
$$

$$
\infer[\rn{field-assign}]
 {\Delta_s, H_s \vdash \rangerstep{(T,P,p)}{x:=e}{(T,P,p)}{x:=e}{field}}
{}
$$

$$
\infer[\rn{field-put}]
 {\Delta_s, H_s\vdash\rangerstep{(T,P,\bot)}{\putfield{r}{f}{e}}{(T\cup \{t\}, P',\bot)}{t := e}{field}}
{\begin{gathered}
r \in R
\qquad
P'=P[(r,f)\leftarrow t] \qquad {\text{fresh\;} t}
\end{gathered}}
$$

$$
\infer[\rn{field-put-var}]
 {\Delta_s, H_s\vdash\rangerstep{(T,P,\bot)}{\putfield{z}{f}{e}}{(T, \varnothing,\top)}{addPuts(P);\putfield{z}{f}{e}}{field}}
{\begin{gathered}
z \notin R 
\end{gathered}}
$$

\[
\quad\infer[\rn{field-get}]
{\Delta_s, H_s \vdash \rangerstep{(T,P,\bot)}{\getfield{x}{r}{f}}{(T,P,\bot)}{x:=t}{field}}
{\begin{gathered}
r \in R \qquad t = \lookup{P}{r,f}
\end{gathered}}
\]

\[
\quad\infer[\rn{field-get-var}]
{\Delta_s, H_s \vdash \rangerstep{(T,P,\bot)}{\getfield{x}{z}{f}}{(T,\varnothing,\top)}{addPuts(P);\getfield{x}{z}{f}}{field}}
{\begin{gathered}
z \notin R 
\end{gathered}}
\]

\[
\quad\infer[\rn{field-invoke}]
{\Delta_s, H_s \vdash \rangerstep{(T,P,\bot)}{\invoke{x}{z}{\gamma}{y}}{(T,\varnothing,\top)}{addPuts(P);\invoke{x}{z}{\gamma}{y}}{field} } 
{}
\]

$$
\infer[\rn{field-comp}]
 {\Delta_s, H_s \vdash \rangerstep{(T, P,\bot)}{s_1;s_2}{(T'',P'',p')}{s_1';s_2'}{field}}
{ \Delta_s, H_s\vdash\rangerstep{(T,P,\bot)}{s_1}{(T',P',p)}{s_1'}{field}
\\
\Delta_s, H_s\vdash\rangerstep{(T',P',p)}{s_2}{(T'',P'',p')}{s_2'}{field}
}
$$
$$
\infer[\rn{field-if-1}] 
 {\Delta_s, H_s \vdash \rangerstep{(T,P,\bot)}{\myif{e}{s_1}{s_2}}{(T'',\hat{P},\bot)}{\myif{e}{s_1'}{s_2'};\hat{s}}{field}}
{ 
\begin{gathered}
(\hat{s},\hat{P}) = merge(e,P',P'')
\\
\Delta_s, H_s\vdash\rangerstep{(T,P,\bot)}{s_1}{(T',P',\bot)}{s_1'}{field}
\qquad
\Delta_s, H_s\vdash\rangerstep{(T',P,\bot)}{s_2}{(T'',P'',\bot)}{s_2'}{field}
\end{gathered}
}
$$
$$
\infer[\rn{field-if-2}] 
 {\Delta_s, H_s\vdash \rangerstep{(T,P,\bot)}{\myif{e}{s_1}{s_2}}{(T'',\varnothing,\top)}{\myif{e}{s_1'}{s_2'}}{field}}
{ 
\begin{gathered}
\Delta_s, H_s\vdash\rangerstep{(T,P,\bot)}{s_1}{(T',\varnothing,\top)}{s_1'}{field}
\qquad
\Delta_s, H_s \vdash\rangerstep{(T',P,\bot)}{s_2}{(T'',\varnothing,\top)}{s_2'}{field}
\end{gathered}
}
$$

$$
\infer[\rn{field-if-3}] 
 {\Delta_s, H_s\vdash \rangerstep{(T,P,\bot)}{\myif{e}{s_1}{s_2}}{(T'',\varnothing,\top)}{\myif{e}{s_1'}{(s_2';addPuts(P''))}}{field}}
{ 
\begin{gathered}
\Delta_s, H_s\vdash\rangerstep{(T,P,\bot)}{s_1}{(T',\varnothing,\top)}{s_1'}{field}
\qquad
\Delta_s, H_s \vdash\rangerstep{(T',P,\bot)}{s_2}{(T'',P'',\bot)}{s_2'}{field}
\end{gathered}
}
$$

$$
\infer[\rn{field-if-4}] 
 {\Delta_s, H_s \vdash \rangerstep{(T,P,\bot)}{\myif{e}{s_1}{s_2}}{(T'',\varnothing,\top)}{\myif{e}{(s_1';addPuts(P'))}{s_2'}}{field}}
{ 
\begin{gathered}
\Delta_s, H_s\vdash\rangerstep{(T,P,\bot)}{s_1}{(T',P',\bot)}{s_1'}{field}
\qquad
\Delta_s, H_s \vdash\rangerstep{(T',P,\bot)}{s_2}{(T'',\varnothing,\top)}{s_2'}{field}
\end{gathered}
}
$$

}}
\caption{Field SSA Transformation Rules}
\label{fig:fieldSSA}
\end{figure}

}

\journalreport{
\begin{figure}[h!t]
    \footnotesize
\fbox{%
\parbox{\textwidth}{%
$$
\begin{gathered}
merge(e,P,P) = (skip, P) \\
merge(e,P_1,P_2) \text{ when } P_1 \neq P_2 \\
\text{Let } (r,f) \in Dom(P_1), \text{ and } (r,f) \in Dom(P_2) \text{ s.t. } P_1(r,f) \neq P_2(r,f), \\
\indent \text{ Let } t_1=P_1(r,f)\qquad t_2=P_2(r,f)\\
P_1'=P_1[(r,f) \leftarrow t] 
\qquad P_2'=P_2[(r,f) \leftarrow t]
\\(s_r', P_r')=merge(e, P_1', P_2'), \text{fresh } t.\\
(t:=\gammaexp{e}{t_1}{t_2};s_r', P_r')
\end{gathered}
$$
\hrule
\hrule
\vspace{5pt}
\begin{align*}
    initialP(H_s)&= \{((r,f), t) \mid (r,f) \in Dom(H_s),  \text{ fresh } t\}\\
    addGets(P)&=\getfield{t_1}{r_1}{f_1};\ldots;\getfield{t_n}{r_n}{f_n}\quad  \forall ((r_i,f_i),t_i) \in P, 1 \le i \le n,  n=|P| \neq 0\\
    addGets(\varnothing)&=skip \\
    addPuts(P)&=\putfield{r_1}{f_1}{t_1};\ldots;\putfield{r_n}{f_n}{t_n} \quad 
    \forall ((r_i,f_i),t_i) \in P, 1 \le i \le n,  n=|P| \neq 0\\
    addPuts(\varnothing)&=skip 
\end{align*}

\hrule
\hrule
\[
\infer[\rn{field-wrap}]
 {\Delta_s, H_s \vdash \rangerstep{(I,O,T)}{s}{(I,O,T'')}{(addGets(P); skip; s'; skip; addPuts(P')}{field-wrap}}
{ 
\begin{gathered}
P=initialP(H_s) \qquad
T'=T\cup Range(P) \qquad
\Delta_s, H_s\vdash\rangerstep{(T',P,\bot)}{s}{(T'',P',p)}{s'}{field}\\
\end{gathered}
}
\]
\hrule
\hrule
\vspace{8pt}

$$
\infer[\rn{field-put}]
 {\Delta_s, H_s\vdash\rangerstep{(T,P,\bot)}{\putfield{r}{f}{e}}{(T\cup \{t\}, P',\bot)}{t := e}{field}}
{\begin{gathered}
r \in R
\qquad
P'=P[(r,f)\leftarrow t] \qquad {\text{fresh\;} t}
\end{gathered}}
$$

\[
\quad\infer[\rn{field-get}]
{\Delta_s, H_s \vdash \rangerstep{(T,P,\bot)}{\getfield{x}{r}{f}}{(T,P,\bot)}{x:=t}{field}}
{\begin{gathered}
r \in R \qquad t = \lookup{P}{r,f}
\end{gathered}}
\]

$$
\quad\infer[\rn{field-invoke}]
{\Delta_s, H_s \vdash \rangerstep{(T,P,\bot)}{\invoke{x}{z}{\gamma}{y}}{(T,\varnothing,\top)}{addPuts(P);\invoke{x}{z}{\gamma}{y}}{field} } 
{}
$$
$$
\quad\infer[\soha{\rn{field-get-var}}]
{\Delta_s, H_s \vdash \rangerstep{(T,P,\bot)}{\getfield{x}{z}{f}}{(T,\varnothing,\top)}{addPuts(P);\getfield{x}{z}{f}}{field}}
{\begin{gathered}
z \notin R 
\end{gathered}}
$$
%
$$
\infer[\rn{field-if-1}] 
 {\Delta_s, H_s \vdash \rangerstep{(T,P,\bot)}{\myif{e}{s_1}{s_2}}{(T'',\hat{P},\bot)}{\myif{e}{s_1'}{s_2'};\hat{s}}{field}}
{ 
\begin{gathered}
\\
\Delta_s, H_s\vdash\rangerstep{(T,P,\bot)}{s_1}{(T',P',\bot)}{s_1'}{field}
\\
\Delta_s, H_s\vdash\rangerstep{(T',P,\bot)}{s_2}{(T'',P'',\bot)}{s_2'}{field}
\qquad
(\hat{s},\hat{P}) = merge(e,P',P'')
\end{gathered}
}
$$
$$
\infer[\rn{field-if-2}] 
 {\Delta_s, H_s\vdash \rangerstep{(T,P,\bot)}{\myif{e}{s_1}{s_2}}{(T'',\varnothing,\top)}{\myif{e}{s_1'}{s_2'}}{field}}
{ 
\begin{gathered}
\Delta_s, H_s\vdash\rangerstep{(T,P,\bot)}{s_1}{(T',\varnothing,\top)}{s_1'}{field}
\\
\Delta_s, H_s \vdash\rangerstep{(T',P,\bot)}{s_2}{(T'',\varnothing,\top)}{s_2'}{field}
\end{gathered}
}
$$

$$
\infer[\shortstack{\texttt{\tiny{field-}}\\
\texttt{\tiny{if-3}}}] 
 {\Delta_s, H_s\vdash \rangerstep{(T,P,\bot)}{\myif{e}{s_1}{s_2}}{(T'',\varnothing,\top)}{\myif{e}{s_1'}{(s_2';addPuts(P''))}}{field}}
{ 
\begin{gathered}
\Delta_s, H_s\vdash\rangerstep{(T,P,\bot)}{s_1}{(T',\varnothing,\top)}{s_1'}{field}
\\
\Delta_s, H_s \vdash\rangerstep{(T',P,\bot)}{s_2}{(T'',P'',\bot)}{s_2'}{field}
\end{gathered}
}
$$

$$
\infer[\shortstack{\texttt{\tiny{field-}}\\
\texttt{\tiny{if-4}}}] 
 {\Delta_s, H_s \vdash \rangerstep{(T,P,\bot)}{\myif{e}{s_1}{s_2}}{(T'',\varnothing,\top)}{\myif{e}{(s_1';addPuts(P'))}{s_2'}}{field}}
{ 
\begin{gathered}
\Delta_s, H_s\vdash\rangerstep{(T,P,\bot)}{s_1}{(T',P',\bot)}{s_1'}{field}
\\
\Delta_s, H_s \vdash\rangerstep{(T',P,\bot)}{s_2}{(T'',\varnothing,\top)}{s_2'}{field}
\end{gathered}
}
$$

}}
\caption{Snippet of Field SSA Transformation Rules}
\label{fig:fieldSSA}
\end{figure}

}

\subsection{Field Transformation is Sound}
In this section, we discuss the soundness of the Field Transformation. We first present a consistency properties over the path-subscript map, then we show the soundness of the recursive judgment followed by the soundness of the wrapper judgment.

For any path-subscript map $P$, local variable mapping $\Delta$, and a $Heap$, we define a $P$-Heap consistency relation $(P,H',\Delta') \rtimes H'$ for all reference and field pairs ($r,f$), as
\[
(P, H', \Delta') \rtimes H \quad \text{iff} \quad \forall (r,f)\in Dom(H). 
H(r, f) = 
\begin{cases}
\Delta'(P(r,f)) & \text{if } (r,f) \in \text{Dom}(P) \\
H'(r,f) & \text{otherwise}
\end{cases}
\]
Intuitively, the relation $(P, H', \Delta') \rtimes H$ defines when a concrete heap $H$ is 
consistent with a symbolic field map $P$, a base heap $H'$, and a local 
variable mapping $\Delta'$. Specifically, for each field-reference pair 
$(r,f)$, $H$ must agree with the value of the corresponding temporary variable 
$P(r,f)$ in $\Delta'$ if $(r,f)$ is tracked in $P$, and fall back to the base 
heap $H'$ otherwise.

A key property of the $\mathtt{merge}$ function is that it preserves the 
P-Heap consistency relation across both branches of a conditional. Intuitively, 
if both branches independently maintain P-Heap consistency with their respective 
field maps $P_1$ and $P_2$, then the $\gamma$-assignments produced by 
$\mathtt{merge}$ yield a unified field map $\hat{P}$ that remains consistent 
with the concrete heap $H$, regardless of which branch was taken. This is 
formalized in Lemma~\ref{lem:field-merge-lemma}.

\begin{lem}
\label{lem:field-merge-lemma}    
$$Dom(P_1)=Dom(P_2)$$
$$ (P_1,H_1',\Delta_1) \rtimes H_1 $$
$$ (P_2,H_2',\Delta_2) \rtimes H_2 $$
$$(\hat{s},\hat{P})=merge(e,P_1,P_2)$$
$$ (\Theta \vdash \concsstep{(\eta,\Delta,H',\nu, \mu)}{\hat{s}}{(\eta_s,\hat{\Delta},H',\nu, \mu)}{}$$
$$\forall x .isUse(x,e) \implies \Delta_1(x)=\Delta_2(x)$$
$$\vdashconcestep{\Delta}{e}{b} \qquad b \in\{true,false\}$$
$$(\Delta,H',H)=
\begin{cases}
(\Delta_1,H_1', H_1) & \text{if b=true}  \\
(\Delta_2,H_2', H_2) & \text{otherwise}\\
\end{cases}$$

then
$$ (\hat{P},H',\hat{\Delta}) \rtimes H $$

\end{lem}

\techreport{
\begin{proof}
    
    We will prove this lemma using induction over the recursive invocation of the $merge(e,P_1,P_2)$ function. Observe that the merge function is terminating because on each recursive call the number of differences between the two merged path-mappings, i.e., $P_1, P_2$ decreases. 
    
    First, consider the case when $P_1=P_2=P$. Let $H_1, H_1', H_2, H_2'$ and $\Delta_1, \Delta_2$ be instances of a concrete heap and local variable mapping such that $(P,H_1',\Delta_1) \rtimes H_1$ and $(P,H_2',\Delta_2) \rtimes H_2$. 
     At this point, we know that the merging will not introduce new changes, i.e., $(skip,P)=merge(e,P,P)$, where $\hat{s}=skip$, and $\hat{P}=P$. We want to show that $(\hat{P},H',\hat{\Delta}) \rtimes H$. There are two possibilities depending on the evaluation of $e$. The first case is when $\vdashconcestep{\Delta_1}{e}{true}$, and $b=true$. In this case, we have that $\concsstep{(\eta,\Delta_1,H_1,\nu, \mu)}{skip}{(\eta,\Delta_1,H_1,\nu, \mu)}{}$. But we have shown that when $(P,H_1',\Delta_1) \rtimes H_1$. Similarly, for the second case where  $b=false$, we have that $\concsstep{(\eta,\Delta_2,H_2,\nu, \mu)}{skip}{(\eta,\Delta_2,H_2,\nu,\mu)}{}$, but also we have $(P,H_2',\Delta_2) \rtimes H_2$, from the premise of the theorem. Thus, we have shown that inn either case  $(\hat{P},H',\hat{\Delta}) \rtimes H$.

     Next, consider the case where $P_1 \ne P_2$. Then, this means that there exist some $(r,f)$ pair whose mapped values are different in both maps, i.e., $P_1(r,f)\ne P_2(r,f)$. Let $t_1=P_1(r,f)$, and $t_2=P_2(r,f)$, where we know that $t_1\ne t_2$. In this case, according to the merge definition, a new fresh variable $t$ is created and both maps are updated to create two new maps $P_1'$ and $P_2'$ with one less difference. More formally, we have that $P_1'=P_1[(r,f)\leftarrow t]$, and $P_2'=P_2[(r,f)\leftarrow t]$. Then, the $merge$ function is recursively called over the remaining differences in $P_1'$ and $P_2'$ such that $(s_r',P_r')=merge(e,P_1',P_2')$. And finally, the function returns the pair $(t:=\gammaexp{e}{t_1}{t_2};s_r',P_r')$. Here $\hat{s}=t:=\gammaexp{e}{t_1}{t_2};s_r'$, and $\hat{P}=P_r'$. Let $\Delta_1,\Delta_2$ be two concrete local mapping, and $H_1, H_2$ be two concrete heap mappings, such that $(P_1,H_1',\Delta_1) \rtimes H_1$, and $(P_2,H_2',\Delta_2) \rtimes H_2$. Then, we want to show that if $(\Theta \vdash \concsstep{(\eta,\Delta_1,H',\nu, \mu)}{\hat{s}}{(\eta,\hat{\Delta},H',\nu, \mu)}{}$, then $(\hat{P},H_1',\Delta_1) \rtimes H_1$, if $\vdashconcestep{\Delta_1}{e}{b}$, and where $b=true$. Similarly, we want to show that if $\Theta \vdash \concsstep{(\eta,\Delta_2,H',\nu, \mu)}{\hat{s}}{(\eta,\hat{\Delta},H',\nu, \mu)}{}$, then $(\hat{P},H_2',\Delta_2) \rtimes H_2$, if $\vdashconcestep{\Delta_2}{e}{b}$, and where $b=false$. 

     To show this proof, will first show that the induction induction hypothesis holds, before we can apply it.
     First, observe that $Dom(P_1')=Dom(P_2')$ since $Dom(P_1)=Dom(P_2)$ from the premise, also $P_1'$ and $P_2'$ are created from $P_1$ and $P_2$, respectively, by updating an exiting mapping and thus not changing their domains.
     
     Next, let $\Delta_1$ be defined by $\Theta \vdash \concsstep{(\eta,\Delta_1,H_1',\nu, \mu)}{t=\gammaexp{e}{t_1}{t_2}}{(\eta,\Delta_1',H_1',\nu, \mu)}{}$. 
     We also know that $\forall x.isUse(x,e) \implies \Delta_1(x)=\Delta_2(x)$ and that 
    $\vdashconcestep{\Delta}{e}{b}$ ,where $b \in\{true,false\}$. Now assume that $b=true$. In this case we know that $\Delta_1'=\Delta_1[t\leftarrow \Delta_1(t_1)]$, using the $\gamma$-expression rule from the concrete semantics. Now observe that $(P_1',H_1',\Delta_1') \rtimes H_1$ holds, where $H_1'=H_1$, since no change on the heap can occur with the assignment statement. Moreover, observe that we have that $(P_1,H_1,\Delta_1) \rtimes H_1$ by the statement of the theorem. In addition, we know that the only change $P_1'$ has over $P_1$ is the update of the single reference-field pair $P_1'=P_1[(r,f)\leftarrow t]$, which will be mapped to the value of $t_1$ in $\Delta_1'$. But since we have $(P_1,H_1,\Delta_1) \rtimes H_1$, then we know that it must be $H_1(r,f)=\Delta_1(P_1(t_1))=\Delta_1'(P_1'(t))=H_1'(r,f)$. Because of the same reasoning, in the case where $b=false$, we will have that $(P_2',H_2',\Delta_2') \rtimes H_2$, where $H_2'=H_2$, again because no change on the heap can occur with the assignment statement.
    
    Next observe that the evaluation of $e$ in both $\Delta_1$ and $\Delta_2$ are the same, i.e, $\vdashconcestep{\Delta_1}{e}{b}$, and $\vdashconcestep{\Delta_2}{e}{b}$. This is because we have that $\forall x.isUse(x,e) \implies \Delta_1(x)=\Delta_2(x)$, by the statement of the theorem, and we also know the newly assigned $t$ is fresh, and thus previous variable mappings remain the same. 
    
    Finally, observe that $(\Theta \vdash \concsstep{(\eta,\Delta',H',\nu, \mu)}{s_r'}{(\eta,\Delta'',H',\nu, \mu)}{}$, then it must be that $\Delta''=\hat{\Delta}$, since $\hat{s}=t:=\gammaexp{e}{t_1}{t_2};s_r'$, and recall that $\Delta'$ must contain the mapping $t$, i.e., $\Delta_1'=\Delta_1[t\leftarrow t_1]$, or $\Delta_2'=\Delta_2[t\leftarrow t_2]$.

    At this point, we have shown that the induction hypothesis hold, and thus we can conclude that $(P_r',H',\hat{\Delta})\rtimes H$ hold. But observe here that $P_r'=\hat{P}$ since $P_r'$ will include the mapping all the mappings from the induction hypothesis in addition to the $(r,f)\leftarrow t$, which we have shown that satisfies the Bigrtimes relation. Thus, we can conclude that $(\hat{P},H',\hat{\Delta})\rtimes H$.
\end{proof}
}

Before establishing the soundness of the field transformation, we state two 
preliminary results. Corollary~\ref{rem:domainP} establishes a domain invariant 
on the field map $P$: after any application of the field judgment, either the 
flag is set to $\top$ and $P'$ is empty, reflecting that field elimination was 
halted by a method invocation, or the flag remains $\bot$ and the domain of 
$P'$ is preserved and equals that of $H_s$. This ensures that $P$ always 
tracks exactly the set of field-reference pairs in $H_s$ throughout the 
transformation. Lemma~\ref{lem:recursive-fieldSSA} establishes that the field 
judgment preserves the generalized SSA property: if $s$ satisfies $SSA(s)$ 
before the transformation, then the rewritten statement $s'$ satisfies 
$SSA(s')$ after. These two results together provide the necessary foundations 
for the soundness argument of the field transformation.

\begin{cor}
\label{rem:domainP}    
[Either $P$ is empty, or its domain is the same to the domain of $H_s$]

If 
$$
\Delta_s, H_s \vdash\rangerstep{(T,P,p)}{s}{(T',P',p')}{s'}{field}
$$
then
$$
p'=\top \wedge P'=\varnothing
$$
or
$$
p'=\bot \wedge Dom(P)=Dom(H_s)=Dom(P')
$$
\techreport{
\begin{proof}
    
It is easy to see why this theorem is true. Initially, the {\tt field-wrap} rule triggers the $\leftarrow_{field}$ judgment, which sets up $P$ to have the same domain as $H_s$.
Next, {\tt field-no-op}, {\tt field-get-var}, {\tt field-put-var}, {\tt field-invoke}, {\tt field-if-2}, {\tt field-if-3}, and {\tt field-if-4} rules, all end up with $p'=\bot$, but we can also see that in these cases $P'=\varnothing$. 

For {\tt field-get}, we can clearly see that $P'$ is the same as $P$, while $p'=\bot$, so the remark holds here too.

The rules that introduce a change on $P'$ that is different from $P$ includes {\tt put-field}, {\tt field-comp} and {\tt field-if-1}. For {\tt put-field}, we $P'$ is created from $P$ as $P'=P[(r,f) \leftarrow t]$, where $r,f$ are a reference, and field pair, while $t$ is a fresh temporary variable. But since $r,f$ already exists in $P$, then $P'$ is actually updating the mapping of an existing $r,f$ pair, rather than introducing a new entry. Thus in this case we have that $Dom(P)=Dom(H_s)=Dom(P')$.

In the {\tt field-comp} rule, the theorem hold straightforward by using induction over the judgments used to transform $s_1$, and $s_2$.

Finally, for {\tt field-if-1} rule, we have that $\Delta_s, H_s \vdash \rangerstep{(T,P,p)}{\myif{e}{s_1}{s_2}}{(T_2,\hat{P},\hat{p})}{\myif{e}{s_1'}{s_2'};\hat{s}}{field}$, such that $(\hat{s},\hat{P}) = merge(e,P_1,P_2), 
\Delta_s, H_s\vdash\rangerstep{(T,P,\bot)}{s_1}{(T_1,P_1,p_1)}{s_1'}{field}, 
\Delta_s, H_s\vdash\rangerstep{(T_1,P,\bot)}{s_2}{(T_2,P_2,p_2)}{s_2'}{field}, p_1=p_2=p=\bot$. Here we have $P'=\hat{P}, p'=\hat{p}$.
By induction, we know that $Dom(P_1)=Dom(P_2)=Dom(P)=Dom(H_s)$. In this case, from the rule we directly have that $p'=\bot$. Then, we want to show that $Dom(P')=Dom(\hat{P})=Dom(H_s)$. Here, we know that $\hat{P}$ is the result of $merge(e,P_1,P_2)$. But, from the definition of the merge function, the function either returns the same mapping if $P_1=P_2$, and it is clear then that $Dom(\hat{P})=Dom(P)$, or it create a new temporary variable $t$ such that $\hat{P}(r,f)=t$, but we know that $r,f$ must have existed in the domain of either $P_1$ or $P_2$. Then, we can conclude that $Dom(\hat{P})=Dom(P_1)=Dom(P_2)=Dom(P)=Dom(H_s)$.
\end{proof}
}
\end{cor}

Lemma~\ref{lem:fieldRecSound} establishes the soundness of the field 
transformation. It states that executing the original statement $s$ on a 
concrete state $(\Delta, H)$ and executing the rewritten statement $s'$ on an 
enriched state $(\Delta', H')$ that is P-Heap consistent with $H$ produce 
the same observable behavior, i.e., the same control state $\eta$, return 
value $\mu$, and path condition $\nu$. Furthermore, the local variable mapping 
is preserved in the sense that $\Delta'' \subseteq \Delta'''$, and the P-Heap 
consistency relation is maintained after execution, ensuring that the rewritten 
heap $H'''$ remains consistent with the updated field map $P'$.

\begin{lem}
\label{lem:recursive-fieldSSA}    
[The judgment $\rightarrow_{field}$ maintains the SSA property]

If 
$$
\Delta_s, H_s \vdash\rangerstep{(T,P,p)}{s}{(T',P',p')}{s'}{field}
$$
$$isSSA(s)$$
then
$$
isSSA(s')
$$
\end{lem}
\techreport{
\begin{proof}
    Rules {\tt field-no-op} and {\tt field-assign} hold trivially since, from the hypothesis, we have $isSSA(s)$, and $s$ does not change with the rewriting.

    Now, consider the rules {\tt field-comp, field-if-2}. For these rules, we can use straightforward induction over the premises of the rules to conclude that the rewritten statement maintains the SSA property.

    Next, consider the {\tt field-put} rule, here $s=\putfield{r}{f}{e}$, and $s'= t:=e$. But note here that the newly added variable $t$, is already fresh (it cannot clash with existing variables), thus it cannot clash with any of the defined/assigned variables.

    Next, consider the {\tt field-get} and {\tt field-invoke}, in both of these rule no new assignments are introduced, thus, the SSA property is maintained there too.

    Now, consider the more interesting case of {\tt field-if-1}. In this case, new statements are added in $\hat{s}$ in the form of a sequential composition of assignment statements of the form $t:=\gammaexp{e}{t_1}{t_2}$. According to the definition of the merge, $t$ must be fresh, i.e., cannot clash with any of the existing variables.
    
    Finally, consider the {\tt field-if-3} and {\tt field-if-4}. In both of these rules, the if-statement is rewritten, where one of its sides is sequentially composed with a sequence of {\tt get-field} statements; i.e., ($s_2';addPuts(P'')$), and ($s_1';addPuts(P')$). But, observe here that $s_2'$ and $s_1$ are both in SSA form by using induction. Moreover, the {\tt put-field} statements that results from $addPuts$ function, do not define/introduce new assignment statements. Thus, in both cases the rewritten statement also maintains the SSA property.
\end{proof}
}

\begin{lem}
\label{lem:fieldRecSound}    
[Recursive Field Rewriting is Sound]

If 
$$
\Delta \subseteq \Delta'
$$
$$ (P,H',\Delta') \rtimes H $$
$$
\Delta_s, H_s\vdash\rangerstep{(T,P,p)}{s}{(T',P',p')}{s'}{field}
$$
$$ (\Theta \vdash \concsstep{(\eta,\Delta,H,\nu,\mu)}{s}{(\eta'',\Delta'',H'',\nu'',\mu'')}{}$$
$$ (\Theta \vdash \concsstep{(\eta,\Delta',H',\nu, \mu)}{s'}{(\eta''',\Delta''',H''',\nu''',\mu''')}{}$$
then 
\begin{gather}
 \eta''=\eta''' \quad \wedge \quad \nu''=\nu''' \quad \wedge \quad \mu''=\mu'''\\
 \Delta'' \subseteq \Delta'''\\
 (P',H''',\Delta''') \rtimes H''
\end{gather}
\end{lem}
\techreport{

\begin{proof}
We proof the lemma using structural induction over the statement $s$. 
    
     Consider the {\tt field-no-op} rule. In this case the lemma statement holds trivially as $s=s'$, i.e., the statement $s$ has not been rewritten, and thus we will have that $\eta_s''=\eta_s''', \mu''=\mu''', \Delta'' \subseteq \Delta'''$
     and $H''=H'''$. We also have that $P=P'=\varnothing$, which means that the second case of the $(P',H''',\Delta''') \rtimes H''$ applies, i.e., $(\varnothing,H'',\Delta''') \rtimes H''$.

\begin{description}
    \item[field-put] We have $\Delta_s, H_s\vdash\rangerstep{(T,P,\bot)}{\putfield{r}{f}{e}}{(T\cup \{t\}, P',\bot)}{t := e}{field}$, where $ r \in R$ and $P'=P[(r,f)\leftarrow t]$, where $t$ is a fresh SSA variable. Executing concretely the {\tt put-field} statement we get $\Theta\vdash\concsstep{(\eta, \Delta, H, \bot)}{\putfield{r}{f}{e}}{(\eta, \Delta,H[(r,f) \leftarrow v_2],\bot)}{}$, such that $\vdashconcestep{\Delta}{e}{v_2}$, $\Delta''=\Delta$, and $H''=H[(r,f)\leftarrow v_2]$. 
From the hypothesis we have $\Delta'$, and $H'$, such that $(P,H',\Delta') \rtimes H$.
Executing the rewritten statement concretely, we get $  \Theta\vdash\concsstep{(\eta, \Delta', H', \bot)}{t := e}{(\eta, \Delta''',H''',\bot)}{}$, where $\Delta'''=\Delta'[t\leftarrow v_1]$, $\vdashconcestep{\Delta'}{e}{v_1}$, and $H'''=H'$. Here, we know that , $\eta'''=\eta=\eta'', \mu'''=\mu=\mu''$ (since $H,\eta$, and $\mu$ do not change). This proves the third conclusion (1) in our theorem.

To show conclusion (2), observe here that the rewriting has not changed the expression $e$. Moreover, here we are evaluating $e$ in $\Delta'$, where we know that $\Delta \subseteq \Delta'$, by the hypothesis of the lemma. Moreover, we know the SSA property is maintained (assume~\ref{lem:recursive-fieldSSA}). Thus, we can conclude that $e$ will evaluate to the same value, more precisely, $v_2=v_1$, and since $t$ is fresh, i.e., $t \notin Dom(\Delta')$, and that the only change between $\Delta'''$ and $\Delta'$ is the new mapping for $t$, then we must have that have $\Delta''\subseteq \Delta'''$.

Finally, to show conclusion (3), observe that from remark~\ref{rem:domainP}, we know that the domain of $P'$, is the same as the domain of $P$, and we have from the hypothesis that $(P,H',\Delta') \rtimes H$. Moreover it must be that an existing reference-field pair in $P'$ was updated rather than a new one was added. Thus, we can observe here that $\Delta'''(t)=v_1=v_2=H''(r,f)$. Thus, we can also conclude that $(P',H''',\Delta''') \rtimes H''$.

\item[field-get] We have $\Delta_s,H_s \vdash \rangerstep{(T,P,\bot)}{\getfield{x}{z}{f}}{(T,P,\bot)}{x:=t}{field}$, where $P'=P, T'=T, r_1=\Delta_s(z), r_1 \in R$, and $t = \lookup{P}{r_1,f}$.
Executing the original statement concretely we get $\Theta\vdash\concsstep{(\eta, \Delta, H, \bot)}{\getfield{x}{z}{f}}{(\eta, \Delta[x \leftarrow v_2],H, \bot)}{}{}$, such that $\vdashconcestep{\Delta}{z}{r_2}$, and $v_2 = \lookup{H}{r_2,f}$. Since we have that $(\Delta,H) \models (\Delta_s,H_s,\pi)$ (assume~\ref{prop:variable-consistent}), and since we assume that all references must be concrete values (assume~\ref{prop:ref-are-conc}), then it must be the case that $r_1=r_2$, which we will refer to as just $r$. 
Observe that $\eta_s''=\eta_s, H''=H, \Delta''=\Delta[x \leftarrow v_2], \nu''=\nu$ and $\mu''=\mu$. Let us assume that we have $\Delta', H'$, such that $(P,H',\Delta') \rtimes H$, and $\Delta \subseteq \Delta'$.

Now, executing the rewritten statement, we get $\Theta\vdash\concsstep{(\eta, \Delta', H', \bot)}{x:=t}{(\eta, \Delta'[x \leftarrow v_2],H', \bot)}{}{}$, such that $v_2 = \lookup{\Delta'}{t}$, $\Delta'''=\Delta'[x \leftarrow v_2]$. and $H'''=H'$. 
Here, we know that , $\eta'''=\eta=\eta'', \mu'''=\mu=\mu''$ (since $H,\eta_s$, and $\mu$ do not change). This proves the third conclusion (1) in our theorem.
 
Now, bserve that since we have that $\Delta \subseteq \Delta'$, and $\Delta''=\Delta[x \leftarrow v_1]$, and $\Delta'''=\Delta'[x \leftarrow v_2]$. Moreover, we can conclude that $v_1=v_2$. This is because, we have that $v_1 = \lookup{H}{r,f} = \Delta'(P(r,f)) = v_2$, where the second equality comes from $(P, H', \Delta') \rtimes H$, and the fact that $(r,f) \in Dom(P)$. In other words, $P$ must contain a mapping for $(r,f)$ pair, which when looked up in $\Delta'$, it must be equal to the  value of the pair in $H$, i.e., $H(r,f)$, consequently, it must be the case that $(P',H''',\Delta''') \rtimes H''$ (conclusion 3 holds).
Finally, due to the SSA property we know that $x \notin Dom(\Delta)$ and similarly $x \notin Dom(\Delta')$. Moreover, we have shown that $v_1=v_2$, thus, we can conclude that $\Delta''\subseteq\Delta'''$, which proves conclusion 2.

\item [field-invoke, field-put-var, field-get-var]
We will show the proof for field-invoke, but all three rules {\tt field-invoke, field-put-var}, and {\tt field-get-var} are follow the same proof structure.

In general, the proof of this rule follows from the observation that the {\tt field-puts} statement, generated by {\tt addPuts(P)}, copies the contents of $P$ into the heap without altering any other environment variable. By the induction hypothesis, we have $(P,H',\Delta') \rtimes H$, which implies that if $P$ is non-empty, then for every reference-field pair, the corresponding temporary variable name holds a concrete value equal to that of the reference-field pair in the heap. Consequently, writing back the values of $P$ into $H$ reproduces $H$ itself. Therefore, the {\tt invoke} statement is executed in exactly the same environment as the original statement, without introducing any changes. As a result, the final state remains identical, and thus all conclusions of the lemma continue to hold.

\item [field-comp]
Consider the {\tt field-comp} rule. Here we will use similarity between the conclusions of the lemma and the two similar premises to show how the property is maintained between multiple statements.  Here, we have
 $s=s_1;s_2$, such that $\Delta_s, H_s\vdash \rangerstep{(T, P,\bot)}{s_1;s_2}{(T',P_2,p_2)}{s_1';s_2'}{field}$, where $\Delta_s,H_s\vdash\rangerstep{(T,P,\bot)}{s_1}{(T_1,P_1,p_1)}{s_1'}{field}$, and $\Delta_s,H_s\vdash\rangerstep{(T_1,P_1,p_1)}{s_2}{(T_2,P_2,p_2)}{s_2'}{field}$. Note that here we have $P'=P_2, p'=p_2$. Observe here that {\tt composition-ret} cannot apply after running the remove final return transformation (assume~\ref{prop:final-ret-no-return}). Thus, the only concrete compositional rule that can apply at this point is {\tt composition}. Here, we get $\Theta\vdash \concsstep{(\eta, \Delta, H, \nu, \bot)}{s_1;s_2}{(\eta^5,\Delta^5, H^5, \nu^5, \mu^5)}{}$, where $\Theta\vdash \concsstep{(\eta, \Delta, H, \bot)}{s_1}{(\eta^4, \Delta^4, H^4, \nu^4, \bot)}{}$, and $\Theta\vdash \concsstep{(\eta^4, \Delta^4, H^4, \nu^4, \bot)}{s_2}{(\eta^5,\Delta^5, H^5,\nu^5,\mu^5)}{}$.
 Let us assume we have $\Delta',H'$ such that, $\Delta\subseteq \Delta'$, and $(P, H',\Delta')\rtimes H$. Now executing the $s_1'$ of the rewritten statement we get $\Theta\vdash \concsstep{(\eta, \Delta', H', \bot)}{s_1'}{(\eta^6, \Delta^6, H^6, \nu^6, \mu^6)}{}$.  Using the induction hypothesis at $s_1$, we know that it must be that $\eta''=\eta^5, \Delta''=\Delta^5, H''=H^5, \nu''=\nu^5$ and $\mu''=\mu^5$, and that $(P_1,H^6,\Delta^6) \rtimes H^4$, $\Delta^4\subseteq \Delta^6$, and that $\eta^6=\eta^4, \nu^6=\nu^4$, and $\mu^6=\mu^4=\bot$.
 Thus, executing $s_2'$ in this concrete state we get $\Theta\vdash \concsstep{(\eta^6, \Delta^6, H^6, \nu^6, \nu^6)}{s_2'}{(\eta^7,\Delta^7, H^7, \nu^7, \mu^7)}{}$. Again using the induction hypothesis at $s_2$ with the premises $(P_1,H^6,\Delta^6) \rtimes H^4$, and $\Delta^4\subseteq \Delta^6$, then we know that it must be the case that $(P_2,H^7,\Delta^7) \rtimes H^5$, and $\Delta^5\subseteq \Delta^7$, $\eta^7=\eta^5,\nu^7=\nu^5$ and $\mu^7=\mu^5$. But, this then mean that 
 $\Theta\vdash \concsstep{(\eta, \Delta, H, \bot)}{s_1';s_2'}{(\eta^7,\Delta^7, H^7, \mu^7)}{}$, where $\Delta'''=\Delta^7, H'''=H^7$. Now, we have just shown that $(P',H''',\Delta''') \rtimes H''$ (conclusion 3), and that $\Delta''=\Delta^5 \subseteq \Delta^7=\Delta'''$ (conclusion 2). And finally that, $\eta'''=\eta^7=\eta^5=\eta'',\nu'''=\nu^7=\nu^5=\nu''$ and $\mu'''=\mu^7=\mu^5=\mu''$ (conclusion 1).

 \item [field-if-1]
 Now, consider the {\tt field-if-1} rule. In this case, we have
 $s=\myif{e}{s_1}{s_2}$, such that $\Delta_s,H_s \vdash \rangerstep{(T,P,\bot)}{\myif{e}{s_1}{s_2}}{(T'',\hat{P},\bot)}{\myif{e}{s_1'}{s_2'};\hat{s}}{field}$, where $(\hat{s},\hat{P}) = merge(e,P_1,P_2)$, $\Delta_s,H_s\vdash\rangerstep{(T,P,\bot)}{s_1}{(T',P_1,\bot)}{s_1'}{field}$, and $\Delta_s,H_s \vdash\rangerstep{(T',P,\bot)}{s_2}{(T'',P_2,\bot)}{s_2'}{field}$. Note that here we have $P'=\hat{P}, p'=\bot$. There are two concrete execution rules that can apply depending on whether the condition of the if-statement evaluates to {\tt true}  or {\tt false}. We will show the proof argument for the former, but the latter follows a very similar argument.
 Now, consider executing the original statement using the both {\tt composition} and {\tt if-phi-true} rule, we get for the first statement $\Theta\vdash\concsstep{(\eta, \Delta, H, \bot)}{\myif{e}{s_1}{s_2}}{(\eta^4,\Delta^4, H^4, \nu^4, \mu^4)}{}$, where $\mu^4=\bot$, $\vdashconcestep{\Delta}{e}{true}$, and $\Theta\vdash\concsstep{(\eta,\Delta, H, \mu)}{s_1}{(\eta^4,\Delta^4, H^4, \nu^4, \mu^4)}{}$. Here we have that $\eta_s''=\eta_s^4, \Delta''=\Delta^4, H''=H^4, \nu''=\nu^4, \mu''=\mu^4=\bot$. Now consider executing the rewritten statement concretely. Let us assume we have $\Delta',H'$ such that, $\Delta\subseteq \Delta'$, and $(P, H',\Delta')\rtimes H$. Due to the subset relationship between $\Delta \subseteq \Delta'$, we know that the expression $e$ must evaluate to the same value {\tt true}, i.e., $\vdashconcestep{\Delta'}{e}{true}$, and therefore, rule {\tt if-phi-true} will still apply. Here we have $\Theta\vdash\concsstep{(\eta,\Delta', H', \mu)}{s_1'}{(\eta^5,\Delta^5, H^5, \nu^5, \mu^5)}{}$, where $\mu^5=\bot$. Using the induction hypothesis at $s_1$, we know that $(P_1,H^5,\Delta^5) \rtimes H^4$, and that $\eta^4=\eta^5, \nu^4=\nu^5, \mu^4=\mu^5$. Thus, we can evaluate $\hat{s}$ using the remaining part of the sequential composition, i.e.; $\Theta\vdash\concsstep{(\eta^5,\Delta^5, H^5, \nu^5, \mu^5)}{\hat{s}}{(\eta^6,\Delta^6, H^6, \nu^6, \mu^6)}{}$. 
 Note that, here we have that $\Delta'''=\Delta^6$, and $H'''=H^6$. 
 However, since by definition $\hat{s}$ is composed of a sequence of assignment statements, then no heap changes can ever occur. Thus we know that $\Delta''=\Delta^4 \subseteq \Delta^6=\Delta'''$. Also, we can conclude that $\eta'''=\eta^6=\eta_s^5=\eta_s^4=\eta_s'', \nu'''=\nu^6=\nu^5=\nu^4=\nu'', \mu'''=\mu^6=\mu^5=\mu^4=\mu''$ (conclusion 1). Observe that $H'''=H^6=H^5$, Now, we want to show that $(\hat{P},H^5,\Delta^6) \rtimes H^4$, but this falls directly from applying lemma~\ref{lem:field-merge-lemma}, where $\vdashconcestep{\Delta}{e}{b}, b=true$ (conclusion 3).

 \item[field-if-2] 
 Here,
 $s=\Delta_s,H_s \vdash \rangerstep{(T,P,\bot)}{\myif{e}{s_1}{s_2}}{(T'',\varnothing,\top)}{\myif{e}{s_1'}{s_2'}}{field}$, where $\Delta_s,H_s\vdash\rangerstep{(T,P,\bot)}{s_1}{(T',\varnothing,\top)}{s_1'}{field}$, and $\Delta_s,H_s \vdash\rangerstep{(T',P,\bot)}{s_2}{(T'',\varnothing,\top)}{s_2'}{field}$. Note that here we have $P'=\varnothing, p'=\top$. There are two concrete execution rules that can apply depending on whether the condition of the if-statement evaluates to {\tt true}  or {\tt false}. We will show the proof argument for the former, the proof of the latter follows a very similar argument.
 Now, consider executing the original statement using the {\tt if-phi-true} rule, we get for the first statement $\Theta\vdash\concsstep{(\eta, \Delta, H, \bot)}{\myif{e}{s_1}{s_2}}{(\eta^4,\Delta^4, H^4, \nu^4, \mu^4)}{}$, where $\mu^4=\bot, \vdashconcestep{\Delta}{e}{true}$, and $\Theta\vdash\concsstep{(\eta,\Delta, H, \mu)}{s_1}{(\eta,\Delta^4, H^4, \bot)}{}$. Here we have that $\eta''=\eta^4, \nu''=\nu^4, \mu''=\mu^4$, and $\Delta''=\Delta^4, H''=H^4$. Now, consider executing the rewritten statement concretely. By hypothesis, we have $\Delta',H'$ such that, $\Delta\subseteq \Delta'$, and $(P, H',\Delta')\rtimes H$. Due to the subset relationship between $\Delta \subseteq \Delta'$, we know that the expression $e$ must evaluate to the same value {\tt true}, i.e., $\vdashconcestep{\Delta'}{e}{true}$, and therefore, rule {\tt if-phi-true} will still apply. Here we have $\Theta\vdash\concsstep{(\eta,\Delta', H', \mu)}{s_1'}{(\eta^5,\Delta^5, H^5, \nu^5, \mu^5)}{}$, where $\mu^5=\bot$. Note, that  we have that $\eta'''=\eta^5, \nu'''=\nu^5, \mu'''=\mu^5, \Delta'''=\Delta^5$, and $H'''=H^5$. Using the induction hypothesis at $s_1$, we know that $\eta'''=\eta^5=\eta_s^4=\eta_s''$, $\nu'''=\nu^5=\nu^4=\nu''$ and $\mu'''=\mu^5=\mu^4=\mu''$ (conclusion 1). $(P'=\varnothing,H^5,\Delta^5) \rtimes H^4$ (conclusion 3) and $\Delta^4\subseteq\Delta^5$ (conclusion 2), which are what we want to prove.

 \item [field-if-3]
 Here, we have that
 $$\Delta_s,H_s \vdash \rangerstep{(T,P,\bot)}{\myif{e}{s_1}{s_2}}{(T_2,\varnothing,\top)}{\myif{e}{(s_1';addPuts(P_1))}{s_2'}}{field}$$
 where $\Delta_s,H_s\vdash\rangerstep{(T,P,\bot)}{s_1}{(T_1,P_1,\bot)}{s_1'}{field}$, and $\Delta_s,H_s \vdash\rangerstep{(T_1,P,\bot)}{s_2}{(T_2,\varnothing,\top)}{s_2'}{field}$. Note that here we have $P'=\varnothing, p'=\top$. There are two concrete execution rules that can apply depending on whether the condition of the if-statement evaluates to {\tt true}  or {\tt false}. We will show the proof argument for the latter, but the former follow a very similar argument to that of the {\tt field-if-2} rule. Now consider executing the original statement using the {\tt if-phi-true} rule, we get for the first statement $\Theta\vdash\concsstep{(\eta, \Delta, H, \bot)}{\myif{e}{s_1}{s_2}}{(\eta^4,\Delta^4, H^4, \nu^4, \mu^4)}{}$, where $\mu^4=\bot, \vdashconcestep{\Delta}{e}{true}$, and $\Theta\vdash\concsstep{(\eta,\Delta, H, \mu)}{s_1}{(\eta^4,\Delta^4, H^4, \nu^4, \mu^4)}{}$. Here we have that $\eta''=\eta^4, \Delta''=\Delta^4, H''=H^4, \nu''=\nu^4, \mu''=\mu^4$. Now, consider executing the rewritten statement concretely. By hypothesis, we have $\Delta',H'$ such that, $\Delta\subseteq \Delta'$, and $(P, H',\Delta')\rtimes H$. Due to the subset relationship between $\Delta \subseteq \Delta'$, we know that the expression $e$ must evaluate to the same value {\tt true}, i.e., $\vdashconcestep{\Delta'}{e}{true}$, and therefore, rule {\tt if-phi-true} will still apply. Now, we have to execute the then-side of the rewritten statement; i.e. $s_1';addPuts(P_1)$. Here, we have $\Theta\vdash\concsstep{(\eta,\Delta', H', \mu)}{s_1'}{(\eta^5,\Delta^5, H^5, \nu^5, \mu^5)}{}$, where $\mu^5=\bot$. Using the induction hypothesis at $s_1$, we know that $(P_1,H^5,\Delta^5) \rtimes H^4$. Now, evaluating the $addPuts(P_1)$, we get $\Theta\vdash\concsstep{(\eta^5, \Delta^5, H^5, \nu^5, \mu^5)}{addPuts(P_1)}{(\eta^6, \Delta^6,H^6, \nu^6, \mu^6)}{}{}$. However, according to the definition of the $addPuts(P_1)$(see Fig.~\ref{fig:fieldSSA}), $H_6$ will be supplemented with the same reference-field mappings as $H_4$, for all reference-field pairs that exists in $P_1$, since we have that $(P_1,H^5,\Delta^5) \rtimes H^4$ initially assumed. Thus, $H_6=H_4$. Also observe that the $addPuts(P_1)$ does not introduce any side effects on the local mapping, since the $addPuts(P_1)$ is a sequence of $put\_field$ statements, and $\eta_s'''=\eta_s^6=\eta_s^5=\eta_s'', \nu'''=\nu^6=\nu^5=\nu'', \mu'''=\mu^6=\mu^5=\mu''$ (conclusion 1), and $\Delta'''=\Delta^6=\Delta^5=\Delta''$ (conclusion 2). Thus, formally we can conclude that $\Theta\vdash\concsstep{(\eta, \Delta^5, H^5, \bot)}{addPuts(P_1)}{(\eta, \Delta^5,H^4, \bot)}{}{}$, which implies that $(\varnothing,H^4,\Delta^5) \rtimes H^4$ holds (conclusion 3).
 
 The proof for the {\tt field-if-4} follows the same argument as this one (the {\tt field-if-3} rule).

\end{description}
\end{proof}
}
\journalreport{

\begin{proof}[Proof Sketch]
We provide proof sketches for the {\tt field-put, field-get} and {\tt field-if-1} cases, 
as they illustrate the key role of the P-Heap consistency relation and the 
merge lemma respectively. The remaining cases and full proof details are 
deferred to the technical report~\cite{techreport}.

More precisely, the  \texttt{putfield} statement is rewritten to $t := e$ where 
$t$ is fresh and $P' = P[(r,f) \leftarrow t]$. Concretely executing the original 
statement updates $H'' = H[(r,f) \leftarrow v]$ where $v$ is the value of $e$ 
under $\Delta$. Executing the rewrite updates $\Delta''' = \Delta'[t \leftarrow v]$ 
and leaves $H'$ unchanged. Conclusion~(1) holds since neither execution affects 
$\eta$ or $\mu$. Conclusion~(2) follows since $e$ is unchanged by the 
transformation and evaluates to the same value under $\Delta \subseteq \Delta'$, 
and $t$ is fresh so $\Delta'' \subseteq \Delta'''$. Conclusion~(3) follows since 
$\Delta'''(t) = v = H''(r,f)$, and the domain of $P'$ is preserved by 
Corollary~\ref{rem:domainP}, so $(P', H''', \Delta''') \rtimes H''$.

The \texttt{getfield} statement is rewritten to $x := t$ where 
$t = P(r,f)$ and $r = \Delta_s(z) \in R$. Concretely executing the original 
statement yields $\Delta'' = \Delta[x \leftarrow v]$ where $v = H(r,f)$, leaving 
$H$ unchanged. Conclusion~(1) holds since neither execution affects $\eta$ or 
$\mu$. For conclusion~(3), observe that by P-Heap consistency $(P, H', \Delta') 
\rtimes H$, we have $\Delta'(t) = \Delta'(P(r,f)) = H(r,f) = v$, so executing 
$x := t$ under $\Delta'$ produces the same value $v$, and since $P' = P$ and 
$H''' = H'$, we immediately get $(P', H''', \Delta''') \rtimes H''$. 
Conclusion~(2) follows since $x \notin Dom(\Delta)$ and $x \notin Dom(\Delta')$ 
by the SSA property, and $v_1 = v_2$, so $\Delta'' \subseteq \Delta'''$.

Finally, in the {\tt field-if-1} rule, the conditional is rewritten to $\myif{e}{s_1'}{s_2'};\hat{s}$ 
where both branches are processed independently and $(\hat{s}, \hat{P}) = 
merge(e, P_1, P_2)$. We show the case where $e$ evaluates to \texttt{true}; 
the \texttt{false} case follows symmetrically. Since $\Delta \subseteq \Delta'$, 
$e$ evaluates to \texttt{true} under both $\Delta$ and $\Delta'$, so the same 
branch $s_1$ (resp. $s_1'$) is taken in both executions. By the induction 
hypothesis applied to $s_1$, conclusions~(1) and~(2) hold after executing 
$s_1'$, and we obtain $(P_1, H^5, \Delta^5) \rtimes H^4$. Executing $\hat{s}$ 
does not modify the heap since it consists solely of variable assignments, so 
$H''' = H^5$ and conclusions~(1) and~(2) are preserved. Conclusion~(3), i.e., 
$(\hat{P}, H''', \Delta''') \rtimes H''$, then follows directly from 
Lemma~\ref{lem:field-merge-lemma} with $b = \texttt{true}$.
\end{proof}
}

\AtBeginEnvironment{thm}{\setcounter{equation}{0}}

We now lift the soundness result from the recursive field judgment to the 
top-level wrapper. Theorem~\ref{lb:field-wrap-theorem} establishes that 
executing the original statement $s$ and its field-transformed counterpart 
$s''$ from the same concrete state produce the same observable behavior: 
the same control state $\eta$, path condition $\nu$, and return value $\mu$. 
Furthermore, the local variable mapping is extended consistently, i.e., 
$\Delta' \subseteq \Delta''$, and the heap is preserved, i.e., $H' = H''$, 
reflecting that all heap accesses have been internalized as variable assignments 
by the transformation.

 \begin{thm}[Field Transformation is Sound]
 \label{lb:field-wrap-theorem}
   $$(\Delta,H) \models_H (\Delta_s, H_s,\pi)$$
  $$\Delta_s, H_s \vdash \rangerstep{(I,O,T)}{s}{(I,O,T')}{s''}{field-wrap}$$
  $$\Theta \vdash \concsstep{(\eta, \Delta, H, \nu, \mu)}{s}{(\eta', \Delta', H', \nu', \mu')}{}$$
  $$\Theta \vdash \concsstep{(\eta, \Delta, H, \nu, \mu)}{s''}{(\eta'', \Delta'', H'', \nu'', \mu'')}{}$$

  then 
  \begin{gather}
 \eta'=\eta'' \quad \wedge \quad \nu'=\nu'' \quad \wedge \quad \mu'=\mu''\\
 \Delta' \subseteq \Delta''\\
 H'=H''
\end{gather}
 \end{thm}

\techreport{
\begin{proof}
    
    First, we have by applying the {\tt field-wrap} that $p=\bot, T'=T\cup Range(P)$ such that $H_s\vdash\rangerstep{(T',P,\bot)}{s}{(T'',P',p)}{s_2}{field}$.
    
     Next, observe that the obtained $s''$ can be understood in three parts, i.e., $s''= s_1;s_2;s_3$, such that $s_1=addGets(P); skip;$, and $s_3=skip; addPuts(P')$. This break down is unique since $addGets$ and $addPuts$ cannot contain a $skip$-statement.
    
    Now, executing concretely the first part of $s''$, we get
    $$\Theta \vdash \concsstep{(\eta,\Delta,H,\nu,\mu)}{s_1}{(\eta_1,\Delta_1,H_1,\nu_1,\mu_1)}{}$$

    Observe that here we have $\eta'=\eta_1, \Delta'=\Delta_1, H'=H_1, \nu'=\nu_1, \mu'=\mu_1$.

    Since $addGets(P)$ is defined as 
    $$addGets(P)=\getfield{t_1}{r_1}{f_1};\ldots;\getfield{t_n}{r_n}{f_n}$$

    $  \forall ((r_i,f_i),t_i) \in P, 1 \le i \le n,  n=|P|$. 
    Then, we know that the concrete execution of $s_1$ only affects the local variable mapping. Therefore, we know that $\eta_1=\eta$, $H_1=H$, $\nu_1=\nu$ and $\mu_1=\mu$. Because $P=initialP(H_s)$ each $t_i$ is a fresh variable, i.e., distinct from any variable in $\Delta$. From Corollary~\ref{rem:domainP}, we know that $Dom(H)=Dom(H_s)=Dom(P)$ and from the premise of the theorem we have  $(\Delta,H) \models (\Delta_s, H_s,\pi)$. Therefore, it must be $\Delta_1=\Delta[t_1 \leftarrow v_1,\ldots,t_n\leftarrow v_n]$, for each $(r_i,f_i) \in Dom(H)$ such that $t_i=P(r_i,f_i)$ and $v_i=H(r_i,f_i)$. 
    
    This implies that we must have $\Delta \subseteq \Delta_1$. 
    Also, recall that $(P,H^*,\Delta^*) \rtimes H$ is defined as 

    \[
        (P, H^*, \Delta^*) \rtimes H \quad \text{iff} \quad \forall (r, f)\in Dom(H).
        H(r, f) = 
        \begin{cases}
        \Delta^*(P(r,f)) & \text{if } (r,f) \in \text{Dom}(P) \\
        H^*(r,f) & \text{otherwise}
        \end{cases}
    \]

    Let us choose $H^*=H$, and $\Delta^*=\Delta_1$. 

    Let $(r_i,f_i)$ be any pair in the domain of $H$. Because the $Dom(H)=Dom(P)$ the first case applies. 
    Let $t_i=P(r_i,f_i)$, then $\Delta_1(t_i)=v=H(r_i,f_i)$, since $\Delta_1$ is known to be equal to $\Delta[t_1 \leftarrow v_1,\ldots,t_n\leftarrow v_n]$. Thus the equality in the definition of $(P,H,\Delta_1) \rtimes H$ holds.
    
    Now, when concretely executing $s_2$ we can apply Lemma~\ref{lem:fieldRecSound} for the soundness of the recursive judgment for the field transformation, such that we have 
    $$p=\bot$$
    $$
    \Delta_1 \subseteq \Delta_1
    $$
    $$ (P,H,\Delta_1) \rtimes H $$
    $$
    \Delta_s, H_s\vdash\rangerstep{(T,P,p)}{s_2}{(T',P',p')}{s_2'}{field}
    $$
    $$ (\Theta \vdash \concsstep{(\eta,\Delta_1,H,\nu,\mu)}{s_1}{(\eta_1,\Delta_1,H_1,\nu_1,\mu_1)}{}$$
    $$ (\Theta \vdash \concsstep{(\eta,\Delta_1,H,\nu, \mu)}{s_2'}{(\eta_2,\Delta_2,H_2,\nu_2,\mu_2)}{}$$
    then we must have that 

    \begin{gather*}
     \eta_1=\eta_2 \quad \wedge \quad \nu_1=\nu_2 \quad \wedge \quad \mu_1=\mu_2\\
     \Delta_1 \subseteq \Delta_2\\
     (P',H_2,\Delta_2) \rtimes H_1\\
    \end{gather*}

    Since we have $\eta'=\eta_1, \Delta'=\Delta_1, H'=H_1, \nu'=\nu_1, \mu'=\mu_1, \Delta \subseteq \Delta_1$ from before, then we can further simplify to obtain

     \begin{gather*}
     \eta'=\eta_1=\eta_2 \quad \wedge \quad \nu'=\nu_1=\nu_2 \quad \wedge \quad \mu'=\mu_1=\mu_2\\
     \Delta \subseteq \Delta_1 \subseteq \Delta_2\\
     (P',H_2,\Delta_2) \rtimes H_1\\
    \end{gather*}

   Finally, executing concretely the third part of $s''$, we get 
    
    $$\Theta \vdash \concsstep{(\eta_2,\Delta_2,H_2,\nu_2,\mu_2)}{s_3}{(\eta_3,\Delta_3,H_3,\nu_3,\mu_3)}{}$$

     where $s_3=addPuts(P')$, which is defined as 
     \[
     \begin{gathered}    addPuts(P')=\putfield{r_1}{f_1}{t_1};\ldots;\putfield{r_n}{f_n}{t_n} \quad 
    \forall ((r_i,f_i),t_i) \in P', 1 \le i \le n,  n=|P'|
    \end{gathered}
    \]

     Here we have that $\eta''=\eta_3,\Delta''=\Delta_3,H''=H_3,\nu''=\nu_3,\mu''=\mu_3$.
     
     Notice that the $addPuts(P')$ operation does not alter any of the environment variables except for the heap. Thus, we obtain $\eta''=\eta_3=\eta_2=\eta_1=\eta', \nu''=\nu_3=\nu_2=\nu_1=\nu', \mu''=\mu_3=\mu_2=\mu_1=\mu'$ (conclusion 1). 
     
     Furthermore, we have $\Delta_3=\Delta_2$, where $\Delta''=\Delta_3, \Delta'=\Delta_1$, and that 
     $\Delta \subseteq \Delta_1 \subseteq \Delta_2$, 
     thus, we can conclude $\Delta'=\Delta_2 \subseteq \Delta_3=\Delta''$ (conclusion 2).

     Finally, we have already established that $(P',H_2,\Delta_2) \rtimes H_1$. By definition, this means that for any reference-field pair $(r_i,f_i) \in \text{Dom}(H_1)$: if $(r_i,f_i) \in \text{Dom}(P')$, then its value is given by $\Delta_2(P'(r_i,f_i))$, such that $\Delta_2(P'(r_i,f_i))=H_1(r_i,f_i)$; otherwise, we have $H_1(r_i,f_i)=H_2(r_i,f_i)$.
/
    Therefore, executing $addPuts(P')$ in state $s_3$ restores the values of all reference-field pairs in $P'$ to the heap $H_3$. This must yield that $H_3$ have all the mappings that were perviously in $P'$, thereby, we must have $H^1=H^3$, establishing conclusion (3).

\end{proof}

}

\journalreport{\begin{proof}[Proof Sketch]
By the {\tt field-wrap} rule, the rewritten statement decomposes as 
$s'' = s_1; s_2; s_3$, where $s_1 = addGets(P);\mathit{skip}$, 
$s_2$ is the recursively rewritten statement, and 
$s_3 = \mathit{skip};addPuts(P')$.

Executing $s_1$ concretely introduces fresh temporaries $t_i$ for each 
$(r_i, f_i) \in Dom(P)$, updating $\Delta$ to $\Delta_1 = \Delta[t_i \leftarrow 
H(r_i,f_i)]$. Since all $t_i$ are fresh, $\Delta \subseteq \Delta_1$, and by 
construction $\Delta_1(P(r,f)) = H(r,f)$ for all $(r,f) \in Dom(H)$, 
establishing $(P, H, \Delta_1) \rtimes H$.

Applying Lemma~\ref{lem:fieldRecSound} to $s_2$ with $(P, H, \Delta_1) \rtimes H$ 
yields $\eta' = \eta_2$, $\nu' = \nu_2$, $\mu' = \mu_2$, $\Delta_1 \subseteq 
\Delta_2$, and $(P', H_2, \Delta_2) \rtimes H_1$.

Finally, executing $s_3 = addPuts(P')$ only modifies the heap, 
leaving $\eta$, $\nu$, $\mu$, and $\Delta$ unchanged, establishing 
conclusion~(1) and conclusion~(2). Since $(P', H_2, \Delta_2) \rtimes H_1$, 
the $addPuts(P')$ restores all field values tracked in $P'$ back to 
the heap, yielding $H'' = H'$ and establishing conclusion~(3).
\end{proof}}

Finally, we can conclude that the following corollary holds:
\begin{cor}
\label{cor:field-out-dont-change}
 input and output are the same before and after rewriting
 $$ \forall o \in O, \Delta''(o)=\Delta'''(o) $$
 \techreport{
 \begin{proof}
     This follows directly from the application of the subset relation from Theorem~\ref{lem:recursive-fieldSSA}, $\Delta''\subseteq \Delta'''$. And since the output variable names has not changed, then we must have that $ \forall o \in O, \Delta''(o)=\Delta'''(o) $.
 \end{proof}
 }
\end{cor}

\subsection{Properties of the Field Transformation}

The field transformation assumes all the guaranteed properties from pervious transformations: $\gamma$-creation, remove final-returns, renaming, method-inlining and substitution transformations.

Given,
  $$\Delta_s, H_s \vdash \rangerstep{(I,O,T)}{s}{(I,O,T')}{s'}{field-wrap}$$

The field transformation assumes the following property:
\begin{itemize}
        \item Field-SSA Transformation maintains the SSA property (Property~\ref{prop:ssa}). More formally, if
        $$
        \Delta_s, H_s \vdash\rangerstep{(I,O,T)}{s}{(I,O,T')}{s'}{field-wrap}
        $$
        $$s'=addGets(P);skip;s'';skip;addPuts(P')$$
        $$SSA(s)$$
        then
        $$SSA(s')$$
        \techreport{
        \begin{proof}
            We know that the $get-fields$ introduced by the $addGets(P)$, creates new assignment statements of the form $\getfield{P(r_1,f_1)}{r_1}{f_1};\ldots;\getfield{P(r_n,f_n)}{r_n}{f_n}, \forall (r_i,f_i) \in Dom(P), 1 \le i \le n,  n=|P|$, where $P=initialP(H_s)$. Observe here that all reference-field pairs are mapped to fresh variable(s) $t_i$, which by definition cannot clash with any of the existing variables.
            From lemma~\ref{lem:recursive-fieldSSA}, we know that $SSA(s'')$. Finally, the last component of the rewritten statement, the $addPuts(P')$, do not introduce new assignment statements. Thus, we can conclude that $SSA(s')$ holds.
        \end{proof}
        }

\end{itemize}

\soha{
\subsection{Array Transformation}}

\soha{While we do not present the formal semantics of the array transformation, we provide in this subsection an intuitive discussion of how it works in the real system. The core goal of memory transformations in our formalism is to eliminate heap operations by replacing them entirely with fresh, uniquely indexed local SSA variables, yielding a intermediate program containing only local variable assignments.}

\soha{
Array elements can be conceptualized directly as fields within an array object instance, where each array index $i$ acts as a distinct field name. Under this view, an {\tt arrayload} or {\tt arraystore} at a concrete index $i$ maps directly to a {\tt getfield} or {\tt putfield} operation for that specific field $i$. }


\soha{Unlike standard field accesses, array transformations must support loads and stores at symbolic indices. When an array store occurs at a symbolic index, the exact element being modified might be unknown statically. To account for this uncertainty, the transformation models array loads and stores as conditional expressions using the $\gamma$-expression.}

\soha{Also, array lengths introduce another key difference from fields, as symbolically sized arrays can theoretically assume any length. In practice, to handle symbolic lengths while guaranteeing termination during symbolic execution, Java Ranger concretizes array sizes to a finite set of candidate lengths and explores each valid size as an independent execution choice.}
\section{$TR_8$: Constant Reference Propagation}
\label{sec:constpropagation}
The constant reference propagation transformation aims to resolve reference 
variables by propagating concrete references to subsequent variables in the IR. Fig.~\ref{fig:constant-prop-rules} shows the rules for the transformation. 
The rules define two judgments. The judgment $T \vdash 
\rangerstepnolookup{(\Delta_r, s)}{\Delta_r'}{const}$ traverses an IR 
statement $s$, accumulating mappings from variables to their resolved concrete 
references in $\Delta_r$, where $T$ is the set of temporary variables whose 
assignments to concrete references are eligible for collection. The wrapper 
judgment $\Delta_s, T \vdash \rangerstep{(\mu_s, \pi_r)}{s}{(\pi_r', \mu_s')}{s'}{const\text{-}wrap}$ 
is the entry point of the transformation, triggering the const judgment on $s$ 
to collect $\Delta_r$, then applying substitution via ${\tt sub-wrap}$ to replace 
all occurrences of the collected reference variables throughout $s$, $\pi_r$, 
and $\mu_s$, producing the rewritten statement $s'$ and updated $\pi_r'$ and 
$\mu_s'$.

resolved concrete references. 
Rule~{\tt const-assign-change} handles the core 
case: if a variable $x \in T$ is assigned a concrete reference $r \in R$, then 
$\Delta_r$ is updated to record this mapping. Rule~{\tt const-assign-no-change} 
handles the remaining assignment cases, where either $x \notin T$ or the 
right-hand side is not a concrete reference, leaving $\Delta_r$ unchanged. 
Rules~{\tt const-if} and {\tt const-comp} propagate the transformation 
structurally, threading $\Delta_r$ through sequential composition while 
processing both branches of a conditional independently. All other statements 
are handled by {\tt const-no-op}, which leaves $\Delta_r$ unchanged.

Rule~{\tt const-wrapper} ties the transformation together. It first applies the 
constant propagation judgment to collect all concrete reference assignments into 
$\Delta_r$, effectively replacing those assignments with $\texttt{skip}$ while 
recording their mappings. Once $\Delta_r$ is fully populated, the wrapper 
immediately applies {\tt sub-wrap} to substitute all occurrences of the collected 
reference variables throughout the IR with their corresponding concrete 
references in $\Delta_r$. This two-phase design ensures that all concrete 
references are first identified and collected in a single pass, and then 
propagated uniformly to all their use sites via substitution, resolving 
previously unresolved reference variables for subsequent transformations such 
as field elimination and method inlining.

Next, rule~{\tt const-assign-change} handles the core 
case: if a variable $x \in T$ is assigned a concrete reference $r \in R$, then 
$\Delta_r$ is updated to record this mapping. Rule~{\tt const-assign-no-change} 
handles the remaining assignment cases, where either $x \notin T$ or the 
right-hand side is not a concrete reference, leaving $\Delta_r$ unchanged. 
Rules~{\tt const-if} and {\tt const-comp} propagate the transformation 
structurally, threading $\Delta_r$ through sequential composition while 
processing both branches of a conditional independently. All other statements 
are handled by {\tt const-no-op}, which leaves $\Delta_r$ unchanged. The 
resulting map $\Delta_r$ is then used by the substitution transformation to 
replace unresolved reference variables with their concrete counterparts in 
subsequent transformations.

 \begin{figure}[h!t]
    \footnotesize
\fbox{%
\parbox{\textwidth}{%

\[
   \infer[\rn{const-wrapper}]
   {\Delta_s,T \vdash \rangerstep{(\mu_s,\pi_r)}{s}{(\pi_r',\mu_s',)}{s'}{const-wrap}}
    { 
    \begin{gathered}
        T \vdash \rangerstepnolookup{(\varnothing,s)}{\Delta_r}{const}
    \qquad
    I'=Dom(\Delta_r) \\ 
          \Delta_r\cup\Delta_s \vdash \rangerstepnolookup{((\pi_r,\mu_s,I'),s)}{((\pi_r',\mu_s',I''),s')}{sub-wrap}
    \end{gathered}
    }
\]

\hrule
\hrule

\[
\infer[\rn{const-assign-change}]
      { T \vdash \rangersteplhsprime{\Delta_r}{x:=r}{\Delta_r'}{const}}
    { 
    x \in T \qquad r \in R
    \qquad
    \Delta_r'=\Delta_r[x\leftarrow r]
    }
\]

\[
\infer[\rn{const-assign-no-change}]
      { T \vdash \rangersteplhsprime{\Delta_r}{x:=e}{\Delta_r}{const}}
    { 
    ( x \notin T) \vee (e \notin R)
    }
\]

\[
    \infer[\rn{const-if}]
     { T \vdash \rangersteplhsprime{\Delta_r}{\myif{e}{s_1}{s_2}}{\Delta_r} {const}}
    {
     \begin{gathered}
        T \vdash \rangersteplhsprime{\Delta_r}{s_1}{\Delta_r'}{const}
        \qquad
       T \vdash \rangersteplhsprime{\Delta_r}{s_2}{\Delta_r''}{const}
    \end{gathered}
    }
\]

\[
 \infer[\rn{const-comp}]
     {T \vdash \rangersteplhsprime{\Delta_r}{s_1;s_2}{\Delta_r''}{const}}
    { 
    \begin{gathered}
    T \vdash \rangersteplhsprime{\Delta_r}{s_1}{\Delta_r'}{const}
    \qquad
    T \vdash \rangersteplhsprime{\Delta_r'}{s_2}{\Delta_r''}{const} 
    \end{gathered}
    }
\]

\[
 \infer[\rn{const-no-op}]
     { T \vdash \rangersteplhsprime{\Delta_r}{s}{\Delta_r}{const}}
    {  \func{noOp}_{const}(s) }
\]
    }}
\caption{Constant Propagation (for References)}
\label{fig:constant-prop-rules}
\end{figure}

\techreport{
 Here the $\noop{const}{s}$ is defined as

\noindent$\noop{const}{\return{e}} = true$\\
$\noop{const}{skip} = true$\\
$\noop{const}{\newobj{z}{c}} = true$\\
$\noop{const}{\putfield{z}{f}{e}} = true$\\
$\noop{const}{\getfield{x}{z}{f}} = true$\\
$\noop{const}{\invoke{x}{z}{g}{y}} = true$\\
$\noop{const}{s} = false \qquad$ otherwise \\

\noindent$\noop{const}{\binaryop{e_1}{e_2}{}} = true$\\
$\noop{const}{\unaryOp{e}{}} = true$\\
$\noop{const}{\gammaexp{e_1}{e_2}{e_3}} = true$\\
$\noop{const}{e} = false \qquad$ otherwise \\
}

\subsection{Constant Reference Propagation is Sound}

We now establish a key invariant of the constant propagation judgment: the 
reference map $\Delta_r$ remains a subset of the concrete local variable 
mapping $\Delta$ throughout the transformation. Intuitively, this ensures 
that every mapping collected in $\Delta_r$ corresponds to a valid concrete 
reference assignment witnessed during execution, and that $\Delta_r$ never 
introduces mappings that conflict with the symbolic state $\Delta_s$. 
Lemma~\ref{lem:const-stmt-sound} formalizes this: given that $\Delta_r 
\subseteq \Delta$ holds before rewriting $s$, it is preserved after rewriting, 
i.e., $\Delta_r' \subseteq \Delta'$, and $\Delta_r$ remains disjoint from 
$\Delta_s$.

\begin{lem}
\label{lem:const-stmt-sound}
    $$T \vdash \rangersteplhsprime{\Delta_r}{s}{\Delta_r'}{const}$$
    $$\Theta \vdash \concsstep{(\eta, \Delta, H, \nu, \mu)}{s}{(\eta', \Delta', H', \nu', \mu')}{} $$
    $$\Delta_r \subseteq \Delta$$
then
$$\Delta_r'\subseteq \Delta'$$
\techreport{
\begin{proof}
    For assignment statements, two rewriting cases are distinguished: const-assign-change and const-assign-no-change. The former applies when the assigned expression evaluates to a concrete reference value and the target variable is a temporary variable; the latter applies otherwise.

    Let us consider the first case, i.e., {\tt const-assign-change}, in this case the rewrite writes to $ T \vdash \rangersteplhsprime{\Delta_r}{x:=r}{\Delta_r'}{const}$, such that $x \in T,
    \Delta_r'=\Delta_r[x\leftarrow r]$. 
    Now, executing the original statement concretely we get $\Theta\vdash\concsstep{(\eta,\Delta, H, \nu,\bot)}{x:=r}{(\eta, \Delta',H,\nu,\bot)}{}$, such that $\Delta'=\update{\Delta}{x}{r}, \vdashconcestep{\Delta}{r}{r}$.  
    Since we have $\Delta_r \subseteq \Delta$, and the only added mapping to both $\Delta'$ and $\Delta_r'$ is the mapping for $x$, i.e., $\update{\Delta}{x}{r}$, then we can conclude that $\Delta_r' \subseteq\Delta'$

    Next consider {\tt const-no-change} rule. In this case, no change occurs on $\Delta_r$; i.e., $\Delta_r'=\Delta_r$. Also, we have from the premise of the lemma that $\Delta_r \subseteq \Delta$. Since, $\Delta'$ is exactly as $\Delta$ with an additional mapping for $x$, then it must be that $\Delta_r=\Delta_r' \subseteq \Delta'$. 
    Also, a similar argument is used for the proof of {\tt const-no-op} rule. 

    Next, consider {\tt const-if} rule. In this case, the rewrite writes to $T \vdash \rangersteplhsprime{\Delta_r}{\myif{e}{s_1}{s_2}}{\Delta_r} {const}$, such that $T \vdash \rangersteplhsprime{\Delta_r}{s_1}{\Delta_r'}{const}$ and $
       T \vdash \rangersteplhsprime{\Delta_r}{s_2}{\Delta_r''}{const}$. 
       There are two cases when executing this statement concretely depending on whether the evaluation for $e$ evaluates to $true$ or $false$. We will show the proof for the first case, but the second case is very similar. 
       Here we have $\vdashconcestep{\Delta}{e}{true}$, which implies that {\tt if-phi-true} must apply. Thus, we will have $\concsstep{(\eta,\Delta, H, \nu,\mu)}{s_1}{(\eta_1,\Delta_1, H_1, \nu_1,\mu_1)}{}$, here we have that $\Delta'=\Delta_1$. Recall from Lemma~\ref{lem:delta-ssa}, it must be that $\Delta \subseteq \Delta_1$. Since we have from the statement of the lemma that $\Delta_r \subseteq \Delta$, then by transitivity, we have that $\Delta_r\subseteq\Delta_1=\Delta'$, which is what we want to prove.

       Finally, consider {\tt const-comp} rule. In this case, we have that $T \vdash \rangersteplhsprime{\Delta_r}{s_1;s_2}{\Delta_{r2}}{const}$, such that $ T \vdash \rangersteplhsprime{\Delta_r}{s_1}{\Delta_{r1}'}{const}$, and $T \vdash \rangersteplhsprime{\Delta_{r1}'}{s_2}{\Delta_{r2}}{const} $. Here we have that $\Delta_r'=\Delta_{r2}$ . Observe here that {\tt const-comp} only applies after removing final return transformation (Property~\ref{prop:final-ret-no-return}). Thus, the only concrete compositional rule that can apply at this point is {\tt composition}. Here, we get $\Theta\vdash \concsstep{(\eta, \Delta, H, \nu, \bot)}{s_1;s_2}{(\eta_2,\Delta_2, H_2, \nu_2, \mu_2)}{}$, where $\Theta\vdash \concsstep{(\eta, \Delta, H, \bot)}{s_1}{(\eta_1, \Delta_1, H_1, \nu_1, \bot)}{}$, and $\Theta\vdash \concsstep{(\eta_1, \Delta_1, H_1, \nu_1, \bot)}{s_2}{(\eta_2,\Delta_2, H_2,\nu_2,\mu_2)}{}$. Here we have $\Delta_2=\Delta'$. By induction on $s_1$, we get $\Delta_{r1} \subseteq \Delta_1$, which allows us to use the lemma over $s_2$ to obtain $\Delta_r'=\Delta_{r2}\subseteq \Delta_2=\Delta'$, which is what we want to prove.
\end{proof}
}

\journalreport{
\begin{proof}[Proof Sketch]
We proceed by induction on the structure of $s$.

\textit{Case} {\tt const-assign-change}: The assignment $x := r$ where $x \in T$ 
and $r \in R$ updates $\Delta_r' = \Delta_r[x \leftarrow r]$ and $\Delta' = 
\Delta[x \leftarrow r]$. Since $\Delta_r \subseteq \Delta$ by hypothesis, and 
both $\Delta_r'$ and $\Delta'$ extend their respective maps with the same 
mapping for $x$, we conclude $\Delta_r' \subseteq \Delta'$.

\textit{Case} {\tt const-assign-no-change} and {\tt const-no-op}: $\Delta_r$ is 
unchanged, and $\Delta'$ extends $\Delta$ with at most one new mapping for $x$. 
Since $\Delta_r \subseteq \Delta \subseteq \Delta'$, the conclusion holds 
immediately.

\textit{Case} {\tt const-if}: Both branches are processed under the same 
$\Delta_r$. We show the case where $e$ evaluates to $\texttt{true}$; the 
$\texttt{false}$ case is symmetric. The concrete semantics selects $s_1$, 
yielding $\Delta' = \Delta_1$. By Lemma~\ref{lem:delta-ssa}, $\Delta \subseteq 
\Delta_1$, and since $\Delta_r \subseteq \Delta$ by hypothesis, transitivity 
gives $\Delta_r \subseteq \Delta' $.

\textit{Case} {\tt const-comp}: since no-return statement Property~\ref{prop:final-ret-no-return} is propagated from the remove final-returns transformation, then only the \texttt{composition} rule can apply concretely. By induction on $s_1$, we 
obtain $\Delta_{r1}' \subseteq \Delta_1$, which satisfies the premise for 
applying the induction hypothesis to $s_2$, yielding $\Delta_r' = \Delta_{r2} 
\subseteq \Delta_2 = \Delta'$.
\end{proof}
}
\end{lem}
Building on Lemma~\ref{lem:const-stmt-sound}, we now lift the soundness result 
to the wrapper judgment. Lemma~\ref{lem:const-wrap-sound} establishes that 
executing the original statement $s$ and its rewritten counterpart $s'$ from 
the same concrete state produce the same observable behavior: the same control 
state $\eta$, heap $H$, path condition $\nu$, and return value $\mu$. 
Furthermore, the local variable mapping is extended consistently, i.e., 
$\Delta' \subseteq \Delta''$, the concrete state remains a model of the symbolic 
state after rewriting, and the return semantics captured by $\pi_r'$ and $\mu_s'$ 
faithfully reflect the original return behavior, i.e., 
$\func{evalReturn}(\Delta', \pi_r, \mu_s) = \func{evalReturn}(\Delta'', \pi_r', 
\mu_s')$.
\begin{lem}
\label{lem:const-wrap-sound}
    $$(\Delta) \models (\Delta_s,\pi)$$
    $$\Delta_s,T \vdash \rangerstep{(\mu_s,\pi_r)}{s}{(\mu_s',\pi_r')}{s'}{const-wrap}$$
    $$\Theta \vdash \concsstep{(\eta, \Delta, H, \nu, \mu)}{s}{(\eta', \Delta', H', \nu', \mu')}{} $$
    $$\mu^*=\func{evalReturn}(\Delta', \pi_r,\mu_s)$$
    $$\Theta \vdash \concsstep{(\eta, \Delta, H, \nu,\mu)}{s'}{(\eta'', \Delta'', H'', \nu'',\mu'')}{}$$
then
$$\Delta'\subseteq \Delta''$$
$$\eta''=\eta', H''=H', \nu''=\nu',\mu''=\mu'$$
$$\mu^*=\func{evalReturn}(\Delta'',\pi_r',\mu_s')$$

\techreport{
\begin{proof}
    Here, we have that ${ T \vdash s \longrightarrow_{const-wrap} \Delta_r}$, such that $  T \vdash \rangerstepnolookup{(\varnothing,s)}{\Delta_r}{const}$, $I'=Dom(\Delta_r)$ and $
          \Delta_s \cup \Delta_r \vdash \rangerstepnolookup{((\pi_r,\mu_s,I'),s)}{(( \pi_r',\mu_s',I''),s')}{sub-wrap}$. 
    In this case, we obviously have $\varnothing \subseteq\Delta$. Thus, we can apply the const-statement lemma~\ref{lem:const-stmt-sound} to conclude that $\Delta_r \subseteq \Delta'$.

    Now, we want to apply Corollary~\ref{cor:subs-wrap-stmt-sound}, to show that apply substitution on the collected references is sound. Let $\Delta_s'=\Delta_s \cup \Delta_r$. First, we show that the premises of Corollary~\ref{cor:subs-wrap-stmt-sound} hold, then we apply it.
    Observe that
    \begin{itemize}
        \item We want to show that $\forall i \in I.\forall x \in Dom(\Delta_s').\neg isUse(i,\Delta_s'(x))$. Since $\Delta_s'=\Delta_s \cup \Delta_r$, then we want to show that this condition holds for both $\Delta_r$ and $\Delta_s$. To see how it holds for $\Delta_r$, observe that $I'=Dom(\Delta_r)$ and that $\Delta_r$ only collects assignments from constants, then it must be that the $Range(\Delta_r) \in V$, i.e., all mapped to expressions, must be concrete values, which means there cannot be any usage of any variable in the $Range(\Delta_r)$, which makes this condition holds for $\Delta_r$. To see why it holds for $\Delta_s$, since $SSA(s)$ (guaranteed property from field transformation, i.e., Property~\ref{prop:ssa}) and the definitions of $I'$ comes from within $s$, then that means there cannot be any uses of it within $Range(\Delta_s)$, thus here too, the condition holds.
        \item since we have that $(\Delta) \models (\Delta_s,\pi)$ holds from the premise of the current lemma. And again because of the SSA property for $s$ we know that $\Delta_r \cap \Delta_s=\varnothing$, thus evaluating expression in $\Delta_s\cup\Delta_r$ must yield the same values, thus we can conclude that $(\Delta) \models (\Delta_s\cup \Delta_r,\pi)$; i.e., $(\Delta) \models (\Delta_s',\pi)$.
        \item finally, we have that $\Delta\subseteq\Delta$.
    \end{itemize}
    The above allows us to apply Corollary~\ref{cor:subs-wrap-stmt-sound}, to conclude that 
    $\Delta'\subseteq \Delta''$, $\eta''=\eta', H''=H', \nu''=\nu',\mu''=\mu'$, $(\Delta'') \models (\Delta_s,\pi)$, $\mu^*=\func{evalReturn}(\Delta'',\pi_r',\mu_s')$, and that $\forall i \in I'.\neg isUse(i,s')$.
\end{proof}

\subsection{Properties of Constant Reference Propagation}
The transformation assumes all properties guaranteed by previous transformations, i.e., $\gamma$-creation, remove final-returns, renaming, method-inlining and field-SSA transformations.

In addition the constant reference propagation transformation guarantees the following property
\begin{itemize}
    \item Rewritten Statement $s'$ satisfies the generalized SSA property (Property~\ref{prop:ssa}). More formally, $SSA(s')$ holds.
\end{itemize}

}

\journalreport{
\begin{proof}[Proof Sketch]
By {\tt const-wrapper}, the transformation first collects all concrete reference 
assignments into $\Delta_r$ via the constant propagation judgment, starting 
from $\varnothing$, then applies {\tt sub-wrap} with $\Delta_s' = \Delta_s \cup 
\Delta_r$ to substitute all collected references throughout $s$.

Since $\varnothing \subseteq \Delta$, Lemma~\ref{lem:const-stmt-sound} gives 
$\Delta_r \subseteq \Delta'$. To apply Corollary~\ref{cor:subs-wrap-stmt-sound}, 
we verify its three premises. First, $\forall i \in I'$ no variable in $I'$ 
appears in $Range(\Delta_s')$: for $\Delta_r$ this holds because its range 
consists entirely of concrete references containing no variable occurrences; 
for $\Delta_s$ this follows from $SSA(s)$ (Property~\ref{prop:ssa}) that enforces that variables can only be defined once. Thus, variables defined within $\Delta_s$ cannot be redefined in $\Delta_r$ if $SSA(s)$ holds, which we showed that it does. Also, observe that $I' = Dom(\Delta_r)$ 
originates from within $s$. 

Finally, $\Delta \subseteq 
\Delta$ holds trivially. Applying Corollary~\ref{cor:subs-wrap-stmt-sound} then 
yields all four conclusions directly.
\end{proof}
}
\end{lem}

\section{Path-Merging Fix-Point Algorithm}
\label{sec:fixpoint}
\begin{algorithm}[h!t]
\DontPrintSemicolon
\SetAlgoLined
{\bf Input}: to-summarize statement $s$\;
{\bf Input}: concrete state $\delta=(\Theta,\eta,\Delta,H,\nu,\mu))$\;
{\bf Data}: Ranger State $\omega=(\Theta_t,\eta_s,\Delta_s, H_s, \mu_s,(\pi,\pi_r,\pi_s), (I,O,T), u)$\;
$\omega$:={\tt initialize-state}($\delta$)$\quad$ where $\omega=(\Theta_t,\eta_s,\Delta_s, H_s, \mu_s,(\pi,\pi_r,\pi_s), (I,O,T), u)$\;
$\rangerstepnolookup{s}{s^1}{\gamma-wrap}$\;
$ s^1 \to_{ret-wrap} (\pi_r^2,s^2)$\;
$\rangerstepnolookup{s^2}{(\mu_s^3,s^3)}{final-wrap}$\;
$I,O\vdash \rangerstepnolookup{((\mu_s^3,\pi_r^2,T,u), s^3)}{((T^4, \mu_s^4, \pi_r^4, u^4), s^4)}{rename-wrap}$\;
$ \Delta_s \vdash \rangerstepnolookup{(\pi_r^4,\mu_s^4,I, s^4)}{((\pi_r^5,\mu_s^5,I^5),s^5)}{sub-wrap}$\;
$(\Theta_t',\eta_s',\Delta_s', H_s', \mu_s', (\pi',\pi_r',\pi_s')', (I',O',T'), u'):=((\Theta_t,\eta_s,\Delta_s,H_s,\mu_s^5,(\pi,\pi_r^5,\pi_s),(I^5,O,T^4),u^4),s^5)$\;
$(\omega',s'):=((\Theta_t',\eta_s',\Delta_s', H_s', \mu_s', (\pi',\pi_r',\pi_s')', (I',O',T'), u'),s^5)$\;
        \Repeat{($\omega_{fix}, s_{fix}$) := ($\omega', s'$)}{
        $(\omega_{fix},s_{fix}):=(\omega',s')$\;
        $\Theta_t',\eta_s',\Delta_s' \vdash \rangerstep{(T',u')}{s'}{(T^6,u^6)}{s^6}{inline-wrap}$\;
		$\Delta_s', H_s' \vdash \rangerstep{(I',O',T^6)}{s^6}{(I',O',T^7)}{s^7}{field-wrap}$\;
        $\Delta_s',T^7 \vdash \rangerstep{(\mu_s',\pi_r')}{s^7}{(\pi_r^8,\mu_s^8,)}{s^8}{const-wrap}$\;
        $(\omega',s'):=((\Theta_t',\eta_s',\Delta_s'\cup\Delta_r, H_s', \mu_s^8, (\pi,\pi_r^8,\pi_s^8), (I',O',T^7), u^8),s^8)$\;
		}
{\bf return} ($\omega_{fix},s_{fix}$) 
\caption{\toolshort $ $ Static-Summary-Instantiation Algorithm}
\label{fig:algorithm}
\end{algorithm}

\subsection{Path-Merging fix-point Algorithm is Sound}

We divide the proof of the path-merging Algorithm into three main parts. First, we show that the transformations before hitting the fix-point loop are sound, i.e., line 5-9. Then, we show that the fix-point loop is sound, \soha{if terminated} (line 12-16).

Before we show the proof, we define $TR_i$ to denote the a particular transformation, more formally:
\begin{align*}
TR_1(\omega,s) &= (\omega,s'), 
    &&\text{where } \rangerstepnolookup{s}{s'}{{\gamma}\text{-}wrap} \\[4pt]
TR_2(\omega,s) &= (\omega',s'), 
    &&\text{where } s \to_{ret\text{-}wrap} (\pi_r',s'),\ 
       \omega'=\omega[\pi_r'/\pi_r] \\[4pt]
TR_3(\omega,s) &= (\omega',s'), 
    &&\text{where } \rangerstepnolookup{s}{(\mu_s',s')}{final\text{-}wrap},\ 
       \omega'=\omega[\mu_s'/\mu_s] \\[4pt]
TR_4(\omega,s) &= (\omega',s'), 
    &&\text{where } I,O\vdash \rangerstepnolookup{((\mu_s,\pi_r,T,u), s)}{((T', \mu_s', \pi_r', u'), s')}{rename\text{-}wrap}, \\
               & &&\phantom{\text{where }} \omega'=\omega[\mu_s'/\mu_s, \pi_r'/\pi_r, T'/T, u'/u] \\[4pt]
TR_5(\omega,s) &= (\omega',s'), 
    &&\text{where } \Delta_s \vdash 
       \rangerstepnolookup{(\pi_r,\mu_s,I, s)}{((\pi_r',\mu_s',I'),s')}{sub\text{-}wrap}, \\
               & &&\phantom{\text{where }} \omega'=\omega[ \pi_r'/\pi_r, \mu_s'/\mu_s, I'/I] \\[4pt]
TR_6(\omega,s) &= (\omega',s'), 
    &&\text{where } \Theta_t,\eta_s,\Delta_s \vdash 
       \rangerstep{(T,u)}{s}{(T',u')}{s'}{inline\text{-}wrap}, \\
               & &&\phantom{\text{where }} \omega'=\omega[\pi_r'/\pi_r, T'/T, u'/u] \\[4pt]
TR_7(\omega,s) &= (\omega',s'), 
    &&\text{where } \Delta_s, H_s \vdash 
       \rangerstep{(I,O,T)}{s}{(I',O',T')}{s'}{field\text{-}wrap}, \\
               & &&\phantom{\text{where }} \omega'=\omega[I'/I, O'/O, T'/T] \\[4pt]
TR_8(\omega,s) &= (\omega',s'), 
    &&\text{where } T \vdash 
       \rangerstep{(\mu_s,\pi_r)}{s'}{(\pi_r',\mu_s')}{s'}{const\text{-}wrap}, \\
               & &&\phantom{\text{where }} \omega'=\omega[(\pi_r'/\pi_r, \mu_s'/\mu_s]
\end{align*}
  We use the notation $TR_n$ to denote the application of a set of transformations, such that 

       $$TR_n(\omega,s) = (\omega', s'), \text{ where } (\omega', s') = TR_n \circ TR_{n-1} \circ \cdots \circ TR_1(\omega,s)$$
       
        \begin{lem} 
        Pre-Fixed-Point Transformations (line 5-9) are Sound.
            \label{lem:pre-fix-sound}
        \end{lem}
        
        Given a statement to path-merge $s$, a concrete state $\delta$ and a symbolic state $\omega$, such that:

        Let $\delta = (\eta, \Delta, H, \nu, \mu)$, $\delta'=(\eta', \Delta', H', \nu', \mu')$,
and $\delta''=(\eta'', \Delta'', H'', \nu'', \mu'')$.

Let $\omega=(\Theta_t,\eta_s,\Delta_s, H_s, \mu_s,(\pi,\pi_r,\pi_s), (I,O,T), u)={\tt initialize\text{-}state}(\delta)$

$\omega'=(\Theta_t',\eta_s',\Delta_s', H_s', \mu_s',(\pi',\pi_r',\pi_s'), (I',O',T'), u')$.

        we have that 
        \begin{gather*}
        \Theta \vdash \concsstep{\delta}{s}{\delta'}{} \tag{premise 1}\\
        (\omega',s') = TR_5 \circ TR_4 \circ TR_3 \circ TR_2 \circ TR_1 (\omega,s)  \tag{premise 2}\\
        (\Delta) \models (\Delta_s,\pi) \tag{premise 3}\\
        (\Delta,H) \models_H (\Delta_s,H_s,\pi)\tag{premise 4}\\
        \forall i \in \func{getInputs}(s).\forall x \in Dom(\Delta_s).\neg isUse(i,\Delta_s(x)) \tag{premise 5}\\
        \eta_s \subseteq \eta \tag{premise 6}\\
        \Theta \vdash \concsstep{\delta}{s'}{\delta''}{} \tag{premise 7}\\
        \end{gather*}
        
        then
        \begin{gather*}
        \eta'=\eta''\wedge H'=H'' \wedge \nu'=\nu'' \tag{conclusion 1}\\
         \forall i \in I'.\Delta'(i)=\Delta''(i) \tag{conclusion 2} \\
        \forall o\in O'. \Delta'(o) =\Delta''(o) \tag{conclusion 3}\\
        \mu''=\bot \tag{conclusion 4}  \\
        \mu' = evalReturn(\Delta'',\pi_r',\mu_s') \tag{conclusion 5}\\
        (\Delta'') \models (\Delta_s',\pi') \tag{conclusion 6}\\
        \eta_s' \subseteq \eta'' \tag{conclusion 7}\\
        (\Delta,H) \models_H (\Delta_s',H_s',\pi') \tag{conclusion 8}\\
        \end{gather*}
        
     \begin{proof}
            Observe that {\tt path-merge}$_{5-9}$ means executing the following transformations:
            $$\rangerstepnolookup{s}{s^1}{\gamma-wrap}$$
            $$ s^1 \to_{ret-wrap} (\pi_r^2,s^2)$$
            $$\rangerstepnolookup{s^2}{(\mu_s^3,s^3)}{final-wrap} $$
            $$I,O\vdash \rangerstepnolookup{((\mu_s^3,\pi_r^2,T,u), s^3)}{((T^4, \mu_s^4, \pi_r^4, u^4), s^4)}{rename-wrap}$$
            $$ \Delta_s \vdash \rangerstepnolookup{(\pi_r^4,\mu_s^4,I, s^4)}{((\pi_r^5,\mu_s^5,I^5),s^5)}{sub-wrap}$$

            where $s'=s^5, \pi_r'=\pi_r^5,\mu_s'=\mu_s^5, I'=I^5, u'=u^4, T'=T^4$ and the other environment variables of the symbolic state remains unchanged, i.e., $\Theta_t'=\Theta_t, \eta_s'=\eta_s, \Delta_s'=\Delta_s, H_s=H_s', \pi_s'=\pi_s, O'=O$.

            First, 
            observe that  premise 3, 4, 5 and 6 satisfy the Property~\ref{prop:variable-consistent}, \ref{prop:heap-consistent},  \ref{prop:initial-deltas-no-I} and \ref{prop:eta-unchanged}, respectively.  Also, observe {\tt initialize-state}($\delta$) also assumes that $\mu=\bot$ (Property~\ref{prop:initial-no-return}).
            
            Recall that transformations $TR_1, TR_2$ and $TR_3$ corresponds to the static transformations of Java Ranger, which we considered their composition in Theorem~\ref{thm:transform-sound}, such that
            
            if
            $$\Theta \vdash \concsstep{(\eta, \Delta, H, \nu, \mu)}{s^3}{(\eta^3, \Delta^3, H^3, \nu^3, \mu^3)}{} $$
            then

            $$\eta'=\eta^3,\Delta'=\Delta^3,H'=H^3,\nu'=\nu^3$$
            $$ \mu' = evalReturn(\Delta^3,\pi_r^2,\mu_s^3)$$
            $$  \mu^3=\bot$$

            Next, since we have that 
            
            $$\Theta \vdash \concsstep{(\eta, \Delta, H, \nu, \mu)}{s^3}{(\eta^3, \Delta^3, H^3, \nu^3, \mu^3)}{} $$
            $$\mu=\mu^3=\bot$$
            $$I,O\vdash \rangerstepnolookup{((\mu_s^3,\pi_r^2,T,u), s^3)}{((T^4, \mu_s^4, \pi_r^4, u^4), s^4)}{rename-wrap}$$
            $$\Theta \vdash \concsstep{(\eta, \Delta, H, \nu, \mu)}{s^4}{(\eta^4, \Delta^4, H^4, \nu^4, \mu^4)}{} $$
            then, we can apply renaming theorem (knowing that we have shown that $\Delta'=\Delta^3$ from before) (Theorem~\ref{thm:renaming-wrap-sound}) to get

            $$ \forall i \in I, \Delta(i)=\Delta^3(i)=\Delta^4(i)=\Delta'(i) $$
            $$ \forall o \in O, \Delta^3(o)=\Delta^4(o)=\Delta'(o) $$
            $$evalReturn(\Delta^3,\pi_r^2,\mu_s^3)=evalReturn(\Delta^4,\pi_r^4,\mu_s^4)$$

            where $\eta^4=\eta^3=\eta', H^4=H^3=H', \nu^4=\nu^3=\nu'=\bot,\mu^4=\mu^3=\bot$.
            
            Also, observe that we have established from before that $\mu' = evalReturn(\Delta^3,\pi_r^2,\mu_s^3)$, and we also have shown that $$evalReturn(\Delta^3,\pi_r^2,\mu_s^3)=evalReturn(\Delta^4,\pi_r^4,\mu_s^4)$$
            
            thus by transitivity, we have $\mu'=evalReturn(\Delta^4,\pi_r^4,\mu_s^4)$            
            
            Next, consider the {\tt sub-wrap}. From premise 4, 5 of the lemma, we have that 

            $$(\Delta) \models (\Delta_s,\pi)$$
            $$\forall i \in I.\forall x \in Dom(\Delta_s).\neg isUse(i,\Delta_s(x))$$

            then if 
            $$\rangerstepnolookup{((\pi_r^4,\mu_s^4, I),s^4)}{((\pi_r^5,\mu_s^5,I^5),s^5)}{sub-wrap}$$
            $$\Theta \vdash \concsstep{(\eta, \Delta, H, \nu,\mu)}{s^5}{(\eta^5, \Delta^5,H^5,\nu^5,\mu^5)}{}$$
            Then, we can apply Corollary~\ref{cor:subs-wrap-stmt-sound}, reusing $\Delta$ in the position of the free local variable mapping (named $\Delta''$ in Corollary~\ref{cor:subs-wrap-stmt-sound}). Since $\Delta \subseteq \Delta$, we get

        $$\eta^5=\eta^4, H^5=H^4, \nu^5=\nu^4,\mu^5=\mu^4$$
        $$\Delta^4\subseteq\Delta^5$$
        $$\forall i \in I.\neg isUse(i,s^5)$$
        $$(\Delta^5) \models (\Delta_s,\pi^5)$$
        $$\func{evalReturn}(\Delta^4,\pi_r^4,\mu_s^4)=\func{evalReturn}(\Delta^5,\pi_r^5,\mu_s^5)$$

        Observe that since this is the last transformation in {\tt path-merge}$_{5-9}$, then we have 
        
        $\eta''=\eta^5, \Delta''=\Delta^5, H''=H^5, \nu''=\nu^5$. 
        
        and by transitivity, we get 
        
        $\eta''=\eta^5=\eta^4=\eta', H''=H^5=H^4=H', \nu''=\nu^5=\nu^4=\nu'$ (conclusion 1). 

        Here, we know that the set of inputs $I$ and set of outputs $O$ are unchanged by any of the transformation, and also must have mappings in $\Delta^4$. Now, as $\Delta^4 \subseteq\Delta^5$, then we know that all mapping for variables in $I$ and $O$ must be exactly the same between $\Delta^4$ and $\Delta^5$. Therefore we can conclude that

        $$ \forall i \in I, \Delta^4(i)=\Delta^5(i)=\Delta''(i)$$
        $$ \forall o \in O, \Delta^4(o)=\Delta^5(o)=\Delta''(o)$$

        since we have shown that $\Delta^3=\Delta'$, by transitivity we can conclude that 
       
        $ \forall i \in I, \Delta'(i)=\Delta''(i)$, and similarly, we conclude that
        $ \forall o \in O, \Delta'(o)=\Delta''(o)$. And since none of the transformations changed the input or the output, then we can conclude that the new input $I'$ and output $O'$ must have
        $ \forall i \in I', \Delta'(i)=\Delta''(i)$ (conclusion 2), and that, 
        $ \forall o \in O', \Delta'(o)=\Delta''(o)$ (conclusion 3). 
        
        Also, we have that $\mu''=\mu^4=\bot$ (conclusion 4). 
        
        Note that here 
        $\pi_r'=\pi_r^4, \mu_s'=\mu_s^4$.
        Finally, by transitivity we have that
        $$\func{evalReturn}(\Delta^5,\pi_r^5,\mu_s^5)=\func{evalReturn}(\Delta^4,\pi_r^4,\mu_s^4)= evalReturn(\Delta^3,\pi_r^2,\mu_s^3)=\mu' \text{(conclusion 5)}$$

        We also observe that the symbolic variable mapping $\Delta_s$, and reference mappings remain unchanged, i.e., $\Delta_s'=\Delta_s,\eta_s'=\eta$. And we concluded from applying Corollary~\ref{cor:subs-wrap-stmt-sound} that
    
        %
        $$(\Delta^5) \models (\Delta_s',\pi) \qquad \text{(conclusion 6)}$$ 
    Also, as shown above, $\eta''=\eta$, and $\eta_s \subseteq \eta$ from premise 7, thus $\eta_s'\subseteq\eta''$ (conclusion 8).

    Finally, since we have $\Delta_s'=\Delta_s$, $H_s'=H_s$ and $\pi'=\pi$, and from premise 4, we can conclude that we have $(\Delta,H) \models_H (\Delta_s',H_s',\pi')$ (conclusion 8) established.
        
        \end{proof}

      \begin{lem} 
        \label{lem:fix-point-body}
        Applying transformations from line 14-16 is sound.
        \end{lem}
        
        Given a statement to path-merge $s$, a concrete state $\delta$ and a symbolic state $\omega$, such that:

        \begin{itemize}
            \item $\delta = (\eta, \Delta, H, \nu, \mu)$
            \item $\delta'=(\eta', \Delta', H', \nu', \mu')$
            \item $\delta''=(\eta'', \Delta'', H'', \nu'', \mu'')$
            \item $\omega=(\Theta_t,\eta_s,\Delta_s, H_s, \mu_s,(\pi,\pi_r,\pi_s), (I,O,T), u)$
            \item $\omega'=(\Theta_t',\eta_s',\Delta_s', H_s', \mu_s',(\pi',\pi_r',\pi_s'), (I',O',T'), u')$.
        \end{itemize}

        we have that 
        \begin{gather*}
        \Theta \vdash \concsstep{\delta}{s}{\delta'}{} \tag{premise 1}\\
        (\omega',s') = TR_8 \circ TR_7 \circ TR_6 (\omega,s)  \tag{premise 2}\\
        (\Delta) \models (\Delta_s,\pi) \tag{premise 3}\\
        \forall i \in \func{getInputs}(s).\forall x \in Dom(\Delta_s).\neg isUse(i,\Delta_s(x)) \tag{premise 4}\\
        \Theta \vdash \concsstep{\delta}{s'}{\delta''}{} \tag{premise 5}\\
        \eta_s \subseteq \eta \tag{premise 6}\\
        \end{gather*}
        
        then
        \begin{gather*}
        \eta'=\eta''\wedge H'=H'' \wedge \nu'=\nu'' \tag{conclusion 1}\\
         \forall i \in I'.\Delta'(i)=\Delta''(i) \tag{conclusion 2} \\
        \forall o\in O'. \Delta'(o) =\Delta''(o) \tag{conclusion 3}\\
        \mu''=\bot \tag{conclusion 4}  \\
        \mu' = evalReturn(\Delta'',\pi_r',\mu_s') \tag{conclusion 5}\\
        \forall i \in I'.\neg isUse(i,s') \tag{conclusion 6}\\
        (\Delta'') \models (\Delta_s',\pi') \tag{conclusion 7}\\
        \eta_s' \subseteq \eta'' \tag{conclusion 8}
        \end{gather*}

    \subsection{Loop Invariant}
    \label{sec:fix-invariant}

    We will define the loop invariant to be a property of the state of the rewriting ($\omega^i$, $s^i$), that says, that if we are in the middle of the writing, such that ($\omega^i$, $s^i$)  is the current rewriting state, then that rewriting is preserving the semantics of the original code, i.e., concretely executing the currently rewritten statement is sound with respect to executing the original statement concretely. We use the superscript $i$ notation to suggest that the loop invariant will hold after $i$ iterations of the fix point loop, in a way that we will define later.

    More formally, we define the loop invariant be a predicate on the state of the rewriting ($\omega^i$, $s^i$), parameterized by the initial concrete state $\delta$, initial symbolic state $\omega$ and the original statement $s$.
    
    We further define:
     \begin{itemize}

         \item we write the components of the original concrete state $\delta$ as $(\eta, \Delta, H, \nu, \mu)$.
         
         \item state after evaluating the original statement: $\delta^\bot = (\eta^\bot, \Delta^\bot, H^\bot, \nu^\bot, \mu^\bot)$, defined by $\Theta \vdash \concsstep{\delta}{s}{\delta^\bot}{} $.

        \item state after evaluating the rewritten statement after executing the fix-loop $i$ times: $\delta^i=(\eta^i, \Delta^i, H^i, \nu^i, \mu^i)$, defined by $\Theta \vdash \concsstep{\delta}{s^i}{\delta^i}{} $.

         \item initial state before any rewriting $\omega=(\Theta_t,\eta_s,\Delta_s, H_s, \mu_s,(\pi,\pi_r,\pi_s), (I,O,T), u)$.  
        \item state of the rewriting statement after executing the fix-loop $i^{th}$ times: $\omega^i=(\Theta_t^i,\eta_s^i,\Delta_s^i, H_s^i, \mu_s^i,(\pi^i,\pi_r^i,\pi_s^i), (I^i,O^i,T^i), u^i)$.
     \end{itemize}

        The loop invariant is the conjunction of the following properties:
        
        \begin{gather*}
        \eta^i=\eta^\bot\wedge H^i=H^\bot \wedge \nu^i=\nu^\bot \tag{invariant 1}\\
        \forall o\in O^i. \Delta^i(o) =\Delta^\bot(o) \tag{invariant 2}\\
        \mu^i=\bot \tag{invariant 3}\\
        \mu^\bot= evalReturn(\Delta^i,\pi_r^i,\mu_s^i) \tag{invariant 4}\\
        \Delta\models (\Delta_s^i,\pi^i) \tag{invariant 5}\\
        (\Delta,H) \models_H (\Delta_s^i,H_s^i,\pi^i) \tag{invariant 6}\\
        \eta_s^i \subseteq \eta^i \tag{invariant 7}\\
        \end{gather*}

We prove that the invariant holds at the entry of the fixpoint loop, i.e., 
before the first iteration. We use superscript $0$ to denote the state upon 
first entering the loop, and superscript $i$ to denote the state after $i$ 
iterations. The invariant is first established at line~13, which marks the 
beginning of the loop body, and we show that it holds at this point before 
any iteration has been executed, i.e., at superscript $0$.

\begin{proof}
    
        We establish the correctness of the loop invariant by induction. 

        \begin{description}
            \item[Base Case:] 
            \begin{itemize} Initially, before any rewriting happens, as per Algorithm~\ref{fig:init-algo}, we have the following assumptions between the concrete state and the symbolic state:
        \begin{itemize}
            \item Assumption Property~\ref{prop:eta-unchanged}: $\eta_s \subseteq \eta$
            \item Assumption Property~\ref{prop:variable-consistent}: $\Delta \models (\Delta_s, \pi)$
            \item Assumption Property~\ref{prop:initial-no-return}: $\mu=\mu_s=\bot$
            \item Assumption Property~\ref{prop:initial-deltas-no-I}: $\forall i \in \func{getInputs}(s).\forall x \in Dom(\Delta_s).\neg isUse(i,\Delta_s(x))$
        \end{itemize}
       
        From Properties~\ref{prop:variable-consistent}, \ref{prop:initial-deltas-no-I}, 
        and \ref{prop:eta-unchanged}, we apply Lemma~\ref{lem:pre-fix-sound} to 
        establish the soundness of the pre-fixpoint transformations $TR_1$--$TR_5$, 
        where $(\omega^0, s^0) = TR_5 \circ TR_4 \circ TR_3 \circ TR_2 \circ TR_1(\omega, s)$, 
        such that
        $$\Theta \vdash \concsstep{\delta}{s}{\delta^\bot}{} $$
        $$\Theta \vdash \concsstep{\delta}{s^0}{\delta^0}{} $$
        $$ \rangerstepnolookup{(\omega, s)}{(\omega^0,s^0)}{TR}$$
        
        Thus, here we get:
        \begin{gather*}
        \eta^0=\eta^\bot\wedge H^0=H^\bot \wedge \nu^0=\nu^\bot \\
        \forall o\in O^0. \Delta^\bot(o) =\Delta^0(o) \\
        \mu^0=\bot   \\
        \mu^\bot = evalReturn(\Delta^0,\pi_r^0,\mu_s^0) \\
       (\Delta^\bot,H^\bot) \models_H (\Delta_s^0,H_s^0,\pi^0) \\
        (\Delta) \models (\Delta_s^0,\pi^0) \\
        \eta_s^0 \subseteq \eta^0 
        \end{gather*}
        which directly satisfies the loop invariant just before entering the loop.
            \end{itemize}
        
        \item [Inductive Case:]
        From the induction hypothesis we know that the loop invariant hold at iteration $i-1$. 
        Now, assuming rewriting has finished $i^{th}$ iteration within the loop. Thus, we will have
        $$ \rangerstepnolookup{(\omega^{i-1}, s^{i-1})}{(\omega^i,s^i)}{TR_{6-8}}$$
        $$\Theta \vdash \concsstep{\delta}{s^{i-1}}{\delta^{i-1}}{} $$
        $$\Theta \vdash \concsstep{\delta}{s^i}{\delta^i}{} $$

        more precisely, we know that

        \begin{gather*}
        \eta^{i-1}=\eta^\bot\wedge H^{i-1}=H^\bot \wedge \nu^{i-1}=\nu^\bot \tag{induction 1}\\
        \forall o\in O^{i-1}. \Delta^{i-1}(o) =\Delta^\bot(o) \tag{induction 2}\\
        \mu^{i-1}=\bot \tag{induction 3}\\
        \mu^\bot= evalReturn(\Delta^{i-1},\pi_r^{i-1},\mu_s^{i-1}) \tag{induction 4}\\
        \Delta\models (\Delta_s^{i-1},\pi^{i-1}) \tag{induction 5}\\
        (\Delta,H) \models_H (\Delta_s^{i-1},H_s^{i-1},\pi^{i-1}) \tag{induction 6}\\
        \eta_s^{i-1} \subseteq \eta^{i-1} \tag{induction 7}
        \end{gather*}

        
        Now we divide the prove induction into three main steps, each of them talking about on of the transformations from  lines 14-16 in Algorithm~\ref{fig:algorithm}, and proving that the invariant holds after each transformation.

        The first transformation within the loop is the method inlining. To use its soundness lemma (Lemma~\ref{lem:inline-wrapper-sound}) we know to satisfy its premise, more precisely, we need to show that,
        given an intermediate concrete and symbolic state after executing the method-inlining as

        $$\Theta \vdash \concsstep{(\eta, \Delta, H, \nu,\mu)}{s^{i-1}}{(\eta^{i-1}, \Delta^{i-1}, H^{i-1} \nu^{i-1},\mu^{i-1})}{} $$
        $${\Theta_t^{i-1},\eta_s^{i-1},\Delta_s^{i-1} \vdash \rangerstep{(T^{i-1},u^{i-1})}{s^{i-1}}{(T_{in},u_{in})}{s_{in}}{inline-wrap}}$$
        $$\Theta \vdash \concsstep{(\eta, \Delta, H, \nu,\mu)}{s_{in}}{(\eta_{in}, \Delta_{in}, H_{in}, \nu_{in}, \mu_{in})}{} $$

        the premise of the inlining wrap lemma (Lemma~\ref{lem:inline-wrapper-sound}) below 
        
        \begin{align*}
        &(\Delta) \models (\Delta_s^{i-1},\pi^{i-1}) &\tag{from induction 5}\\
        &\eta_s^{i-1}\subseteq \eta^{i-1} &\tag{from induction 7}
        \end{align*}

        We can now apply Lemma~\ref{lem:inline-wrapper-sound} for soundness of the inlining wrapper to get the following. 
        
        \begin{align*}
        &\Delta^{i-1}\subseteq \Delta_{in} &\tag{inline 1}\\
        &\eta_{in}=\eta^{i-1},H_{in}=H^{i-1},\nu_{in}=\nu^{i-1},\mu_{in}=\mu^{i-1} &\tag{inline 2}\\
        &\eta_s^{i-1} \subseteq \eta_{in} &\tag{inline 3}\\
        &\func{evalReturn}(\Delta_{i-1},\pi_r^{i-1},\mu_s^{i-1})=\func{evalReturn}(\Delta_{in},\pi_r^{i-1},\mu_s^{i-1}) &\tag{inline 4}\\
        &(\Delta_{in}) \models (\Delta_s^{i-1},\pi^{i-1}) &\tag{inline 5}
        \end{align*}

        Observe that in the inlining transformation, the only two environment variables that are changing is $T$ and $u$ (from $T^{i-1}$ and $u^{i-1}$ to $T_{in}$ and $u_{in}$), the other environment variables stays the same as those in iteration $i-1$. That is 
        \begin{align*}
         \Theta_{t_{in}}=\Theta_t^{i-1},\eta_{s_{in}}=\eta_s^{i-1},\Delta_{s_{in}}=\Delta_s^{i-1}, H_{s_{in}}=H_s^{i-1}, \mu_{s_{in}}=\mu_s^{i-1}\\
         \pi_{in}=\pi^{i-1},\pi_{r_{in}}=\pi_r^{i-1},\pi_{s_{in}}=\pi_s^{i-1},I_{in}=I^{i-1},O_{in}=O^{i-1}   \tag{inline 6}
        \end{align*}

        The inlining transformation preserves the loop invariant. 
        
        - Inline invariant~1 ($\eta_{in} = \eta^\bot \wedge H_{in} = H^\bot \wedge 
        \nu_{in} = \nu^\bot$) follows by combining (inline~2), which gives $\eta_{in} = \eta^{i-1}$, $H_{in} = H^{i-1}$, $\nu_{in} = \nu^{i-1}$, with (induction~1), which gives $\eta^{i-1} = \eta^\bot$, $H^{i-1} = H^\bot$, $\nu^{i-1} = \nu^\bot$, and applying transitivity.
        
        -Inline Invariant~2 ($\forall o \in O_{in}.\ \Delta_{in}(o) = \Delta^\bot(o)$) follows by combining (inline~1), which gives $\Delta^{i-1} \subseteq \Delta_{in}$, with the observation that $O_{in} = O^{i-1}$, ensuring that all output mappings are preserved in $\Delta_{in}$. Applying (induction~2) transitively then establishes that every output variable in $O_{in}$ carries the same mapping as in $\Delta^\bot$.

        - Inline invariant~3 ($\mu_{in}=\bot$) can be easily established from inline 6, which states that $\mu_{s_{in}}=\mu_s^{i-1}$, then we can use transitivity from induction~3 ($\mu^{i-1}=\bot$), to conclude that inline invariant~3 ($\mu_{in}=\bot$) holds after the method inlining transformation.
        
        - Inline Invariant~4 ($\mu^\bot= evalReturn(\Delta_{in},\pi_{r_{in}},\mu_{s_{in}})$) follows directly by transitivity from (inline~4), which gives $\func{evalReturn}(\Delta^{i-1}, \pi_r^{i-1}, \mu_s^{i-1}) = \func{evalReturn}(\Delta_{in}, \pi_r^{i-1}, \mu_s^{i-1})$, and (induction~4), which gives $\mu^\bot = \func{evalReturn}(\Delta^{i-1}, \pi_r^{i-1}, \mu_s^{i-1})$.

        - Inline Invariant~5 ($\Delta \models (\Delta_{s_{in}}, \pi_{in})$) holds because none  of the symbolic states need in this function has changed. More precisely, (inline~6) gives  $\Delta_{s_{in}} = \Delta_s^{i-1}$ and $H_{s_{in}} = H_s^{i-1}$, and  applying transitivity over (induction~5), which gives  $\Delta \models (\Delta_s^{i-1}, \pi^{i-1})$, establishes the invariant.

        - Inline Invariant~6 ($(\Delta,H) \models_H (\Delta_{s_{in}},H_{s_{in}},\pi_{in}) $) holds again since none of the symbolic states needed in this function has changed. More precisely (inline-6) gives $\Delta_{s_{in}}=\Delta_s^{i-1}, H_{s_{in}}=H_s^{i-1}, \pi_{in}=\pi^{i-1}$, then applying transitivity over (induction~6) which gives $(\Delta,H) \models_H (\Delta_s^{i-1},H_s^{i-1},\pi^{i-1}) $, we establish the invariant.

        - Finally, Inline Invariant~7 ($\eta_{s_{in}} \subseteq \eta_{in} $) directly holds from inline~3 which states that $\eta_s^{i-1} \subseteq \eta_{in}$, as well as inline~6 ($\eta_{s_{in}}=\eta_s^{i-1}$).

        Thus, we have established that the method-inlining preserves the loop invariant. More specifically,
            
        \begin{gather*}
        \eta_{in}=\eta^\bot\wedge H_{in}=H^\bot \wedge \nu_{in}=\nu^\bot \tag{inline-inv 1: from inline 2, induction 1}\\
        \forall o\in O_{in}. \Delta_{in}(o) =\Delta^\bot(o) \tag{inline-inv 2: from inline 1, induction 2}\\
        \mu_{in}=\bot \tag{inline-inv 3: from inline 2, induction 3}\\
        \mu^\bot= evalReturn(\Delta_{in},\pi_{r_{in}},\mu_{s_{in}}) \tag{inline-inv 4: from inline 4, induction 4}\\
        \Delta\models (\Delta_{s_{in}},\pi_{in}) \tag{inline-inv 5: from inline 6, induction 5}\\
        (\Delta,H) \models_H (\Delta_{s_{in}},H_{s_{in}},\pi_{in}) \tag{inline-inv 6: from inline 6, induction 6}\\
        \eta_{s_{in}} \subseteq \eta_{in} \tag{inline-inv 7: from inline 3 and  6}\\
        \end{gather*}
        
        The next transformation is the one at line 15, the field SSA transformation. Here, we apply Lemma~\ref{lb:field-wrap-theorem}, this is because from (inline 6) we have satisfied the first premise for Lemma~\ref{lb:field-wrap-theorem} and we know that the field transformation rewriting will produce:
        
         $$\Delta_s^{i-1}, H_s^{i-1} \vdash \rangerstep{(I^{i-1},O^{i-1},T_{in})}{s_{in}}{(I^{i-1},O^{i-1},T_{field})}{s_{field}}{field-wrap}$$

         Finally, executing concretely the field-transformed statement, we get:
         
         $$\Theta \vdash \concsstep{(\eta, \Delta, H, \nu, \mu)}{s_{field}}{(\eta_{field}, \Delta_{field}, H_{field}, \nu_{field}, \mu_{field})}{}$$

        Given the above, we applying Lemma~\ref{lb:field-wrap-theorem} to obtain

         \begin{gather*}
         \eta_{field}=\eta_{in} \quad \wedge \quad \nu_{field}=\nu_{in} \quad \wedge \quad \mu_{field}=\mu_{in} \tag{field 1}\\
         \Delta_{in} \subseteq \Delta_{field} \tag{field 2}\\
         H_{field}=H_{in} \tag{field 3}
        \end{gather*}

        From here we can see that the field transformation is preserving the loop invariant, for similar reasoning that we have established before when presenting the proof that the method inlining preserves the loop invariant. 
        
        More precisely. we have that 
        invariant 1 satisfied by transitivity from field 1 and field 3. Invariant 2 is satisfied from field 2, i.e., output variables are still the same when comparing them after executing the non-rewritten statement, versus the field rewritten statement. Invariant 3, 4, 5, 6 and 7 are all satisfied since none of these symbolic states have changed by field SSA transformation, i.e., 
        \begin{gather*}
        \Theta_{t_{field}}=\Theta_{t_{in}},\eta_{s_{field}}=\eta_{s_{in}},\Delta_{s_{field}}=\Delta_{s_{in}}, H_{s_{field}}=H_{s_{in}}, \mu_{s_{field}}=\mu_{s_{in}},\pi_{field}=\pi_{in}\\
        \pi_{r_{field}}=\pi_{r_{in}},\pi_{s_{field}}=\pi_{s_{in}}, I_{field}=I_{in},O_{filed}=O_{in}, u_{field}=u_{in}    \tag{field 4}
        \end{gather*}
        then we can conclude by transitivity from the previous transformation, that the field transformation also preserves invariants 3-7. 

        More precisely, we have

        \begin{gather*}
        \eta_{field}=\eta^\bot\wedge H_{field}=H^\bot \wedge \nu_{field}=\nu^\bot \tag{field-inv 1: from field 1, 3, inline-inv 1}\\
        \forall o\in O_{field}. \Delta_{field}(o) =\Delta^\bot(o) \tag{field-inv 2: from field 2, inline-inv 2}\\
        \mu_{field}=\bot \tag{field-inv 3: from field 1, inline-inv 3}\\
        \mu^\bot= evalReturn(\Delta_{field},\pi_{r_{field}},\mu_{s_{field}}) \tag{field-inv 4: from field 1, 2, 4, inline-inv 4}\\
        \Delta\models (\Delta_{s_{field}},\pi_{field}) \tag{field-inv 5: from field 4, inline-inv 5}\\
        (\Delta,H) \models_H (\Delta_{s_{field}},H_{s_{field}},\pi_{field}) \tag{field-inv 6: from field  4, inline-inv 6}\\
        \eta_{s_{field}} \subseteq \eta_{field} \tag{field-inv 7: from field 1, inline-inv 7}\\
        \end{gather*}
        
        Next, we want to apply the reference constant propagation transformation. We have already proven that the field transformation maintains invariant 5 (field-inv 5), that is 

        $$\Delta \models (\Delta_{s_{field}},\pi_{field})$$

        and given a constant propagation rewriting and the execution of the its rewritten statement, such that we have


        $$\Delta_{s_{field}},T_{field} \vdash \rangerstep{(\mu_{s_{field}},\pi_{r_{field}})}{s_{field}}{(\mu_{s_{const}},\pi_{r_{const}})}{s}{const-wrap}$$
        $$\Theta \vdash \concsstep{(\eta, \Delta, H, \nu,\mu)}{s_{const}}{(\eta_{const}, \Delta_{const}, H_{const}, \nu_{const},\mu_{const})}{}$$
        $$\mu^*=\func{evalReturn}(\Delta_{const}, \pi_{r_{field}},\mu_{s_{field}})$$
        
        Now, applying the constant propagation lemma~\ref{lem:const-wrap-sound}, we get

        \begin{gather*}
        \Delta_{field}\subseteq \Delta_{const} \tag{const 1}\\
        \eta_{const}=\eta_{field}, H_{const}=H_{field}, \nu_{const}=\nu_{field},\mu_{const}=\mu_{field} \tag{const 2}\\
        \mu^*=\func{evalReturn}(\Delta_{const},\pi_{r_{const}},\mu_{s_{const}}) \tag{const 3}\\
        \end{gather*}

        Observe here that, except from $\pi_{r_{const}},\mu_{s_{const}}$, none of the symbolic environment variables have changed from the constant propagation rewriting, more precisely, we have that
        \begin{align*}
        &\Theta_{t_{const}}=\Theta_{t_{field}},\eta_{s_{const}}=\eta_{s_{field}},\Delta_{s_{const}}=\Delta_{s_{field}},
        H_{s_{const}}=H_{s_{field}}&\\
        &\mu_{s_{const}}=\mu_{s_{field}}, \pi_{const}=\pi_{field},\pi_{r_{const}}=\pi_{r_{field}},\pi_{s_{const}}=\pi_{s_{field}}&\\        &I_{const}=I_{field},O_{filed}=O_{field},T_{const}=T_{field}, u_{const}=u_{field}&\tag{const 4}
        \end{align*}
        
        Now, we can prove that the loop invariant hold from:
        \begin{gather*}
        \eta_{const}=\eta^\bot\wedge H_{const}=H^\bot \wedge \nu_{const}=\nu^\bot \tag{const-inv 1: from const 1, field-inv 1}\\
        \forall o\in O_{const}. \Delta_{const}(o) =\Delta^\bot(o) \tag{const-inv 2: from const 4, field-inv 2}\\
        \mu_{const}=\bot \tag{const-inv 3: from const 2, field-inv 3}\\
        \mu^\bot= evalReturn(\Delta_{const},\pi_{r_{const}},\mu_{s_{const}}) \tag{const-inv 4: from const 3, field-inv 4}\\
        \Delta\models (\Delta_{s_{const}},\pi_{const}) \tag{const-inv 5: from const 1, 4, field-inv 5}\\
        (\Delta,H) \models_H (\Delta_{s_{const}},H_{s_{const}},\pi_{const}) \tag{const-inv 6: from const 1, 4, field-inv 6}\\
        \eta_{s_{const}} \subseteq \eta_{const} \tag{const-inv 7: from const 2, field-inv 7}\\
        \end{gather*}

        Observe here that this is the last transformation within the loop, thus we know that that 
        $\eta^i=\eta_{const}, \Delta^i=\Delta_{const}, H^i=H_{const}, \nu^i=\nu_{const}, \mu^i=\mu_{const}$.
        Similarly, we have $\Theta_t^i=\Theta_{const},\eta_s^i=\eta_{s_{const}},\Delta_s^i=\Delta_{s_{const}}, H_s^i=H_{s_{const}}, \mu_s^i=\mu_{s_{const}},\pi^i=\pi_{const},\pi_r^i=\pi_{r_{const}},\pi_s^i=\pi_{s_{const}}, I^i=I_{const},O^i=O_{const},T^i=T_{const}, u^i=u_{const}$.
        
        This means that we have just proved the loop invariant after iteration $i$.
        
    \end{description}            
\end{proof}

 \soha{Thus, we can conclude that conditioned on exiting the fix-point loop, the Java Ranger's state is sound with respect to the concrete semantics. }

 \soha{\subsection{Termination and Other Limitations}}

 \begin{description}
     \item[- Termination of the Formalized Fix-Point Loop:] \soha{In our formalization, termination is not guaranteed in all cases. Specifically, recursion presents a primary source of non-termination, as recursive methods can undergo unbounded inlining (see inlining rules in Fig.~\ref{fig:inlinerules}).}
 
    \soha{Another source of non-termination arises from arrays with symbolic lengths; because symbolic array bounds can be arbitrarily large integers, summarizing statements over such arrays may not terminate.}
    
     
\item [- Termination and Implementation Limitations in Practice:]
 \soha{In practice, Java Ranger guarantees termination by employing practical bounded heuristics. To prevent unbounded method inlining, Java Ranger restricts recursive calls to a configurable maximum depth. Similarly, to manage symbolic-sized arrays, Java Ranger concretizes array lengths to a predefined set of specific sizes and independently explores program paths under each size.}
 
 \soha{Furthermore, observe that object references do not induce non-termination, because both our formalization and the practical implementation restrict summarization exclusively to code regions where all references are concretely resolved.}
 
\end{description}

\section{\soha{Case Study: Uncovering Field SSA Bug Through Formal Verification}}

\begin{tcolorbox}[
    listing only,
    listing options={style=mystyle},
    colback=white,
    colframe=black!30,
    top=2pt, bottom=2pt
]
\begin{lstlisting}[caption={Reproducing Field SSA-Transformation Bug}, label={lst:fieldSSABug}, captionpos=b]
class Ref {
    public Ref self;
    int x = 0;
    public Ref(int x) {
      this.x = x;
      this.self = this;
    }
  }
  
 private void start() {
    int i = symVar(a) //assigned to symbolic var "a"
    Ref ref = new Ref(i);
    int xCopy = ref.x;
    if (i < 0) {
      ref.self.x += 2;
      xCopy = ref.x;
    }
    assert i < 0 ? xCopy == i + 2 : true; //assertion fails
  }
\end{lstlisting}
\end{tcolorbox}

\soha{Listing~\ref{lst:fieldSSABug} presents a minimal program that reproduces the Field-SSA transformation bug uncovered during the construction of our formal semantics. The underlying functionality is straightforward: increment a field ({\tt x}) by $2$ and then assert its updated value. The program is intentionally structured to expose the flaw in field SSA transformation.}

\soha{More precisely, line~12 allocates an object {\tt ref} on the heap (let us assume it holds a concrete reference $471$) whose field {\tt x} is initialized to a symbolic integer $a$ (line~11), and whose reference field {\tt self} points directly back to {\tt ref} (line~6). Within the {\tt i < 0} branch, the program increments {\tt x} via the chained reference {\tt ref.self.x += 2} (line~15) and immediately copies the updated field to a local variable {\tt xCopy} (line~16). Under correct semantics, $x$ should be incremented by $2$, satisfying the assertion on line~18.}

\soha{This example fails under JR's previous Field-SSA rules because they rewritten field accesses as soon as an operand became locally resolvable, without enforcing strict sequential program order across iterations. At the start of the branch, the receiver {\tt ref.self} in the statement {\tt ref.self.x += 2} requires two levels of reference resolution and is not yet resolved to its concrete heap address ($471$). Rather than stalling subsequent operations, the old rules skipped ahead during the first iteration of the fixpoint loop to the next statement, {\tt xCopy = ref.x}. Because {\tt ref} required only one level of referencing, it was already resolved; JR thus immediately evaluated the read in this first iteration, assigning $xCopy := a$ using the stale initial symbolic value $a$. The write to $x$ ($a + 2$) was deferred and processed only in the second iteration of the fixpoint loop, after {\tt ref.self} was finally resolved to $471$. This out-of-order processing across iterations reordered the heap read before the heap write, completely disconnecting the local copy from the intended update and causing the assertion $xCopy == a + 2$ to fail.}

\soha{Although this example appears simple, this bug remained hidden for years across extensive empirical test suites and benchmarks. The bug only surfaces when two different levels of referencing point to the same underlying heap object, and the chained reference (which requires multiple iterations to resolve) appears before a direct reference to that same object in the program order. In typical benchmarks, objects are either accessed at the same depth, resolved immediately, or cleared before subsequent reads. This subtle aliasing order was exposed only while constructing the formal induction proof for Field-SSA soundness, where maintaining the P-Heap consistency relation required consistency between the symbolic environment and the concrete heap across every fixpoint iteration.}

\soha{To resolve this unsoundness, we updated the formal Field-SSA semantics to enforce sequential evaluation order across un-evaluated accesses. The updated rules suspend field elimination whenever complex field accesses or invocation boundaries are encountered. Technically, the transformation flushes the current field map to the concrete state by emitting $\text{addPuts}(P)$ statements and propagates a control flag $\top$ to halt further SSA rewriting for subsequent statements in the region. This mechanism forces all pending updates back to the heap, ensuring that subsequent fixpoint iterations restore the exact concrete state before evaluating dependent reads.}

\soha{While the formal Field-SSA semantics presented in this paper have been updated to prevent this out-of-order execution, the corresponding engineering fix for the implementation of JR is currently tracked under issue~$\#42$~\cite{bug} and is actively underway.}

\subsection{Conclusion}

In this paper, we presented the first formal treatment of path-merging symbolic 
execution, targeting Java Ranger as a representative tool. We formalized each 
of its code transformations and proved their soundness with respect to a 
simplified version of the Java concrete semantics, establishing that the 
path-merging process preserves program semantics. We hope this work lays the 
groundwork for future efforts toward formally verifying the correctness of 
symbolic execution tools more broadly, enforcing rigorous soundness guarantees 
for tools that are increasingly trusted in safety-critical verification 
contexts.
\begin{appendices}
\section{Properties}

Here, we list properties that are shared among semantics.

\begin{prop}
\label{prop:ssa}
    For a statement $s$ to be processed by Java Ranger, we assume that the statement obeys the Single Static Assignment property (SSA) for all local variable mappings; i.e., $SSA(s)$, such that the $SSA$ predicate is recursively defined as
\begin{align*}
    &SSA(s_1;s_2) = ASSN(s_1) \cap ASSN(s_2) = \emptyset  \, \land SSA(s_1) \, \land SSA(s_2) \\
    &SSA(if \,\,e \,\, then \,\, s_1 \,\, else \,\, s_2 ; \phi_1(x_1,e^1_1,e^2_1);\cdots;\phi(x_n,e^1_n;e^2_n)) 
    = \\
    & \qquad\qquad\qquad\qquad(ASSN(s_1) \cup ASSN(s_2)) \cap \{x_1,\cdots,x_2\} = \emptyset  \, 
    \land SSA(s_1) \, \land SSA(s_2)   \\
    &SSA(s) =\top \qquad \text{otherwise}
\end{align*}
where the $ASSN$ function returns the set of all assigned to variables. Observe that this extends to implicit assignments from the {\tt invoke}, {\tt getfield} and {\tt new} statements that collects the returned value.
\techreport{
\begin{align*}
    &ASSN(x := e) = \{x\} \\
    &ASSN(s_1;s_2) = ASSN(s_1) \cup ASSN(s_2)\\
    &ASSN(skip) = \emptyset\\
    &ASSN(if \,\,e \,\, then \,\, s_1 \,\, else \,\, s_2 ; \phi_1(x_1,e^1_1,e^2_1);\cdots;\phi(x_n,e^1_n;e^2_n))= \\
    & \qquad\qquad\qquad\qquad\qquad\qquad\qquad ASSN(s_1) \cup ASSN(s_2) \cup \{x_1,\cdots,x_n\} \\
    &ASSN(\invoke{x}{z}{g}{e})) = \{x\}\\
    &ASSN(\return{e}) = \emptyset\\
    &ASSN(\putfield{z}{f}{e}) = \emptyset\\
    &ASSN(\newobj{z}{c}) = \{ z\}\\
    &ASSN(\getfield{x}{z}{f}) = \{x\}\\
\end{align*}
}
\end{prop}
\begin{prop}[Statement $s$ always has at most a single return within an if-statement]
\label{prop:final-at-most-final-if-return}, i.e., $\func{atMostSingleIfReturn}(s) $. 
Let us define the function $atMostSingleIfReturn$ such that 
\begin{align*}
\func{atMostSingleIfReturn}(s)=\neg \func{hasReturn}(s) \vee \func{singleIfReturn}(s)\end{align*}

where $\func{singleIfReturn}$ and $\func{hasReturn}$ are defined to capture the case where the statement $s$ has at most a single return or none at all (see technical report for complete detail~\cite{techreport})
\techreport{
\begin{align*}
&\func{singleIfReturn}(s_1;s_2)= \neg \func{hasReturn}(s_1) \wedge \func{singleIfReturn}(s_2)\\
&\func{singleIfReturn}(\myif{e_1}{\return{e_2}}{skip})= true \\
&\func{singleIfReturn}(s) = false \qquad  \text{otherwise} \\
&\func{hasReturn}(\return{e}) = true\\
&\func{hasReturn}(\myif{e}{s_1}{s_2}) = \func{hasReturn}(s_1) \vee \func{hasReturn}(s_2)\\
&\func{hasReturn}(s_1;s_2)= \func{hasReturn}(s_1) \vee \func{hasReturn}(s_2)\\
&\func{hasReturn}(s) = false \qquad \text{otherwise}\\   
\end{align*}
}
\end{prop}

 \begin{prop}
[Methods in $\Theta$ must have a return statement] 
    \label{prop:theta-stmt-has-return}
    All invoked statements in $\Theta$ must contain a {\tt return} statement. More formally, $hasReturns(s)$ must hold.
 \end{prop}

\begin{prop}[Class object references]
\label{prop:class-references}
For each class $c \in C$ there exists a meta-class $c_c \in C$ and a class object reference $r_c \in R$ such that $\eta(r_c) = c_c$ and $\Delta(z_c) = r_c$. We write the class name directly to refer to its class object reference (e.g., $\mathit{Integer}$ denotes $z_{\mathit{Integer}}$). Finally, superscripts on environment variables denote successive versions, i.e., $\Delta^1, \Delta^2, \Delta^3$ abbreviate $\Delta', \Delta'', \Delta'''$ rather than exponentiation.
\end{prop}

\begin{prop}[Formal argument usage]
\label{prop:y-must-be-used-in-s}
All formal arguments of every method body in $\Theta$ are used in its statement:
$$\forall (\overrightarrow{y}.s) \in \mathit{range}(\Theta).\; \forall y \in \overrightarrow{y}.\; \mathit{isUse}(y,s)$$
\end{prop}

\begin{prop}[No initial concrete return value]
\label{prop:initial-no-return}
    The concrete state cannot be in the middle of returning from a method. More formally, it must be that 
    $$\mu=\bot$$
\end{prop}

\begin{prop}[$\eta_s=\eta$]
\label{prop:eta-unchanged}
    The reference typing mapping is carried over unchanged from the concrete semantics, and is never symbolic.
\end{prop}
\begin{prop}
\label{prop:no-sym-ref}
    No symbolic references: all references $r \in R$ in $\Delta_s$ and $H_s$ are concrete. That is, for all $(Id,n) \in \text{dom}(\Delta_s)$, if $\Delta_s(Id,n) \in R$ then $\Delta_s(Id,n)$ is a concrete reference, and similarly for $H_s(r,f) \in R$.
\end{prop}

\begin{prop}[Consistent Local Variable Mapping with Symbolic State]
\label{prop:variable-consistent}
    Concrete local variable mapping $\Delta$ is consistent with the symbolic local variable mapping $\Delta_s$ with respect to the path condition $\pi$. 

    More formally, we have that 
  \begin{align*}
    \Delta &\models (\Delta_s, \pi) \triangleq \concestep{\Delta}{\pi}{true}{} \wedge  \forall x \in \mathrm{Dom}(\Delta_s).\concestep{\Delta}{\Delta_s(x)}{v}{} \wedge v=\Delta(x)
\end{align*}
\end{prop}

\begin{prop}[Concrete Heap Consistent with Symbolic State]
\label{prop:heap-consistent}
    Concrete heap $H$ is consistent with the symbolic local variable mapping $\Delta_s$, and symbolic heap $H_s$ with respect to the path condition $\pi$. 

    More formally, we have that 
  \begin{align*}
    (\Delta,H) &\models_H (H_s,\pi) \triangleq \forall (r,f) \in \mathrm{Dom}(H).\concestep{\Delta}{H_s(r,f)}{v'}{} \wedge v'=H(r,f)
\end{align*}
\end{prop}

\begin{prop}[Symbolic local variable mapping is consistent with respect to the input $I$]
\label{prop:initial-deltas-no-I}
    Input variables are not used in any of the mapped expressions in initial $\Delta_s$.
    More formally, if $I$ is defined as $I=\func{getInput}(s)$, for some statement $s$ that is going to be summarized, we have that 
    $$\forall i \in \func{getInputs}(s).\forall x \in Dom(\Delta_s).\neg isUse(i,\Delta_s(x))$$ 
\end{prop}

\begin{prop}[No $\phi$ statements appear in  $s$]
\label{prop:no-phi}
\end{prop}

Let us define the function $atMostSingleIfReturn$ such that 
\begin{align*}
\func{atMostSingleIfReturn}(s)=\neg \func{hasReturn}(s) \vee \func{singleIfReturn}(s)\end{align*}

where $\func{singleIfReturn}$ and $\func{hasReturn}$ are defined to capture the case where the statement $s$ has at most a single return or none at all.
\journalreport{ See technical report for complete detail~\cite{techreport}}
\techreport{
\begin{align*}
&\func{singleIfReturn}(s_1;s_2)= \neg \func{hasReturn}(s_1) \wedge \func{singleIfReturn}(s_2)\\
&\func{singleIfReturn}(\myif{e_1}{\return{e_2}}{skip})= true \\
&\func{singleIfReturn}(s) = false \qquad  \text{otherwise} \\
&\func{hasReturn}(\return{e}) = true\\
&\func{hasReturn}(\myif{e}{s_1}{s_2}) = \func{hasReturn}(s_1) \vee \func{hasReturn}(s_2)\\
&\func{hasReturn}(s_1;s_2)= \func{hasReturn}(s_1) \vee \func{hasReturn}(s_2)\\
&\func{hasReturn}(s) = false \qquad \text{otherwise}\\   
\end{align*}
}
\begin{prop}[Rewritten statement always has at most a single return within an if-statement]
\label{prop:early-at-most-final-if-return}
 
$$\func{atMostSingleIfReturn}(\hat{s}) $$
\end{prop}

\begin{prop}[If statement $s$ has a return statement then $s$ must have a specific form]
\label{prop:final-s-has-a-specific-form}
For some statement $s''$, we have that

 $$\func{hasReturn}(s) \implies  (s=s'';\myif{\pi_r}{\return{e}}{skip}) \wedge \pi_r \neq false$$

\end{prop}

\begin{prop}[$\pi_r$ is true on all return paths] 
\label{prop:final-pi-r-iff-mu-s}
More formally, 

If
$$\Theta \vdash ((\eta,\Delta,H,\perp),s) \Rightarrow_c (\eta',\Delta',H',\mu) $$
then
$$\mu \neq \bot \text{ iff } (\Delta',\pi_r)\rightsquigarrow_c true$$
\end{prop}

\begin{prop}[ statement $s'$ has no return statements] 
\label{prop:final-ret-no-return}
More formally, $\neg hasReturn(s')$ is valid.
\end{prop}

\begin{prop}[Fresh variables have $u+1$ suffices] 
\label{prop:fresh-vars-increments-u}
    More formally, $(x,\bot)$ is in $s$, then all its occurrences with $s'$ must have been renamed to $(x,u+1)$. 
\end{prop}

\begin{prop}[Old input not in $s'$]
\label{prop:renaming-no-old-input}
More formally, Old input does not appear, i.e., is not used in $s'$. More formally if $(x,j) \in I$, then $\neg isUse((x,j),s')$
\end{prop}

\begin{prop}[Variables in original statement are not indexed]
\label{prop:initial-variables-not-rename}
    Given a statement $s$ that has not yet been proceed by JR, i.e., JR's algorithm has not yet run yet, it must be that all variables in $s$ has not been renamed before. More formally for all $x\in isUse(s)$, then $x$ must be of the form $(a,\bot)$, where $a$ is the name of the variable and $\bot$ location indicates that the variable has not been renamed before.
\end{prop}




\begin{prop}[Invoked statements have been prepared and stored in $\Theta_t$]
\label{prop:inline-theta-is-transformed}
    This transformation assume that all invoked methods have been transformed and stored in $\Theta_t$, where for each $c\in C$, $g \in G$ such that $\overrightarrow{y}.s=\Theta(c,g)$, we have a transformed definition of $\overrightarrow{y}.s$ of the form $(\mu_s,\pi_r),\overrightarrow{y}.s'=transform(s)$(refer to section~\ref{sec:statictransform} for more detail). 
\end{prop}


\begin{prop}[Reference variables always refer to concrete references]
\label{prop:ref-are-conc}
    $\forall z. r=\Delta_s(z).r \in R$, then it must be that $r$ is concrete. 
\end{prop}



\section{Static Transform}
\label{sec:statictransform}
\begin{algorithm}[H]
\DontPrintSemicolon
\SetAlgoLined
{\bf Input}: to-summarize statement $s$\;
    $s^1$ :=  $\gamma$-creation($s$)\;
    $(\pi_r^2,s^2)$ := eliminate early-return($s^1$)\;
    $(\mu_s^3, s^3)$ := remove final-return($s^2$)\;
{\bf return} ($(\mu_s^3,\pi_r^2), s^3$)
\caption{transform definition: an environment variable  without a superscript indicate the initial state}
\label{fig:transform}
\end{algorithm}

\begin{thm}[Transform is Sound]
\label{thm:transform-sound}
For 
any statement $s$, such that $\mu_s=\bot$.
Also, we assume any concrete $\delta=(\eta, \Delta, H, \nu, \mu)$

If

$$\Theta \vdash \concsstep{(\eta, \Delta, H, \nu, \mu)}{s}{(\eta', \Delta', H', \nu, \mu')}{} $$
%
$$((\mu_s,\pi_r),s') =transform(s)$$
$$\Theta \vdash \concsstep{(\eta, \Delta, H, \nu, \mu)}{s'}{(\eta'', \Delta'', H'', \nu'', \mu'')}{} $$
then
$$\eta'=\eta'',\Delta'=\Delta'',H'=H'',\nu'=\nu''$$
$$
\mu' = evalReturn(\Delta'',\pi_r,\mu_s),  \mu''=\mu'
$$

\begin{proof}
To proof this theorem, we apply various transformation from the transform function shown in Fig.~\ref{fig:transform}. 
We know that according to Corollary~\ref{cor:gamma-wrap-sound}, we have that
$$     
\rangerstepnolookup{s}{s^1}{\gamma-wrap}$$
$$\Theta \vdash \concsstep{(\eta, \Delta, H, \nu, \mu)}{s}{(\eta', \Delta', H', \nu', \mu')}{} $$
such that  
$$\Theta \vdash \concsstep{(\eta, \Delta, H, \nu, \mu)}{s^1}{(\eta', \Delta', H', \nu', \mu')}{}$$

The next transformation that the {\tt transform} function has is the eliminate early-return transformation on the rewritten statement $s^1$. In particular, we apply theorem~\ref{thm:return-wrap-sound} to get 
$$ s^1 \to_{ret-wrap} (\pi_r^2,s^2)$$
$$\Theta \vdash ((\eta,\Delta,H,\perp),s^1) \Rightarrow_c (\eta_2,\Delta_2,H_2,\mu_2) 
$$
then
$$\Theta \vdash ((\eta,\Delta,H,\perp),s^2) \Rightarrow_c (\eta_2,\Delta_2,H_2,\mu_2)$$

But, we have just shown from applying the $\gamma$-transformation theorem~\ref{cor:gamma-wrap-sound} that $s^1$ must evaluate in $(\eta,\Delta,H,\perp)$ to $(\eta', \Delta', H', \nu', \mu')$, then we can conclude that $(\eta_2=\eta',\Delta'=\Delta',H_2=H',\mu_2=\mu')$. 

The last transformation applied by the {\tt transform} function is the remove final-return transformation. In particular we can apply Corollary~\ref{cor:final-ret-sound} which states that 

$$\rangerstepnolookup{s^2}{(\mu_s^3,s^3)}{final-wrap} $$
$$\Theta \vdash \concsstep{(\eta, \Delta, H, \nu,\mu)}{s^2}{(\eta_3, \Delta_3, H_3, \nu_3, \mu_3)}{} $$

then 
$$\Theta \vdash \concsstep{(\eta, \Delta, H, \nu,\mu)}{s^3}{(\eta_4, \Delta_4, H_4, \nu_4,\mu_4)}{}$$
$$\eta_4=\eta_3,\Delta_4=\Delta_3,H_4=H_3,\nu_4=\nu_3,\mu^4=\mu$$
$$
\mu_3 = evalReturn(\Delta_4,\pi_r^2,\mu_s^3)$$
Here we have that $\eta''=\eta_4,\Delta''=\Delta_4,H''=H_4,\nu''=\nu_4, \mu''=\mu_4$, and $\pi_r=\pi_r^2,\mu_s=\mu_s^3$.
Given that we just established, from before, that $\Theta \vdash ((\eta,\Delta,H,\perp),s^2) \Rightarrow_c (\eta',\Delta',H',\mu')$. 
Therefore, from Corollary~\ref{cor:final-ret-sound}, we have $\eta''=\eta_4=\eta_3=\eta',\Delta''=\Delta_4=\Delta_3=\Delta',H''=H_4=H_3=H',\nu''=\nu_4=\nu_3=\nu'$, and that $  evalReturn(\Delta_4,\pi_r^2,\mu_s^3)=\mu_3=\mu'$, which is what we wanted to prove. 
\end{proof}
\end{thm}




\end{appendices}

\clearpage
\bibliography{sn-bibliography}

@inproceedings{dart,
  author = {Kasper Luckow and Marko Dimja\v{s}evi\'c and Dimitra Giannakopoulou and Falk Howar
    and Malte Isberner and Temesghen Kahsai and Zvonimir Rakamari\'c and Vishwanath Raman},
  title = {{JDart}: A Dynamic Symbolic Analysis Framework},
  booktitle = {Proceedings of the 22nd International Conference on Tools and Algorithms
    for the Construction and Analysis of Systems (TACAS)},
  series = {Lecture Notes in Computer Science},
  volume = {9636},
  publisher = {Springer},
  address = {Berlin, Heidelberg},
  editor = {Marsha Chechik and Jean-Fran{\c{c}}ois Raskin},
  year = {2016},
  pages = {442--459}
}

@Article{spf,
author="P{\u{a}}s{\u{a}}reanu, Corina S.
and Visser, Willem
and Bushnell, David
and Geldenhuys, Jaco
and Mehlitz, Peter
and Rungta, Neha",
title="Symbolic PathFinder: integrating symbolic execution with model checking for Java bytecode analysis",
journal="Automated Software Engineering",
year="2013",
month="Sep",
day="01",
volume="20",
number="3",
pages="391--425",
issn="1573-7535",
doi="10.1007/s10515-013-0122-2",
url="https://doi.org/10.1007/s10515-013-0122-2"
}

@inproceedings{adaptorsynth,
   author = {{Sharma}, Vaibhav and {Hietala}, Kesha and {McCamant},
		 Stephen},
    title = "{Finding Substitutable Binary Code By Synthesizing Adaptors}",
     year = 2018,
    month = apr,
booktitle = "11th IEEE Conference on Software Testing, Validation and Verification (ICST)"
}

@inproceedings{veritesting,
 author = {Avgerinos, Thanassis and Rebert, Alexandre and Cha, Sang Kil and Brumley, David},
 title = {Enhancing Symbolic Execution with Veritesting},
 booktitle = {Proceedings of the 36th International Conference on Software Engineering},
 series = {ICSE 2014},
 year = {2014},
 isbn = {978-1-4503-2756-5},
 location = {Hyderabad, India},
 pages = {1083--1094},
 numpages = {12},
 NOTurl = {http://doi.acm.org/10.1145/2568225.2568293},
 doi = {10.1145/2568225.2568293},
 acmid = {2568293},
 publisher = {ACM},
 address = {New York, NY, USA}
}

@InProceedings{cute,
author="Sen, Koushik
and Agha, Gul",
editor="Ball, Thomas
and Jones, Robert B.",
title="CUTE and jCUTE: Concolic Unit Testing and Explicit Path Model-Checking Tools",
booktitle="Computer Aided Verification",
year="2006",
publisher="Springer Berlin Heidelberg",
address="Berlin, Heidelberg",
pages="419--423",
isbn="978-3-540-37411-4"
}

@proceedings{feedbackinvariantdiscovery,
title = {ISSTA 2014: Proceedings of the 2014 International Symposium on Software Testing and Analysis},
year = {2014},
isbn = {9781450326452},
publisher = {Association for Computing Machinery},
address = {New York, NY, USA},
location = {San Jose, CA, USA}
}

@INPROCEEDINGS{syminfer,
author={T. {Nguyen} and M. B. {Dwyer} and W. {Visser}},
booktitle={2017 32nd IEEE/ACM International Conference on Automated Software Engineering (ASE)},
title={Symlnfer: Inferring program invariants using symbolic states},
year={2017},
volume={},
number={},
pages={804-814},
doi={10.1109/ASE.2017.8115691},
ISSN={null},
month={Oct}}

@inproceedings{HansenSS2009,
  author    = {Trevor Hansen and
               Peter Schachte and
               Harald S{\o}ndergaard},
  title     = {State Joining and Splitting for the Symbolic Execution of Binaries},
  booktitle = {Runtime Verification, 9th International Workshop, {RV} 2009, Grenoble,
               France, June 26-28, 2009. Selected Papers},
  pages     = {76--92},
  year      = {2009},
  NOTurl       = {https://doi.org/10.1007/978-3-642-04694-0\_6},
  NOTdoi       = {10.1007/978-3-642-04694-0\_6},
}

@inproceedings{kuznetsov,
 author = {Kuznetsov, Volodymyr and Kinder, Johannes and Bucur, Stefan
 and Candea, George},
  title = {Efficient State Merging in Symbolic Execution},
 booktitle = {Proceedings of the 33rd ACM SIGPLAN Conference on
 Programming Language Design and Implementation},
  series = {PLDI '12},
 year = {2012},
  isbn = {978-1-4503-1205-9},
 location = {Beijing, China},
  pages = {193--204},
 numpages = {12},
 publisher = {ACM},
  address = {New York, NY, USA},
}

@inproceedings{multise,
  author = {Sen, Koushik and Necula, George and Gong, Liang and Choi,
		 Wontae},
  title = {MultiSE: Multi-path Symbolic Execution Using Value
		Summaries},
  booktitle = {Proceedings of the 2015 10th Joint Meeting on
	 Foundations of Software Engineering},
  series = {ESEC/FSE 2015},
  year = {2015},
  isbn = {978-1-4503-3675-8},
  location = {Bergamo, Italy},
  pages = {842--853},
  numpages = {12},
  NOTurl = {http://doi.acm.org/10.1145/2786805.2786830},
  doi = {10.1145/2786805.2786830},
  acmid = {2786830},
  publisher = {ACM},
  address = {New York, NY, USA},
}

@inproceedings{necula2000translation,
  author       = {George C. Necula},
  editor       = {Monica S. Lam},
  title        = {Translation validation for an optimizing compiler},
  booktitle    = {Proceedings of the 2000 {ACM} {SIGPLAN} Conference on Programming
                  Language Design and Implementation (PLDI), Vancouver, British Columbia,
                  Canada, June 18-21, 2000},
  pages        = {83--94},
  publisher    = {{ACM}},
  address = {New York, NY, USA},
  year         = {2000},
  NOTurl          = {https://doi.org/10.1145/349299.349314},
  doi          = {10.1145/349299.349314},
  bibsource    = {dblp computer science bibliography, https://dblp.org}
}

@inproceedings{pnueli1998translation,
  author       = {Amir Pnueli and
                  Michael Siegel and
                  Eli Singerman},
  editor       = {Bernhard Steffen},
  title        = {Translation Validation},
  booktitle    = {Tools and Algorithms for Construction and Analysis of Systems, 4th
                  International Conference, {TACAS} '98, Held as Part of the European
                  Joint Conferences on the Theory and Practice of Software, ETAPS'98,
                  Lisbon, Portugal, March 28 - April 4, 1998, Proceedings},
  series       = {Lecture Notes in Computer Science},
  pages        = {151--166},
  publisher    = {Springer},
address="Berlin, Heidelberg",
  year         = {1998},
  NOTurl          = {https://doi.org/10.1007/BFb0054170},
  doi          = {10.1007/BFB0054170},
  bibsource    = {dblp computer science bibliography, https://dblp.org}
}

@article{leroy2009formal,
  author       = {Xavier Leroy},
  title        = {Formal verification of a realistic compiler},
  journal      = {Commun. {ACM}},
  volume       = {52},
  number       = {7},
  pages        = {107--115},
  year         = {2009},
  NOTurl          = {https://doi.org/10.1145/1538788.1538814},
  doi          = {10.1145/1538788.1538814},
  bibsource    = {dblp computer science bibliography, https://dblp.org}
}

@inproceedings{kumar2014cakeml,
  author       = {Ramana Kumar and
                  Magnus O. Myreen and
                  Michael Norrish and
                  Scott Owens},
  editor       = {Suresh Jagannathan and
                  Peter Sewell},
  title        = {{CakeML}: a verified implementation of {ML}},
  booktitle    = {The 41st Annual {ACM} {SIGPLAN-SIGACT} Symposium on Principles of
                  Programming Languages, {POPL} '14, San Diego, CA, USA, January 20-21,
                  2014},
  pages        = {179--192},
  publisher    = {{ACM}},
   address = {New York, NY, USA},
 year         = {2014},
  NOTurl          = {https://doi.org/10.1145/2535838.2535841},
  doi          = {10.1145/2535838.2535841},
  bibsource    = {dblp computer science bibliography, https://dblp.org}
}

@article{bilardi2003algorithms,
  author       = {Gianfranco Bilardi and
                  Keshav Pingali},
  title        = {Algorithms for computing the static single assignment form},
  journal      = {J. {ACM}},
  volume       = {50},
  number       = {3},
  pages        = {375--425},
  year         = {2003},
  NOTurl          = {https://doi.org/10.1145/765568.765573},
  doi          = {10.1145/765568.765573},
  bibsource    = {dblp computer science bibliography, https://dblp.org}
}

@inproceedings{schwartz2010all,
  author       = {Edward J. Schwartz and
                  Thanassis Avgerinos and
                  David Brumley},
  title        = {All You Ever Wanted to Know about Dynamic Taint Analysis and Forward
                  Symbolic Execution (but Might Have Been Afraid to Ask)},
  booktitle    = {31st {IEEE} Symposium on Security and Privacy, {SP} 2010, 16-19 May
                  2010, Berleley/Oakland, California, {USA}},
  pages        = {317--331},
  publisher    = {{IEEE} Computer Society},
  address      = {Washington, DC, USA},
  year         = {2010},
  NOTurl          = {https://doi.org/10.1109/SP.2010.26},
  doi          = {10.1109/SP.2010.26},
  bibsource    = {dblp computer science bibliography, https://dblp.org}
}

@article{deboer2021symbolic,
  author       = {De Boer, Frank S. and
                  Marcello M. Bonsangue},
  title        = {Symbolic execution formally explained},
  journal      = {Formal Aspects Comput.},
  volume       = {33},
  number       = {4-5},
  pages        = {617--636},
  year         = {2021},
  NOTurl          = {https://doi.org/10.1007/s00165-020-00527-y},
  doi          = {10.1007/S00165-020-00527-Y},
  bibsource    = {dblp computer science bibliography, https://dblp.org}
}

@article{voogd2025compositional,
  author       = {Erik Voogd and
                  {\AA}smund Aqissiaq Arild Kl{\o}vstad and
                  Einar Broch Johnsen and
                  Andrzej Wasowski},
  title        = {Compositional symbolic execution semantics},
  journal      = {Theor. Comput. Sci.},
  volume       = {1044},
  pages        = {115263},
  year         = {2025},
  NOTurl          = {https://doi.org/10.1016/j.tcs.2025.115263},
  doi          = {10.1016/J.TCS.2025.115263},
  bibsource    = {dblp computer science bibliography, https://dblp.org}
}

@article{porncharoenwase2022formal,
  author       = {Sorawee Porncharoenwase and
                  Luke Nelson and
                  Xi Wang and
                  Emina Torlak},
  title        = {A formal foundation for symbolic evaluation with merging},
  journal      = {Proc. {ACM} Program. Lang.},
  volume       = {6},
  number       = {{POPL}},
  pages        = {1--28},
  year         = {2022},
  url          = {https://doi.org/10.1145/3498709},
  doi          = {10.1145/3498709},
  bibsource    = {dblp computer science bibliography, https://dblp.org}
}

@inproceedings{lopes2015provably,
  author       = {Nuno P. Lopes and
                  David Menendez and
                  Santosh Nagarakatte and
                  John Regehr},
  editor       = {David Grove and
                  Stephen M. Blackburn},
  title        = {Provably correct peephole optimizations with {Alive}},
  booktitle    = {Proceedings of the 36th {ACM} {SIGPLAN} Conference on Programming
                  Language Design and Implementation, Portland, OR, USA, June 15-17,
                  2015},
  pages        = {22--32},
  publisher    = {{ACM}},
  address = {New York, NY, USA},
  year         = {2015},
  NOTurl          = {https://doi.org/10.1145/2737924.2737965},
  doi          = {10.1145/2737924.2737965},
  bibsource    = {dblp computer science bibliography, https://dblp.org}
}

@inproceedings{lopes2021alive2,
  author       = {Nuno P. Lopes and
                  Juneyoung Lee and
                  Chung{-}Kil Hur and
                  Zhengyang Liu and
                  John Regehr},
  editor       = {Stephen N. Freund and
                  Eran Yahav},
  title        = {{Alive2}: bounded translation validation for {LLVM}},
  booktitle    = {{PLDI} '21: 42nd {ACM} {SIGPLAN} International Conference on Programming
                  Language Design and Implementation, Virtual Event, Canada, June 20-25,
                  2021},
  pages        = {65--79},
  publisher    = {{ACM}},
  address = {New York, NY, USA},
  year         = {2021},
  NOTurl          = {https://doi.org/10.1145/3453483.3454030},
  doi          = {10.1145/3453483.3454030},
  bibsource    = {dblp computer science bibliography, https://dblp.org}
}

@inproceedings{vanhattum2024lightweight,
  author       = {Alexa VanHattum and
                  Monica Pardeshi and
                  Chris Fallin and
                  Adrian Sampson and
                  Fraser Brown},
  editor       = {Rajiv Gupta and
                  Nael B. Abu{-}Ghazaleh and
                  Madan Musuvathi and
                  Dan Tsafrir},
  title        = {Lightweight, Modular Verification for {WebAssembly}-to-Native Instruction
                  Selection},
  booktitle    = {Proceedings of the 29th {ACM} International Conference on Architectural
                  Support for Programming Languages and Operating Systems, Volume 1,
                  {ASPLOS} 2024, La Jolla, CA, USA, 27 April 2024- 1 May 2024},
  pages        = {231--248},
  publisher    = {{ACM}},
  address = {New York, NY, USA},
  year         = {2024},
  NOTurl          = {https://doi.org/10.1145/3617232.3624862},
  doi          = {10.1145/3617232.3624862},
  bibsource    = {dblp computer science bibliography, https://dblp.org}
}

@inproceedings{brown2020towards,
  author       = {Fraser Brown and
                  John Renner and
                  Andres N{\"{o}}tzli and
                  Sorin Lerner and
                  Hovav Shacham and
                  Deian Stefan},
  editor       = {Alastair F. Donaldson and
                  Emina Torlak},
  title        = {Towards a verified range analysis for {JavaScript} {JITs}},
  booktitle    = {Proceedings of the 41st {ACM} {SIGPLAN} International Conference on
                  Programming Language Design and Implementation, {PLDI} 2020, London,
                  UK, June 15-20, 2020},
  pages        = {135--150},
  publisher    = {{ACM}},
  address = {New York, NY, USA},
  year         = {2020},
  NOTurl          = {https://doi.org/10.1145/3385412.3385968},
  doi          = {10.1145/3385412.3385968},
  bibsource    = {dblp computer science bibliography, https://dblp.org}
}

@inproceedings{zhao2013formal,
  author       = {Jianzhou Zhao and
                  Santosh Nagarakatte and
                  Milo M. K. Martin and
                  Steve Zdancewic},
  editor       = {Hans{-}Juergen Boehm and
                  Cormac Flanagan},
  title        = {Formal verification of {SSA}-based optimizations for {LLVM}},
  booktitle    = {{ACM} {SIGPLAN} Conference on Programming Language Design and Implementation,
                  {PLDI} '13, Seattle, WA, USA, June 16-19, 2013},
  pages        = {175--186},
  publisher    = {{ACM}},
  address = {New York, NY, USA},
  year         = {2013},
  NOTurl          = {https://doi.org/10.1145/2491956.2462164},
  doi          = {10.1145/2491956.2462164},
  bibsource    = {dblp computer science bibliography, https://dblp.org}
}

@article{symbolicsurvey,
  author    = {Baldoni, Roberto and Coppa, Emilio and D'Elia, Daniele Cono and Demetrescu, Camil and Finocchi, Irene},
  title     = {A Survey of Symbolic Execution Techniques},
  journal   = {ACM Comput. Surv.},
  volume    = {51},
  number = {3},
  articleno = {50},
  publisher = {ACM},
  address = {New York, NY, USA},
  year = {2018}
}

@inproceedings{javaranger,
author = {Sharma, Vaibhav and Hussein, Soha and Whalen, Michael W. and McCamant, Stephen and Visser, Willem},
title = {Java Ranger: statically summarizing regions for efficient symbolic execution of Java},
year = {2020},
isbn = {9781450370431},
publisher = {Association for Computing Machinery},
address = {New York, NY, USA},
url = {https://doi.org/10.1145/3368089.3409734},
doi = {10.1145/3368089.3409734},
booktitle = {Proceedings of the 28th ACM Joint Meeting on European Software Engineering Conference and Symposium on the Foundations of Software Engineering},
pages = {123–134},
numpages = {12},
location = {Virtual Event, USA},
series = {ESEC/FSE 2020}
}

@misc{svcomp2026,
  title        = {{SV-COMP} 2026: 15th International Competition on Software Verification},
  howpublished = {\url{https://sv-comp.sosy-lab.org/2026/}},
  year         = {2026},
  note         = {Held at {TACAS} 2026, Torino, Italy, April 13, 2026}
}

@misc{svcomp2026rules,
  title        = {{SV-COMP} 2026: Definitions and Rules},
  howpublished = {\url{https://sv-comp.sosy-lab.org/2026/rules.php}},
  year         = {2026},
  note         = {15th International Competition on Software Verification, held at {TACAS} 2026, Torino, Italy}
}

@misc{techreport,
      title={Technical Report: A Formal Semantics for Java Symbolic Evaluation using Large-Block Encoding}, 
      author={Soha Hussein and Stephen McCamant and Kelton OBrien and Kuen-Bang Hou and Michael Whalen and Vaibhav Sharma},
      year={2026},
      eprint={2608.04513},
      archivePrefix={arXiv},
      primaryClass={cs.SC},
      url={https://arxiv.org/abs/2608.04513}, 
}

@inproceedings{ramos,
 author = {Ramos, David A and Engler, Dawson R.},
 title = {Practical, Low-effort Equivalence Verification of Real Code},
 booktitle = {Proceedings of the 23rd International Conference on Computer Aided Verification},
 series = {CAV'11},
 year = {2011},
 isbn = {978-3-642-22109-5},
 location = {Snowbird, UT},
 pages = {669--685},
 url = {http://dl.acm.org/citation.cfm?id=2032305.2032360},
 acmid = {2032360},
 publisher = {Springer-Verlag},
 address = {Berlin, Heidelberg},
}

@inproceedings{contractdr,
author = {Hussein, Soha and Rayadurgam, Sanjai and McCamant, Stephen and Sharma, Vaibhav and Heimdahl, Mats},
title = {Counterexample-Guided Inductive Repair of Reactive Contracts},
year = {2022},
isbn = {9781450392877},
publisher = {Association for Computing Machinery},
address = {New York, NY, USA},
url = {https://doi.org/10.1145/3524482.3527650},
doi = {10.1145/3524482.3527650},
booktitle = {Proceedings of the IEEE/ACM 10th International Conference on Formal Methods in Software Engineering},
pages = {46–57},
numpages = {12},
location = {Pittsburgh, Pennsylvania},
series = {FormaliSE '22}
}

@inproceedings{driller,
  title={{Driller: Augmenting Fuzzing Through Selective Symbolic Execution}},
  author={Nick Stephens and John Grosen and Christopher Salls and Audrey Dutcher and Ruoyu Wang and Jacopo Corbetta and Yan Shoshitaishvili and Christopher Kruegel and Giovanni Vigna},
  booktitle = {23rd Annual Network and Distributed System Security Symposium, {NDSS}
               2016, San Diego, California, USA, February 21-24, 2016},
  month={February},
  address={San Diego, CA},
  year={2016},
  publisher={The Internet Society},
  pages={1--16}
}

@inproceedings{angr,
  author    = {Yan Shoshitaishvili and Ruoyu Wang and Christopher Salls and
               Nick Stephens and Mario Polino and Andrew Dutcher and
               John Grosen and Siji Feng and Christophe Hauser and
               Christopher Kr{\"{u}}gel and Giovanni Vigna},
  title     = {{SOK:} (State of) The Art of War: Offensive Techniques in Binary Analysis},
  booktitle = {{IEEE} Symposium on Security and Privacy, {SP} 2016, San Jose, CA,
               USA, May 22-26, 2016},
  pages     = {138--157},
  publisher = {{IEEE} Computer Society},
  address   = {Los Alamitos, CA, USA},
  year      = {2016},
  url       = {https://doi.org/10.1109/SP.2016.17},
  doi       = {10.1109/SP.2016.17},
  bibsource = {dblp computer science bibliography, https://dblp.org}
}

@inproceedings{transport,
 author = {Sun, Wei and Xu, Lisong and Elbaum, Sebastian},
 title = {Improving the Cost-effectiveness of Symbolic Testing Techniques for Transport Protocol Implementations Under Packet Dynamics},
 booktitle = {Proceedings of the 26th ACM SIGSOFT International Symposium on Software Testing and Analysis},
 series = {ISSTA 2017},
 year = {2017},
 isbn = {978-1-4503-5076-1},
 location = {Santa Barbara, CA, USA},
 pages = {79--89},
 url = {http://doi.acm.org/10.1145/3092703.3092706},
 doi = {10.1145/3092703.3092706},
 acmid = {3092706},
 publisher = {ACM},
 address = {New York, NY, USA},
}

@online{bug,
  author       = {Java Ranger},
  title        = {Java Ranger Issue \#49},
  year         = {2024},
  url          = {https://github.com/vaibhavbsharma/java-ranger/issues/49},
  note         = {GitHub Issue}
}

@INPROCEEDINGS{testjr,
  author={Hussein, Soha and McCamant, Stephen and Sherman, Elena and Sharma, Vaibhav and Whalen, Mike},
  booktitle={2023 IEEE/ACM International Conference on Automation of Software Test (AST)}, 
  title={Structural Test Input Generation for 3-Address Code Coverage Using Path-Merged Symbolic Execution}, 
  year={2023},
  volume={},
  number={},
  pages={79-89},
  doi={10.1109/AST58925.2023.00012}}

@INPROCEEDINGS{jrfuzz,
  author={Hussein, Soha and McCamant, Stephen},
  booktitle={2025 40th IEEE/ACM International Conference on Automated Software Engineering Workshops (ASEW)}, 
  title={Improving Automated Program Verification for Java Programs with Fuzzing}, 
  year={2025},
  volume={},
  number={},
  pages={153-160},
  doi={10.1109/ASEW67777.2025.00037}}

\end{document}